\documentclass[titlepage, a4paper, oneside, doublespace]{book}
\usepackage{amsmath,amsthm,latexsym,amssymb,amsfonts,epsfig,dsfont,MnSymbol,proof}

\usepackage[all,arc,dvips,arrow,curve,cmactex]{xy}
\usepackage{amsfonts,amsmath,amssymb,graphics,psfrag}

\newcommand\Sets{{\bf Sets}}
\newcommand\op{{\rm op}}
\newcommand\ps[1]{\underline{#1}}
\newcommand\down[1]{\downarrow\!{#1}}
\def\d{{\cal D}}
\def\c{{\cal C}}

\def\md{\mathcal{D}}

\newcommand\Sub{{\rm Sub}}
\newcommand\Om{\underline{\Omega}}
\newcommand\w{\mathfrak{w}}
\newcommand\ket[1]{\,|#1\rangle}
 \newcommand\od{\underline{\mathbb{O}}}
\newcommand\ui{\underline{\mathbb{I}}}

\makeatletter \@addtoreset{equation}{section} \makeatother

\newtheorem{Theorem}{Theorem}[section]
\newtheorem{Definition}{Definition}[section]
\newtheorem{Lemma}{Lemma}[section]
\newtheorem{Corollary}{Corollary}[section]

\newtheorem{Conjecture}{Conjecture}[section]
\newtheorem{Diagram}{Diagram}[section]
\newtheorem{Proof}{Proof}[section]
\newtheorem{example}{Example}[section]
\newtheorem{Axiom}{Axiom}[section]
\newcommand\Sig{\underline{\Sigma}}
\newcommand\V{\mathcal{V}}

\newcommand\Hi{{\cal H}}
\def\be{\begin{equation}}
\def\ee{\end{equation}}
\def\ba{\begin{eqnarray}}
\def\ea{\end{eqnarray}}
\def\mc{\mathcal{C}}
\def\mv{\mathcal{V}}
\def\mh{\mathcal{H}}
\def\mb{\mathcal{B}}
\def\us{\underline{\Sigma}}

\def\uo{\underline{\mathcal{O}}}

\def\ug{{\underline{G}}}
\def\uom{\underline{\Omega}}
\def\Nl{{\mathchoice
{\setbox0=\hbox{$\displaystyle\rm N$}\hbox{\hbox to0pt
{\kern0.4\wd0\vrule height0.9\ht0\hss}\box0}}
{\setbox0=\hbox{$\textstyle\rm N$}\hbox{\hbox to0pt
{\kern0.4\wd0\vrule height0.9\ht0\hss}\box0}}
{\setbox0=\hbox{$\scriptstyle\rm N$}\hbox{\hbox to0pt
{\kern0.4\wd0\vrule height0.9\ht0\hss}\box0}}
{\setbox0=\hbox{$\scriptscriptstyle\rm N$}\hbox{\hbox to0pt
{\kern0.4\wd0\vrule height0.9\ht0\hss}\box0}}}}
\def\Zl{{\mathchoice
{\setbox0=\hbox{$\displaystyle\rm Z$}\hbox{\hbox to0pt
{\kern0.4\wd0\vrule height0.9\ht0\hss}\box0}}
{\setbox0=\hbox{$\textstyle\rm Z$}\hbox{\hbox to0pt
{\kern0.4\wd0\vrule height0.9\ht0\hss}\box0}}
{\setbox0=\hbox{$\scriptstyle\rm Z$}\hbox{\hbox to0pt
{\kern0.4\wd0\vrule height0.9\ht0\hss}\box0}}
{\setbox0=\hbox{$\scriptscriptstyle\rm Z$}\hbox{\hbox to0pt
{\kern0.4\wd0\vrule height0.9\ht0\hss}\box0}}}}
\def\Ql{{\mathchoice
{\setbox0=\hbox{$\displaystyle\rm Q$}\hbox{\hbox to0pt
{\kern0.4\wd0\vrule height0.9\ht0\hss}\box0}}
{\setbox0=\hbox{$\textstyle\rm Q$}\hbox{\hbox to0pt
{\kern0.4\wd0\vrule height0.9\ht0\hss}\box0}}
{\setbox0=\hbox{$\scriptstyle\rm Q$}\hbox{\hbox to0pt
{\kern0.4\wd0\vrule height0.9\ht0\hss}\box0}}
{\setbox0=\hbox{$\scriptscriptstyle\rm Q$}\hbox{\hbox to0pt
{\kern0.4\wd0\vrule height0.9\ht0\hss}\box0}}}}
\def\Rl{{\mathchoice
{\setbox0=\hbox{$\displaystyle\rm R$}\hbox{\hbox to0pt
{\kern0.4\wd0\vrule height0.9\ht0\hss}\box0}}
{\setbox0=\hbox{$\textstyle\rm R$}\hbox{\hbox to0pt
{\kern0.4\wd0\vrule height0.9\ht0\hss}\box0}}
{\setbox0=\hbox{$\scriptstyle\rm R$}\hbox{\hbox to0pt
{\kern0.4\wd0\vrule height0.9\ht0\hss}\box0}}
{\setbox0=\hbox{$\scriptscriptstyle\rm R$}\hbox{\hbox to0pt
{\kern0.4\wd0\vrule height0.9\ht0\hss}\box0}}}}
\def\Cl{{\mathchoice
{\setbox0=\hbox{$\displaystyle\rm C$}\hbox{\hbox to0pt
{\kern0.4\wd0\vrule height0.9\ht0\hss}\box0}}
{\setbox0=\hbox{$\textstyle\rm C$}\hbox{\hbox to0pt
{\kern0.4\wd0\vrule height0.9\ht0\hss}\box0}}
{\setbox0=\hbox{$\scriptstyle\rm C$}\hbox{\hbox to0pt
{\kern0.4\wd0\vrule height0.9\ht0\hss}\box0}}
{\setbox0=\hbox{$\scriptscriptstyle\rm C$}\hbox{\hbox to0pt
{\kern0.4\wd0\vrule height0.9\ht0\hss}\box0}}}}
\def\Hl{{\mathchoice
{\setbox0=\hbox{$\displaystyle\rm H$}\hbox{\hbox to0pt
{\kern0.4\wd0\vrule height0.9\ht0\hss}\box0}}
{\setbox0=\hbox{$\textstyle\rm H$}\hbox{\hbox to0pt
{\kern0.4\wd0\vrule height0.9\ht0\hss}\box0}}
{\setbox0=\hbox{$\scriptstyle\rm H$}\hbox{\hbox to0pt
{\kern0.4\wd0\vrule height0.9\ht0\hss}\box0}}
{\setbox0=\hbox{$\scriptscriptstyle\rm H$}\hbox{\hbox to0pt
{\kern0.4\wd0\vrule height0.9\ht0\hss}\box0}}}}
\def\Ol{{\mathchoice
{\setbox0=\hbox{$\displaystyle\rm O$}\hbox{\hbox to0pt
{\kern0.4\wd0\vrule height0.9\ht0\hss}\box0}}
{\setbox0=\hbox{$\textstyle\rm O$}\hbox{\hbox to0pt
{\kern0.4\wd0\vrule height0.9\ht0\hss}\box0}}
{\setbox0=\hbox{$\scriptstyle\rm O$}\hbox{\hbox to0pt
{\kern0.4\wd0\vrule height0.9\ht0\hss}\box0}}
{\setbox0=\hbox{$\scriptscriptstyle\rm O$}\hbox{\hbox to0pt
{\kern0.4\wd0\vrule height0.9\ht0\hss}\box0}}}}

\begin{document}
\begin{titlepage}
\begin{center}
{\Huge{\bf Lectures on Topos Quantum Theory}} \vspace{0.5cm}
 {\Large\\\vspace{.5cm}
    Dr. Cecilia Flori\\\vspace{.5cm}
Perimeter Institute for Theoretical Physics}\\ \vspace{1cm}  
\Large{Graduate Course 2012\\\vspace{0.3cm}
For The University of Waterloo} \\\vspace{0.5cm}
\end{center}
\tableofcontents

\end{titlepage}



\chapter{Introduction}
\label{che:topostheory}
``We can't solve problems by using the same kind of thinking we used when we created them."\\
(Einstein)\\

The great revolution of the nineteenth century started with the theory of special and general relativity and culminated in quantum theory. However, up to date, there are still some fundamental issues with quantum theory that are yet to be solved. Nonetheless a great deal of effort in fundamental physics is spent on an elusive theory of quantum gravity which is an attempt to combine the two above mentioned theories which seem, as they have been formulated, to be incompatible. In the last five decades, various attempt to formulate such a theory of quantum gravity have been made, but none have fully succeeded in becoming \emph{the} quantum theory of gravity. One possibility of the failure for reaching an  agreement on a theory of quantum gravity might be presence of unresolved fundamental issues already present in quantum theory. Most approaches to quantum gravity adopt standard quantum theory as there starting point, with the hope that the unresolved issues of the theory will get solved along the way. However, it might be the case that these fundamental issues should be solved \emph{before} attempting to define a quantum theory of gravity. 

If one adopts this point of view, the questions that come next are: i) which are the main conceptual issues in quantum theory ii) How can these issues be solved within a new theoretical frame work of quantum theory.\\
Chris Isham, Andreas D\"oring, Jeremy Butterfield and others have proposed that the main issues in the standard quantum
formalism are: (A) the use of critical mathematical ingredients which seem to assume certain properties of space and/or
time which are not entirely justified. In particular it could be the case that such a priori assumptions of space and time are not compatible with a theory of quantum gravity. (B) The instrumental interpretation of quantum theory that denies the possibility of talking about systems without reference to an external observer. A consequence of this issue is the problematic notion of a closed system in quantum cosmology. \\

A possible way to overcome the above mentioned issues is through a reformulation of quantum theory in terms of a different mathematical framework called topos theory. The reason for choosing topos theory is that it `looks like' sets and is equipped with an internal logic. As we will explain in detail in the following chapters, both these features are desirable, because they will allow for a reformulation of quantum theory which is more realist (thus solving issue (B)) and which does not rest on a priori assumptions about the nature of space and time.\\ 
The hope is that such a new formulation of quantum theory will shed some light on how a quantum theory of gravity should look like.
\\

The main idea in the topos formulation of normal
quantum theory \cite{andreas1,andreas2,andreas3,andreas4,andreas5,isham1,isham2,isham3,isham4,isham5}
is that using topos theory to redefine the mathematical
structure of quantum theory leads to a reformulation of quantum
theory in such a way that it is made to `look like' classical
physics. Furthermore, this reformulation of quantum theory has the
key advantages that (i)  no fundamental role is played by the
continuum; and (ii) propositions can be given truth values without
needing to invoke the concepts of `measurement' or `observer`.
Let us 
analyse the reasons why such a reformulation is needed  in the
first place. These concern  quantum theory general and quantum
cosmology in particular.
\begin{itemize}
\item As it stands quantum theory is non-realist. From a mathematical perspective this is reflected in the Kocken-Specher theorem \footnote{
\textbf{Kochen-Specker Theorem}: if the dimension of $\Hi$ is
greater than 2, then there does not exist any valuation function
$V_{\vec{\Psi}}:\mathcal{O}\rightarrow\Rl$ from the set
$\mathcal{O}$ of all bounded self-adjoint operators $\hat{A}$ of
$\Hi$ to the reals $\Rl$ such that  for all
$\hat{A}\in\mathcal{O}$ and all $f:\Rl\rightarrow\Rl$, the following holds $V_{\vec{\Psi}}(f(\hat{A}))=f(V_{\vec{\Psi}}(\hat{A}))$.}. This theorem implies that any statement
regarding state of affairs, formulated within the theory, acquires
meaning contractually, i.e., after measurement. This implies that
it is hard to avoid the Copenhagen interpretation of quantum
theory, which is intrinsically non-realist.
\item Notions of `measurement' and `external observer' pose problems when dealing with cosmology. In fact, in this case there can be no external observer since we are dealing with a closed system. But this then implies that the concept of `measurement' plays no fundamental role, which in turn implies that the standard definition of probabilities in terms of relative frequency of measurements breaks down.
\item The existence of the Planck scale suggests that there is no \emph{a priori} justification for the adoption of the notion of a continuum in the quantum theory used in formulating quantum gravity.
\end{itemize}

These considerations led Isham and D\"oring to search for a
reformulation of quantum theory that is more realist\footnote{By a
`realist' theory we mean one in which the following conditions are
satisfied: (i) propositions form a Boolean algebra; and (ii)
propositions can always be assessed to be either true or false. As
will be delineated in the following, in the topos approach to
quantum theory both of these conditions are relaxed, leading to
what Isham and D\"oring called  a \emph{neo-realist} theory.} than
the existing one. It turns out that this can be achieved through
the adoption of topos theory as the mathematical framework with
which to reformulate Quantum theory.

One approach to reformulating quantum theory in a more realist way
is to re-express it in such a way that it `looks like' classical
physics, which is the paradigmatic example of a realist theory.
This is precisely the strategy adopted by the authors in
\cite{andreas1}, \cite{andreas2}, \cite{andreas3}, \cite{andreas4}
and \cite{andreas5}. Thus the first question is what is the
underlining structure which makes classical physics a realist
theory?

The authors identified this structure with the following elements:
\begin{enumerate}
\item The existence of a state space $S$.

\item  Physical quantities are represented by functions from the state space to the reals. Thus each physical quantity, $A$, is represented by a function
\begin{equation}
f_A:S\rightarrow \Rl
\end{equation}
\item Any propositions of the form ``$A\in \Delta$'' (``The value of the quantity A lies in the subset $\Delta\in\Rl$'') is represented by a subset of the state space $S$: namely, that subspace for which the proposition is true. This is just
\begin{equation}
f_A^{-1}(\Delta)=\{s\in S| f_A(s)\in\Delta\}
\end{equation}
The collection of all such subsets forms  a Boolean algebra,
denoted ${\rm Sub}(S)$.

\item States $\psi$ are identified with Boolean-algebra homomorphisms
\begin{equation}\label{equ:state}
\psi:{\rm Sub}(S)\rightarrow \{0,1\}
\end{equation}
from the Boolean algebra ${\rm Sub}(S)$ to the two-element
$\{0,1\}$. Here, $0$ and $1$ can be identified as `false' and
`true' respectively.

The identification of states with such maps follows from
identifying propositions with subsets of $S$. Indeed, to each
subset $f_A^{-1}(\{\Delta\})$, there is associated a
characteristic function
$\chi_{A\in\Delta}:S\rightarrow\{0,1\}\subset\Rl$ defined by
\begin{equation}
\chi_{A\in\Delta}(s)=\begin{cases}1& {\rm if}\hspace{.1in}f_A(s)\in\Delta;\\
0& \text{otherwise}  \label{eq:mam}.
\end{cases}
\end{equation}
Thus each state $s$ either lies in $f_A^{-1}(\{\Delta\})$ or it
does not. Equivalently, given a state $s$ every proposition about
the values of physical quantities in that state is either true or
false. Thus \ref{equ:state} follows
\end{enumerate}

The first issue in finding quantum analogues of 1,2,3, and 4 is to
consider the appropriate mathematical framework in which to
reformulate the theory. As previously mentioned the choice fell on
topos theory. There were many reasons for this, but a paramount
one is that in any topos (which is a special type of category)
distributive logic arise in a natural way: i.e., a topos has an
internal logical structure that is similar in many ways to the way
in which Boolean algebras arise in set theory. This feature is
highly desirable since requirement 3 implies that the subobjects
of our state space (yet to be defined) should form some sort of
logical algebra.

The second issue is to identify which  topos is the right one to
use. Isham et al achieved this by noticing that the possibility of
obtaining a `neo-realist' reformulation of quantum theory lied in
the idea of a \emph{context}. Specifically, because of the
Kocken-Specher theorem, the only way of obtaining quantum
analogues of requirements 1,2,3 and 4 is by defining them with
respect to commutative subalgebras (the `contexts') of the
non-commuting algebra, $\mathcal{B(H)}$, of all bounded operators
on the quantum theory's Hilbert space. Thus `locally' with respect to these contexts quantum theory 
effectively behaves classically. So the idea is to try and define each quantum object locally in terms of
these abelian contexts. The key feature, however, is that the collection of all this {\it contexts} or {\it classical snapshots} form a category 
ordered by inclusion. This implies that although one defines each quantum object locally, the global information is never lost since it is put back into the picture by the categorical structure of the collection of all these classical snapshots. \\
Hence the task is to find a topos which allows you to define a quantum object as (roughly speaking) a collection of classical approximations. As we will see this can be done thought the topos of presheaves over the category of abelian algebras. \\
In terms of this topos of presheaves, quantum theory can be re-defined so that it retains some realism and its interpretation is not riddled with the above mentioned conceptual problems.

\chapter{Lecture 1}
The first lecture will deal with what the main interpretational problems of canonical quantum theory are. 

In particular, we will analyse how the mathematical formalism of quantum theory leads to a non realist interpretation of the theory. The focus will lie in understanding and analysing the Kochen-Specker theorem (K-S theorem), which can be thought of as the main mathematical underlying reason why quantum theory is non realist. The interpretation which comes out is the well known Copenhagen interpretation of quantum theory, which is an instrumentalist interpretation. However, such an interpretation leads to many conceptual problems. 

Topos quantum theory is a way of overcoming such problems by re-define quantum theory in the novel language of topos theory. The advantage of this language is that it renders the theory more realist, thus solving the above mentioned problems. In the process, however, one ends up with a multivalued/intuitionistic logic rather than a Boolean logic.

\section{Conceptual Problems in Quantum Theory}
The first natural question to ask is: why do we need a topos reformulation of quantum theory? The short answer is that such a reformulation is needed since it represents a candidate for solving certain conceptual issues present in quantum theory, which derive from the mathematical formulation of the theory. In particular, the canonical mathematical formulation of quantum theory leads to an interpretation which has many conceptual obstacles to a fully coherent theory.

In order to understand the above mentioned issues present in quantum theory one needs to first analyse in depths i) what a theory of physics really is and ii) what is it trying to achieve.
Obviously we will not be able to answer these questions fully or give them the attention they deserve since this would go beyond the scope of theis course. However, we will try and give a brief overview of the situation.

\subsection{What is a Theory of Physics and What is it Trying to Achieve?}
A theory of physics can be seen as a mathematical model which tries to answer three of the fundamental questions humanity has been and still is struggling to answer: 
\begin{enumerate}
\item what is a thing. (Heidegger)
\item How are ``things" related to one another.
\item How do we know 1) and 2).

\end{enumerate}
The first two questions are related to ontological\footnote{Ontology comes from the Greek word meaning ``being, that which is" and indicates the study of what things are in themselves and what can be said to exist.} issues while the third is of an epistemological\footnote{Epistemology comes from the Greek word meaning ``knowledge" and it is the study concerning what is knowledge and how do we gain knowledge.} nature.

The two main theories which presuppose to answer the above questions are
\begin{itemize}
\item Classical physics.
\item Quantum theory.
\end{itemize}

The way in which these two theories have answered the above questions is by defining a mathematical model which is supposed to describe nature. The interpretation of this mathematical model then, in turn, gives rise to a philosophical view of the world. In classical physics the mathematical model developed is in accordance with our common believes of the world, which provides the desired answer. In fact it is arguable whether our common believes modelled classical theory, but we will not delve into this. On the other hand things are not so straightforward in quantum theory in which, as we will see, the mathematical formalism of the theory seems, at time, to defy our commons sense.

In any case, in order to fully understand how a philosophical picture of the world can be derived though the mathematical formulation of a theory of physics we need to refine the questions 1-3 defined above (see \cite{topos1} for an in depth discussion). 
In particular, any theory of physics worthy that name, should address the following issues:
\begin{enumerate}
\item What is the system under investigation.
\item Ontological status of physical terms.
\item Epistemological status of physical terms.
\item Relation between the mathematical model and the physical world.
\item How physical statements can be verified or falsified.
\item Nature of space-time.
\item Meaning of probabilities (if they arise in the theory).
\end{enumerate}

As we will see, the different answers given to the above issues by classical theory and quantum theory, respectively, will hilight the radical differences between the two theories and the different interpretations each of them gives to the ``outside world".

\begin{center}
\fbox{Mathematical tools used to describe a physical system encode philosophical position regarding the world.}
\end{center}

In the following subsection we will briefly analyse how questions 1-7 are dealt with in classical physics. However we will not address all of them, since that would require more than a course (probably several) to do so, but we will only focus on certain crucial aspects which are essential in understanding the philosophical position of classical theory

\subsection{Philosophical Position of Classical Theory.}

When studying classical physics we develops an image of the world which is in accordance to our common sense. In fact classical theory is such that 
\begin{enumerate}
\item [i)] properties can be ascribed to a system at any given time and do not depend on the act of measuring; 
\item [ii)] the underlying logic is Boolean (classical) logic\footnote{Boolean logic will be described later on in the course. For now we will simply say that Boolean logic is the logic we use in our every day thinking and in our language. Such a logic is characterised by the fact that i) it is distributive, ii)  it only has two truth values $\{ true, false\}$ and iii) the logical connectives are our linguistic logical connectives: ``and ", ``exclusive or", ``not", ``if then". } which is the same logic we employ in our language. 
\end{enumerate}
Generally speaking a theory with the above mentioned properties is called a \emph{realist theory}. We will analyse later on, in more details, what exactly are the underlying assumptions which make classical physics a realist theory but, for now, it suffices to say that the realism of classical physics is associated  to the way in which objects and knowledge of such objects are mathematically expressed.

The realism of classical theory implies that a \emph{thing} is defined\footnote{It is worth noting that our own language reflects a realist view of the world: `` The tree \emph{is} three meters tall".} in terms of a bundle of properties which are said to belong to the \emph{thing} (system). The type of properties that we are dealing with are of two kinds:
\begin{enumerate}
\item \emph{Internal properties} which belong exclusively to the system, for example the mass, charge etc.
\item \emph{External properties} which define relations to other systems, for example position, velocity etc.
\end{enumerate}

When defining what a thing is one usually considers internal properties.

In classical physics epistemological questions are answered through the process of measurement. In fact \emph{measurement} enables us to know the values of a given system. However, although in the actual process of measurement there is a momentary distinction between object and subject, such a  distinction has a purely functional role, not an actual distinction. In fact both object and subject, as viewed from a classical perspective, exist out there independently of one another. 
This, in turn, implies that no special role is ascribed to measurement, i.e. in classical physics measurement is just another form of interaction.

Finally it is worth mentioning that, generally, classical theory is thought of as being a deterministic theory\footnote{In a stochastic approach the realist conditions i) and ii) at the beginning of the section still hold.}, i.e. given initial state of a system at a given time, it is possible to predict with certainty the state of the system at a subsequent time.

From the above discussion it emerges that the mathematical structure of classical theory implies an interpretation of the theory which is realist. In fact the mathematical model of classical theory induces a conceptual descriptions of various elements in the theory which, in turn, imply a realist philosophy of the outside world. 

So the natural question to ask at this point is: what are the mathematical constructs whose definition ( and in particular the way in which they are defined) directly imply a realist interpretation of classical theory?

\noindent
The answer to this question will be given in detail in subsequent lessons, but for now we will restrict ourselves in answering it in a very conceptual way, so as to give a general idea of the relation between mathematical constructs and induced philosophical ideas.

The elements/concepts whose mathematical description render classical theory a realist theory are the following:
\begin{enumerate}
\item \emph{State space}. In classical physics the state space $S$ is defined to be the collection of all states $s_i\in S$ of the system, such that each $s_i$ at a given time $t_i$ embodies all the properties of the system at that time.
\item Definition of physical quantities (see lecture 6)
\item \emph{Definition of propositions}. See lecture 6
\item \emph{Boolean logic}. The logic governing classical propositions is Boolean logic which is a distributive logic, which admits only two truth values: $\{ true,\; false\}$. Verification of such truth values is done through the measurement interaction. 
\item \emph{Probabilities}. Classical probabilities are defined as follows:
\be
\frac{\text{positive possible outcomes}}{\text{all admissible outcomes}}\nonumber
\ee
\end{enumerate}

Later in the course we will describe, in details, how the above classical concepts are mathematically represented but, for now, it suffices to say that it is precisely the way in which the above elements of the theory are  mathematically expressed which renders classical theory a realist theory. In fact, when considering quantum theory we will see how the same elements are mathematically  described in a very different way. This will induce a different conceptual understanding of such elements which, in turn, will imply a different philosophical interpretation of the theory. 

\subsection{Philosophy Behind Quantum Theory}
What can be said about the philosophical position of quantum theory?  If we analyse the mathematical formalism of quantum theory we immediately realise that the theory is non-realist (with the definition of realist given above). In fact the above conditions [i), ii)] do not strictly hold\footnote{We will clarify this later on in the lecture.} in quantum theory, since the formalism of the latter implies a clear distinction between measuring apparatus and measuring system, such that the act of measuring gets ascribed a special status. In this setting, measurement becomes a means for assigning a probabilistic spread of outcomes rather than a means to determine properties of the system. In fact the very concept of properties ceases to have its common sense meaning since its definition is now intertwined with the act of measurement. It is as if properties acquire the status of latent attributes which are brought into existence by the act of measurement, but which can not be said to exist independently of such measurement. Therefore it becomes meaningless to talk about a physical system as possessing properties. The interpretation that results is the so called \emph{instrumentalist interpretation} of quantum theory which is a non-realist interpretation.

So the feature of quantum theory which render it non-realist can be summarised as follows:
\begin{enumerate}
\item Properties can not be said to be possessed by a system a priori . All that can be said is that after a measurement is performed the system ``acquires" the ``latent" properties (state-vector reductio). 
\item Any statement regarding `states of affairs' about a system can only be made a posteriori after measurement. However such statements can not be regarded as describing properties of the system, on the contrary, it describes probabilities of possible measurement outcomes.
\item Measurement becomes a very special type of interaction. 
\item Clear distinction between observe and observed system.
\item States are not seen as bearers of physical properties but are simply the most efficient tools to enable one to determine/compute predictions for possible measurements, i.e. predictions of probabilities of outcomes not outcomes themselves.
\item Quantum theory is deterministic \emph{but} what evolves are now predicted probabilities of measurement results, \emph{not} actual measurements.
\item Relative frequency interpretation of probabilities.
\end{enumerate}

The above features of quantum theory which directly derive from the mathematical representation of the theory imply a non realist interpretation of quantum theory.
Such an interpretation, although works for some situations, causes various conceptual problems in the context of quantum gravity and quantum cosmology. These difficulties are the following:
\begin{itemize}
\item Notions of `measurement' and `external observer' pose problems when dealing with cosmology. In fact, in this case there can be no external observer since we are dealing with a closed system. But this then implies that the concept of `measurement' plays no fundamental role which, in turn, implies that the standard definition of probabilities, in terms of relative frequency of measurements, breaks down.
\item The existence of the Planck scale suggests that there is no \emph{a priori} justification for the adoption of the notion of a continuum in the quantum theory used in formulating quantum gravity.
\item Standard quantum theory employs in its formulation the use of a fixed spatio-temporal structure (fixed background). This is needed to make sense of its instrumentalist interpretation, i.e. it needs a space-time in which to make a measurement. This fixed background seems to cause problems in quantum gravity where one is trying to make measurements of space-time properties. In fact, if the action of making a measurement requires a space time background, what does it mean to measure space time properties? 
\item Given the concept of superposition present in quantum theory, by applying such concept to quantum gravity we would have to account for the occurrence of quantum superpositions of eignestate properties of space, time and space-time.
\end{itemize}

In the following section we will analyse one of the main theorems (another one would be Bell's inequality) which states the impossibility of quantum theory, as it is canonically expressed, to be a realist theory.

\section{Kochen Specker Theorem}

The Kochen-Specker theorem derives from the incompatibility of two assumptions regarding 
observables in quantum theory, namely \cite{topos1} \cite{topos4}.
\begin{enumerate}
\item The need of assigning simultaneous values to all observables in $\mathcal{O}$ (collection of all self-adjoint operators on $\mh$).
\item The need for the values of observables to be ``mutually exclusive and collectively 
exhaustable"\footnote{Mutually \emph{exclusive} means that only one value of an observable can be realised at a given time, while collectively \emph{exhaustible} means that at least one of the values has to be realised at a given time.} \cite{topos4}.
\end{enumerate}
It follows that the Kochen-Specker theorem is related to the existence, in
quantum theory, of a value function (to be defined) $V:\mathcal{O}\rightarrow\Rl$ from the 
set of self-adjoint 
operators (which are the quantum analogues of physical quantities) to the Reals. 
\subsection{Valuation Function}
To understand what a valuation function is let us first analyse how it is defined in classical theory.
\subsubsection{Valuation Function in Classical Theory}
Before giving the definition of what a valuation function is in classical physics we first of all have to define how a physical quantity is mathematically described in classical physics. 

Namely, in classical physics physical quantities are represented by functions from the state space to the reals. Thus, each physical quantity, $A$, is represented by a function $f_A:S\rightarrow \Rl$, such that for each state $s_i\in S$, $f_A(s_i)\in \Rl$ represents the value of $A$ given the state $s_i$. This association of physical quantities with real valued functions on the state space is 1:1 ( one-2-one: for each quantity $A$ there corresponds one and only one function $f_A$).

Given the definition of physical quantity in terms of maps on the state space, the definition of valuation function in classical physics is straightforward. In particular, a valuation function is defined, for each state $s_i $ in the state space $S$, as a map
\be
V_{s_i}:\mathcal{O}\rightarrow \Rl
\ee 
from the set of observables (physical quantities) $\mathcal{O}$ to the reals, such that for each $A\in \mathcal{O}$ we obtain:
\be
A\mapsto V_{s_i(A)}:=f_A(s_i)
\ee
where $V_{s_i(A)}$ represents the value of the physical quantity $A$ given the state $s_i$.

A condition such a valuation function has to satisfy is the so called \emph{functional composition condition} (FUNC) which is defined as follows:
\be\label{equ:cfunc}
\text{ for any }h:\Rl\rightarrow \Rl, \;\;V_{s_i}(h(A))=h(V_{s_i}(A))
\ee
In this equation $h(A)\in \mathcal{O}$ and is defined in terms of composition of functions:
\be
h(A):=h\circ f_A:S\xrightarrow{f_A}\Rl\xrightarrow{h}\Rl
\ee
If $A$ represents the physical quantity \emph{energy}, and $h$ is a function which defines the square, i.e. $h(A)=A^2$, what equation \ref{equ:cfunc} would mean is: ``the value of the energy squared is equal to the square of the value of the energy".

\subsubsection{Valuation Function in Quantum Theory}
If we were to mimic classical theory then we would define a valuation function as follows: \\
for each state $\psi\in H$ the valuation function is a function from the set of self-adjoint operators (quantum analogues of physical quantities) to the reals
\be
V_{|\psi\rangle}:\mathcal{O}\rightarrow \Rl
\ee

Such that, for each state $|\psi\rangle$, $V_{|\psi\rangle}$ assigns to each self-adjoint 
operator $\hat{A}\in\mathcal{O}$, a real number $V_{|\psi\rangle}(\hat{A})\in \Rl$ that represents the value of $\hat{A}$
for the state $|\psi\rangle$ of the system. 

However this definition of a valuation function only makes sense if $|\psi\rangle$ is an eigenvector of $\hat{A}$. Other than that special case, the above definition of a valuation function does not really make sense. So the question is how to generalise it for an arbitrary state $|\psi\rangle$? A possible generalisation is the following:
\begin{Definition}
A valuation function for quantum theory is a map $V:\mathcal{O}\rightarrow \Rl$ which satisfies the following two conditions:
\begin{enumerate}
\item [i)] For each $\hat{A}\in\mathcal{O}$, $V(\hat{A})\in \Rl$ represents the value of the operator $\hat{A}$ and it belongs to the spectrum of $\hat{A}$.
\item [ii)] FUNC. For all $h:\Rl \rightarrow \Rl$ the following holds
\be
V(h(\hat{A}))=h(V(\hat{A}))
\ee
\end{enumerate}

\end{Definition}
Any function satisfying the above conditions is a valuation function.

At this point it is worth understanding, explicitly, what $h(\hat{A})$ is.\\\\
\textbf{What is $\mathbf{h(\hat{A})}$?}\\\\
Given a self-adjoint operator $\hat{A}$ we have two situations:
\begin{enumerate}
\item [i)] Let $|\psi\rangle\in \mh$ be an eigenvector of $\hat{A}$, i.e. $\hat{A}|\psi\rangle=a|\psi\rangle$. It is then straightforward to define the following:
\be
\hat{A}^2|\psi\rangle=a^2|\psi\rangle;\;\;\hat{A}^3|\psi\rangle=a^3|\psi\rangle
\ee
Thus, generalising for any polynomial function $Q$ we obtain
\be
Q(\hat{A})|\psi\rangle=Q(a)|\psi\rangle
\ee

Provided $Q(a)$ is well defined. Given the above we are justified in defining, for any function $h:\Rl\rightarrow\Rl$ the following:
\be
h(\hat{A})|\psi\rangle=h(a)|\psi\rangle
\ee
Again, provided $h(a)$ is well defined, (ex. not infinite).

\item [ii)] We now would like to generalise it to arbitrary states, not just eigenvectors. To this end we recall that the set of eigenvectors of a self-adjoint operator forms an orthonormal basis for $\mh$. This means that any state $|\psi\rangle\in \mh$ can be written in terms of such an orthonormal basis. 
Thus, considering the case in which $\hat{A}$ has a discrete spectrum (all that follows can be easily generalised for the continuum case\footnote{In lecture 6/7 we will explain, in detail, what e spectral decomposition is, but for now we will simply state that each self-adjoint operator $\hat{A}$ can be written as
$
\hat{A}=\int_{\sigma(A)}\lambda d\hat{E}^{\hat{A}}_{\lambda}
$
Such an expression is called the spectral decomposition of $\hat{A}$. Here $\sigma(A)\subseteq \Rl$ represents the spectrum of the operator $\hat{A}$ and $\{\hat{E}^{\hat{A}}_{\lambda}|\lambda\in\sigma(\hat{A})\}$ is the spectral family of $\hat{A}$. In the discrete case we would have $\hat{A}=\sum_{\sigma(A)}\lambda \hat{P}^{\hat{A}}_{\lambda}$, where the projection operators $ \hat{P}^{\hat{A}}_{\lambda}$ project on subspaces of the Hilbert space for which the states $\psi$ have value $\lambda$ of $A$. }) the \emph{spectral decomposition} of $\hat{A}$ is 
\be
\hat{A}:=\sum_{m=1}^M a_m\hat{P}_m
\ee 
where $\{a_1\cdots a_m\}$ is the set of eigenvalues of $\hat{A}$, while each $\hat{P}_m$ is the projection operator onto the subspace of eigenvectors with eigenvalue $a_m$. In particular
\be
\hat{P}_m:=\sum_{j=1}^{d(m)}|a_m, j\rangle\langle a_m,j|
\ee
Here $j=1\cdots d(m)$ labels the degenerate eigenvectors with common eigenvalue $a_m$. In this setting any state $|\psi\rangle $ can be written as follows
\be
|\psi\rangle=\sum_{m=1}^M\sum_{j=1}^{d(m)}\langle a_m, j|\psi\rangle|a_m, j\rangle
\ee

Keeping this in mind, and inspired by case (i) above we define 
\ba\label{ali:ha}
h(\hat{A})|\psi\rangle&:=&\sum_{m=1}^M\sum_{j=1}^{d(m)}h(a_m)\langle a_m, j|\psi\rangle|a_m, j\rangle\\
&=&\sum_{m=1}^Mh(a_m)\hat{P}_m|\psi\rangle
\ea
Obviously this makes sense iff $f(a_m)$ is well defined.

Since definition \ref{ali:ha} is valid for all $|\psi\rangle\in \mh$ it follows that
\be
h(\hat{A}):=\sum_{m=1}^M\sum_{j=1}^{d(m)}|a_m, j\rangle\langle a_m, j|=\sum_{m=1}^Mh(a_m)\hat{P}_m
\ee
\end{enumerate}

\textbf{The Conditions FUNC Entails}\\\\
The conditions on the valuation function implied by FUNC are:
\begin{enumerate}
\item The $\emph{sum rule}$
\begin{equation}
V(\hat{A}+\hat{B})=V(\hat{A})+V(\hat{B})   \label{eq:sum}
\end{equation}
where $\hat{A}$ and $\hat{B}$ are such that $[\hat{A},\hat{B}]=0$.
\begin{proof}
To prove the above result we need the following theorem:
\begin{Theorem}:
Given any pair of self-adjoint operators $\hat{A}$ and $\hat{B}$, such that $[\hat{A},\hat{B}]=0$ and 
two functions $f,g:\Rl\rightarrow\Rl$, then there exists a third operator $\hat{C}$
such that $\hat{A}=f(\hat{C})$ and $\hat{B}=g(\hat{C})$.
\end{Theorem}
Given two commuting operators $\hat{A}$ and $\hat{B}$ from the above theorem it follows 
that $\hat{A}=f(\hat{C})$ and
$\hat{B}=g(\hat{C})$, therefore there exists a function $h=f + g$ ($(f+h)(x):=f(x)+h(x)$) such that 
$\hat{A}+\hat{B}=h(\hat{C})$,
therefore
\begin{align*}
V(\hat{A}+\hat{B})&=V(h(\hat{C}))\\
&=h(V(\hat{C}))\\
&=f(V(\hat{C}))+g(V(\hat{C}))\\
&=V(f(\hat{C}))+V(g(\hat{C}))\\
&=V(\hat{A})+V(\hat{B})
\end{align*}
\end{proof}

\item The $\emph{product rule}$
\begin{equation}
V(\hat{A}\hat{B})=V(\hat{A})V(\hat{B})   \label{eq:product}
\end{equation}
where $\hat{A}$ and $\hat{B}$ are such that $[\hat{A},\hat{B}]=0$
\begin{proof}
Given $\hat{A}=f(\hat{C})$ and $\hat{B}=g(\hat{C})$ there exists a function $k=f g$ ($fg(x):=f(x)g(x)$) such that
$\hat{A}\hat{B}=k(\hat{C})$
therefore
\begin{align*}
V(\hat{A}\hat{B})&=V(k(\hat{C}))\\
&=k(V(\hat{C}))\\
&=f(V(\hat{C})) g(V(\hat{C}))\\
&=V(f(\hat{C}))\cdot V(g(\hat{C}))\\
&=V(\hat{A}) V(\hat{B})
\end{align*}
\end{proof}
\end{enumerate}
As a consequence of the product and sum rules we obtain the following equalities:
\begin{equation}
\begin{split}
V_{|\psi\rangle}(\hat{\mathds{1}})&=1\\
V_{|\psi\rangle}(\hat{0})&=0\\
V_{|\psi\rangle}(\hat{P})&=0\hspace{.1in}or\hspace{.1in}1    \label{eq:cecia}
\end{split}
\end{equation}
\begin{proof}
\begin{enumerate}
\item Given any physical quantity $B$ (with associated self-adjoint operator $\hat{B}$), from the product rule we have that, for $\hat{A}:=\hat{\mathds{1}} $, the following relation holds:
\be
V(\hat{\mathds{1}} \hat{B})=V(\hat{B})=V(\hat{\mathds{1}})V(\hat{B})=V(\hat{B})
\ee
This implies that $V(\hat{\mathds{1}})=1$.
\item Given any physical quantity $B$ (with associated self-adjoint operator $\hat{B}$), from the sum rule we have that, for $\hat{A}:=\hat{0} $, the following relation holds:
\be
V(\hat{0} +\hat{B})=V(\hat{B})=V(\hat{0})+V(\hat{B})=V(\hat{B})
\ee
This implies that $V(\hat{0})=0$.
\item Given a projection operator $\hat{P}$ we know that $\hat{P}^2=\hat{P}$ therefore
\be
(V(\hat{P})^2=V(\hat{P}^2)=V(\hat{P})
\ee
It follows that 
\be
V(\hat{P})=1\text{ or }0
\ee
\end{enumerate}
\end{proof}

Since quantum propositions can be expressed as projection operators (the reason will be explained later on in the course), what the last result implies is that, for any given state $|\psi\rangle$,  the valuation function can only assign value \emph{true} or \emph{false} to propositions. 
\\
\\

Since the set of all eigenvectors of a self-adjoint operator $\hat{A}$ forms an 
orthonormal basis for $\mathcal{H}$, then we can define the resolution of unity in terms of the projection operators corresponding to the eignevectors:
\begin{equation}
\hat{1}=\sum_{m=1}^M\hat{P}_m   \label{eq:unity1}
\end{equation}
From $\ref{eq:sum}$, $\ref{eq:product}$, $\ref{eq:cecia}$ and $\ref{eq:unity1}$ we conclude 
(for discreet case but it can easily be extended to the continuous case)
\begin{equation}
V(\hat{1})=V\left(\sum_{m=1}^M\hat{P}_m\right)=\sum_{m=1}^M V(\hat{P}_m)=1   \label{eq:ci}
\end{equation}
What this equation means is that one and only one of the projectors that form the resolution of 
unity gets assigned the value 1 (true), while the rest gets assigned the value 0 (false), i.e. the 
value assignment
is said to be `` mutually exclusive and collectively exhaustive" 
\cite{topos4}.
However, the Kochen-Specker theorem will show that it is impossible to give simultaneous values to all observables associated 
with a set of self-adjoint operators,
in such a way that the values are ``mutually exclusive and collectively exhaustive". 
Since the property of values of being ``mutually exclusive and collectively exhaustive" is a consequence
of FUNC, it is worth analysing how the condition FUNC is derived from the formalism of Quantum 
theory.\\
FUNC is a direct consequence of three assumptions and a principle 
present in quantum theory :
\begin{itemize}
\item [1)] \textbf{Statistical functional compositional principle}:given a self-adjoint operator $\hat{A}$
that represents an observable A and a function $f:\Rl\rightarrow\Rl$,
then for an arbitrary real number $a$ we have the following equality:
\begin{equation*}
prob[V(f(\hat{A}))=a]=prob[f(V(\hat{A}))=a]
\end{equation*}
In order to prove the above principle we have to define 
the relation between projector operators and their respective 
characteristic functions.\\
Let us consider the following characteristic function $\chi_{r}$ such that
\begin{equation*}
\chi_{r}(t)=\begin{cases}1& {if\hspace{.1in}t=r}\\
0& otherwise
\end{cases}
\end{equation*}
It then follows that, given a self-adjoint operator $\hat{A}$, whose spectral decomposition (assume discrete) contains 
the spectral
projector $\hat{P}_m$, one can write: 
\begin{equation}
\chi_r(\hat{A}):=\sum_{m=1}^M\chi_r(a_m)\hat{P}_m=\begin{cases} \hat{P}_m & {if\hspace{.1in}a_m=r}\\
0&  otherwise
\end{cases}        \label{eq:chi}
\end{equation}
What equation $\ref{eq:chi}$ uncovers is that $\chi_r(\hat{A})=\hat{P}_m$ iff r is the eigenvalue $a_m$ of 
$\hat{A}$. 
Moreover, given a function $f:\sigma(\hat{A})\rightarrow\Rl$ (where $\sigma(\hat{A})$ represents the spectrum of $\hat{A}$) we have:
\begin{equation}
\chi_r(f(\hat{A}))=\chi_{f^{-1}(a)}(\hat{A})  \label{eq:chi2}
\end{equation} 
We know that the statistical algorithm \cite{topos1} for projection operators is
\begin{equation}
prob(V(\hat{(A)}=a_m)=Tr(\hat{P}_m\hat{P}_{|\psi\rangle})   \label{eq:tr}
\end{equation}
where $\hat{P}_{|\psi\rangle}:=|\psi\rangle\langle\psi|$.
This means that if a measurement of an observable A is made on a system in state $|\psi\rangle$,
then the probability of obtaining as a result the eigenvalue $a_m$ is given by $\ref{eq:tr}$.\\ 
Therefore from $\ref{eq:chi}$ and $\ref{eq:tr}$ we get
\begin{equation*}
prob(V(\hat{A})=a_m)=Tr(\chi_{a_m}(\hat{A})\hat{P}_{|\psi\rangle})
\end{equation*}
We can now prove the statistical functional compositional principle.
\begin{proof}
Using equations: $\ref{eq:chi}$ ,$\ref{eq:tr}$ and $\ref{eq:chi2}$
we can write the statistical algorithm for projector operators as follows:
\begin{align*}
prob(V(f(A))=b)&=
Tr((\chi_{f^{-1}(b)}(\hat{A})\hat{P}_{|\psi\rangle})\\
&=Tr(\hat{P}_{f^{-1}(b)}\hat{P}_{|\psi\rangle})\\
&=prob(V_{|\psi\rangle}(A)=f^{-1}(b))
\end{align*}
but
\begin{equation*}
V(A)=f^{-1}(b)\hspace{.2in}\Leftrightarrow\hspace{.2in}f(V(A))=b
\end{equation*}
therefore
\begin{equation*}
prob(V(f(A))=b)=prob(f(V(A))=b)
\end{equation*}
\end{proof}
\item [2)]\textbf{Non-contextuality}: the value of observables is independent of the measurement context,
i.e. the value of each observable is independent of any other observables evaluated at the same 
time.
\item [3)] \textbf{Value definiteness}: observables possess definite values at all times.
\item [4)]\textbf{Value realism}: to each real number $\alpha$, such that $\alpha=prob(V(\hat{A})=\beta)$,
for an operator $\hat{A}$ there corresponds an observable A with value $\beta$.
\end{itemize}
From the above conditions (1),(2),(3) and (4) the FUNC condition follows.
\begin{proof}
Consider an observable B represented by the self-adjoint operator $\hat{B}$.
From (3) we deduce that $\hat{B}$ possesses a value: $V(\hat{B})=b$.
Given a function $f:\Rl\rightarrow\Rl$ we obtain the quantity $f(V(\hat{B}))=f(b)=a$.
Applying (1) we get
$prob[f(V(\hat{B}))=a]=prob[V(f(\hat{B}))=a]$ which means that there exists a self-adjoint operator of the form 
$f(\hat{B})$. 
From (4) it then follows that the corresponding observable for $f(\hat{B})$ has value $(a)$, therefore 
$f(V(\hat{B}))=V(f(\hat{B}))$.
From (2) this result is unique, therefore FUNC follows.
\end{proof}
We now 
state the Kochen-Specker theorem. 
 \begin{Theorem}
\textbf{Kochen-Specker Theorem}:
if the dimension of $\mathcal{H}$ is greater than 2 then, there does not exist any valuation function
$V:\mathcal{O}\rightarrow\Rl$ from the set $\mathcal{O}$ of all bounded self-adjoint operators 
$\hat{A}$ of $\mathcal{H}$ to the reals $\Rl$, such that the functional composition principle  
is satisfied for all $\hat{A}\in\mathcal{O}$.
\end{Theorem}
Another way of stating the theorem which seems more useful for developing a proof is the following:
\begin{Theorem}
\textbf{Kochen-Specker Theorem}:
Given a Hilbert space $\mathcal{H}$ such that $dim(\mathcal{H})>2$ and a set $\mathcal{O}$ of 
self-adjoint operators $\hat{A}$ which represent observables, then 
the following two statements are contradictory:
\begin{enumerate}
\item all observables associated with projectors in $\mathcal{O}$ have values simultaneously, i.e. 
they are mapped uniquely onto the reals.
\item the values of observables follow the functional composition principle (FUNC).
\end{enumerate}
\end{Theorem}
\section{Proof of Kochen-Specker Theorem}
There are various proofs of the Kochen-Specker Theorem. We will report a simplified 
version of the proof due to Kernaghan (1994) \cite{topos12}.\\

In the previous section we saw that the properties of the valuation functions implied that $V$ can only assign the value \emph{true} or \emph{false} to any projection operator $\hat{P}$, in such a way that this assignment is mutually exclusive and collectively exhaustible. 

A special case would be when $\hat{P}_i:=|e_i\rangle\langle e_i|$ where $\{|e_1\rangle,|e_2\rangle, \cdots, |e_n\rangle\}$ is an orthonormal basis of the Hilbert space $\mh^n$. In this setting the valuation function must assign the value 1 to only one of the projection operators and zero to all the rest. Moreover, if the same projection operator belongs to two different ONB, the value assigned to this projection operator by $V$ has to be the same, independently to which set it is considered to belong. This is what it is meant by \emph{non-contextuality}. Kernaghan, in his proof of the K-S theorem considers a real 4 dimensional
Hilbert space $\mathcal{H}_4$ (there is no loss in generality in considering the Hilbert space to 
be real). He then chooses 11 sets of 4 orthogonal vectors. Each vector is contained in either 2 of these sets or 4, so that there is some correlations between the ONB.
The Kochen-Specker theorem then reduces to a colouring problem, i.e. ``within 
every set of orthogonal vectors in $\mathcal{H}_4$ exactly one must be coloured white (1, true) while 
the remaining black (0, false)". Writing down this collection of vectors, we would end up with the 
following table where each column denotes a set of 4 orthogonal vectors in a particular ONB:

\begin{center}
\hspace{-.4in}\begin{tabular}{|c|c|c|c|c|c|c|c|c|c|c|c|} \hline 
\tiny{$|e_1\rangle$}&\tiny{1,0,0,0}&\tiny{1,0,0,0}&\tiny{1,0,0,0}&\tiny{1,0,0,0}&\tiny{-1,1,1,1}&\tiny{-1,1,1,1}&\tiny{1,-1,1,1}&\tiny{1,1,-1,1}&\tiny{0,1,-1,0}&\tiny{0,0,1,-1}&\tiny{1,0,1,0} \\ \hline
\tiny{$|e_2\rangle$}&\tiny{0,1,0,0}&\tiny{0,1,0,0}&\tiny{0,0,1,0}&\tiny{0,0,0,1}&\tiny{1,-1,1,1}&\tiny{1,1,-1,1}&\tiny{1,1,-1,1}&\tiny{1,1,1,-1}&\tiny{1,0,0,-1}&\tiny{1,-1,0,0}&\tiny{0,1,0,1} \\ \hline
\tiny{$|e_3\rangle$}&\tiny{0,0,1,0}&\tiny{0,0,1,1}&\tiny{0,1,0,1}&\tiny{0,1,1,0}&\tiny{1,1,-1,1}&\tiny{1,0,1,0}&\tiny{0,1,1,0}&\tiny{0,0,1,1}&\tiny{1,1,1,1}&\tiny{1,1,1,1}&\tiny{1,1,-1,-1} \\ \hline
\tiny{$|e_4\rangle$}&\tiny{0,0,0,1}&\tiny{0,0,1,-1}&\tiny{0,1,0,-1}&\tiny{0,1,-1,0}&\tiny{1,1,1,-1}&\tiny{0,1,0,-1}&\tiny{1,0,0,-1}&\tiny{1,-1,0,0}&\tiny{1,-1,-1,1}&\tiny{1,1,-1,-1}&\tiny{1,-1,-1,1} \\ \hline
\end{tabular}
\end{center}

We now want to assign value true (colour white) to one and only one projection operator associated to each vector in each column. This requirement represents condition $\ref{eq:ci}$ above. However, it is easy to see from the table that such condition ($\ref{eq:ci}$) is not satisfied.
In fact, if it were satisfied we would end up with 11 entries being coloured white,  
since each column would have exactly one entry that is coloured white and 
there are 11 columns. But, since each vector appears twice we end up with an even number of white entries which is greater than 11.  Therefore, we conclude that it is impossible to obtain a colouring of a set 
of orthogonal vectors that is consistent with condition $\ref{eq:ci}$. Remember that we have assumed a non-contextual assignment of the entries, 
i.e. we are assuming that same vectors get assigned same colour independently of the column they 
belong to.

Although this is a very simplified version of the proof of the Kochen-Specker theorem, the main 
idea 
is the same as the main idea in the original proof, namely:
given a set of orthogonal vectors in $\mathcal{H}$ it is impossible to assign to each of them a 
set of numbers
$\{1,0,0,0....0\}$ where only one entry is equal to 1, i.e. it is impossible to give simultaneous 
values to all observables 
while respecting the FUNC condition.  
\section{Consequences of the Kochen-Specker Theorem}
The implications of the Kochen-Specker theorem is that one or both of the following two assumptions must be dropped:
\begin{enumerate}
\item [i] The set of truth values is represented by $\{0,1\}$.
\item [ii] The functional composition principle.
\end{enumerate}

\noindent
In the topos approach we abandon the idea that the set of truth values is only $\{0,1\}$. In fact in this approach we utilise a multivalued logic which, in turn, will imply the adoption of an intuitionistic logic.

\noindent
On the other hand, abandoning FUNC would entail abandoning some or all of the three assumptions:
\begin{enumerate}
\item\emph{Non-contextuality}
\item\emph{Value Definiteness}
\item\emph{Value Realism}
\end{enumerate}
from which it derives.\\
In particular, if the FUNC principle gets dropped, then quantum theory turns out to be contextual and non-realist. To understand this situation we first have to introduce the notion of simultaneously measurable observables:

\begin{Definition}
Given two observables $\hat{A}$ and $\hat{B}$ we say that they are \textbf{simultaneously measurable} 
iff $[\hat{A}\hat{B}]=0$.
\end{Definition}
Let us now consider two observables: $\hat{A}$, $\hat{B}$, such that $[\hat{A},\hat{B}]\neq0$, and $\hat{A}$ and $\hat{B}$ have a common projection in their spectral decomposition 
\ba
\hat{A}&=&a_1\hat{P}+a_2\hat{P}_{a_2}+a_3\hat{P}_{a_3}\\
\hat{B}&=&b_1\hat{P}+b_2\hat{P}_{b_2}+b_3\hat{P}_{b_3}
\ea
From equation $\ref{eq:chi}$ it 
follows that $\hat{P}$ can be expressed in terms of $\hat{A}$ or of $\hat{B}$, i.e.:
\begin{align*}
\hat{P}&=\chi_{a_1}(\hat{A}):=\sum_{m=1}^{M}\chi_{a_1}(a_m)\hat{P}_{a_m}\\
\hat{P}&=\chi_{b_1}(\hat{B}):=\sum_{m=1}^{M}\chi_{b_1}(b_m)\hat{P}_{b_m}
\end{align*}

Since commuting operators correspond to orthogonal operators, then, if we choose to express 
$\hat{P}$ in terms of $\hat{A}$ i.e.
$\hat{P}=\chi_{a_1}(\hat{A})$ the commuting operators of $\hat{P}$ are: $\hat{P}_{a_2}$,
$\hat{P}_{a_3}$,$\hat{P}\vee\hat{P}_{a_2}$,$\hat{P}\vee\hat{P}_{a_3}$ and $\hat{P}^{\perp}$.
If, instead, we choose to express $\hat{P}$ in terms of $\hat{B}$, i.e. 
$\hat{P}=\chi_{b_1}(\hat{B})$, then the commuting operators of $\hat{P}$
would be: $\hat{P}_{b_2}$, $\hat{P}_{b_3}$,$\hat{P}\vee\hat{P}_{b_2}$,$\hat{P}\vee\hat{P}_{b_3}$
and $\hat{P}^{\perp}$.

Now, as a consequence of the FUNC condition we obtain that
\ba
V(\hat{P})&=&V(\chi_{a_1}(\hat{A})=\chi_{a_1}(V(\hat{A}))\\
V(\hat{P})&=&V(\chi_{b_1}(\hat{B})=\chi_{b_1}(V(\hat{B}))
\ea
which implies that $\chi_{a_1}(V(\hat{A}))=\chi_{b_1}(V(\hat{B}))$.

However if the FUNC condition does not hold, then we have that
\be
\chi_{a_1}(V(\hat{A}))\neq\chi_{b_1}(V(\hat{B}))
\ee
What this implies is that the value of $\hat{P}$ will depend on whether $\hat{P}$ is considered as belonging to the spectral decomposition of $\hat{A}$ or that of $\hat {B}$. In fact let us assume that  $\chi_{a_1}(V(\hat{A}))=1$, which means that $V(\hat{A})=a_1$. Since $a_1 $ is the eigenvalue with corresponding projection operator $\hat{P}$, then it follows that $V(\hat{P})=1$, since $\hat{P}$ projects on the subspace of eigenvectors which have eigenvalue precisely $a_1$. However, if $\chi_{a_1}(V(\hat{A}))\neq\chi_{b_1}(V(\hat{B}))$, then $\chi_{b_1}(V(\hat{B}))\neq 1$ which must imply (since it is a characteristic function) that $\chi_{b_1}(V(\hat{B}))=0$. If this is the case then $\hat{P}$ is false, i.e. $V(\hat{P})=0$.

%
%
%

As a consequence of the above, a physical quantity $A$ is not represented by a unique operator in quantum theory. On the contrary, each operator has different \emph{meaning} depending on what other operators are considered at the same time. This implies that the quantisation map $A\mapsto\hat{A}$ is one to many.

The contextuality derived from dropping FUNC has great impact on the `realism' of quantum theory. In fact, when one says that a given quantity has a certain value, we mean that that quantity ``possesses" that value,  and the concept of ``possession" is independent of the context chosen. However, if our theory is contextual, what does it mean exactly that a quantity has a given value? It would seem that in a contextual theory there is not room for a realist interpretation. In fact, if we measured two pair of quantities $(A, B)$ and $(A, C)$ and obtained the values $(a, b)$ and $(a^{'}, c)$, respectively, such that $a\neq a^{'}$, then what values does the quantity $A$ actually posses\footnote{It should be noted that the probabilistic predictions of quantum theory are not affected by the notion of contextuality. In fact, the result of measuring a property $A$ of a system does not depend on what else is measured at the same time, since the probability of obtaining $a_m$ as the value of $A$ will always be $\langle|\psi|\hat{P}_{a_m}|\psi\rangle$.}?

Thus, if one drops the FUNC principle we would end up with a non-realist contextual interpretation of quantum theory, which clashes with the realism of classical physics and our common sense.

\noindent
The question then arises: what if, instead, we dropped the first assumption, namely if we allowed for the truth values to be in some larger set other than the set $\{0,1\}$? This is precisely what is done in topos quantum theory. In fact in this setting the FUNC principle is conserved, but the set of truth values is replaced by some larger set than simply $\{0,1\}$ leading to a multivalued logic. The interpretation we end up with is not strictly realist, due to the multivalued nature of the resulting logic. However we reach a more realist interpretation of the theory since now it makes sense to say that values are possessed by quantities in a context independent way. 

At this stage it should be pointed out that in the topos formulation of quantum theory there will be the notion of \emph{contextuality} albeit its interpretation will be very different, thus it will not impinge on the notion of realism (of the interpretation). 

\chapter{Lecture 2/3}
 I this lectures I will explain how, in order to work with category theory, one has to change from an internal description of mathematical objects (as is done in set theory) to an external/relational description. To fully understand this change in perspective I will describe how the axioms of a group can be given by both an internal perspective and an external one.

\noindent
I will then give an axiomatic definition of what a category is, introducing also the concept of a subcategory. This will be augmented with various examples of categories and a list of the most common categories which appear in physics.

\noindent 
I will conclude with the concept of duality, very important in category theory and in what we will be doing in the rest of the course.

\section{Change of Perspective}
``Category Theory allows you to work on structures without the need first to pulverize them into set theoretic dust" (Corfiel).\\

The above quotation explains, in a rather pictorial way, what \emph{category theory} and, in particular \emph{topos theory} are really about.
In fact, \emph{category theory} and, in particular, \emph{topos theory} allow to abstract from the specification of points (elements of a set) and 
functions between these points to a universe of discourse in which the basic elements are arrows, and any property is given 
in terms of compositions of arrows.

The reason for the above characterisation is that the underlining philosophy behind category theory (and topos theory) is that of describing mathematical objects from an external point of view, i.e. in terms of relations.\\
This is in radical contrast to \emph{set theory}, whose approach is essentially internal in nature. By this we mean that the basic/primitive notions of what sets are and the belonging relations between sets, are defined in terms of the elements which belong to the sets in question, i.e. an internal perspective.

In order to be able to implement the notion of external definition we first need to define two important notions i) the notion of a map or arrow, which is simply an abstract characterisation\footnote{By abstract characterisation here we mean a notion that does not depend on the sets or objects between which the arrow is defined} of the notion of a function between sets; ii) the notion of an ``equation" in categorical language. We will first start with the notion of a map.

Given two general objects $A$ and $B$ ( not necessarily sets) an arrow $f$ is said to have domain $A$ and codomain $B$ if it goes from $A$ to $B$, i.e. $f:A\rightarrow B$. It is convention to denote $A=dom(f)$ and $B=cod(f)$.

\noindent
We will often draw such an arrow as follows:
\be
A\xrightarrow{f}B
\ee

Given two arrows $f:A\rightarrow B$ and $g:B\rightarrow C$, such that $cod(f)=dom(g)$ then we can \emph{compose} the two arrows obtaining $g\circ f:A\rightarrow C$. The property of \emph{composition} is drawn as follows:
\be
A\xrightarrow{f}B\xrightarrow {g}C
\ee

For each object (or set $A$) there always exists and \emph{identity arrow} $id_A:A\rightarrow A$
\be
id_A:A\rightarrow A\text{, such that }id_A(a)=a\forall \;\;a\in A
\ee

The collection of arrows between various objects satisfy two laws:
\begin{enumerate}
\item [i)] \emph{Associativity}: given three arrows $f:A\rightarrow B$, $g:B\rightarrow C$ and $h:C\rightarrow D$ with appropriate domain and codomain relations,  we then have
\be
h\circ (g\circ f)=(h\circ g)\circ f
\ee
\item [ii)] \emph{Unit law}: given $f:A\rightarrow B$, $id_A:A\rightarrow A$ and $id_B:B\rightarrow B$ the following holds
\be
f\circ id_A=f=id_B\circ f
\ee
\end{enumerate}

The next step is to try and define the analogue of an equation in an abstract categorical language. This is done through the notion of commutative diagrams. So, what is a diagram? A diagram is defined as follows:
\begin{Definition}
A  graph is a collection of vertices $\bullet_{x_i}$ and directed edges $e:\bullet_{x_1}\rightarrow \bullet_{x_2}$ where $\bullet_{x_1}$ is the source vertex while $\bullet_{x_2}$ is the target. If the vertices are labelled by objects $X_i$ and the edges are labeled by an arrow, such that each $e:\bullet_{x_1}\rightarrow \bullet_{x_2}$ is now labelled as $f:X_1\rightarrow X_2$, then we say that the graph is actually a diagram.
\end{Definition}

A typical diagram will be either a triangle or a square
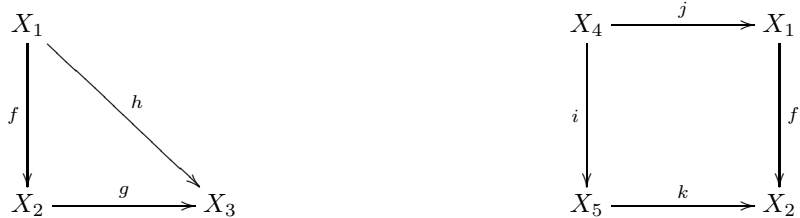
\[\xymatrix{
X_1\ar[dd]_{f}\ar[ddrr]^h&&&&&&\\
&&&&&&\\
X_2\ar[rr]^{g}&&X_3&&&&
}
\xymatrix{
X_4\ar[dd]_{i}\ar[rr]^{j}&&X_1\ar[dd]^f\\
&&\\
X_5\ar[rr]^k&&X_2\\
}\]

An ``equation" is then given by the concept of commutativity, i.e. we say that the above diagrams commute iff
\be
g\circ f=h;\;\;\;\;f\circ j=k\circ i
\ee
Care should be taken since commutativity is not as strict a condition as one might think. In particular, if we have the following commuting diagram
 \[\xymatrix{
A\ar[rr]^g&&B\ar@<3pt>[rr]^{\;\;\;\;\;\;\; f} \ar@<-3pt>[rr]_{\;\;\;\;\;\;\; h}&&C\\
}\]
we can only imply that $f\circ g=h\circ g$, but not that $f=h$

Obviously diagrams can be combined together to form a bigger diagram, as long as the rules pertaining composition of arrows hold. Thus, considering the commuting diagrams in the above example we can combine them to obtain
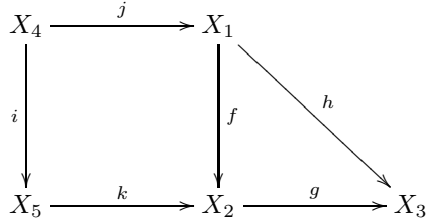
\[\xymatrix{
X_4\ar[dd]_{i}\ar[rr]^{j}&&X_1\ar[dd]^f\ar[ddrr]^h&&\\
&&&&\\
X_5\ar[rr]^k&&X_2\ar[rr]^{g}&&X_3\\
}\]

An important theorem in diagram language is the following: 

\noindent
consider the diagram 
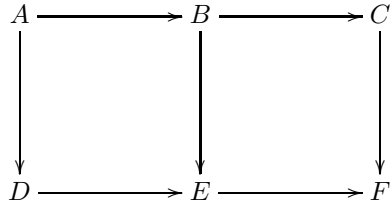
\[\xymatrix{
A\ar[dd]\ar[rr]&&B\ar[dd]\ar[rr]&&C\ar[dd]\\
&&&&\\
D\ar[rr]&&E\ar[rr]&&F\\
}\]
If any two of the three rectangles commute then so will the remaining one.
\\

Now that we know how arrows and ``equations" are defined we are ready to give some examples of how the same concepts can be described both internally, using set theory and externally, using a categorical language.
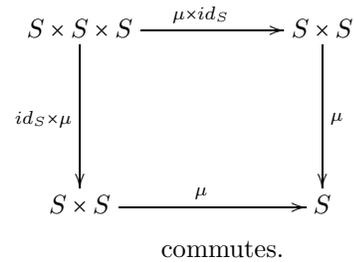
\begin{center}
\begin{tabular}{l l c r } 
&internal &  & external \\\hline
Element&$a\in S$ & $\stackrel{a=f(*)}{\leftrightarrow} $ &$ \{*\}\xrightarrow{f} S $\\ 
Subset&$A\subseteq S $&$\leftrightarrow{} $&$S\rightarrow \{0,1\} $\\
Associative binary operation& $a\cdot(b\cdot c) =(a\cdot b)\cdot c $ $\forall a,b,c\in S$ &$\leftrightarrow$ &$\mu:S\times S\rightarrow S$, such that
\end{tabular}
\end{center}
\[\xymatrix{
&&&&&&&&&&S\times S\times S\ar[rr]^{\mu\times id_{S}}\ar[dd]_{id_{S}\times \mu}&&S\times S\ar[dd]^{\mu}\\
&&&&&&&&&&&&\\
&&&&&&&&&&S\times S\ar[rr]^{\mu}&&S\\
}\] 
\hspace{5in}commutes.\\\\

It is interesting how the definition and the axioms of a group can be described in an external way. In particular, we have
\begin{enumerate}
\item \emph{Associativity}: ($\forall g_1, g_2, g_3\in G$, $g_1(g_2g_3)=(g_1g_2)g_3$).
\be
\mu\circ (id_G\times\mu)=\mu\circ (\mu\times id_G)
\ee
where 
\ba
\mu:G\times G&\rightarrow& G\\
\langle g_1, g_2\rangle&\mapsto& g_1g_2
\ea 
and 
\ba
id_G:G&\rightarrow& G\\
g&\mapsto&g
\ea
Equivalently, associativity can be defined as the condition by which the diagram below commutes
\[\xymatrix{
G\times G\times G\ar[rr]^{\mu\times id_{G}}\ar[dd]_{id_{G}\times \mu}&&G\times G\ar[dd]^{\mu}\\
&&\\
G\times G\ar[rr]^{\mu}&&G\\
}\]
\item \emph{Identity Element}: ($\forall g\in G$, $ge=eg=g$)
\be
\mu\circ (id_G\times k_e)=\mu\circ(k_e\times id_G)=id_G
\ee
where 
\ba
k_e: G&\rightarrow &G\\
g&\mapsto&k_e(g):=e
\ea
is the constant map which maps each element to the identity element. So explicitly the identity element condition is equivalent to the fact that the following diagram commutes
\[\xymatrix{
G\ar[rrdd]^{id_G}\ar[rr]^{id_G\times k_e}\ar[dd]_{k_e\times id_G}&&G\times G\ar[dd]^{\mu}\\
&&\\
G\times  G\ar[rr]^{\mu}&&G\\
}\]

\item \emph{Inverse}: ($\forall g\in G$, $gg^{-1}=g^{-1}g=e$).\\
The existence of an inverse can be written as follows:
\be
\mu\circ(id_G\times i)\circ \Delta=\mu\circ(i\times id_G)\circ \Delta=k_e
\ee
where the diagonal map $\Delta$ is
\ba
\Delta:G&\rightarrow& G\times G\\
g&\mapsto&(g,g)
\ea
while the inverse map is
\ba
i:G&\rightarrow& G\\
g&\mapsto&g^{-1}
\ea
and 
\ba
k_e: G&\rightarrow &G\\
g&\mapsto&k_e(g):=e
\ea
is the constant map which maps each element to the identity element.
Equivalently, the condition of having an inverse can be represented by the following commuting diagram:
\[\xymatrix{
G\ar[rrdddd]^{k_e}\ar[rr]^{\Delta}\ar[dd]_{\Delta}&&G\times G\ar[dd]^{id_G\times i}\\
&&\\
G\times G\ar[dd]^{i\times id_G}&&G\times G\ar[dd]^{\mu}\\
&&\\
G\times G\ar[rr]^{\mu}&&G
}\]

\end{enumerate}
\section{Axiomatic Definition of a Category}
\begin{Definition}
A (small\footnote{A category $\c$ is called small if $Ob(\c)$ is a Set }) category $\mathcal{C}$ consists of the following elements:
\begin{itemize}
\item[1.] A collection $Ob(\c)$ of $\c$-objects
\item[2.] For any two objects $a, b\in Ob(\c)$, a set $Mor_{\c}(a, b)$ of $\c$-arrows (or $\c$-morphisms) from $a$ to $b$
\item[3.] Given any three objects $a, b, c\in\c$, a map which represents composition operation
\ba
\circ:Mor_{\c}(b,c)\times Mor_{\c}(a,b)&\rightarrow& Mor_{\c}(a,c)\\
(f,g)&\mapsto&f\circ g
\ea
Composition is \emph{associative}, i.e. for $f\in Mor_{\c}(b,c)$, $g\in Mor_{\c}(a,b)$ and $h\in Mor_{\c}(c,d)$ we have
\be
h\circ(f\circ g)=(h\circ f)\circ g
\ee
which in diagrammatic form is the statement that the following diagram commutes
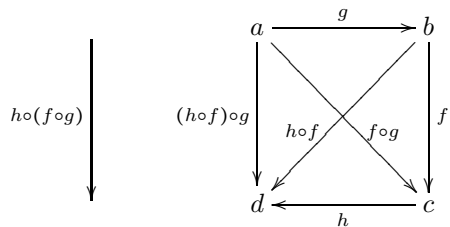
\[\xymatrix{
\ar[dd]_{h\circ(f\circ g)}&&a\ar[rr]^{g}\ar[ddrr]_{\;\;\;\;h\circ f\;\;\;}\ar[dd]_{(h\circ f)\circ g}&&b\ar[dd]^{f}\ar[ddll]^{\;\;\;f\circ g}\\
&&&&\\
&&d&&c\ar[ll]^{h}\\
}\]
\item[4.] For each object $b\in\c$ an identity morphisms $id_{b}\in Mor_{\c}(b,b)$, such that the following \emph{Identity law} holds: for all $g\in Mor_{\c}(a,b)$ and $f\in Mor_{\c}(b, c)$ then $f=f\circ id_b$ and $g=id_b\circ g$. In diagrammatic form this is represented by the fact that the diagram
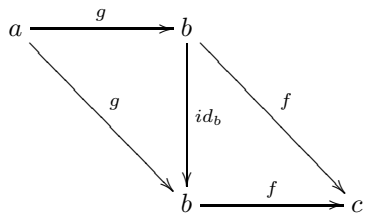
\[\xymatrix{
a\ar[rr]^{g}\ar[ddrr]^{g}&&b\ar[dd]^{id_b}\ar[ddrr]^{f}&&\\
&&&&\\
&&b\ar[rr]^{f}&&c\\
}\]
commutes.
\end{itemize}
\end{Definition}
So, a category is essentially a collection of diagrams for which certain ``equations" (commutative relations) hold.
\begin{Definition}
$\d$ is a \emph{subcategory} of $\c$, denoted $\d\subseteq \c$, if:
\begin{itemize}
\item[i)]
$Ob(\d)\subseteq Ob(\c)$ as sets.
\item[ii)]
For any two objects $c,d\in Ob(\d)$, then $Mor_{\d}(c,d)\subseteq Mor_{\c}(c,d)$.
\end{itemize}
\end{Definition}
Thus a subcategory is a sub-collection of objects with a sub-collection of graphs containing these objects.
\begin{Definition}
$\d$ is a \emph{full subcategory} of $\c$ if an extra requirement is satisfied:\\
(iii) for any $\d$-objects $a$ and $d$, then $Mor_{\d}(a,b)=Mor_{\c}(a,b)$.
\end{Definition}
Keeping with our graph description, a full subcategory is a sub-collection of objects but has the same collection of graphs containing these objects

\subsection{Examples of Categories}
\begin{example}
\textbf{Two object category}\\
Simple example 
of a two element 
category is the following:
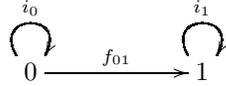
\[\xymatrix{
0\ar@(ul,ur)[]^{i_0}\ar[rr]^{f_{01}}&&1\ar@(ul,ur)^{i_1}
}\]
This category has 3 arrows:
\begin{itemize}
\item $i_0:0\rightarrow 0$ identity on 0.
\item $i_1:1\rightarrow 1$ identity on 1.
\item $f_{01}:0\rightarrow 1$.
\end{itemize}
it is easy to see that the composite arrows are:
$i_0 \circ  i_0=i_0$ ,$i_1 \circ  i_1=i_1$ ,$i_1 \circ  f_{01}=f_{01}$ and $f_{01} \circ  i_1=f_{01}$. 
\end{example}
\begin{example}
\textbf{Poset}
A poset is a set in which the elements are related by a partial order, i.e. not all elements are related to each other. The definition of a poset is as follows:
\begin{Definition}
Given a set $P$ we call this a poset iff a partial order $\leq$ is defined on it.  A partial order is a binary relation $\leq$ on a set $P$, which has the following properties:
\begin{itemize}
\item Reflexivity: $a\leq a$ for all $a\in P$.
\item Antysimmetry: if $a\leq b$ and $b\leq a$, then $a=b$.
\item Transitivity: If $a\leq b$ and $b\leq c$, then $a\leq b$.
\end{itemize}

\end{Definition}
An example of a poset is any set with an inclusion relation defined on it. Another example is $\Rl$ with the usual ordering defined on it.  

A poset forms a category whose objects are the elements of the poset and, given any two elements $p, q$, there exists a map $p\rightarrow q$ iff $p\leq q$ in the poset ordering.
We will be using such a (poset) category quite often when defining a topos description of quantum theory. Thus, it is worth pointing out the following:
\begin{Definition}
Given two partial ordered sets $P$ and $Q$, a map/arrow $f:P\rightarrow Q$ is a partial order homomorphisms (otherwise called monotone functions or order preserving functions) if
\be
\forall x,y\in P\;\;x\leq y\Longleftrightarrow f(x)\leq f(y)
\ee
\end{Definition}
Homomorphisms are closed under composition. A trivial example of partial order homomorphisms is given by the identity
maps.
\end{example}
\begin{example}
\textbf{Comma Category}\\
This category (also called slice category) has as objects arrows with fixed domain or codomain,
for example $C\downarrow\Rl$ is a comma category with:
\begin{itemize}
\item Objects: given $A,B\in C$, the respective objects in $C\downarrow\Rl$ are arrows whose codomain is $\Rl$, i.e. $f:A\rightarrow\Rl$ and $g:B\rightarrow\Rl$ also written as: $(A,f)$ and $(B,g)$.
\item Morphisms: given two objects $f$ and $g$ we define an arrow between them as the arrow $k:A\rightarrow B$ in $\c$ such that,

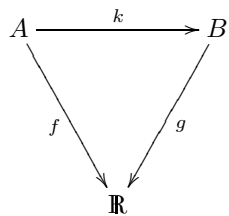
\[\xymatrix{
A\ar[rr]^k\ar[ddr]_f&&B\ar[ddl]^g\\
&&\\
&\Rl& \\
}\]

commutes in $C\downarrow\Rl$
\end{itemize}
The above definition of arrows in $C\downarrow\Rl$ implies the following:
\begin{itemize}
\item \emph{Composition}:
given the two arrows $j:A\rightarrow B$ and $i:B\rightarrow C$, their composition is defined by the following commutative diagram: 
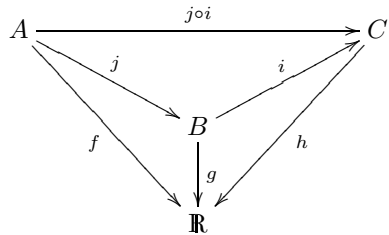
\[\xymatrix{
A\ar[rrrr]^{j \circ i}\ar[drr]^j\ar[ddrr]_f&&&&C\ar[ddll]^h\\
&&B\ar[rru]^i\ar[d]^g&&\\
&&\Rl&&\\
}\]
Basically you just glue triangles together.
\item \emph{Identity}: given an element $f:A\rightarrow\Rl$, its 
identity arrow is: $id_A:(A,f)\rightarrow(A,f)$
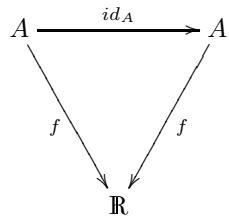
\[\xymatrix{
A\ar[rr]^{id_A}\ar[ddr]_f&&A\ar[ddl]^f\\
&&\\
&\Rl& \\
}\]

\end{itemize}
It is interesting to note that, given a category $\c$ for any element $A\in \c$ we can form the comma category $\c/A$ ($\c\downarrow A$) where objects in $\c/A$ are all morphisms in $\c$ with codomain $A$, while arrows between two objects $f:B\rightarrow A$ and $g:D\rightarrow A$ are commutative diagrams in $\c$, i.e.
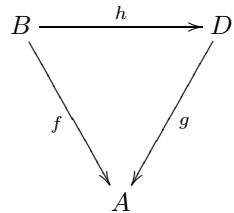
\[\xymatrix{
B\ar[rr]^{h}\ar[ddr]_f&&D\ar[ddl]^g\\
&&\\
&A& \\
}\]
For an object $f:B\rightarrow A$ the identity arrow is simply $id_B:B\rightarrow B$.\\
\end{example}

\begin{example}
\textbf{Monoid}\\
A \emph{monoid} $\mathcal{M}$ is a one object category equipped with a binary operation on that object and a unit 
element. In particular. the definition is as follows:
\begin{Definition} 
A monoid $\mathcal{M}$ is a triplet (M, *, i) such that,
\begin{itemize}
\item M is a Set.
\item * is a map $M\times M\rightarrow M$ which is associative.
\item $i\in M$, such that $\forall\hspace{.02in}x\in M$ $i * x = x * i = x$, where $i$ is the two sided identity.
\end{itemize}
\end{Definition}
The $*$-map can be identified either with the tensor product or with the direct sum or with the direct product according to which category 
M one is tacking into consideration.

\noindent 
Examples of monoinds are $(\Nl, +, 0)$ and any type of group $(G, \cdot, e)$.

\noindent
It is also possible to compare two different monoids as follows:
\begin{Definition}
Given two monoids $M$ and $N$  a map $h : M\rightarrow  N$ is saied to be a monoid homomorphism iff
\be
 \forall m_1,m_2\in M\;\;\;h(m_1 *m_2)=h(m_1)*h(m_2),\;\;	h(1)=1
\ee
\end{Definition}
\end{example}

We will now give a list of various categories which are frequently used in physics.
\begin{center}
\begin{tabular}{ l c r } 
Category & Objects  &morphisms \\\hline
\textbf{Sets} & Sets &functions\\ 
\textbf{Top}&Topological space & Continuous maps\\
\textbf{Gr}&Groups &Homomorphisms of groups\\
\textbf{Ab}& Abelian groups& Homomorphisms of groups\\
$\mathbf{Vect_K}$& Vector spaces over a filed K& K-linear maps\\
\textbf{Man}& Manifolds& Smooth maps\\
\textbf{Pos}& Partially ordered sets& Monoton functions\\
\textbf{N}& N& Natural numbers\\
\textbf{Set} is a descrete category & $x\in Set $& Identity arrows\\
\textbf{Preoder}:P& $x\in P$& At most one arrow between any two objects.\\

\end{tabular}
\end{center}

\section{The Duality Principle}
A very important notion in category theory is the notion of \emph{duality}. In particular, for any statement (or ``equation") $Y$ expressed in categorical language  its dual $Y^{\op}$ is obtained by replacing the domain with the codomain and the codomain by the domain and reversing the order of arrow composition , i.e. $h=g\circ f$ becomes $h=f\circ g$. 

Thus, all arrows and diagrams in $Y$ have the reverse direction in $Y^{\op}$, and the construction /notion described by $Y^{\op}$ is saied to be dual to the notion described by $Y$. Moreover, we also have the notion of a dual category.

\begin{Definition}
Given a category $\c$ the dual $\c^{op}$ is defined as follows:\\
\be
Ob(\c^{op}):=Ob(\c)\;\;Mor_{\c^{op}}(a, b):=Mor_{\c}(b,a)
\ee
the composition law is: 

\noindent
given $f\in Mor_{\c^{op}}(a,b)$ and $g\in Mor_{\c^{op}}(b,c)$, then 
\be
g\circ_{\c^{op}}f:=f\circ_{\c} g
\ee
\end{Definition}
It is easy to see that $(\c^{op})^{op}=\c$ for any category.\\
Therefore, given the construction $Y^{op}$ referred to a category $\c$, this can be considered as the construction $Y$ applied to the dual category $\c^{op}$.\\
The notion of opposite categories leads to the very important notion of \emph{duality principle}, by which a statement $Y$ is true in $\c$ iff its dual $Y^{\op}$ is true in $\c^{op}$. This principle allows us to prove various things simultaneously. By this we mean that if we have a statement $X$, which holds in the category $\c$, then we immediately know that the statement $X^{op}$ holds for $\c^{op}$. Moreover, if we derive a theorem $T$ from the axioms of category theory, then such a theorem holds for any category $\c$. However, by duality $T^{op}$ holds for every category $\c^{op}$. But each category can be written as the opposite of some other category ($(\c^{op})^{op}=\c$), therefore $T^{op}$ holds for all categories.
Then the duality principle allows us to derive a universal theorem from a specific instance of it.\\
In what follows we will see many examples of statements, theorems and their duals. 

\section{Arrows in a Category}
In this section we will explain the notions of injective, surjective and bijective map in a categorical language, i.e. from an external point of view.
\subsection{Monic Arrow}
 Monic arrow is the ``arrow-analogue" of an \emph{injective function}.
\begin{Definition}
A $\c$-arrow $f:a\rightarrow b$ is monic if for any pair of parallel arrows 
$g:c\rightarrow a$, $h:c\rightarrow a$, the equality $f \circ g = f\circ h$ implies that $h=g$, i.e. $f$ is left 
cancellable.
Monic arrows are denoted as:

\[\xymatrix{
*++{a}\ar@{>->}[rr]&&b\\
}\]

\end{Definition}
We now want to show how it is possible to derive a monic arrow from an injective one and vice versa in $\mathbf{Sets}$.
\begin{proof}
Consider the sets $A,B,C$, an injective function $f:A\rightarrow B$ (i.e. if $f(x)=f(y)$, then $ x=y$) and a pair of parallel functions 
$g:C\rightarrow A$ and $h:C\rightarrow A$, such that 
\[\xymatrix{
C\ar[rr]^{g}\ar[dd]_{h}&&A\ar[dd]^{f}\\
&&\\
A\ar[rr]_f&&B\\
}\]
commutes, i.e. $f \circ g = f\circ h$.\\ 
Now if 
\begin{align*}
\hspace{.2in}x\in C\Longrightarrow \hspace{.2in}&f \circ g(x)= f \circ h(x)\\
&f(g(x))=f(h(x))
\end{align*}
Since f is injective it follows that $g(x)=h(x$), i.e $f$ is left cancellable.
Vice versa, let $f$ be left cancellable, and consider the following diagram:

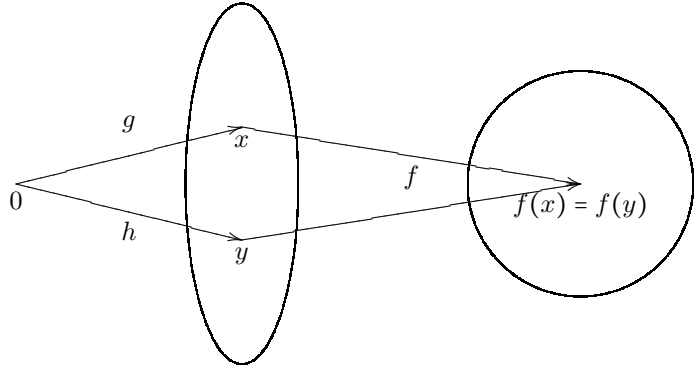
\[\begin{xy}0;/r15mm/:
,(0,0)="x",*@{},*+!U{0}
,(2,0)="z",*@{},*+!{}
,{\ellipse(0.5,1.6){}}
,(2,0.5)="t",*@{},*+!U{x}
,(2,-0.5)="f",*@{},*+!U{y}
,{\ar"x";"t"}
,{\ar"x";"f"}
,(5,0)="p",*@{},*+!U{f(x)=f(y)}
,{\ellipse(1,1){}}
,{\ar"t";"p"}
,{\ar"f";"p"}
,(1,0.5),*@{},*+!{g}
,(1,-0.5),*@{},*+!{h}
,(3.5,0),*@{},*+!{f}
 \end{xy}\]

From the above diagram it is easy to deduce that $f\circ g=f \circ h$, since $f(x)=f(y)$. Given that $x=g(0)$ and $y=h(0)$ by construction, and $f$ is left cancellable by assumption, we get: $g=h$, therefore $x=y$ for $f(x)=f(y)$, i.e. $f$ is injective.
\end{proof}
\subsection{Epic Arrow}
An epic arrow is the ``arrow-analogue" of a \emph{surjective function}.
\begin{Definition}
 An arrow $f:a\rightarrow b$ in a category $\c$ is epic in $\c$ if, for any parallel pair $g:b\rightarrow c$ and $h:b\rightarrow c$ of arrows, the equality $g \circ f =h \circ f$ implies that $h=g$, i.e. $f$ is right cancellable.
Monic arrows are denoted as:\[\xymatrix{
*++{a}\ar@{->>}[rr]&&b\\
}\]
\end{Definition}
An epic is a dual of a monic.\\
In \textbf{Sets} the epic arrows are the surjective set functions. 
\begin{proof}
Let us consider three sets $A,B,C$ such that the set function $f:A\rightarrow B$ is surjective but it is not right cancellable, i.e. given two functions $g,h: B\rightarrow C$ although $h\circ f=g\circ f$, $h\neq g$. What this implies is that there exists an a element $y\in B$ such that $h(y)\neq g(y)$. However since $f$ is surjective $y=f(x)$ for some $x\in A$, then, $h\circ f(x)\neq g\circ f(x)$ which contradicts the assumption that $h\circ f=g\circ f$.
\end{proof}
\subsection{Iso Arrow}
An iso arrow is the ``arrow-analogue" of a \emph{bijective function}.
\begin{Definition}
A $\c$-arrow $f:a\rightarrow b$ is iso, or invertible in $\c$ if there is a $\c$-arrow $g:b\rightarrow a$, such that $g \circ f=1_a$ and $f \circ g=1_b$. Therefore, $g$ is the inverse of $f$ i.e. $g=f^{-1}$.
\end{Definition}
\begin{Theorem}
$g$ is unique.
\end{Theorem}
\begin{proof}
Consider any other $g^{'}$ such that $g^{'} \circ f=1_a$ and $f \circ g^{'}=1_b$, then we have
$$g^{'}=1_a \circ g^{'}= (g \circ f) \circ g^{'} = g \circ (f \circ g^{'}) = g \circ 1_b = g$$
\end{proof}
An iso arrow has the following properties:
\begin{enumerate}
\item \emph{An iso arrow is always monic}.
\begin{proof}
Consider an iso $f$, such that $f\circ g = f \circ h$ ($f:a\rightarrow b$ and $g,h:c\rightarrow a$),
then \begin{align*}
g& = 1_a \circ g = (f^{-1} \circ f) \circ g = f^{-1} \circ (f \circ g) \\
&= f^{-1} \circ (f \circ h) = (f^{-1} \circ f) \circ h = h 
\end{align*}
therefore $f$ is left cancellable.
\end{proof}
\item \emph{An iso arrow is always epic} 
\begin{proof}
Consider an iso $f$ such that $g \circ f = h \circ f$  ($f:a\rightarrow b$ and $g,h:b\rightarrow c$)
\begin{align*}
g& = g \circ 1_b = g \circ (f \circ f^{-1}) = (g \circ f) \circ f^{-1} = (h \circ f) \circ f^{-1}\\
& = h \circ (f \circ f^{-1}) = h
\end{align*}
therefore $f$ is right cancellable
\end{proof}
\end{enumerate}

\textbf{Note}: not all arrows which are monic and epic are iso, for example: 
\begin{enumerate}
\item An inclusion map is both monic and epic, but it is not iso, otherwise it would have an inverse and, as a set function, it would have to be a bijection, but it is not.
\item In \emph{poset}, eventhough all functions are monic and epic, the only iso is the identity map.\\
In fact, consider a function $f:p\rightarrow q$, this implies that $p\leq q$. If $f$ is an iso, then $f^{-1}:q\rightarrow p$ exists, therefore, $q\leq p$. However, from the antisymmetry property $p\leq q$ and $q\leq p$ imply that $p=q$, therefore $f=1_p$ is a unique arrow.
\end{enumerate}
Iso arrows are used to determine isomorphic objects within a given category.
\begin{Definition}
Given two objects $a,b\in \c$, we say that they are isomorphic $a\simeq b$ if there exists an iso $\c$-arrow between them. 
\end{Definition}

As we have seen from the above definitions, we managed to give an external characterisation for the set theoretic concepts of injective, surjective and bijective functions. 
\section{Elements and Their Relations in a Category}
In this section we will describe certain fundamental constructions or elements present in category theory. While reading this, it is useful to try and understand what the corresponding elements would be in \textbf{Sets}.
\subsection{Initial Object}\label{sini}
\begin{Definition}  
An \emph{initial object} in a category $\c$ is a $\c$-object $0$ such that, 
for every $\c$-object $A$, there exists one and only one $\c$-arrow from $0$ to $A$.
\end{Definition}
An initial object is unique up to isomorphism, i.e. all initial objects in a category are isomorphic. To see this, consider two initial objects $A_1\in \c$ and $A_2\in \c$. Being both initial, we have the unique arrows $f_1:A_1\rightarrow A_2$ and $f_2:A_2\rightarrow A_1$. Moreover, the fact that they are initial implies that it is possible to uniquely compose the above arrows obtaining $f_1\circ f_2=id_{A_2}$ and $f_2\circ f_1=id_{A_1}$. Therefore $f_1$ and $f_2$ are isomorphic functions and $A_1\simeq A_2$.\\\\
\textbf{Examples}
\begin{enumerate}
\item 
In $C\downarrow\Rl$ the initial object is $f:\emptyset\rightarrow\Rl$, 
such that the following diagram commutes:
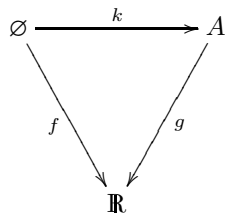
\[\xymatrix{
\emptyset\ar[rr]^k\ar[ddr]_f&&A\ar[ddl]^g\\
&&\\
&\Rl& \\
}\]
\item In $\Sets$ the initial object is the $\emptyset$ element.
\item In \textbf{Pos}, the initial object is the poset $(\emptyset , \emptyset)$.
\item In  \textbf{Top}, the initial object is  the space $(\emptyset,\{\emptyset\})$.
\item In $\mathbf{Vect_k}$, the one-element space $\{0\}$ is the initial object.
\item In a poset, the initial object is the least element with respect to the ordering. 
\end{enumerate}
An initial object is the dual of a terminal object. 

\subsection{Terminal Object}\label{sterm}
\begin{Definition} 
A \emph{terminal object} in a category $\c$ is a $\c$-object 1 such that, 
given any other $\c$-object A, there exists one and only one $\c$-arrow from A 
to 1. 
\end{Definition}
\textbf{Examples} 
\begin{enumerate}
\item
in $C\downarrow\Rl$ the terminal object is ($\Rl$, $id_{\Rl}$), such that the diagram 
\[\xymatrix{
A\ar[rr]^k\ar[ddr]_f&&\Rl\ar[ddl]^{id_{\Rl}}\\
&&\\
&\Rl& \\
}\]
commutes ($\therefore$ k=f)
\item In $\Sets$ a terminal object is a singleton $\{*\}$, since given any other element $A\in \Sets$ there exist 1 and only 1 arrow $A\rightarrow\{*\}$.
\item In \textbf{Pos} the poset $(\{*\}, \{(*, *)\})$ is the terminal object.
\item  In \textbf{Top}, the space$(\{*\},\{\emptyset, \{*\}\})$ is the terminal object.
\item  In $\mathbf{Vect_k}$, the one-element space $\{0\}$ is the terminal object.
\item  In a poset, the terminal object is the greatest element with respect to the ordering.
\end{enumerate}
Given the notion of a terminal object we can now define the notion of an element of a $\c$-object. Note that, so far, the definition of every categorical object that was introduced never rested on specific characteristic of its composing elements. This is because, as stated above, concepts in category theory are defined externally. In fact, it is the case that certain objects in a given category do not have elements. We will return to this later. For now we will give the categorical description of what an element of an object actually is.
\begin{Definition}
Given a category $\c$, with terminal object $1$, then an element of a $\c$-object b is a C-arrow $x:1\rightarrow b$.
\end{Definition}
\begin{example}
In $\mathbf{\Sets}$, an element $x\in A$, can be identified with the singleton subset $\{*\}$, therefore with an arrow 
$\{*\}\rightarrow A$ from the terminal object to A.
\end{example}
\section{Products}
We will now give the external/categorical description of the cartesian product. Such a definition will be a general notion of what a product is, which will be valid in any category independent of the details of that category. This is, in fact, one of the powerful aspects of category theory: \emph{an abstract characterization of objects in terms of universal properties}. In this way, definitions become independent of the peculiarity of individual cases, becoming a more objective, universally valid construction. One can compare the level of abstraction in category theory with the level of abstraction in differential geometry, where one defined objects without the use of a specific coordinate reference frame.

Let us now turn to our task of defining what a product is in categorical language. It should be pointed out that, as with all the other objects defined so far, if a product exists, it is uniquely up to isomorphism. Given a particular category, we can then verify whether or not the product exists in that category. It is only at this point that the particularity of the category in question enters the game, i.e. only into the proof of existence\footnote{The proof of existence is done by constructing an object and verifying if it satisfies the requirements of being a product.
}.All the useful properties of the product follow from the general definition. 
So, what is a product?
\begin{Definition}
A \textit{product} of two objects A and B in a category $\c$ is a third $\c$-object
$A\times B$ together with a pair of $\c$-morphisms (arrows):
\be
pr_A:A\times B\rightarrow A\;\;\;\; pr_B:A\times B\rightarrow B
\ee
such that, given any other pair of  $\c$-arrows $f:C\rightarrow A$ and $g:C\rightarrow B$,
there exists a unique arrow $\langle f,g\rangle:C\rightarrow A\times B$,
such that the following diagram commutes\footnote{Note that an arrow drawn as 
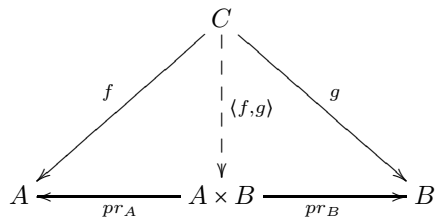
\[\xymatrix{&&\ar@{-->}[rr]&&}\] indicates uniqueness, up to isomorphisms of that arrow.}
\[\xymatrix{
&&C\ar[rrdd]^g\ar@{-->}[dd]^{\langle f,g\rangle}\ar[lldd]_f&&\\
&&&&\\
A&&A\times B\ar[rr]_{pr_B}\ar[ll]^{pr_A}&&B\\
}\]

i.e.
\begin{equation*}
pr_A \circ \langle f,g\rangle=f\hspace{.2in}and\hspace{.2in}pr_b \circ\langle f,g\rangle=g
\end{equation*}
\end{Definition}
Given two products we would now like to know if and how is possible to relate them. To this end one needs to introduce the concept of a map between two product objects. Such a map will be called a product map. The definition is straightforward.
\begin{Definition}
Consider a category $\c$ which allows products. Then consider two $\c$-arrows $f:A\rightarrow B$ and $g:C\rightarrow D$. The product map $f\times g:A\times C\rightarrow B\times D$ is the $\c$-arrow $\langle f\circ pr_A, g\circ pr_B\rangle$. Such an arrow is the unique arrow which makes the following diagram commute:
\[\xymatrix{
&&C\ar[rr]^g&&D\\
&&&&\\
A\times C \ar[uurr]^{pr_C}\ar[ddrr]^{pr_A}\ar@{-->}[rrrr]^{\langle f\circ pr_A, g\circ pr_B\rangle}&&&&B\times D\ar[uu]^{pr_D}\ar[dd]^{pr_B}\\
&&&&\\
&&A\ar[rr]_{f}&&B\\
}\]
\end{Definition}
\begin{Theorem}
In $\Sets$ the product of two sets always exists and it is the cartesian product with projection maps.
\end{Theorem}
\begin{proof}
Given a set $R$ with maps $q_1:R\rightarrow S$, $q_2:R\rightarrow T$, then the map 
\ba
\psi:R&\rightarrow &S\times T\\
r&\mapsto& (q_1(r), q_2(r))
\ea
would satisfy the commutativity property of the cross product, i.e. $\forall r\in R$
\ba
(p_1\circ \psi)(r)&=&q_1(r)\\
(p_2\circ\psi)(r)&=&q_2(r)
\ea
We now need to prove its uniqueness. This is done as follows: 

\noindent
if there exists a $\phi$ which satisfies $p_1\circ \phi=q_1$ and $p_1\circ\phi=q_2$ then, for all $r\in R$ we have
\be
\psi(r)=(q_1(r),q_2(r))=(p_1(\psi(r)),p_2(\psi(r)))=\phi(r)
\ee
where the last equality holds, since $(s,t)=p_1(s,t), p_2(s,t))$ for all $(s,t)\in S\times T$.\\
It follows that $\psi$ is unique.
\end{proof}
We now want to show that the products are commutative, i.e. $A_1\times A_2\simeq A_2\times A_1$. To this end let us consider each product separately. Being products, there exist unique arrows $i$ and $j$ such that the following diagrams commute:
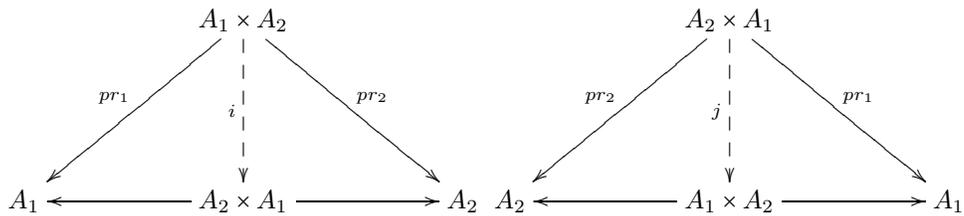
\[\xymatrix{
&&A_1\times A_2\ar@{-->}[dd]_{i}\ar[ddll]_{pr_1}\ar[ddrr]^{pr_2}&&\\
&&&&\\
A_1&&A_2\times A_1\ar[rr]\ar[ll]&&A_2
}
\xymatrix{
&&A_2\times A_1\ar@{-->}[dd]_{j}\ar[ddll]_{pr_2}\ar[ddrr]^{pr_1}&&\\
&&&&\\
A_2&&A_1\times A_2\ar[rr]\ar[ll]&&A_1\\
}\]
Composition of these diagrams in both orders gives us the following commuting diagrams:
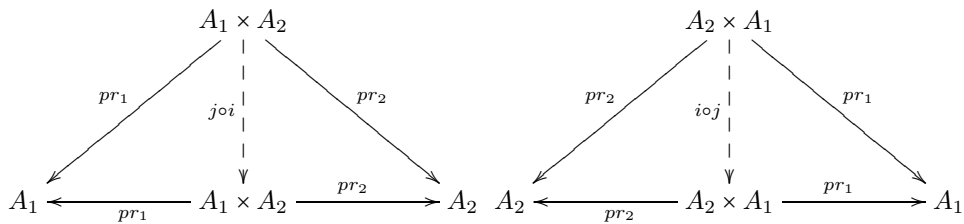
\[\xymatrix{
&&A_1\times A_2\ar@{-->}[dd]_{j\circ i}\ar[ddll]_{pr_1}\ar[ddrr]^{pr_2}&&\\
&&&&\\
A_1&&A_1\times A_2\ar[rr]^{pr_2}\ar[ll]^{pr_1}&&A_2
}
\xymatrix{
&&A_2\times A_1\ar@{-->}[dd]_{i\circ j}\ar[ddll]_{pr_2}\ar[ddrr]^{pr_1}&&\\
&&&&\\
A_2&&A_2\times A_1\ar[rr]^{pr_1}\ar[ll]^{pr_2}&&A_1\\
}\]
It is clear that $j\circ i=id_{A_1\times A_2}$ and $i\circ j=id_{A_2\times A_1}$, thus $i$ and $j$ are isomorphic.\\
The proof given above can be easily extended to give the associativity of the product operation in general, i.e. for an arbitrary amount of factors.\\
It should be noted that the product of an empty set of objects is just the terminal object, and the product of the family consisting of a single object $A$ is $A$ itself with projection $1_A:A\rightarrow A$.
\subsubsection{Examples}
\begin{itemize}
\item In \textbf{Pos}, products are cartesian products with the pointwise order.
\item In \textbf{Top}, products are cartesian products with the product topology. 
\item In $\mathbf{Vect_k}$, products are direct sums. 
\item In a poset products are greatest lower bounds.
\end{itemize}
\subsection{Co-Products}
We now define the categorical/external definition of disjoint union.
\begin{Definition}
A \textit{co-product} of two objects A and B in a category $\c$ is a third $\c$-object
$A + B$ together with a pair of $\c$-arrows:
\be
i_A:A\rightarrow A + B\;\;\;\;\;i_B:B\rightarrow A + B
\ee
such that, given any other pair of  $\c$-arrows $f:A\rightarrow C$ and $g:B\rightarrow C$,
there exists a unique arrow $[f,g]:A + B\rightarrow C$
which makes the following diagram commute:
\[\xymatrix{
A\ar[rr]^{i_A}\ar[rrdd]_f&&A+B\ar[dd]^{[f,g]}&&B\ar[ll]_{i_B}\ar[ddll]^g\\
&&&&\\
&&C&&\\
}\]
i.e. the co-product is the dual of the product. In the above, the arrows $i_A$ and $i_B$ indicate canonical injection maps.
\end{Definition}
Again, it is possible to define a map between two co-products. In fact, in defining such a map one can simply dualise the definition of the product map, thus obtaining the following definition:
\begin{Definition}
Assuming that co-products exist in $\c$, we consider two $\c$-arrows $f:A\rightarrow B$ and $g:C\rightarrow D$. The co-product map $f+g:A+C\rightarrow B+D$ is the unique $\c$-arrow $[i_D\circ f, i_B\circ g]$, such that the following diagram commutes:
\[\xymatrix{
A\ar[dd]_{i_A}\ar[rr]^f&&D\ar[dd]^{i_D}\\
&&\\
A\times C\ar@{-->}[rr]^{[ i_D\circ f, i_B\circ g]}&&B+ D\\
&&\\
C\ar[uu]^{i_C}\ar[rr]_g&&B\ar[uu]_{i_B}\\
}\]
\end{Definition}
 \begin{Theorem}
 In $\mathbf{Sets}$ the co-product of any two elements $X, Y\in \Sets$ always exists and it is the disjoint union
 \be
 X\coprod Y=\{(x,j)\in (X\cup Y)\times\{0,1\}|x\in X \text{ iff } j=0\;,\; x\in Y\text{ iff } j=1\}
 \ee
 We then have
 \ba
 i_X:X&\rightarrow& X\coprod Y\\
 x&\mapsto &(x,0)
 \ea
 and 
 \ba
 i_Y:Y&\rightarrow& X\coprod Y\\
 y&\mapsto &(y,1)
 \ea
 \end{Theorem}
%
It should be noted that the co-product of an empty set is the initial object.
\subsubsection{Examples}
\begin{itemize}
\item In \textbf{Pos}, co-products are identified with disjoint unions (with the inherited orders).
\item In \textbf{Top}, co-products are identified with topological disjoint unions.
\item  In $\mathbf{Vect_k}$, co-products are identified with direct sums.
\item  In a poset, co-products are identified with least upper bounds.
\end{itemize}
\subsection{Equaliser}
We now describe the categorical analogue of the concept of the largest set for which two functions coincide. This is the concept of an equaliser.
 \begin{Definition}
 Given a category $\c$, a $\c$-arrow $i:E\rightarrow A$ is an equaliser of a pair of $\c$-arrows $f,g:A\rightarrow B$ if
 \begin{enumerate}
 \item $f\circ i=g\circ i$.
 \item Given another $\c$-arrow $h:C\rightarrow A$ such that $f\circ h=g\circ h$, there is exactly one $\c$-arrow $k:C\rightarrow E$ such that the following diagram commutes:
 \[\xymatrix{
E\ar[rr]^i&&A\ar@<3pt>[rr]^{\;\;\;\;\;\;\; f} \ar@<-3pt>[rr]_{\;\;\;\;\;\;\; g}&&B\\
&&&&\\
 &C\ar@{-->}[uul]^{k}\ar[uur]^h&&&\\
}\]
 i.e. $i\circ k=h$
 
 \end{enumerate}
 \end{Definition}
 In $\Sets$, the equaliser of a pair of maps $f,g:A\rightarrow B$ is the largest subset for which the two maps coincide, i.e.
 \be
 \{x\in A|f(x)=g(x)\}\subseteq A
 \ee
 \subsection{Co-Equaliser}
 Dual to the equaliser, there exists the co-equaliser which is defined as follows:
 \begin{Definition}
 Given two $\c$-arrows $f,g:A\rightarrow B$, the $\c$-arrow $h:B\rightarrow C$ is a co-equaliser of $f$ and $g$ if the following conditions hold:
 \begin{itemize}
 \item[1)] $h\circ f=h\circ g$
 \item [2)] Given any other $\c$-arrow $h^{'}:B\rightarrow C^{'}$ such that $h^{'}\circ f=h^{'}\circ g$, there exists a unique $\c$-arrow $u:C\rightarrow C^{'}$ such that the following diagram commutes
 \[\xymatrix{
A\ar@<3pt>[rr]^{\;\;\;\;\;\;\; f} \ar@<-3pt>[rr]_{\;\;\;\;\;\;\; g}&&B\ar[rr]^h\ar[rrdd]^{h^{'}}&&C\ar@{-->}[dd]^u\\
&&&&\\
 &&&&C^{'}\\
}\]

 \end{itemize}
 
 \end{Definition}
 In $\mathbf{Sets}$ the co-equaliser of a pair of maps $f,g:A\rightarrow B$ is the quotient of $B$ by the least equivalence relation for which $f(x)=g(x)$ for all $x\in A$. The condition of being the least equivalence relation is required by condition 2) above.
 
 The construction of such an equivalence relation is done as follows:
 
 \noindent
 Consider the co-equaliser
 \be
 S=\{\langle f(x), g(x)\rangle|x\in A\}\subseteq B\times B
 \ee
 Although it is a co-equaliser it is not necessarily an equivalence relation on $B$. However, one can construct a minimal equivalence relation on $B$ which contains $S$. In particular, such an equivalence relation $R\subseteq B\times B$ would be such that
 \begin{itemize}
 \item $S\subseteq R$.
 \item Given any other equivalence relations $T$ on $B$ such that $S\subseteq T$ then, $R\subseteq T$.
 \end{itemize}
 Such a relation is obviously 
 \be
 S=\{x\in A|f(x)=g(x)\}
 \ee
\subsection{Limits and Co-Limits}
In the description of the elements/objects that we gave so far we have always utilised the notion of a \emph{universal property}, which the object in question had to satisfy. In particular, we never constructed any object in terms of the characteristics of its elements, but rather through its relations to other objects, such that these relations had to satisfy a \emph{universal property}, thus obtaining a \emph{universal construction}. As can be easily deduced from the universal constructions given above, these are unique up to isomorphisms.

So, what are these universal constructions and universal properties?

\noindent
The precise characterisation of these concepts is given in terms of the notion of \emph{diagrams}, \emph{limits} and \emph{co-limits} of these diagrams.
\begin{Definition}
Given a category $\c$, a diagram $D$ in $\c$ is defined to be a collection of $\c$-objects $a_i\in C$ ($i\in I$) and a collection of $\c$-arrows $a_i\rightarrow a_j$ between some of the $\c$-objects above.
\end{Definition}

Using the notion of graphs given at the start of this lecture, a diagram can be defined as one graph in the collection of graphs composing a category.

Now a special type of diagram $D$ is the $D$-cone, i.e. a cone for a diagram $D$. This consists of a $\c$-object $c$ and a $\c$-arrow $f_i:c\rightarrow a_i$ one for each $a_i\in D$, such that 
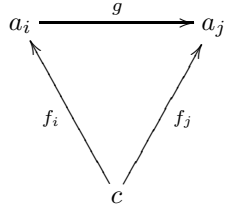
\[\xymatrix{
a_i\ar[rr]^g&&a_j\\
&&\\
&c\ar[uur]_{f_j}\ar[uul]^{f_i}&&\\
}\]
commutes when ever $g$ is an arrow in the diagram $D$. \\
A cone is denoted as $\{f_i:c\rightarrow a_i\}$ and $c$ is called the vertex of the cone.
We now come to the definition of a limit.
\begin{Definition}
A limit for a diagram $D$ is a $D$-cone $\{f_i:c\rightarrow a_i\}$ such that, given any other $D$-cone $\{f^{'}_i:c^{'}\rightarrow a_i\}$, there is only one $\c$-arrow $g:c^{'}\rightarrow c$ such that, for each $a_i\in D$ the following diagram commutes:
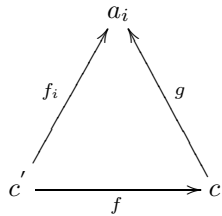
\[\xymatrix{
&a_i&\\
&&\\
c^{'}\ar[uur]^{f_i}\ar[rr]_f&&c\ar[uul]_{g}\\
}\]
\end{Definition}
The limiting cone of a diagram $D$ has the \emph{universal property} with respect to all other $D$-cones, in the sense that any other $D$-cone factors out through the limiting cone.\\

\textbf{Examples}:

\begin{example}
i) The product of two objects $A$ and $B$ in $\c$, defined above, is actually the limiting cone of the diagram containing only two elements $A$ and $B$ and no arrows, i.e. it is the limiting cone of the arrowless diagram
\[\xymatrix{
A&&B\\
}\]
In fact a cone for this diagram is given by any $\c$-object $C$ together with two arrows $f:C\rightarrow A$ and $g:C\rightarrow B$ giving the cone

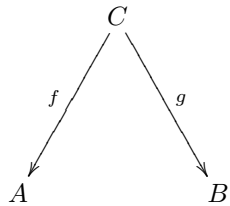
\[\xymatrix{
&C\ar[ddr]^g\ar[ddl]_f&\\
&&\\
A&&B\\
}
\]
Now, in order for this cone to be a limiting cone, we require that any other cone factors through it. This means that given another cone 
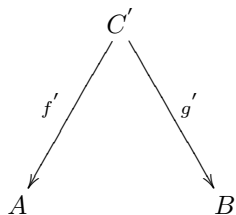
\[\xymatrix{
&C^{'}\ar[ddr]^{g^{'}}\ar[ddl]_{f^{'}}&\\
&&\\
A&&B\\
}
\]
there exists a unique map $h:C^{'}\rightarrow C$, such that
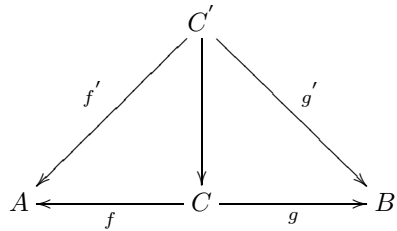
\[\xymatrix{
&&C^{'}\ar[ddrr]^{g^{'}}\ar[ddll]_{f^{'}}\ar[dd]&&\\
&&&&\\
A&&C\ar[rr]_g\ar[ll]^f&&B\\
}
\]
But this is precisely the definition of the product, i.e. $C=A\times B$.
\end{example}
\begin{example}
The terminal object is the limiting cone of the empty diagram.
\end{example}
\begin{example}
 The equaliser is the limiting cone of the diagram
 
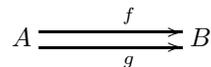
\[\xymatrix{
A\ar@<3pt>[rr]^{\;\;\;\;\;\;\; f} \ar@<-3pt>[rr]_{\;\;\;\;\;\;\; g}&&B\\
}\]

\end{example}
By duality we also have the notion of a \emph{co-limit} whose definition requires (as expected) the notion of a co-cone (dual to a cone).\\
Given a diagram $D$ a co-cone consists of an object $c$ and arrows $\{f_i\rightarrow c\}$, one for each element $a_i\in D$. A co-cone is denoted $\{f_i:a_i\rightarrow c\}$. \\
We now define a co-limit as follows:
\begin{Definition}
A co-limit of $D$ is a co-cone with the (co)-universal property that given any other $D$-cone $\{f^{'}_i:a_i\rightarrow c^{'}\}$ there exists one and only one map $f:c\rightarrow c^{'}$ such that the following diagram commutes
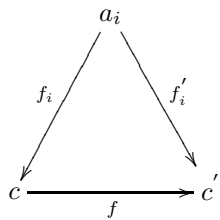
\[\xymatrix{
&a_i\ar[ddr]^{f^{'}_i}\ar[ddl]_{f_i}&\\
&&\\
c\ar[rr]_{f}&&c^{'}\\
}\]
for all $a_i\in D$
\end{Definition}
From the duality principle one can figure out what exactly the co-product, initial object and co-equaliser, are.

\section{Categories in Quantum Mechanics}
In this section we will delineate different categories that arise in quantum theory, however  we will not go into the details of how each of these categories is used. The aim is simply to show that category theory arises in many more contexts than one can imagine. The list of examples of categories in quantum theory is by no means complete.

\subsection{The Category of Bounded Self Adjoint Operators}
\begin{Definition} \cite{isham1} \cite{isham2}
the \textbf{Set $\mathcal{O}$} of bounded self-adjoint operators is a \textbf{category}, such that
\begin{itemize}
\item the objects of $\mathcal{O}$ are the self-adjoint operators;   
\item given a function $f:\sigma(\hat{A})\rightarrow\Rl$ (from the spectrum 
of $\hat{A}$ to the Reals) such that $\hat{B}=f(\hat{A})$
then there exists a morphism $f_{\mathcal{O}}:\hat{B}\rightarrow\hat{A}$ in $\mathcal{O}$ between operators 
$\hat{B}$ and $\hat{A}$. 
\end{itemize}
\end{Definition}
To show that the category $\mathcal{O}$, so defined, is a category, 
we need to show that it 
satisfies the identity
law and composition law. 
This can be shown in the following way:
\begin{itemize}
\item \textit{Identity Law}:
given any $\mathcal{O}$-object $\hat{A}$, the identity arrow is defined as the arrow 
$id_{\mathcal{O}_A}:\hat{A}\rightarrow\hat{A}$ that corresponds to the arrow $id:\Rl\rightarrow\Rl$
in $\Rl$.
\item \textit{Composition Condition}:
given two $\mathcal{O}$-arrows $f_{\mathcal{O}}:\hat{B}\rightarrow\hat{A}$ and 
$g_{\mathcal{O}}:\hat{C}\rightarrow\hat{B}$, such that $\hat{B}= f(\hat{A})$ and 
$\hat{C}= g(\hat{B})$, then, the composite function $f_{\mathcal{O}}o g_{\mathcal{O}}$ in 
$\mathcal{O}$ corresponds to the composite function $f o g:\Rl\rightarrow\Rl$
in $\Rl$.
\end{itemize}
The category  $\mathcal{O}$, as defined above, represents a pre-ordered set.
In fact, the function $f:\sigma(\hat{A})\rightarrow\Rl$ is unique up to isomorphism, therefore 
it follows that for any two objects in 
$\mathcal{O}$ there exists, 
at most, one morphism between them, i.e. $\mathcal{O}$ is a pre-ordered set.
However, $\mathcal{O}$ fails to be a poset since it lacks the 
antisymmetry property .
In fact, it can be the case that two operators $\hat{B}$ and $\hat{A}$ in $\mathcal{O}$ are 
such that $\hat{A}\neq\hat{B}$, but they are related by $\mathcal{O}$-arrows
$f_{\mathcal{O}}:\hat{B}\rightarrow\hat{A}$ and 
$g_{\mathcal{O}}:\hat{A}\rightarrow\hat{B}$ in such a way that:
\begin{equation}
g_{\mathcal{O}}\circ f_{\mathcal{O}}=id_B\hspace{.1in}and\hspace{.1in}f_{\mathcal{O}}\circ g_{\mathcal{O}}
=id_A      \label{eq:eq}
\end{equation}

It is possible to transform the set of self-adjoint operators into a poset by defining a new 
category $[\mathcal{O}]$ in
which the objects are taken to be equivalence classes of operators, whereby two operators are
considered to be equivalent if the $\mathcal{O}$-morphisms, relating them, satisfies equation
 \ref{eq:eq}.
\subsection{Category of Boolean Sub-Algebras}
\begin{Definition} \label{def:W} \cite{isham1} \cite{isham2}
The \textbf{category $\mathcal{W}$} of Boolean sub-algebras of the lattice $P(\mathcal{H})$ has:
\begin{itemize}
\item as objects, the individual Boolean sub-algebras, i.e. elements $W\in\mathcal{W}$ which represent 
spectral algebras associated with different operators. 
\item as morphisms, the arrows between objects of $\mathcal{W}$, such that a morphism 
$i_{W_1W_2}:W_1\rightarrow W_2$ exists iff $W_1\subseteq W_2$.
\end{itemize}
\end{Definition}
From the definition of morphisms it follows that there is, at 
most, one morphisms between any two elements of $\mathcal{W}$, therefore $\mathcal{W}$
forms a poset under sub-algebras inclusion $W_1\subseteq W_2$.\\
To show that $\mathcal{W}$, as defined above, is indeed a category, we need to define the identity
arrow and the composite arrow.
The identity arrow in $\mathcal{W}$ is defined as $id_W:W\rightarrow W$, which corresponds 
to $W\subseteq W$, whereas, given two $\mathcal{W}$-arrows $i_{W_1W_2}:W_1\rightarrow W_2$ 
$(W_1\subseteq W_2)$ and 
$i_{W_2W_3}:W_2\rightarrow W_3$ $(W_2\subseteq W_3)$ the composite 
$i_{W_2W_3}\circ i_{W_1W_2}$ corresponds to $W_1\subseteq W_3$. 

\begin{example}
An example of the category $\mathcal{W}$ can be formed in the following way:\\
consider a category consisting of four objects (operators): $\hat{A}$,$\hat{B}$,$\hat{C}$,$\hat{1}$
such that the spectral decomposition is:
\begin{align*}
\hat{A}&=a_1\hat{P}_1+a_2\hat{P}_2+a_3\hat{P}_3\\
\hat{B}&=b_1(\hat{P}_1\vee\hat{P}_2)+b_2\hat{P}_3\\
\hat{C}&=C_1(\hat{P}_1\vee\hat{P}_3)+c_2\hat{P}_2
\end{align*}
then the spectral algebras are the following:
\begin{align*}
W_A&=\{\hat{0},\hat{P}_1,\hat{P}_2,\hat{P}_3,\hat{P}_1\vee\hat{P}_3,\hat{P}_1\vee\hat{P}_2,\hat{P}_3\vee\hat{P}_2,\hat{1}\}\\ 
W_B&=\{\hat{0},\hat{P}_3,\hat{P}_1\vee\hat{P}_2,\hat{1}\}\\
W_C&=\{\hat{0},\hat{P}_2,\hat{P}_1\vee\hat{P}_3\hat{1}\}\\
W_1&=\{\hat{1}\}
\end{align*}
The relation between the spectral algebras is given by the following diagram:
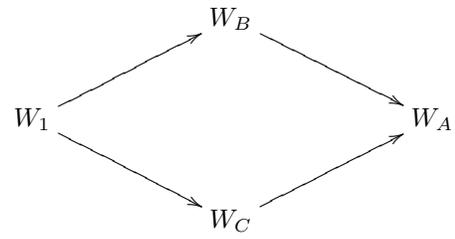
\[\xymatrix{
&&W_B\ar[rrd]&&\\
W_1\ar[rru]\ar[rrd]&&&&W_A\\
&&W_C\ar[rru]&&
}\]
where the arrows are subset inclusions.
\end{example}
\chapter{Lecture 4}

In this lecture I will describe how it is possible to define maps between categories. There are two types of such maps called \emph{covariant functors} and \emph{contravariant functors}. I will describe both and give examples of both. We will then abstract a bit further and define maps between functors them selves. These are called \emph{Natural Transformations}. Such transformations will enable us to define equivalent categories.

\section{Functors and Natural Transformations}
\label{sfunc}
So far we have introduced the notion of a category.  However, if we can not compare categories together we could not do much in terms of category theory. Thus, there must be a way of comparing categories or, at least, define maps between them. This is done through the notion of a functor.
Generally speaking a functor is a transformation from one category $\c$
to another category
$\d$, such that
the categorical structure of the domain $\c$ is preserved, i.e. gets mapped onto the structure of the codomain category
$\d$.\\
There are two types of functors:
\begin{enumerate}
\item \textbf{Covariant Functor}
\item \textbf{Contravariant Functor}
\end{enumerate}
\subsection{Covariant Functor}
 \begin{Definition}:
A \textbf{covariant functor} from a category $\c$ to a category $\d$ is a map 
$F:\c\rightarrow\d$ that assigns to each $\c$-object $a$, a 
$\d$-object F(a) and to each $\c$-arrow $f:a\rightarrow b$ a $\d$-arrow
$F(f):F(a)\rightarrow F(b)$, such that the following are satisfied:
\begin{enumerate}
\item $F(1_a)=1_{F(a)}$
\item $F(f \circ g)=F(f) \circ F(g)$ for any $g:c\rightarrow a$ 
\end{enumerate}
\end{Definition}

It is clear from the above that a covariant functor is a transformation that preserves both:
\begin{itemize}
\item The domain's and the codomain's identities.
\item The composites of functions, i.e. it preserves the direction of the arrows.
\end{itemize}
A pictorial description if a covariant functor is as follows:
 \[\xymatrix{
a\ar[rr]^f\ar[rrdd]_h&&b\ar[dd]^g\\
&&\\
&&c\\
}
\xymatrix{\ar@{=>}[rr]^F&&}
\xymatrix{
F(a)\ar[rr]^{F(f)}\ar[rrdd]_{F(h)}&&F(b)\ar[dd]^{F(g)}\\
&&\\
&&F(c)\\
}\]
\subsubsection{Examples}
\begin{example}
\emph{Identity functor}:\\ $id_{\c}:\c\rightarrow\c$ is such that $id_{\c}A=A$ for all $A\in \c$ and $id_{\c}(f)=f$ for all $\c$-arrows $f$. Similarly one can define the insertion functor for any subcategory $\d\subseteq \c$.
This is trivially defined as follows
\ba
I:\d&\rightarrow&\c\\
A&\mapsto&A\\
(f:A\rightarrow B)&\mapsto&(f:A\rightarrow B)
\ea
Given such a definition it follows that
\ba
I(id_A)&=&id_{I(A)}\\
I(f\circ g)&=&=I(f)\circ I(g)
\ea

\end{example}

\begin{example}
\emph{Power set functor}:

\noindent
$P:\Sets\rightarrow \Sets$ assigns to each object $X\in \Sets$ its power set\footnote{Sets of all subsets of X.} $PX$, and to each map $f:X\rightarrow Y$ the map $P(f):PX\rightarrow PY$, which sends each subset $S\subseteq X$ to the subset $f(S)\subseteq Y$.
\end{example}

\begin{example}
\emph{Forgetful functor}:\\ Given a category $\c$ with some structure on it, for example the category of groups $\mathbf{Grp}$, the \emph{forgetful functor} $F:\mathbf{Grp}\rightarrow \Sets$ takes each group to its underlining set forgetting about the group structure, and each $\c$-arrow to itself.
\end{example}

\begin{example}
\emph{Hom functor}:\\Given any $\c$-object $A$, then the \emph{Hom functor} $\c(A,-):\c\rightarrow \Sets$ takes each object $B$ to the set of all $\c$-arrows $\c(A,B)$ from $A$ to $B$, and to each $\c$-arrow $f:B\rightarrow C$ it assigns the map 
\ba
\c(A,f):\c(A, B)&\rightarrow &\c(A,C)\\
g&\mapsto& \c(A, f)(g):=f\circ g
\ea
such that the following diagram commutes
 \[\xymatrix{
A\ar[rr]^g\ar[rrdd]_{f\circ g}&&B\ar[dd]^f\\
&&\\
&&C\\
}
\]
\end{example}

\begin{example}
\emph{Free Group Functor}:\\
Given the categories $\Sets$ and $\mathbf{Grp}$ the free group functor is a functor $F:\Sets\rightarrow \mathbf{Grp}$ which assigns, to each set $A\in \Sets$, the free group\footnote{ A group $G$ is called free if there exists a subset $S\subseteq G$, such that any element of G can be uniquely written as a product of finitely many elements of $S$ and their inverses.} generated by $A$ and, to each morphism $f$, the induced homomorphism between the respective groups which coincides with $f$ on the free generators. 
\end{example}

\begin{example}
\emph{Functors between Preorders}\\ 
Given two preorders $(P, \leq)$, $(Q, \leq) $. A covariant functor $F: (P,\leq)\rightarrow (Q,\leq)$ is defined as a covariant functor $F:P\rightarrow Q$ which is order preserving, i.e.
\be
\forall p_1, p_2\in P\text{ if } p_1\leq p_2\text{ then } F(p_1)\leq F(p_2)
\ee
It can be easily seen that indeed the above map satisfies the conditions of being a functor. It follows that, in this case, $F$ is simply a monotone map.
\end{example}

\begin{example}
Given two monoids $ (M,*, 1)$, $(N, *, 1)$,  a covariant functor $F : (M,*, 1)\rightarrow (N, *, 1)$ is such that it maps $M$ to $N$, i.e. $F(M)=N$ since monoids are categories with a single object. The functoriality condition is then defined as follows\footnote{Recall that in a monoid with set $M$ maps have both domain and codomain equal to $M$ (technically these type of morphisms are called endomorphisms), i.e. $m_i:M\rightarrow M$ and represent elements of $M$.  }:

\be
\forall m_1, m_2\in M\;\;F(m_1*m_2)=F(m_1) * F(m_2)\text{ and } F(1)=1
\ee

Hence, a covariant functor between monoids is just a monoid homomorphism.
\end{example}
\begin{example}
Given a group $(G,*, 1)$ (which, as we previously saw, can be considered as a monoid), a covariant functor $F:G\rightarrow Sets$ represents the action of $G$ on a set $X\in Sets$. In particular $F(G)=X$ and each map $g_i:G\rightarrow G$ gets mapped to an endofunction on $X$, i.e. $F(g_i):=g_1* -:X\rightarrow X$. The functorial condition then amounts to the following:
\be
 \forall g_1, g_2\in G\;\;\;F(g_1* g_2)=F(g_1)\circ F(g_2)\text{ and } F(1)=id_X
\ee
Therefore, given any $x\in X$ the above maps imply 
\be
(g_1* g_2)x=g_1\cdot g_1\cdot x\text{ and } 1\cdot x=x
\ee 
Thus $F$ defines an action of $G$ on $X$.

\end{example}
\subsection{Contravariant Functor}
Let us now analyse the other type of functor: \emph{contravariant functor}
\begin{Definition}

A \textbf{contravariant functor} from a category $\c$ to a category $\d$ is a 
map $X:\c\rightarrow\d$ that assigns to each $\c$-object a a 
$\d$-object X(a) and to each $\c$-arrow $f:a\rightarrow b$ a $\d$-arrow
$X(f):X(b)\rightarrow X(a)$, such that the following conditions are satisfied:
\begin{enumerate}
\item $X(1_a)=1_{X(a)}$
\item $X(f \circ g)=X(g) \circ X(f)$ for any $g:c\rightarrow a$ 
\end{enumerate}
\end{Definition}
A diagrammatic representation of a contravariant functor is the following:

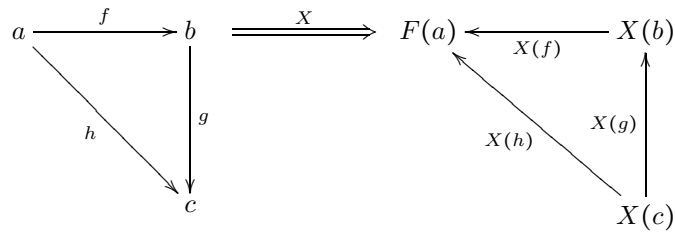
\[\xymatrix{
a\ar@{>}[rr]^f\ar[rrdd]_h&&b\ar[dd]^g\\
&&\\
&&c\\
}
\xymatrix{\ar@{=>}[rr]^X&&}
\xymatrix{
F(a)&&X(b)\ar[ll]^{X(f)}\\
&&\\
&&X(c)\ar[lluu]^{X(h)}\ar[uu]^{X(g)}\\
}\]

\vspace{.2in}
Thus, a contravariant functor in mapping arrows from one category 
to the next reverses
the directions of the arrows, by mapping domains to codomains and vice versa. A contravariant functor is also called a \emph{presheaf}. These types of functors will be the principal objects which we will study when discussing quantum theory in the language of topos theory.
\subsubsection{Examples}
\begin{example}
\textit{Contravariant power set functor} is a functor $\tilde{P}:\Sets\rightarrow \Sets$ which assigns to each set $X$ its power set $\tilde{P}(X)$ and, to each arrow $f:X\rightarrow Y$ the inverse image map $\tilde{P}(f):\tilde{P}(Y)\rightarrow \tilde{P}(X)$, which sends each set $S\in P(Y)$ to the inverse image $f^{-1}(S)\in P(X) $.
\end{example}

\begin{example}
\textit{Contravariant Hom-functor}: \\
For any object $A\in\c$ we define the \emph{contravariant Hom functor} to be the functor $\c(-,A):\c\rightarrow \Sets$, which assigns to each object $B\in \c$  the set of $\c$-arrows $\c(B,A)$ and, to each $\c$-arrow $f:B\rightarrow C$, it assigns the function
\ba
\c(f,A):\c(C,A)&\rightarrow& \c(B,A)\\
g&\mapsto&\c(B,A)(f):=g\circ f
\ea
such that the following diagram commutes
 \[\xymatrix{
B\ar[rr]^f\ar[rrdd]_{g\circ f}&&C\ar[dd]^g\\
&&\\
&&A\\
}
\]
\end{example}

\subsubsection{Characterising Functors}
Irrespectively of whether we are talking about covariant or contravariant functors, there are several properties which distinguish different functors, these are the following:
\begin{Definition}
A functor $F:\c\rightarrow \d$ is called
\begin{enumerate}
\item \emph{Faithful} if 
\be
F:Mor_{\c}(x,y)\rightarrow Mor_{\d}(Fx,Fy)\;\;\;
\forall x,y\in Ob(\c)
\ee
is injective
\item \emph{Full} if 
\be
F:Mor_{\c}(x,y)\rightarrow Mor_{\d}(Fx,Fy)\;\;\;
\forall x,y\in Ob(\c)
\ee
is surjective
\item
\emph{fully faithful}
if 
\be
F:Mor_{\c}(x,y)\rightarrow Mor_{\d}(Fx,Fy)\;\;\;
\forall x,y\in Ob(\c)
\ee
is bijective
\item \emph{Forgetful}
if $F$ takes each $\c$ object to its underlining set, forgetting about any structure which might be present in $\c$, while mapping each $\c$-arrow to itself. Thus, all that is remembered by this forgetful functor is the fact that the $\c$-arrows are set functions.
\item\emph{Essentially surjective}.
A functor $F:\c\rightarrow \d$ is \emph{essentially surjective} (or dense) if each object $d\in \d$ is isomorphic to an object of the form $F(c)$ for some object $c\in\c$.
\item \emph{Embedding} if $F$ is full, faithful, and injective on objects.
\item \emph{An equivalence} if $F$ is full, faithful, and essentially surjective.
\item \emph{An isomorphism} if there exists another functor $G : \d \rightarrow \c$ such that
\be
G \circ F = Id_{\c},\;\;	F\circ  G = Id_{\d} 
\ee

\end{enumerate}
\end{Definition}
It should be noted that the definition of a full and faithful functor only requires that there is a bijection between the morphisms of the categories, not between the objects. In fact if $F:\c\rightarrow \d$ is a full and faithful functor, than it could be the case that:
\begin{itemize}
\item[i] there exists some $y\in Ob(\d)$, such that there is no object $x\in Ob(\c)$ for which $y=F(x)$, i.e. $F$ is not surjetive on objects.
\item [ii] Given two elements $x_1, x_2\in Ob(\c)$, then if $F(x_1)=F(x_2)$ this does not entail that $x_1=x_2$, i.e. $F$ is not injective on objects.
\end{itemize}

\subsubsection{Preservation and Reflection}
So far we have classified functors according to how they act the collection of objects and morphisms seen as sets. However, one can abstract a little more and try understanding how functors behave on more complex structures, such as properties of arrows. In particular consider a property $P$ of arrows, a functor $F : \c \rightarrow \d$ \emph{preserves} $P$ iff
\be
f \text{ satisfies } P\Longrightarrow F(f)\text{ satisfies } P
\ee
$F$ \emph{reflects} $P$ iff
\be
F(f) \text{ satisfies }  P\Longrightarrow f \text{ satisfies } P
\ee

\subsection{Natural Transformations} \label{snat}
So far we have defined categories and maps between them called \emph{functors}. We will now abstract a step more and define maps between functors. These are called \emph{natural transformations}.
 \begin{Definition} 
A \textbf{natural transformation} from $Y:\c\rightarrow \d$ to 
$X:\c\rightarrow \d$ 
is an  assignment of an arrow 
$N:Y\rightarrow X$ that associates to each object A in $\c$ an arrow 
$N_A:Y(A)\rightarrow X(A)$ in $\d$ such that, for any 
$\c$-arrow $f:A\rightarrow B$ the following diagram commutes
\[\xymatrix{
A\ar[dd]^f&&Y(B)\ar[rr]^{N_B}\ar[dd]_{Y(f)}&&X(B)\ar[dd]^{X(f)}\\
&&&&\\
B&&Y(A)\ar[rr]_{N_A}&&X(A)\\
}\]
i.e.
\begin{equation*}
N_A \circ Y(f)=X(f) \circ N_B
\end{equation*}
\end{Definition}
Here $N_A:Y(A)\rightarrow X(A)$ are the components on N, while N is the \emph{natural transformation}.\\
From this diagram it is clear that the two arrows $N_A$ and $N_B$ turn the Y-picture of $f:A\rightarrow B$
into the respective X-picture.
If each $N_A$ ($A\in\c$) is an isomorphism, then $N$ is a \emph{natural isomorphism}
\be
N:Y\xrightarrow{\simeq}X
\ee
\subsubsection{Examples}
\begin{example}
Consider the operation of taking the dual of a vector space defined over some field $K$. This operation is actually a functor as follows
\ba
*:Vect_K&\rightarrow& Vect_K^*\\
V&\mapsto&V^*:=Hom_k(V,K)
\ea
Moreover, given a linear map $f:V\rightarrow W$ we obtain the map $f^*:W^*\rightarrow V^*$ such that $f^*(\phi)=\phi\circ f$ where $\phi\in W^*$. By reiterating this functor we can define a double dual functor as follows:
\ba
**:Vect_K&\rightarrow& Vect_K\\
V&\mapsto& V^{**}
\ea
such that $v^{**}(f)=f(v)$ for $f\in V^*$ and $v\in V$.\\
It is then possible to define a natural transformation between the identity functor 
$1_{Vect_K}:Vect_K\rightarrow Vect_k$ and the double dual functor as follows:
\be
N:1_{Vect_k}\rightarrow **
\ee
whose components are
\[\xymatrix{
V\ar[dd]^f&&1_{Vect_k}(V)\ar[rr]^{N_V}\ar[dd]_{1_{Vect_k}(f)}&&V^{**}\ar[dd]^{f^{**}}\\
&&&&\\
W&&1_{Vect_k}(W)\ar[rr]^{N_W}&&W^{**}\\
}\]
\end{example}
 
 \begin{example}
 Given a map $f : A\rightarrow B$ in a category $\c$, we obtain a natural transformation between covariant Hom-functors as follows:
 \ba
 \c(f,-):\c(B,-)\rightarrow \c(A,-)
 \ea
 such that for each $C\in\c$ we obtain
 \ba\label{ali:def}
 \c(f,C) :\c(B,C)&\rightarrow& \c(A,C)\\
 (g:B\rightarrow C)&\mapsto&\c(f,C)(g):=(g\circ f:A\rightarrow C)
 \ea
 To show that $\c(f,-)$ is indeed a natural transformation we need to show that for all $h:C\rightarrow D$ the following diagram commutes
 \[\xymatrix{
\c(B,C)\ar[rr]^{\c(B, h)}\ar[dd]_{\c(f, C)}&&\c(B, D)\ar[dd]^{\c(f, D)}\\
&&\\
\c(A, C)\ar[rr]^{\c(A,h)}&&\c(A,D)\\
}\]
Chasing the diagram around we have on the one hand
\be
\c(A,h)(\c(f,C)(g))=\c(A,h)(g\circ f)=h\circ (g\circ f)
\ee
and on the other hand
\be
\c(f,D)(\c(B,h)(g))=\c(f,D)(h\circ g)=(h\circ g)\circ f
\ee
 such that 
 \be
 h\circ (g\circ f)=(h\circ g)\circ f
 \ee
 This equality follows from associativity.
  \end{example}
 \begin{example}
 Given a map $f:A\rightarrow B$ in a category $\c$ it is possible to define a natural transformation between contravariant Hom-functors as follows:
 \be
 \c(-,f):\c(-,B)\rightarrow \c(-,A)
 \ee
 That this is indeed a well defined natural transformation. In fact, given any object $C\in \c$ the action of $\c(-,f)$ is
\ba
\c(C,f):\c(C, A)&\rightarrow& \c(C, B)\\
h&\mapsto&\c(C,f)(h)=f\circ h
\ea
While for morphisms $g:C\rightarrow D$ we get
\ba
\c(D,f):\c(D,A)&\rightarrow& \c(D, B)\\
k&\mapsto&f\circ k
\ea
such that the following diagram commutes

 \[\xymatrix{
\c(C,A)\ar[rr]^{\c(C, f)}&&\c(C, B)\\
&&\\
\c(D, A)\ar[rr]^{\c(D,f)}\ar[uu]^{\c(g, A)}&&\c(D,B)\ar[uu]^{\c(g, B)}\\
}\]
Thus $\c(C,f)\circ \c(g, A)(k)=\c(C,f)(k\circ g)=f\circ (k\circ g)$ while $\c(g,B)\circ \c(D,f)(k)=\c(g,B)(f\circ k)=(f\circ k)\circ g$ by associativity $(f\circ k)\circ g=f\circ (k\circ g)$

 \end{example}
 We will state and prove a lemma which is a version of the very important \emph{Yoneda Lemma} which will be analysed in details in subsequent lectures.
 \begin{Lemma}\label{lem:simpleyoneda}
 Given a category $\c$ and two objects $A,B \in \c$, then for each natural transformation $t:\c(A,-)\rightarrow \c(B,-)$ for covariant functors, there exists a unique $f:B\rightarrow A$ in $\c$ such that $t=\c(f,-)$
 \end{Lemma}
 \begin{Proof}
 Let us define
 \be
\Big(f:B\rightarrow  A\Big):=t_A(id_A)
 \ee
Since $t:\c(A,-)\rightarrow \c(B,-)$  is a natural transformation, given any arrow $g:A\rightarrow C$ we have the following commutative diagram
 \[\xymatrix{
\c(A,A)\ar[rr]^{\c(A, g)}\ar[dd]_{t_A}&&\c(A, C)\ar[dd]^{t_C}\\
&&\\
\c(B, A)\ar[rr]^{\c(B,g)}&&\c(B,C)\\
}\]
 Chasing the diagram around we have 
 \be
t_C( \c(A, g)(id_A))=t_C(g)
 \ee
 On the other hand
 \be
 \c(B,g)(t_A(id_A))=\c(B,g)f=g\circ f
 \ee
 thus, from the requirement of commutativity 
 \be
 t_C(g)=g\circ f
 \ee
 However from the definition in equation \ref{ali:def} we know that 
 \be
 \c(f,C)(g):=g\circ f
 \ee 
 thus
 \be
 t_C=\c(f,-)
 \ee
 To prove uniqueness we need to show that if $\c(f, -)=\c(f^{'},-)$ then $f=f^{'}$. To this end consider
 \ba
 f=id_A\circ f=\c(f,A)(id_A)=\c(f^{'},A)(id_A)=id_A\circ f^{'}=f^{'}
 \ea

 \end{Proof}
\subsection{Equivalence of Categories}
Now that we have defined maps between categories, i.e. functors and maps between functors, i.e. natural transformations, it is possible to compare two categories and see if they are equivalent or not. To this end we need the notion of isomorphic functors.
\begin{Definition}
Two functors $F,G:\c\rightarrow \d$ are said to be naturally isomorphic if there exists a natural transformation $\eta:F\rightarrow G $, which is invertible.  
\end{Definition}
We can now define equivalent categories in terms of naturally isomorphic functor. 
\begin{Definition}
Two categories $\c$ and $\d$ are said to be equivalent if there exists functors $F:\c\rightarrow \d$ and $G:\d\rightarrow\c$, such that the functors $F\circ G\simeq id_{\d}$ and $G\circ F\simeq id_{\c}$ are naturally isomorphic to the identities. In this case $F$ is called an \emph{equivalence of categories}.
\end{Definition}

\chapter{Lecture 5/6}

In this lecure I will first of all introduce the category of functors whose objects are functors (lecture 4) and whose morphisms are natural transformations (lecture 4). Such a category is denoted as $D^C$ for $C$ and $D$ being categories themselves. As we will see later on in the course, for a particular choice of $D$ and $C$ we obtain the quantum topos.
I will then give the axiomatic definition of what a topos is. In particular I will focus on two main objects which are present in a topos, these are: 
\begin{enumerate}
\item [i)] Sub-object Classifier;
\item [ii)] Internal logic: Heyting algebra.
\end{enumerate}
These will be very important when defining the topos version of quantum theory.

\section{The Functor Category }
We will now introduce a type of category which is very important for the topos formulation of quantum theory since i) it is actually a topos, ii) for an appropriate choice of base category it will be the topos in terms of which quantum theory is defined.

The functor category is, in a way more abstract than any of the categories we have encountered so far since it has as objects contravariant functors, and as maps natural transformations. So it is one level higher in abstraction of the `standard' category, which had as objects simpler elements with not so much structure and as maps simpler function between these elements. On the other hand the elements in the functor categories are complicated object being themselves maps between categories, thus they carry a lot of structure. The maps between these objects are now required to preserve such complicated structures. However, on a general level, we still simply have a collection of objects with maps between them, the only difference is that now  the hidden level of complexity has to be taken into consideration when defining any categorical construction. So, for example, as we will see in the examples below,  when defining the product, the equaliser, etc, the structure of the individual elements of the category has to be taken into consideration.

\noindent
The definition of the functor category is as follows:
\begin{Definition}
Given two categories $\c$ and $\d$, the functor category $\d^{\c}$ has: 
\begin{itemize}
\item \emph{Objects}: all functors of the form $F:\c\rightarrow \d$.
\item\emph{Morphisms}: natural transformations between the above mentioned functors.
\end{itemize}
\end{Definition}
Given two natural transformations $N_1:F\rightarrow G$ and $N_2:G\rightarrow H$ in $\d^{\c}$ composition is defined as follows
\be
N_2\circ N_1:F\rightarrow H
\ee
such that, for each $C\in\c$ the individual components are
\ba
(N_2\circ N_1)_C&:&F(C)\rightarrow H(C)\\
&:&F(C)\xrightarrow{(N_1)_C}G(C)\xrightarrow{(N_2)_C}H(C)
\ea
Composition is associative

Of particular importance for us is the case in which $\d$ is the category $\Sets$. In fact $\Sets^{\c}$, for a particular $\c$, will be the category which we will use to describe quantum theory. In particular, for quantum theory we will be using the presheaves $\Sets^{\c^{\op}}$\footnote{Note that $\c^{op}$ represents the 
opposite of the category $\c$. Objects in $\c^{op}$ are the same as the objects in 
$\c$, while the morphisms are the inverse of the morphisms in $\c$,
i.e. $\exists$ a $\c^{op}$-morphisms $f:A\rightarrow B$ iff $\exists$ a $\c$-
morphisms $f:B\rightarrow A$. Thus the category $\Sets^{\c^{\op}}$ has as objects covariant functors from $\c^{\op}$ to $\Sets$ or, equivalently, contravariant functors from $\c$ to $\Sets$.} for an appropriate $\c$.

So let us analyse $\Sets^{\c^{\op}}$
\section{The Functor Category with Domain Sets}
We will now describe in detail the functor category $\Sets^{\c^{\op}}$ for some category $\c$. The reason we are interested in this category is two fold:\\
i) It is actually a topos\\
ii) It will be the topos in which quantum theory will be defined.

Given a contravariant functor  between a category $\c$ and $\mathbf{Sets}$, we can form a 
category $\Sets^{\c^{\op}}$ such that we have the following:
\begin{itemize}
\item Objects:
all contravariant functors $P:\c\rightarrow \Sets$ 
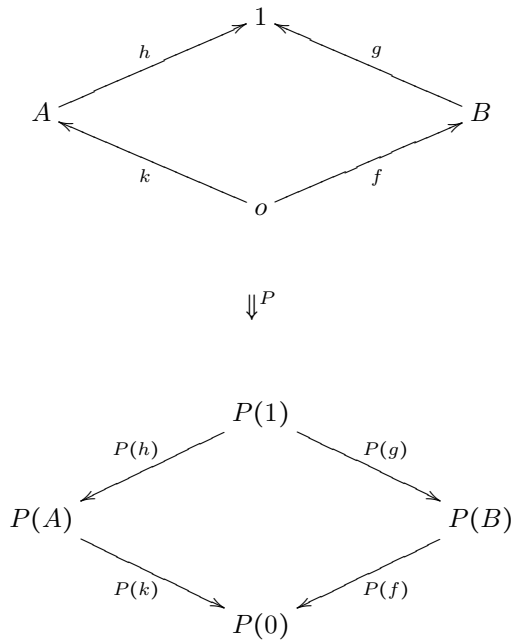
\[\xymatrix{
&&1&&\\
A\ar[rru]^h&&&&B\ar[llu]_g\\
&&o\ar[rru]_f\ar[llu]^k&&\\
&&\Downarrow^P&&\\
&&P(1)\ar[rrd]^{P(g)}\ar[lld]_{P(h)}&&\\
P(A)\ar[rrd]_{P(k)}&&&&P(B)\ar[lld]^{P(f)}\\
&&P(0)&&\\
}\]
\item Arrows: 
all natural transformation $N:P\rightarrow P^{'}$ between contravariant functors such that, given a function 
$f:D\rightarrow C$ in $\c$ the following diagram commutes
\[\xymatrix{
P(C)\ar[rr]^{P(f)}\ar[dd]_{N_C}&& P(D)\ar[dd]^{N_D}\\
&&\\
P^{'}(C)\ar[rr]_{P^{'}(f)}&& P^{'}(D)\\
}\]
\end{itemize}
The morphisms satisfy the following conditions:
\begin{itemize}
\item \emph{Identity maps} for each object X in $\Sets^{\c^{op}}$ are identified with natural transformations  
$i_X$, 
whose components $i_{X_A}$ are the identity maps of the set X(A) in $\Sets$. 
\item \emph{Composition of maps} in $\Sets^{\c^{op}}$ :\\
Consider the functors X,Y and Z that belong to $\Sets^{\c^{op}}$, such that there exist maps\footnote{Here it is intended natural transformations, but we will often simply call them maps. The specification of what type of maps we are considering should be clear from the context.} 
$X\xrightarrow{N}Y$
and $Y\xrightarrow{M}Z$ between them. We can then form a new map $X\xrightarrow{M o N}Y$, whose 
components would be $(M \circ N)_A=M_A \circ N_A$, i.e. graphically we would have 

\[\xymatrix{
X(A)\ar[rrrr]^{X(f)}\ar[dd]_{N_A}&&&&X(B)\ar[dd]^{N_B}\\
&&&&\\
Y(A)\ar[rrrr]^{Y(f)}\ar[dd]_{M_A}&&&&Y(B)\ar[dd]^{M_B}\\
&&&&\\
Z(A)\ar[rrrr]^{Z(f)}&&&&Z(B)\\
}\]

\end{itemize}

Depending on the choice $\c$, there are a number of relevant elements of the category $\Sets^{\c^{\op}}$ in quantum theory. Here we will only give a few example of such \emph{presheaves} (see  \cite{isham1}, \cite{isham2}, \cite{isham4}). A more in depth analysis will be given in subsequent lectures.

\subsection{Spectral Presheaf on the Category of Self-Adjoint Operators with Discrete Speactra}
\begin{Definition}  \label{def:sigma}
 \textbf{Spectral Presheaf} on $\mathcal{O}_d$ (subcategory of $\mathcal{O}$, in which the 
operators have discrete spectra)\footnote{The condition of the spectrum being discrete implies
that given a Borel function $f:\sigma(\hat{A})\subseteq \Rl\rightarrow\Rl$, then 
$\sigma(f(\hat{A}))=f(\sigma(\hat{A}))$.} $\Sigma:\mathcal{O}_d\rightarrow\Sets$ is defined such that:
\begin{enumerate}
\item Objects $\hat{A}\in\mathcal{O}$ get mapped to $\Sigma(\hat{A})=\sigma(\hat{A})$ where 
$\sigma(\hat{A})$ is the spectrum of $\hat{A}$.
\item Morphisms $f_{\mathcal{O}}:\hat{B}\rightarrow\hat{A}$ in $\mathcal{O}_d$, such that 
$\hat{B}=f(\hat{A})$ ($f:\sigma(\hat{A})\subseteq\Rl\rightarrow\Rl$), gets mapped to $\Sigma(f_{\mathcal{O}}):\Sigma(\hat{A})\rightarrow\Sigma(\hat{B})$,
which is equivalent to  $\Sigma(f_{\mathcal{O}}):\sigma(\hat{A})\rightarrow\sigma(\hat{B})$
and is defined by $\Sigma(f_{\mathcal{O}})(\alpha):=f(\alpha)$ for all $\alpha\in\sigma(\hat{A})$.
\end{enumerate}
\end{Definition}
In order to prove that $\Sigma$, as defined above is indeed a presheaf, we need to prove that, given
any function $g_{\mathcal{O}}:\hat{C}\rightarrow\hat{B}$ such that $\hat{C}=g(\hat{B})$
then, the following equation is satisfied:
\begin{equation*}
\Sigma(g_{\mathcal{O}})\circ\Sigma(f_{\mathcal{O}})=\Sigma( f_{\mathcal{O}}\circ g_{\mathcal{O}})
\end{equation*}
\begin{proof}
If we consider the composite function $f_{\mathcal{O}}\circ g_{\mathcal{O}}=h_{\mathcal{O}}:\hat{C}\rightarrow\hat{A}$
from the definition of $\Sigma$ we have
\begin{align*}
\Sigma(f_{\mathcal{O}})&:\sigma(\hat{A})\rightarrow\sigma(\hat{B})\\
\Sigma(g_{\mathcal{O}})&:\sigma(\hat{B})\rightarrow\sigma(\hat{C})\\
\Sigma(f_{\mathcal{O}}\circ g_{\mathcal{O}})=\Sigma(h_{\mathcal{O}})&:
\sigma(\hat{A})\rightarrow\sigma(\hat{C})
\end{align*}
therefore
\begin{align*}
\Sigma(f_{\mathcal{O}}o g_{\mathcal{O}})(\alpha)&=\Sigma(h_{\mathcal{O}})(\alpha)\\
&=h(\alpha)\\
&=(g \circ f)(\alpha)\\
&=g \circ(\Sigma(f_{\mathcal{O}})(\alpha))\\
&=[\Sigma(g_{\mathcal{O}})\circ\Sigma(f_{\mathcal{O}})](\alpha)
\end{align*}

\end{proof}
\subsubsection{Example of Spectral Presheaf}
Let us consider a simple category whose elements are defined by
\begin{align*}
\hat{A}&=a_1P_1+a_2P_2+a_3P_3\\
\hat{B}&=b_1(P_1\vee P_3)+b_2P_2\\
\hat{C}&=c_1P_3+c_2(P_1\vee P_2)
\end{align*}
This can be represented in the following diagram:
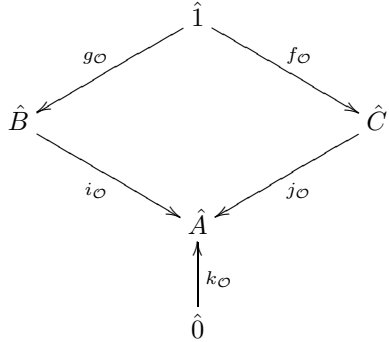
\[\xymatrix{
&&\hat{1}\ar[rrd]^{f_{\mathcal{O}}}\ar[lld]_{g_{\mathcal{O}}}&&\\
\hat{B}\ar[rrd]_{i_{\mathcal{O}}}&&&&\hat{C}\ar[lld]^{j_{\mathcal{O}}}\\
&&\hat{A}&&\\
&&\hat{0}\ar[u]_{k_{\mathcal{O}}}&&
}\]
The elements of the presheaf $\Sigma$ are the following:
\begin{align*}
\sigma(\hat{A})&=\{a_1,a_2,a_3\}\\
\sigma(\hat{B})&=\{b_1,b_2\}\\
\sigma(\hat{C})&=\{c_1,c_2\}
\end{align*}
From definition \ref{def:sigma} it follows that, for example, the map $j_{\mathcal{O}}:\hat{C}\rightarrow\hat{A}$ 
gets mapped to $\Sigma(j_{\mathcal{O}}):\sigma(\hat{A})\rightarrow\sigma(\hat{C})$ such that, component-wise, we get the following mapping:
\begin{align*}
\Sigma(j_{\mathcal{O}})(a_3)&=c_1\\
\Sigma(j_{\mathcal{O}})(a_1)&=c_2\\
\Sigma(j_{\mathcal{O}})(a_2)&=c_2
\end{align*} 

\subsection{The Dual Presheaf on $W$ }
Another presheaf which can de defined in quantum theory is the dual presheaf on the category $W$ of Boolean algebras under sub-algebra inclusion. In particular we have:
\begin{Definition}
The dual presheaf on $W$ is the contravariant functor $D : W \rightarrow \Sets$ defined as follows:
\begin{itemize}
\item  On objects: $D(W)$ is the dual of $W$; thus it represents the set $Hom(W,\{0,1\})$ of all homo- morphisms from the Boolean algebra $W$ to the Boolean algebra $\{0, 1\}$.
\item On morphisms: given $i_{W_2,W_1} : W_2\rightarrow W_1$ then $D(i_{W_2,W_1}) : D(W_1)\rightarrow D(W_2)$ is defined by $D(i_{W_2W_1})(\chi) := \chi|_{W_2}$ where $\chi|_{W_2}$	denotes the restriction of $\chi\in D(W_1)$ to the sub-algebra $W_2\subseteq  W_1$.
\end{itemize}
\end{Definition}

\section{Topos Theory}
In this lecture we will analyse, in details, what a topos is. The very hand wavy definition of a topos is that of a category with extra properties. What these extra properties are we will see later on, the important thing for the time being is what this extra properties imply. The implications of these extra properties are that they make a topos ``look like" $\Sets$, in the sense that any mathematical operation which can be done in set theory can be done in a general topos.

In the previous lecture we gave an account (not complete) on how set theoretical structures can be given an external characterisation through category theory. Although it is true that all of the set theoretic constructions can be defined in a categorical language, however, it is not true that all categories have all these set theoretical constructions.

\noindent
A topos, on the other hand, is a category for which all the categorical versions of set constructs exist and are well defined. It is precisely in this sense that a topos ``looks like" $\Sets$.

Before giving the axiomatic definition of what a topos is we, first of all, need to define certain extra constructs of category theory, which are required to be present for a given category to be a topos.  
\subsection{Exponentials}
Given two sets $A$ and $B$, let us imagine we would like to group all arrows between them, i.e.,  we would then define the set $\{f_i:A\rightarrow B\}$. An important property of the set $\{f_i:A\rightarrow B\}$ is that
\be
\text{ if } A, B\in \Sets \text{ then }\{f_i:A\rightarrow B\}\in \Sets
\ee
We will call this object (which in this particular situation is the set $\{f_i:A\rightarrow B\}$) an \emph{exponential} and we will denoted it as $B^A$. 

We now would like to abstract the characterisation of such an object for a general category. Thus we would like to define an object $B^A$ with the properties i) if $A\in\c$ and $B\in \c$ then $B^A\in\c$; ii) it represents a certain relation between $A$ and $B$.

Since in categorical language objects are defined according to the relations with other objects, we will define $B^A$ in terms of what it `does' operationally. To this end let us consider all three objects involved: $A$, $B$, $B^A$. Taking inspiration from what they actually represent in $\Sets$, a possible relation between then can be defined as follows:
\ba
ev:B^A\times A&\rightarrow &B\\
(f_i, a)&\mapsto&ev(f_i,a):=f_i(a)
\ea 
This definition seems very plausible, however, for it to make sense in the categorical world it has to be \emph{universal}, in the sense that any other arrow $g:C\times A\rightarrow B$ will factor through $ev$ in a unique way. Thus, there will exist a unique arrow 
$\hat{g}:C\rightarrow B^A$, which makes the following diagram commute
\[\xymatrix{
B^A\times A\ar[rr]^{ev}&&B\\
&&\\
C\times A\ar@{-->}[uu]^{\hat{g}\times 1_A}\ar[rruu]_g&&\\
}\]
where 
\ba
\hat{g}:C&\rightarrow& B^A\\
c&\mapsto&\hat{g}(c):=g(c,-)
\ea
such that
\ba
g(c,-):A&\rightarrow &B\\
a&\mapsto&g(c,a)
\ea
Thus the map $\hat{g}$ assigns to any $c\in C$ a function $A\rightarrow B$ by taking $g$ and keeping the first term fixed at $c$, while ranging over the elements of $A$. Thus for each $c\in C$ we have that $\hat{g}=g_c:=g(c,-)$.

We are now ready to give the abstract categorical definition of what an exponential is.
\begin{Definition} 
Given two $\c$-objects $A$ and $B$, their \emph{exponentiation} is a $\c$-object $B^A$ together with an \emph{evaluation map} $ev:B^A\times A\rightarrow B$ with the property that, given any other $\c$-object $C$ and 
$\c$-arrow $g:C\times A\rightarrow B$, there exists a unique arrow 
$\hat{g}:C\rightarrow B^A$, such that the following diagram commutes
\[\xymatrix{
B^A\times A\ar[rr]^{ev}&&B\\
&&\\
C\times A\ar@{-->}[uu]^{\hat{g}\times 1_A}\ar[rruu]_g&&\\
}\]
\end{Definition} 

Therefore for any $\langle c,a\rangle\in C\times A$  we get that 
\ba
ev\circ (\hat{g}\times 1_A)(\langle c,a\rangle= ev(\langle\hat{g}(c),a\rangle)= g_c(a)=g(\langle c, a\rangle)
\ea

Now that we have abstractly defined what an exponentiation is, we would like to know what its elements are. Remember that we started off with the example in $\Sets$ and in that case we knew that the objects in $B^A$ are maps $\{f_i:A\rightarrow B\}$. We then abstracted from this particular example and defined a general notion of an exponential in terms of the universal property of the exponential map. However, such definition did not relay on what type of elements $B^A$ had, if any. We now would like to go back full circle and see what we can say about the elements of $B^A$. We know from previous lectures that an element in any object $C$ is identified with a map $1\rightarrow C$. This correspondence will be used in the following definition.
\begin{Definition} \label{def:one}
Objects of $B^A$ are in one-to-one correspondence with maps of the form $f:A\rightarrow B$.
To see this let us consider the following commuting diagram
\[\xymatrix{
B^A\times A\ar[rr]^{ev}&&B\\
&&\\
1\times A\ar@{-->}[uu]^{\hat{f}\times 1_A}\ar[rruu]_{f^{'}}&&\\
}\]
where $f^{'}:1\times A\rightarrow B$ is unique given $f$.
But $1\times A\equiv A$, therefore to each element of $B^A$ ($\hat{f}:1\rightarrow B^A$)  there corresponds a unique function $f:A\rightarrow B$.
\end{Definition}
The above definition of \emph{exponentials} allows us to define a subset of categories called \emph{cartesian closed categories}. The precise definition is as follows:
\begin{Definition}
A category $\c$ is said to be cartesian closed (CCC) if it has terminal object, products and exponentials.
\end{Definition}
Examples of CCC are $\Sets$ and Boolean algebras, seen as categories. For those who are interested in logic, a Boolean algebra is defined as a CCC as follows:
\begin{itemize}
\item Products are given by conjunctions
\be
A \wedge B
\ee
\item Exponentials are implications 
\be
A\Rightarrow B :=\neg A\vee B
\ee
\item Evaluation is Modus Ponens,
\be
(A\Rightarrow B)\wedge A\leq B 
\ee
\item Universality is the Deduction Theorem,
\be
C \wedge A\leq B \Leftrightarrow C\leq A\Rightarrow  B .
\ee
\end{itemize}

\subsubsection{Examples of Exponentiation}
In the following, we will give examples of cartesian closed categories and how the exponentials in each of them are formed.
\begin{example}
In $\mathbf{Sets}$, given two objects $A$ and $B$, the exponential $B^A$ is defined as follows:
\be
B^A=\{f|f\text{is a function from A to B}\}
\ee
In this case the evaluation map 
would be the
following: $$ev(\langle f,x\rangle)=f(x) \text{ with } x\in A$$ 
\end{example}
\begin{example}
Consider the category \textbf{Finord} of all finite ordinals. In such a category we have as objects numbers $0,1,\cdots,n$ where 
\ba
0&=&\emptyset\\
1&=&\{0\}\\
2&=&\{0,1\}\\
3&=&\{0,1,2\}\\
&\vdots&\\
n&=&\{1,2,\cdots, n-1\}
\ea
The maps are then simply the set functions between these cardinals. Given two such elements $n$ and $m$, then the exponential would be $n^m$ where $m$ is the exponent. Such an element should be considered as a finite ordinal with $n^m$ elements in it.
\end{example}
\subsection{Pullback}\label{spull}
We will now define another construction which is present in any topos. This is the notion of pullback or fiber product.
\begin{Definition} 
A \textbf{pullback} or \textbf{fibered product} of a pair of functions $f:A\rightarrow C$ and $g:B\rightarrow C$ (with common codomain) in a category 
$\c$, is a pair of $\c$-arrows $h:D\rightarrow A$ and $k:D\rightarrow B$, such that 
the following conditions are satisfied:
\begin{enumerate}
\item $f \circ h=g \circ k$ i.e the following diagram commutes

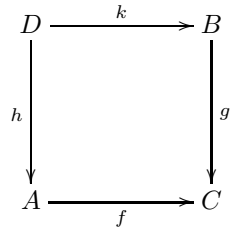
\[\xymatrix{
D\ar[rr]^{k}\ar[dd]_{h}&&B\ar[dd]^{g}\\
&&\\
A\ar[rr]_f&&C\\
}\]
  
One usually writes $D=A\times_C B$. 
\item Given two functions $i:E\rightarrow A$ and $j:E\rightarrow B$, where $f \circ i=g \circ j$ 
then, there exists a unique $\c$-arrow $l$ from $E$ to $D$ such that the outer rectangle of 
the following diagram commutes
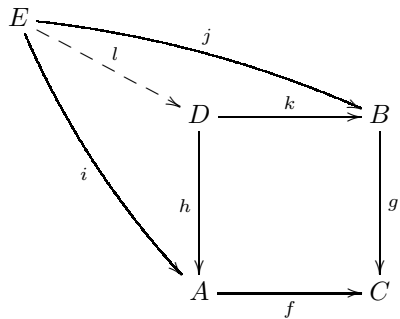
\[\xymatrix{
E\ar@/^/[rrrrd]^j\ar@{-->}[rrd]^l\ar@/_/[rrddd]_i&&&&\\
&&D\ar[rr]^k\ar[dd]_h&&B\ar[dd]^g\\
&&&&\\
&&A\ar[rr]_f&&C\\
}\]

i.e.
\begin{equation*}
i=h \circ l\hspace{.2in}j=k \circ l
\end{equation*} 
We then say that $f$ (respectively $g$) has been pulled back along $g$ (respectively $f$). 
\end{enumerate}
\end{Definition}
A pullback of a pair of $\c$-arrows $A\xrightarrow{f} C\xleftarrow{g} B$ is a limit of the diagram
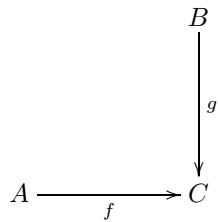
\[\xymatrix{
&&B\ar[dd]^g\\
&&\\
A\ar[rr]_f&&C\\
}\]
In fact a cone\footnote{Strictly speaking one has two cones $C\xleftarrow{j}D\xrightarrow{i}A$ and $C\xleftarrow{j}D\xrightarrow{h}B$, but when composing the diagrams it turns out that $j=g\circ h=f\circ i$, thus we can omit the $j$ map. } for this diagram is a pair of $\c$-arrows $A\xleftarrow{i}D\xrightarrow{h}B$ which compose to give the commutative diagram
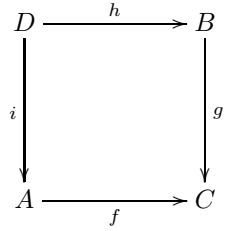
\[\xymatrix{
D\ar[rr]^h\ar[dd]_i&&B\ar[dd]^g\\
&&\\
A\ar[rr]_f&&C\\
}\]
A limiting cone is such that given any other two cones $A\xleftarrow{k}E\xrightarrow{l}D$ and $B\xleftarrow{j}E\xrightarrow{l}D$ it factors uniquely through it, i.e. it gives rise to the following commuting diagram.
\[\xymatrix{
E\ar@/^/[rrrrd]^j\ar@{-->}[rrd]^l\ar@/_/[rrddd]_k&&&&\\
&&D\ar[rr]^h\ar[dd]_i&&B\ar[dd]^g\\
&&&&\\
&&A\ar[rr]_f&&C\\
}\]

\subsubsection{Examples of Pullback}
\begin{example}
If $A, C, D$ and $B$ are sets, then
\be
D=A \times_C B=\{(a,b)\in A \times B | f(a)=g(b)\}\subseteq A \times B\}
\ee
with maps 
\ba
k:A \times_C B&\rightarrow& B\\
(a,b)&\mapsto& b
\ea
and 
\ba
h:A \times_C B&\rightarrow& A\\
(a,b)&\mapsto& a
\ea
satisfies the conditions of being a pullback.\\
\begin{proof}
Given a set E with maps $j:E\rightarrow B$ and $i:E\rightarrow A$, then the map $l:E\rightarrow A \times_C B$, ($e \longmapsto (i(e),j(e)$) would make the diagram 
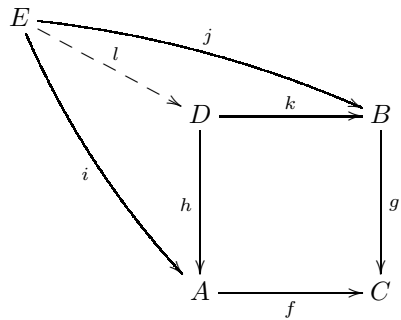
\[\xymatrix{
E\ar@/^/[rrrrd]^j\ar@{-->}[rrd]^l\ar@/_/[rrddd]_i&&&&\\
&&D\ar[rr]^k\ar[dd]_h&&B\ar[dd]^g\\
&&&&\\
&&A\ar[rr]_f&&C\\
}\]
commute. In fact we have the following identities for all $e\in E$, $h\circ l(e)=i(e)$ and $k\circ l(e)=j(e)$. 
\end{proof}
Moreover $l$ is unique since, given any other map $m:E\rightarrow A \times_C B$, such that $h\circ m=i$ and $k\circ m=j$, then for all $e\in E$ the following holds:
\be
l(e)=(i(e),j(e))=(h(m(e)),k(m(e)))=m(e)=(h(a,b),k(a,b))=(a,b)
\ee
Therefore $l$ is unique.
\end{example}

\subsection{Pushouts} \label{spus}
As usual, any notation in category theory has a dual, thus we will now define the dual of a pullback which is a push out.
\begin{Definition} 
A \textbf{Pushout} or \textbf{fibered co-product} of a pair of functions $f:A\rightarrow B$ and $g:A\rightarrow C$ in a category 
$\c$ is a pair of $\c$-arrows $h:B\rightarrow D$ and $k:C\rightarrow D$, such that 
the following conditions are satisfied:
\begin{enumerate}
\item $h \circ f=k \circ g$, i.e the following diagram commutes

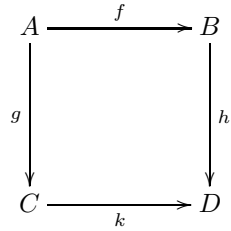
\[\xymatrix{
A\ar[rr]^{f}\ar[dd]_{g}&&B\ar[dd]^{h}\\
&&\\
C\ar[rr]_k&&D\\
}\]
  
One usually writes $D=C+_A B$. 
\item Given two functions $i:B\rightarrow E$ and $j:C\rightarrow E$, where $i \circ f=j\circ g$, 
then there exists a unique $\c$-arrow $l$ from $D$ to $E$, such that the outer rectangle of 
the following diagram commutes
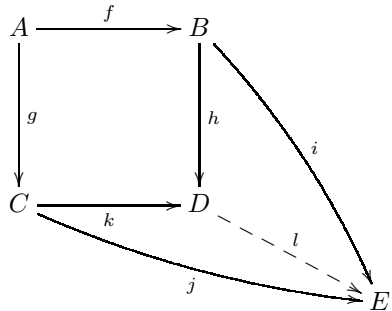
\[\xymatrix{
A\ar[rr]^f\ar[dd]^g&&B\ar[dd]^h\ar@/^/[rrddd]^i&&\\
&&&&\\
C\ar[rr]_k\ar@/_/[rrrrd]_j&&D\ar@{-->}[rrd]^l&&\\
&&&&E\\
}\]

i.e.
\begin{equation*}
i=l \circ h\hspace{.2in}j=l \circ k
\end{equation*} 
We then say that f (respectively g) has been pushed out along g (respectively f).
\end{enumerate}
\end{Definition}
\subsubsection{Examples}
\begin{example}
In $\Sets$, given three sets $A, B, C$, the set $D=C+_A B$ always exists and it is identified with the disjoint union of $C$ and $B$, i.e.
\begin{equation*}
D=C\amalg B:=\Big\{(x,t)\in (C\cup B)\times\{0,1\}|\begin{cases}x\in B& if\hspace{.1in}t=0\\
x\in C& if\hspace{.1in}t=1
\end{cases}\Big\}
\end{equation*}
where, in this case, the arrows $h$ and $k$ are defined as follows:
\ba
h:B&\rightarrow& C\amalg B\\
b&\mapsto& (b,0)
\ea and 
\ba
k:C&\rightarrow& C\amalg B\\
c&\mapsto& (c,1)
\ea
We now want to prove that $C\amalg B$, as defined, does indeed satisfy the conditions of a pushout.
\begin{proof}:\\
Given a set $E$ and two maps $j:C\rightarrow E$, $i:B\rightarrow E$, we define the map $l:C\amalg B\rightarrow E$, such that $(c,1)\longmapsto j(c)$ and $(b,0)\longmapsto i(b)$.\\
It is then easy to see that the diagram 
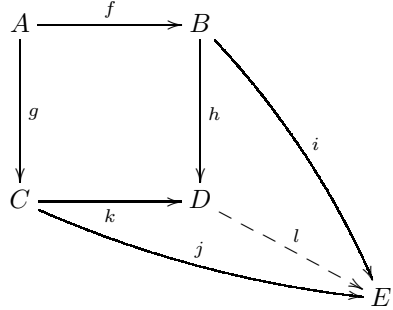
\[\xymatrix{
A\ar[rr]^f\ar[dd]^g&&B\ar[dd]^h\ar@/^/[rrddd]^i&&\\
&&&&\\
C\ar[rr]_k\ar@/_/[rrrrd]^j&&D\ar@{-->}[rrd]^l&&\\
&&&&E\\
}\]
commutes. In fact we have the following:
\be
(l\circ h)(b)=l(b,0)=i(b)
\ee and 
\be
(l\circ k)(b)=l(c,1)=j(c)
\ee
The second step in the proof is showing that the map $l$ is unique. In fact, given another map $m:C\amalg B\rightarrow E$, such that $m\circ h=i$ and $m\circ k=j$, then we would have the following equality:
\be
l(b,0)=i(b)=m(h(b))=m(b,0)
\ee and 
\be 
l(c,1)=j(c)=m(k(c))=j(c)
\ee
This shows that $l$ is unique.
\end{proof}
\end{example}

\subsection{Sub-Objects}\label{sssub}
Everyone is familiar with the notion of a subset in the category $\Sets$. We now would like to generalise this notion and describe it in categorical language, i.e. in terms of relation.

Let us start with $\Sets$ and see how much we can abstract from the already known definition of subset. Consider two sets $A$ and $B$ such that $A\subseteq B$. This means that there is an inclusion map $f:A\hookrightarrow B$. In categorical language the map $f$ is monic. So if we went the other way round and considered a monic arrow $f:A\rightarrowtail B$, this would determine a subset of $B$, namely $Im f:=\{f(x)|x\in A\}$. Thus $Im f\subseteq B$ and $Im f\simeq A$.

What this means is that the domain of a monic arrow is isomorphic to a subset of the codomain of the arrow, i.e. up to isomorphisms the domain of a monic arrow is a subset of the codomain.

\noindent
As you might have noticed, going from a subset of a set to the set itself requires a change of type, we are in a different object (namely going from domain object to codomain object). This implies that in the categorical version of subset, there is no feasible way to say that the same element $x$ is in both a set and a subset of the set. All we can say is that two elements are isomorphic. This might seem striking at first since we are used to think in $\Sets$-language terms, where object are determined/defined by the elements which comprise them thus, saying that two sets are the same means saying that the elements which compose them are the same. However, in categorical language ``the same" now becomes ``are isomorphic" since all is defined in terms of relations, i.e. arrows between objects. 

So, preliminary, we will define a sub-object of a $\c$ object $B$ as an arrow in $\c$ which is monic and which has codomain $B$.

\noindent
However, this is not the end of the story, since we are not taking into account that some objects might be isomorphic which in categorical language means the same, thus, we should only really consider them once. Therefore our definition of a sub-object should take into account the existence of equivalent objects.
In order to do so we need the following definition:
\begin{Definition}
Given an arrow $ f : A\rightarrow  B$ in some category $\c$, if for some arrow $g:C\rightarrow B$ in $\c$, there exists another arrow $h:A\rightarrow C$ in $\c$ such that $f=g\circ h$, then we say that $f$ factors through $g$, since $f=g\circ h$ can be solved for $h$.
\end{Definition}
Diagrammatically what factorising means is that the following diagrams commute
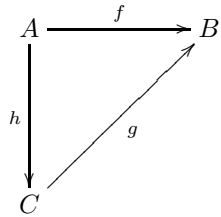
\[\xymatrix{
A\ar[rr]^f\ar[dd]_h&&B\\
&&\\
C\ar[rruu]_g&&\\
}\]

We now consider two sub-objects $g:C\rightarrow B$ and $f:A\rightarrow B$, such that they both factor through each other, i.e., the following diagram commutes
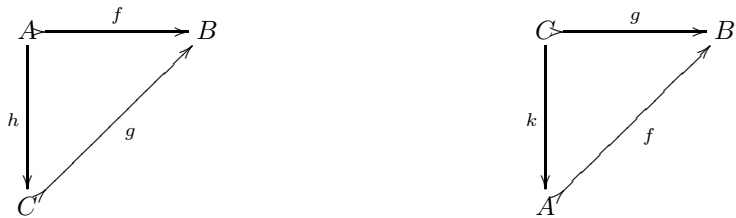
\[\xymatrix{
A\ar@{>->}[rr]^f\ar[dd]_h&&B&&&&C\ar@{>->}[rr]^g\ar[dd]_k&&B\\
&&&&&&&&\\
C\ar@{>->}[rruu]_g&&&&&&A\ar@{>->}[rruu]_f&&\\
}\]

We then consider such sub-objects as equivalent. This leads to the definition of the following equivalence relation.
\begin{Definition}\label{def:equi}
Given two monic arrows $f, g$ with the same codomain we say that they are equivalent $f\sim g$ iff they factor though each other.
\end{Definition}

Given the above we are now ready to define the categorical version of a sub-object.
\begin{Definition}
In a category $\c$, a sub-object of any object in $\c$ is an equivalence class of monic arrows under the equivalence condition $\sim $ defined in \ref{def:equi}. We will denote such an equivalence class as $[f]$. The sub-object is a proper sub-object if it does not contain $id_{\c}$.
\end{Definition}

This definition of sub-objects allows us to define the collection of all sub-objects of a given object as a poset under subset inclusion. In particular, we know from set theory that the collection of all subsets of a given set $X$ is a poset, such that there is an arrow between any two such subsets $A$ and $B$ iff $A\subseteq B$. Thus, diagrammatically, we have

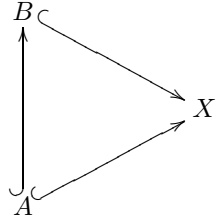
\[\xymatrix{
B\ar@{^{(}->}[drr]&&\\
&&X\\
A\ar@{^{(}->}[urr]\ar@{^{(}->}[uu]&&\\
}\]
 Therefore abstracting such a definition to categorical language we say that two monic arrows $f:A\rightarrowtail B$ and $g:C\rightarrowtail B$ are such that $f\subseteq g$ iff the following diagram commutes
 
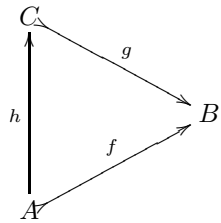
\[\xymatrix{
C\ar@{>->}[drr]^g&&\\
&&B\\
A\ar@{>->}[urr]^f\ar[uu]^h&&\\
}\]
i.e. $f=g\circ h$. (h is monic\footnote{ Given $f:A\rightarrow B$, and $h, k:D\rightarrow A$, assume that $f\circ h=f\circ k$. Then consider $g:B\rightarrow C$ such that $(g\circ f)$ is monic. It follows that
\be
(g\circ f)\circ h=g\circ (f\circ h)=g\circ (f\circ k)=(g\circ f)\circ k\text{ implies that } h=k
\ee
}

If we now consider all sub-objects of a given object $B$ 
\be
Sub(B):=\{[f]|f \text{ is monic and } \text{cod}f=B\}\text{ where }[f]=\{g|f\simeq g\}
\ee
then such a collection of sub objects forms a poset under subset inclusion defined by 
\be
[f]\subseteq [g]\text{ iff } f\subseteq g
\ee
\begin{proof}\label{pro:poset}\hspace{.1in}\\
\begin{itemize}
\item Reflexive: $[f]\subseteq [f]$ implying that $f\subseteq f$. The latter is satisfied since

 
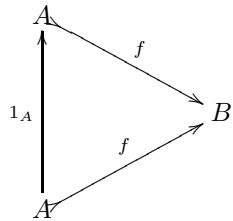
\[\xymatrix{
A\ar@{>->}[drr]^f&&\\
&&B\\
A\ar@{>->}[urr]^f\ar[uu]^{1_A}&&\\
}\]
\item Transitive: $[f]\subseteq [g]$ and $[g]\subseteq [h]$ then $[f]\subseteq [h]$. The fact that $[f]\subseteq [g]$ and $[g]\subseteq [h]$ implies that the diagram
\[\xymatrix{
C\ar@{>->}[ddrr]^h&&\\
&&\\
B\ar@{>->}[rr]^g\ar[uu]^i&&D\\
&&\\
A\ar@{>->}[uurr]_f\ar[uu]^{k}&&\\
}\]
commutes, where $f\in [f]$, $g\in[g]$ and $h\in[h]$. If $f\circ g=k$ and $g=h\circ i$ then it follows that $f=h\circ (i\circ k)$ thus $f\subseteq h$ and $[f]\subseteq [h]$
\item Antisymmetric:  If $[f]\subseteq [g]$ and $[g]\subseteq [f]$ then $f\subseteq g$ and $g\subseteq f$. This implies that $f\sim g$ thus $[f]=[g]$. \\
If we had only considered equivalence classes then $f\subseteq g$ and $g\subseteq f$ would only imply $f\sim g$ but not $f=g$. This would only be the case iff the only arrow allowing the factorisation of $f$ via $g$ (or other way round) would be the identity arrow.

\end{itemize}

\end{proof}

From now on, when we will talk about sub-objects it will be implicit that we are referring to equivalence classes of sub-objects, even though it is not explicitly stated.

It is interesting to note that sub-objects, given by a categorical definition, are not the same as subsets, but each subset determines and is determined by a unique sub-object. In fact we have the following definition of a sub-object in $\Sets$:
\begin{Definition}
In  $\Sets$ a sub-object is an equivalence class of injections (set theoretic equivalence of monic map).
\end{Definition}

The 1:1 (one to one) correspondence between sub-objects and sets in set theory is given by the statements in lemma \ref{lem:subsets} and \ref{lem:subsets2}. 

In particular , lemma \ref{lem:subsets} shows that given a sub-object $O$ of a set $S$ there corresponds a unique subset $I$ of $S$. On the other hand lemma \ref{lem:subsets2} shows the reverse, given a subset $I\subseteq S$ there corresponds a unique sub-object $O$ of $S$.

\begin{Lemma}\label{lem:subsets}\hspace{.1in}\\
Given a set $S$ and a sub-object $O$ (equivalence class of injective maps with codomain S) we have the following:
\begin{enumerate}
\item [a)] any two injections $f:A\rightarrow S$ and $g:B\rightarrow S$ which are in $O$ (i.e. they are equivalent) have the same image $I$ in S.
\item [b)] The inclusion $i : I\rightarrow S$ is equivalent to any injection in $O$, thus it is an element of $O$.
\item [c)] If $j : J\rightarrow S$ is an inclusion of a subset $J$ into $S$ that is in $O$, then $I = J$ and $i = j$. 
\item [d)] It follows from the above statements that every sub-object $O$ of $S$ contains one and only one injective map which represents the inclusion of a subset of $S$ into $S$. This  subset is the image of any element of $O$. 
\end{enumerate}
\end{Lemma}
\begin{Lemma}\label{lem:subsets2}\hspace{.1in}\\
On the other hand, given an inclusion map $i:I\rightarrow S$ of a subset $I$ into $S$, we have that:
\begin{enumerate}
\item $i$ is injective, thus it is an element of some sub-object $O$ of $S$.
\item Any two distinct equivalence classes are disjoint, thus $i$ can not belong to two different sub-objects.
\item Hence each subsets of $S$ (and their inclusion maps) defined a unique sub-object of $S$.
\end{enumerate}
\end{Lemma}
\begin{proof}
We now proof both of the above lemmas.
\begin{enumerate}
\item 
$\mathcal{O}:=[f_i]$, $f_i:A\rightarrow S$ where $f_1\sim f_2$ iff

\[\xymatrix{
A\ar[rr]^{f_1}\ar[dd]^g&&S\\
&&\\
B\ar[uurr]^{f_2}&&\\
}\text{    and   }
\xymatrix{
A\ar[rr]^{f_1}&&S\\
&&\\
B\ar[uu]^h\ar[uurr]^{f_2}
}\]

Since $f_1\circ h=f_2$ ( $im(f_1\circ h)=im(f_2)$) and $f_1$ and $f_2$ are monic then $h$ is monic which in Sets means it is injective. Thus $im(h)\subseteq A$ which in Sets can be written as $im(B)\subseteq A$ it follows that $im(f_1)\subseteq im(f_2)$. On the other hand $f_2\circ g=f_1$ thus $im(f_2\circ g)=im(f_1)$  $g$ is monic. Therefore $im(g)\subseteq $ B implying that $im(f_1)\subseteq im(f_1)$. Thus $im(f_1)=im(f_2)$.

\item $j:J\rightarrow S$ is an inclusion in $\mathcal{O}$, thus an injection, i.e. an element. Each element in $\mathcal{O}$ has the same image therefore $im(j)=I$. but $im(j)=J$ since $J\subseteq S$ thus $J=I$ and $j=i$.

\end{enumerate}

\end{proof}
The difference between the categorical definition of a sub-object, as an equivalence class of monic arrows and the standard definition of a subset in set theory is quite important at a conceptual level, although, from an operational point of view they are equivalent. To understand this consider the integers $\Zl$. In standard set theory they are simply a subset of the reals $\Rl$, such that each integer is actually a real. On the other hand in category theory sub-object relations only require the existence (or definition) of a monic map between $\Zl$ and $\Rl$. Thus, in this case, an integer needs not be a real number, it could, in principle, be something different. Reiterating, in standard set theory, 
the image of each integer is represented by the same integer, so that the integers are a subset of the real numbers. From the categorical point of view, instead,  this condition can be relaxed, and all that is needed is that there exists a monic map $m:\Zl\rightarrow\Rl$. If one wants to recover the standard definition in which such monic arrow picks out the same integer, then an extra condition has to be placed, namely, that $m$ is an equivalent monic, i.e. picks out the same sub-object. However, from an operational point of view the integers, as defined by a monic map $\Zl\rightarrow\Rl$ behave the same way, whether this extra condition is taken into consideration or not.

\subsection{Sub-Object Classifier (Truth Object)}
Now that we have defined what a sub-object is we would like to understand how to identify sub-objects.
To this end let us consider a specific example in $\mathbf{Sets}$. Here we have the following isomorphisms which we will prove later.
\be
Sub(S)\simeq 2^S:=\{S\rightarrow 2=\{0,1\}\}
\ee
In fact, given a subset $A$ of $S$, i.e $A\subseteq S$, the notion of being a subset can be expressed mathematically
using the so called characteristic function: $\chi_A:S\rightarrow\{0,1\}$, which is 
defined as follows:
\begin{equation}\label{equ:character}
\chi_A(x)=\begin{cases}0& if\hspace{.1in}x\notin A\\
1& if\hspace{.1in}x\in A  
\end{cases}
\end{equation}
(here we interpret 1=true and 0=false). The role of the characteristic function is to determine which 
elements belong to a certain subset.

Remembering that in any category sub-objects are identified as monic arrows, we define the value true in terms of the following monic map:
\ba
true:1=\{0\}&\rightarrow& 2=\{0,1\}\\
0&\rightarrow& 1
\ea Given this definition,
it can be easily seen that 
\be
A=\{x|x\in S\text{ and } \chi_A(x)=1\}=\chi_A^{-1}(1)
\ee
This equation is equivalent to the statement that the following diagram
\begin{Diagram} \label{dia:11}
 \[\xymatrix{
*++{A}\ar@{^{(}->}[rr]\ar[dd]^{!}&&S\ar[dd]^{\chi_A}\\
&&\\
1\ar[rr]_{\text{true}}&&2\\
}\]
\end{Diagram}
is a pullback, i.e. $A$ is the pullback of $true:1\rightarrow \{0,1\}$ along $\chi_A$. In fact, if we consider the following diagram
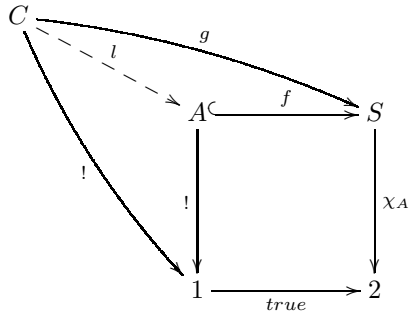
\[\xymatrix{
C\ar@/^/[rrrrd]^g\ar@{-->}[rrd]^l\ar@/_/[rrddd]_{!}&&&&\\
&&A\ar@{^{(}->}[rr]^f\ar[dd]_{!}&&S\ar[dd]^{\chi_A}\\
&&&&\\
&&1\ar[rr]_{true}&&2\\
}\]
such that the outer square commutes, we then have that for any $c\in C$; $\chi_A(g(c))=true(!(c))=1$. From the definition of $A$ above it follows that $g(c)\in A$. Thus we can define the map $l$ for each $c\in C$ as $l(c):=g(c)$. 
Obviously such choice makes the whole diagram commute and is the only arrow that would do so.
It follows that $A\subseteq S$ iff the diagram \ref{dia:11} is a pullback. $\chi_A$ is the only arrow that makes such a diagram a pullback. 
\begin{proof}
Given the pullback diagram 
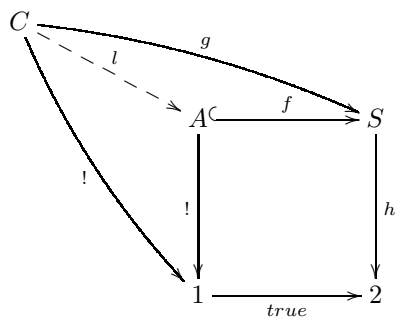
\[\xymatrix{
C\ar@/^/[rrrrd]^g\ar@{-->}[rrd]^l\ar@/_/[rrddd]_{!}&&&&\\
&&A\ar@{^{(}->}[rr]^f\ar[dd]_{!}&&S\ar[dd]^{h}\\
&&&&\\
&&1\ar[rr]_{true}&&2\\
}\]
We want to show that $h=\chi_A$. Since the diagram is indeed a pullback and $f$ a monic, we have that for all $x\in A$, $f(x)=true(!(x))=1$ therefore $A=f^{-1}(1)$, but this is precisely the definition of the characteristic map $\chi_A$.

\end{proof}

When diagrams like this arise:
 \[\xymatrix{
*++{A}\ar@{^{(}->}[rr]^f\ar[dd]_{!}&&S\ar[dd]^{\phi}\\
&&\\
1\ar[rr]_{\text{true}}&&2\\
}\]
we say that $\phi$ classifies $f$ or the sub-object represented by $f$.

We can now give an abstract characterisation of what it means to classify sub-objects.

\begin{Definition} \label{def:clas}
Given a category with a terminal object 1, a \textbf{sub-object classifier} is an object $\Omega$, 
together with a monic arrow $\mathcal{T}:1\rightarrow\Omega$ (topos analogue of the set theoretic arrow \emph{true}) such that, given a monic $\c$-arrow 
$f:a\rightarrow b$, there exists one and only one $\chi_f$ arrow, which makes the following
 diagram 
\[\xymatrix{
*++{a}\ar@{>->}[rr]^f\ar[dd]_{!}&&b\ar[dd]^{\chi_f}\\
&&\\
1\ar[rr]_{\mathcal{T}}&&\Omega\\
}\]
a pullback.

\end{Definition}
\vspace{.5in}
\begin{Axiom} \label{axi:1}
Given a category $\c$ with sub-object classifier $\Omega$ and sub-objects, there exists (in $\c$) an isomorphisms
\be
y:Sub_{\c}(X)\simeq \c(X,\Omega)\hspace{.1in}\forall X\in C
\ee
\end{Axiom}
In order to prove the above axiom we need to show that y is a) monic, b) epic c) has an inverse. Since the prove of the above 
theorem in topos is quite complicated and needs definitions not yet given we will use an analogous proof in $\mathbf{Sets}$, which 
essentially uses the 
same strategy as the proof in topos, but, is much more intuitive.
In $\mathbf{Sets}$ we can write the above axiom as follows:
\begin{Axiom}
The collection of all subsets of S denoted by $\mathcal{P}(S)$ and 
the collection of all maps from S to the set $\{0,1\}=2$ denoted by $2^S$ are isomorphic. This means that the function $y:\mathcal{P}(S)\rightarrow2^S$, 
which in terms of single elements of $\mathcal{P}(S)$ is $A\rightarrow\chi_A$, is a bijection.
\end{Axiom}

\begin{proof}
Let us consider the Diagram \ref{dia:11}
\begin{enumerate}
\item [a)] \emph{y is injective (1:1)}: 
 consider the case in which $\chi_A=\chi_B$, where
 \begin{equation*}
 \chi_B(x)=\begin{cases}1& iff\hspace{.1in}x\in B\\
 0& iff\hspace{.1in}x\notin B
 \end{cases}
 \end{equation*}
 while $\chi_A$ was defined as in equation \ref{equ:character}.
 
 Since the two functions are the same, they both associate the same domain to the codomain 1, therefore $A=B$ ($A=\chi^{-1}_A(1)=\chi_B^{-1}(1)=B$).
 \item [ b)] \emph{y is surjective (onto)}:
 given any function $f\in2^S$,
 then there must exist a subset $A$ of $S$, such that $A_f=\{x:x\in S\hspace{.1in}and\hspace{.1in} f(x)=1\}$, 
 i.e. $A_f=f^{-1}(\{1\})$.  Therefore $f=\chi_{A_f}$.
 \item [c)] The inverse is simply given by $\chi^{-1}(1)$.
 \end{enumerate}
\end{proof}
\begin{Corollary}
The domain of the arrow \emph{true}: $1\rightarrow \Omega$ is always the terminal object. 
\end{Corollary}

\begin{proof}
Let us assume that instead of the terminal object being the domain of the \emph{true} map we have a general element $Q$, obtaining $true:Q\rightarrow \Omega$. We now want to consider the identity arrow $id_A$ on an object $A$. Being an identity, such a map is iso, thus monic. We then want to analyse what sub-objects this map defines. We already know the answer since it is simply an identity map, but nonetheless we apply the sub-object classifier procedure and define the pull back diagram

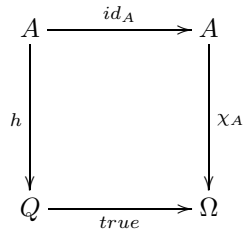
\[\xymatrix{
A\ar[rr]^{id_A}\ar[dd]_{h}&&A\ar[dd]^{\chi_A}\\
&&\\
Q\ar[rr]_{true}&&\Omega\\
}\]
Thus we obtain that $\chi_A=true\circ h$.

On the other hand, given any arrow $k:A\rightarrow Q$, we can define the following pullback classifying diagram
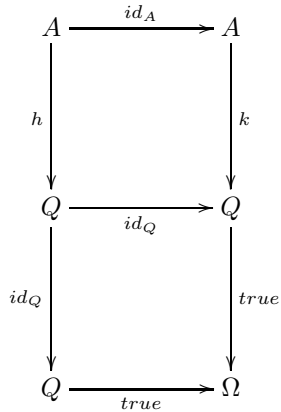
\[\xymatrix{
A\ar[rr]^{id_A}\ar[dd]_{h}&&A\ar[dd]^{k}\\
&&\\
Q\ar[dd]_{id_Q}\ar[rr]_{id_Q}&&Q\ar[dd]^{true}\\
&&\\
Q\ar[rr]_{true}&&\Omega\\
}\]
This means that in this case $true\circ k=true\circ h$. We know that \emph{true} is monic, thus $k=h$. What this means is that for a given object (in this case) $A$ there is one and only one (up to isomorphisms) arrow from $A$ to $Q$, this $Q$ must be the terminal object.
\end{proof}

\subsection{What is a Sub-Object Classifier in Topos ?} \label{ssub}
In the case of $\mathbf{Sets}$, $\Omega\cong\{0,1\}$ therefore the elements of $\Omega$ are simply 0 and 1, which can be identified with the values false and true in the context of the theory of logic. This is not the case for a general 
topos. 

Since we are mainly interested in the topos $\Sets ^{\c^{\op}}$, which will be the topos in which quantum theory will be expressed (and will be explained in details in Lecture 9), we will analyse what the elements of the sub-object classifier are in this case. To fully understand the nature of these elements it is not necessary that you know the details of the topos $\Sets ^{\c^{\op}}$. For now it suffices to know that it is actually a category with objects and morphisms. The details of either are irrelevant at this point.

Given the category $\Sets ^{\c^{\op}}$, the \emph{elements} of the \emph{sub-object classifier} are \emph{sieves}.\\
In what follows we will first define what a sieve is and, then, we will show that they can be identified with elements of $\Omega$.
\subsection{Sieve} \label{ssieve}
\begin{Definition} 
A \textbf{sieve} on an object $A\in\c$ is a collection $S$ of morphisms in $\c$ 
whose codomain is $A$ and such that, if $f:B\rightarrow A\in S$ then, given any morphisms 
$g:C\rightarrow B$ we have
$f \circ g\in S$, i.e. $S$ is closed under left composition:
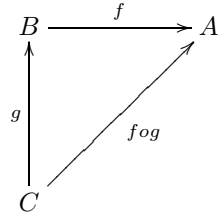
\[\xymatrix{
B\ar[rr]^{f}&&A\\
&&\\
C\ar[uu]^{g}\ar[rruu]_{fog}&&\\
}\]
\end{Definition} 
It is also possible to define maps between different sieves when the objects these sieves are defined on are related in some way. 
For example, if we have a $\c$-arrow $f:A\rightarrow B$ then it is possible to define a map from the set of all sieves on $B$, which we denote $\Omega(B)$ to the set $\Omega(A)$ of all sieves on $A$ as follows:
\ba\label{ali:sieve}
\Omega_{BA}:\Omega(B)&\rightarrow&\Omega(A)\\
S&\mapsto&\Omega_{BA}(S):=\{g|cod(g)=A\text{ and } f\circ g\in S\}
\ea
Sometimes you can symbolically write $\Omega_{BA}(S)=S\cap \downarrow A$ where $\downarrow A$ is the principal sieve on $A$, i.e. the sieve that contains the 
identity morphism of $A$, therefore it is the biggest sieve on $A$\footnote{Essentially $\downarrow A$ is the sieve which contains all possible $\c$-arrows, which has as codomain $A$, i.e. $\downarrow A:=\{f_i\in \c|cod(f_i)=A\}$.}.

An important property of sieves is the following: \\
if $f:B\rightarrow A$ belongs to $S$ which is a 
sieve on $A$, then the pullback of $S$ by $f$ determines the principal sieve on B, i.e.
\begin{equation*}
f^*(S):=\{h:C\rightarrow B|f \circ h \in S\}=\{h:C\rightarrow B\}=\downarrow B   
\end{equation*}
For example
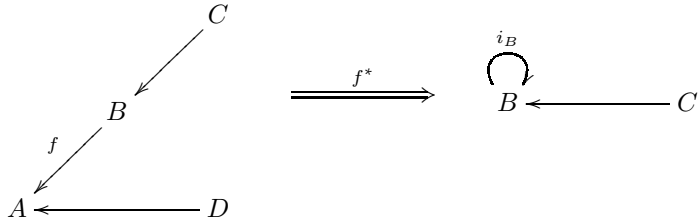
\[\xymatrix{
&&C\ar[dl]\\
&B\ar[dl]_f&\\
A&&D\ar[ll]\\
}\hspace{.2in}
\xymatrix{
&&\\
\ar@{=>}[rr]^{f^*}&&}\hspace{.2in}
\xymatrix{
&&\\
B\ar@(ul,ur)[]^{i_B}&&C\ar[ll]
}\]
An important property of sieves is that the set of sieves defined on an object forms a Heyting algebra with partial ordering given by subset inclusion. The fact that the set of sieves forms a Heyting algebra is very important since, as we will see later on, such an algebra will represent the logic of truth values. Thus the next question to address is: what is an Heyting algebra? The answer to this question will be the topic of the next section.
\subsubsection{Heyting Algebra}

\begin{Definition} \label{shey}
 A Heyting Algebra \textbf{H} is a relative pseudo complemented distributive lattice.
 
 \end{Definition}
 We will explain these attributes one at a time.
 The definition of a lattice was already given in the handout of lecture 3, but for sake of completeness we will nontheless restate it here.
  \begin{Definition}
 Given a poset $(L,\leq)$, we say that this is a lattice if the following conditions are satisfied.
 \begin{itemize}
 \item [1)] Given any two elements $a, b\in L $ it is always possible to define a third element $a\vee b\in L$ called the join or least upper bound or supremum.
 \item [2)] Given any two elements $a, b\in L $ it is always possible to define a third element $a\wedge b\in L$ called the meet or greatest lower bound or infimum.
 \end{itemize}
 \end{Definition}
 because the elements of the lattice involved in defining join and meet are two, the operations $\vee$ and $\wedge$ are the binary operations of the lattice.
 
 If only the first condition holds we say that $L$ is a join-semilattice, if only the second holds $L$ is a meet-semilattice.

A lattice $L$ is saied to be \emph{distributive} if for any $a_i\in L$ the following relations hold: 
\begin{align*}
a_1\wedge(a_2\vee a_3)&=(a_1\wedge a_2)\vee(a_1\wedge a_3)\\
a_1\vee(a_2\wedge a_3)&=(a_1\vee a_2)\wedge(a_1\vee a_2)
\end{align*}
 In order to understand the property of being a \emph{relative pseudo complemented lattice} we first of all have to introduce the notion of least upper bound (l.u.b) (and dually of greatest lower bound (g.l.b)) of a set. We are already acquainted with the notion of the l.u.b for two elements $a, b$ in a lattice $L$ with ordering $\leq$. This is simply given by the element $a\vee b$. Similarly, the g.l.b. is  $a\wedge b$. But how do we define such notions with respect to a set $A\subseteq L$? The definition is quite intuitive: the g.l.b. of a set $A$ is an element $c\in L$, such that for all $a\in A$, $c\leq a$ and given any other element $b\leq a$ for all $a\in A$ then $b\leq c$. The condition of $c$ being the g.l.b. of $A$ is denoted by $c\leq A$. Moreover, we say that $c$ is the \emph{greatest element} of $A$ if $c$ is the g.l.b and $c \in A$.

 Dually the l.u.b $y\in L$ of $A$, denoted $A\leq y$ is such that for all $a\in A$, $a\leq y$ and, given any other element $z\in L$, such that $A\leq z$, then $y\leq z$.  
We can now define the notion of a relative pseudo-complement as follows:
\begin{Definition}
  $L$ is a relative-pseudo complemented lattice iff for each two elements $a, b\in L$ there exists a third element $c$, such that 
  \begin{enumerate}
\item $a\wedge c\leq b$
\item $\forall x\in L\hspace{.2in}x\leq c\hspace{.2in}iff\hspace{.2in}a\wedge x\leq b$
\end{enumerate}
where $c$ is defined as the \textit{pseudo complement} of $a$ relative to $b$ i.e., the 
greatest element of the set $\{x:a\wedge x\leq b\}$, and it is denoted 
as $a\Rightarrow b$, i.e.
\begin{equation}
a\rightarrow b=\bigvee\{x:x\wedge a\leq b\}
\end{equation}
\end{Definition}
%
%
If in the above definition we replace $b$ with the element $0$, then we obtain the notion of pseudo-complement.
\begin{Definition}
Given a lattice $(L,\leq)$ with a zero element, the pseudo-complement of $a$ is the greatest element of $L$ disjoint from $a$, i.e. the greatest element of the set $\{x\in L|a\wedge x=0\}$. The pseudo complement of $a$ will be denoted as $a\Rightarrow 0$
\end{Definition}
 If every element of $L$ has a pseudo-complement, then $L$ is a pseudo-complemented lattice.
 
 The pseudo complement in a Heyting algebra is identified with the negation operation, i.e. $\neg a:=a\Rightarrow 0$.
%
From the above definition of negation operation ($\neg a:a\Rightarrow 0$) in a Heyting algebra,  we obtain the following corollary:
\begin{Corollary}
Given any element S of an Heyting algebra, we have the following: 
\be
S\vee\neg S\leq1
\ee
\end{Corollary}
\begin{proof}
Let us consider $S\vee\neg S$. This represents the least upper bound of S and $\neg S$ 
therefore, given any other element $S_1$ in the Heyting algebra, such that $S\leq S_1$ and 
$\neg S\leq\ S_1$, then, $S\vee\neg S\leq S_1$. But, since for any S we have $S\leq1$ and $\neg S\leq1$,
it follows that $S\vee\neg S\leq1$.
\end{proof}
\subsection{Understanding the Sub-Object Classifier in a Topos}
As previously stated, the role of a sub-object classifier is to identify sub-objects. This is done in terms of associating for each sub-object $A$, of a given object $X$, an element of the sub-object classifier. For $\Sets$ this is very straightforward since in this case $\Omega=\{0,1\}$, which can be interpreted as the values true and false. Thus, in this case, each sub-object $A$ is uniquely identified in terms of the elements $x$ which belong to them, i.e. for which $\chi_A(x)=1$. 

The fact that in $\Sets$ the set of truth values is simply $\{0,1\}$ implies that the logic which is derived is a Boolean logic, i.e. a classical logic. This is basically the logic that each of us adopts when speaking any western language. It is the logic of the classical world, and the logic of the western way of reasoning. 

On the other hand, for a general topos the sub-object classifier will not simply be the two valued set $\{0,1\}$ as it is for $\Sets$. This `complicates' (depending on the point of view) the situation however, at an interpretative level, the role of the sub-object classifier is still unchanged and the elements still represent truth values. Since these truth values will, in general, not be simply true or false, we will not end up with a Boolean Logic. Instead, we will end up with an intuitionistic logic which, mathematically, is represented by a Heyting algebra described above.

In particular, in classical logic, either a statement is true or it is false there is nothing in between. This is not the case for intuitionistic logic. If one thinks about it, examples of intuitionistic logic can be found in our language. In fact, although the way we reason is governed by classical logic, some aspects of our language can, in a way, be considered more intuitionistic in nature. For example, statements regarding more subjective issues. Consider the statement ``I am tired" or `` I am hungry" the truthfulness of these statements is not simply true or false. In fact, if you asked the question ``are you tired?" or ``are you hungry?" in most cases you will not get a simple yes or no answer but you could get something more elaborated such as ``I am a little tired" (``I am a bit hungry") or ``I am not so tired" (`` I am not so hungry"), extremely , not at all  etc. This is because the expressions tired, hungry, can be quantified, i.e. graded.  At the extreme points there will be the answers ``yes I am tired" (true) and ``no, I am not tired" (false) and similarly for ``hungry". However, as we have just seen, there will be many different statements in between (different truth values). 

Can these truth values be quantified?  The answer is ``yes", as long as we remain in the real of language since we know that, for example, the expressions ``very", ``a little", ``quite", have the following relations in terms of strength:
\be
\text{a little}\le\text{quite}\le\text{very}
\ee 
In this case if you are ``very hungry" but just ``a little tired", the statement `` you are hungry" is more true than the statement  ``you are tired".
However, this analysis relies only on the meaning that language has given to such expressions. The question is if it would be possible to define relations of truthfulness on a more objective ground. To this end let us consider the statement `` the 1 litre bottle is full" which we symbolically denote as $X$. Let us also consider two identical bottles $b_1$ and $b_2$ one half full and the other three quarters full, respectively. Obviously, the statement $X$, when referred to $b_1$, is less true than when referred to $b_2$. The objectiveness of such truth values can be defined by measuring how much water there is in each bottle (1/2 and 3/4). 

Alternatively, one can also assert how true the proposition $X$ is in a more operational way by measuring how much water needs to be added, such that the proposition $X$ is true in the classical sense, i.e. such that there is 1 litre of water in the bottle. Thus, if proposition $X$ is true when there is exactly one litre of water in a bottle, such that no more water has to be added, it is then possible to define the truth value of $X$, given a system $b_i$,  as $(1-\text{water to be added})$. This tells you precisely how far away you are from the truth, i.e. from one litre of water.

Thus, the truth value of the proposition $X$, given the `states' $b_1$ and $b_2$ is now given by 
\be
v(X, b_1)=(1-1/2)\le v(X, b_2)=(1-1/4)
\ee
Obviously, the extreme points of such truth values are $0$ and $1$; for example when the bottle is empty (1 litre of water has to be added $v(X, b)=1-1=0$) and when the bottle is full (no water has to be added  $v(X, b)=1-0=1$), respectively.

This way of defining truth values is based on the idea of how much one has to \emph{change} the system (fill up the bottles) in order for the \emph{untouched} proposition to be true. However, it is also possible to proceed the other way, i.e. \emph{generalise} the proposition so that the now \emph{unchanged} system satisfies the more general proposition. Therefore, keeping with our water bottle\footnote{I don't know about you but I'm getting very thirsty writing this :)} examples, let us assume we have a litre bottle $b_1$ which has 1/10 of a litre of water in it.  Again, we have the proposition $X$ stating ``the bottle is full". We can't say that $X$ is true given $b_i$, but we can't even say it is false since $b_1$ is not empty. So the question is what can we say about the truth value of $X$ given $b_1$ without changing the system? Well, we simply do the `inverse' procedure as we did before, i.e. we generalise the proposition $X$.

\noindent
To understand this it is convenient to re-write $X$ as ``$b_i$ has 1 litre of water". We then start generalising such a proposition by subtracting amounts of water to the bottle and forming new propositions. For example, we consider the proposition $X_{1/2}= ``b_i$ has (1-1/2) litres of water". However, given our bottle $b_1$ which has only 1/8 of water we still can not say that $X_{1/2}$ is true. So we keep generalising the proposition till we reach the proposition $X_{7/8}=``b_1$ has (1-7/8) litres of water". Such a proposition is true given $b_1$. 

This example shows how it is possible to identify truth values of propositions, given a state, in terms of how much the original proposition had to be generalised such that this new coarse grained proposition is true (given that state).

 Thus, in our case, the truth value of the original proposition $X$, given the state $b_i$, is defined in terms of how much we have to generalise $X$, in this case we coarse grain it to $X_{7/8}$, such that $X_{7/8}$ is true given $b_i$. Mathematically this is given by the equation
\be
v(X, b_1):=1-inf\{x|v(X_{x}, b_1)=1\}=1-7/8=1/8
\ee
It is now straight forward to see how truth values can be compared. In particular, if we have a bottle which contained $1/2$ litre of water then its truth value would be
\be
v(X, b_2):=1-inf\{x|v(X_{x}, b_2)=1\}=1-1/2=1/2
\ee
It then follows that $v(X, b_2)\ge v(X, b_1)$.

In this schema we obviously don't just have `true' or `false' as truth values, but we have many more truth values, whose limiting points are precisely 1 (true) and 0 (false).
 
Although this is a very general description of how different truth values, other than `true' or `false', can be obtained, it still helps shedding some light on how intuitionistic logic works. In such a logic truth values are not constricted to be only true or false but there are many more values in between. Obviously, the two limiting points will then coincide with the classical notions of true and false.

What the above implies is that, differently from a Boolean algebra, in a Heyting the law of excluded middle does not hold: $S\vee \neg S\leq 1$ (in a Boolean logic we would simply have $S\vee\neg S=1$). Thus, differently from ordinary logic where something is either true or false, in intuitionistic logic this is not the case. As a consequence, the sum of a proposition and its negation doesn't give the `whole story', which implies that double negation does not give back the original element, i.e. $\neg\neg S\geq S$. On the other hand strict equality would hold for Boolean algebras.

Above we have given a very intuitive definition of what a sub-object classifier does in a general topos, while the technical definition was given in \ref{def:clas}. However the important point to understand is that in a general topos the sub-object classifier together with the Heyting algebras of sub-objects enables us to render mathematically precise the notion of a proposition being nearly true, almost true etc. Moreover, it allows for a well defined mathematical notion of how `far away' from the truth a proposition, given a state, actually is. This `distance' from the truth is in the end what a truth value represents in a general topos.

In particular, in the yet to be defined topos $\Sets^{\c^{\op}}$, the elements of the sub-object classifier are sieves and these represent truth values, thus the bigger the sieve is the truer a proposition will be. The principal sieve, which is the biggest sieve, represents the analogue of the classical value \emph{true}, while the empty sieve represents the classical value \emph{false}.

\subsection{Axiomatic Definition of a Topos}
Now that we have defined all the above constructions we are ready to give the axiomatic definition of a topos.
\begin{Definition}
An elementary \textbf{topos} is a category with all finite limits, exponentials and a sub-object classifier.
\end{Definition}
An alternative but equivalent definition is:
\begin{Definition} \label{def:topos}
A \textbf{topos} is a category $\tau$ with the following extra properties:
\begin{itemize}
\item [1.] $\tau$ has an initial object $(0)$. 
\item [2.] $\tau$ has a terminal object $(1)$ . 
\item [3.] $\tau$ has pullbacks.
\item [4.] $\tau$ has pushouts. 
\item [5.] $\tau$ has exponentiation, i.e. $\tau$ is such that for every pair of objects X and Y in $\tau$,
the map $Y^X$ exists.
\item [6.] T has a sub-object classifier.
\end{itemize}
\end{Definition}
It is straightforward to see that condition (1) and (3) above are equivalent to stating that a topos $\tau$ has all finite limits. By duality, if $\tau$ has all finite limits it also has all finite co-limits. This requirement is equivalent to conditions (2) and (4).
\chapter{Lecture 7}
In this lecture I will describe some categorical constructs present in the topos of contravariant presheeaves $Sets^{\c^{op}}$. This topos is important since, for a particular choice of $\c$ is will be the topos we will utilised to express quantum theory.

\section{Categorical Constructs in the Topos of Presheaves Taking Values in Sets}
We will now show how some of the categorical construct, which we delineated in lecture 2, apply to the category $\Sets^{\c^{op}}$.
\subsection{Pullbacks}
Pullbacks exist in any category of presheaves $\Sets^{\c^{op}}$.
In fact, if $X,Y,B\in \Sets^{\c^{op}}$, then $P\in \Sets^{\c^{op}}$ is a pullback in $\Sets^{\c^{op}}$ iff
\vspace{.3in}
\[\xymatrix{
P(C)\ar[rr]^{k}\ar[dd]_{h}&&Y(C)\ar[dd]^{g}\\
&&\\
X(C)\ar[rr]_l&&B(C)\\
}\]

is a pullback in set.
This implies that $P(C)=(X \times_B Y)C\simeq X(C)\times_{B(C)} Y(C)$\footnote{Note that the last product is all in Sets}.\\
The fact that a pullback in $\Sets^{C^{op}}$ is defined in terms of a pullback in $\Sets$, implies that the former always exists, since the latter does. In fact, given any object $C\in\c$ and any three functors X, Y, B, it is always possible to construct (in $\Sets$) a diagram as the one above, where $P(C)=(X \times_B Y)C\simeq X(C)\times_{B(C)} Y(C)$. This implies that it is always possible to define a functor $P:\c\rightarrow \Sets$ which assigns, for each $C\in\c$ an object $P(C)=(X \times_B Y)C\simeq X(C)\times_{B(C)} Y(C)$ (a set) and, for each arrow $f:A\rightarrow C$ in $\c$ the unique arrow $P(f):P(C)\rightarrow P(A)$ in $\Sets$, which makes the following cube a pullback in $\Sets^{\c^{op}}$
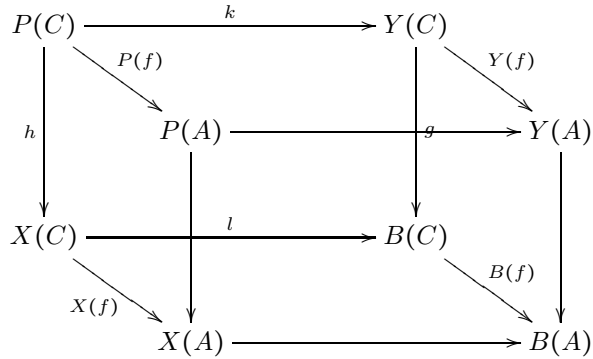
\[\xymatrix{
P(C)\ar[rrr]^{k}\ar[rd]^{P(f)}\ar[dd]_{h}&&&Y(C)\ar[dd]^{g}\ar[rd]^{Y(f)}&\\
&P(A)\ar[dd]\ar[rrr]&&&Y(A)\ar[dd]\\
X(C)\ar[rrr]^l\ar[rd]_{X(f)}&&&B(C)\ar[rd]^{B(f)}&\\
&X(A)\ar[rrr]&&&B(A)\\
}\]
Where, again the outer square is a pullback in $\Sets$ and, as such, it always exists.
\subsection{Sub-objects}
A sub-object of a presheaf is defined as follows: 
\begin{Definition} \label{def:sub}
Y is a \textbf{sub-object of a presheaf} X if there exists a natural transformation
$i:Y\rightarrow X$ 
which is defined, 
component wise, as $i_A:Y(A)\rightarrow X(A)$ and where $i_A$ defines a subset embedding 
i.e. $Y(A)\subseteq X(A)$.
\end{Definition}
Since Y is itself a presheaf, the maps between the objects of Y are the restrictions of the 
corresponding maps between the objects of X.
This can be easily seen with the aid of the following diagram:

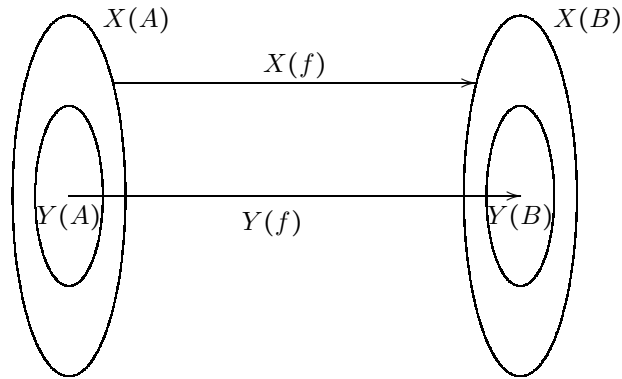
\[\begin{xy} 0;/r15mm/:
,(0,0)="x",*@{},*+!U{Y(A)}
,{\ellipse(0.5,1.6){}}
,{\ellipse(0.3,0.8){}}
,(1.8,-0.3),*@{},*+!{Y(f)}
,(0.6,1.5),*@{},*+!{X(A)}
,(3.6,1)="z",*@{},*+!{}
,(0.4,1)="y",*@{},*+!{}
,{\ar"y";"z"}
,(4.6,1.5),,*@{},*+!{X(B)}
,(2,1.1),,*@{},*+!{X(f)}
,(4,0)="p"*@{},*+!U{Y(B)}
,{\ellipse(0.5,1.6){}}
,{\ellipse(0.3,0.8){}}
,{\ar"x";"p"}
 \end{xy}\]

An alternative way of expressing this condition is through the following commutative diagram:

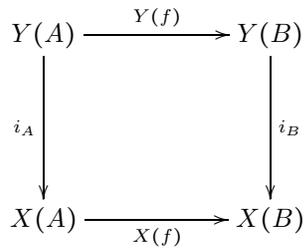
\[\xymatrix{
Y(A)\ar[rr]^{Y(f)}\ar[dd]_{i_A}&&Y(B)\ar[dd]^{i_B}\\
&&\\
X(A)\ar[rr]_{X(f)}&&X(B)\\
}\]
\subsection{Initial and Terminal Object}
An initial object in the topos of presheaves is defined as follows:
\begin{Definition} See Exercise
%
\end{Definition}
An initial object is the dual of a terminal object. 
A terminal object in the Topos of presheaves $\Sets^{\c^{op}}$ 
is defined as follows:
\begin{Definition}
A \textbf{terminal object} in $\Sets^{\c^{op}}$ is the constant functor 
$1:\c\rightarrow \Sets$
that maps every $\c$-object to the one element Set $\{*\}$ and every 
$\c$-arrow to 
the identity arrow on $\{*\}$.
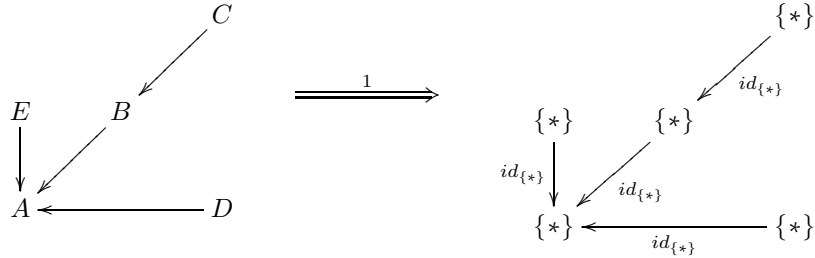
\[\xymatrix{
&&C\ar[dl]\\
E\ar[d]&B\ar[dl]&\\
A&&D\ar[ll]\\
}\hspace{.2in}
\xymatrix{
&&\\
\ar@{=>}[rr]^{1}&&}\hspace{.2in}
\xymatrix{
&&\{*\}\ar[dl]^{id_{\{*\}}}\\
\{*\}\ar[d]_{id_{\{*\}}}&\{*\}\ar[dl]^{id_{\{*\}}}&\\
\{*\}&&\{*\}\ar[ll]^{id_{\{*\}}}\\
}\]

\end{Definition}

\subsection{Sub-object Classifier in the Topos of Presheaves}
We will now describe the most important object in $\Sets^{\c^{op}}$: the sub-object classifier. As previously described this will allow us to define truth values in our topos representation of quantum theory.
\begin{Definition}\label{def:subclas}
A \textbf{Sub-object Classifier} $\Omega$ is a presheaf 
$\Omega:\c\rightarrow\Sets$ such that:
\begin{itemize}
\item To each object $A\in\c$
there corresponds an object $\Omega(A)\in\Sets$, which represents the set 
of all sieves on A. 
\item  To each 
$\c$-arrow $f:B\rightarrow A$, there corresponds a
$\Sets$-arrow $\Omega(f):\Omega(A)\rightarrow\Omega(B)$, such that
$\Omega(f)(S):=\{h:C\rightarrow B|f \circ h\in S\}$ is a sieve on B, where $\Omega(f)(S):=f^*(S)$. 
\end{itemize}
\end{Definition}
We now want to show that this definition of sub-object classifier is in agreement with definition 1.10 in lecture 5.
In order to do that we need to define the analogue of arrow 
\textit{true} ($\top$) and of the \textit{characteristic function} in a topos.
\begin{Definition}
$\top:1\rightarrow\Omega$ is the natural transformation that has components $\top_A:\{*\}\rightarrow\Omega(A)$ 
given by $\top_A(*)=\downarrow A$ = principal sieve on $A$. 
\end{Definition}
To understand how $\top$ works, let us consider a monic arrow $\sigma:F\rightarrow X$ in $\Sets^{\c^{op}}$, 
which is defined component-wise by $\sigma_A:F(A)\rightarrow X(A)$ and represents subset inclusion.\\
Now let us define the character $\chi^{\sigma}:X\rightarrow\Omega$ of $\sigma$ which is a natural transformation in the category of presheaves, such that the components $\chi_A^{\sigma}$ represent functions from 
$X(A)$ to $\Omega(A)$, as depicted in:
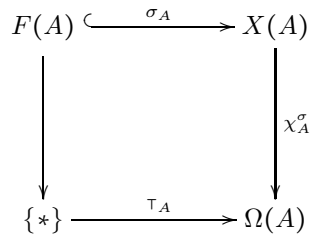
\[\xymatrix{
*++{F(A)}\ar@{^{(}->}[rr]^{\sigma_A}\ar[dd]&&X(A)\ar[dd]^{\chi^{\sigma}_A}\\
&&\\
\{*\}\ar[rr]^{\top_A}&&\Omega(A)\\
}\]
where $T\{*\}=\downarrow A$.\\
From the above diagram we can see that $\chi^{\sigma}_A$ assigns to each element $x\in X(A)$
a sieve $\chi^{\sigma}_A(x)\in \Omega(A)$ on $A$. For an arrow $f: B\rightarrow A$ in $\c$ to belong to the sieve $\chi^{\sigma}_A(x)$ on $A$ we require that
the following diagram commutes
  \[\xymatrix{
*++{F(A)}\ar@{^{(}->}[rr]^{\sigma_A}\ar[dd]_{F(f)}&&X(A)\ar[dd]^{X(f)}\\
&&\\
*++{F(B)}\ar@{^{(}->}[rr]^{\sigma_B}&&X(B)\\
}\]            
Such that $F(f)$ is the restriction of $X(f)$ to $F(A)$, since $F$ is a sub-presheaf of $X$. Therefore
\begin{equation}
\chi_A^{\sigma}(x):=\{f:B\rightarrow A|X(f)(x)\in F(B)\}    \label{eq:char}
\end{equation}
This condition is expressed by the following diagram:
  
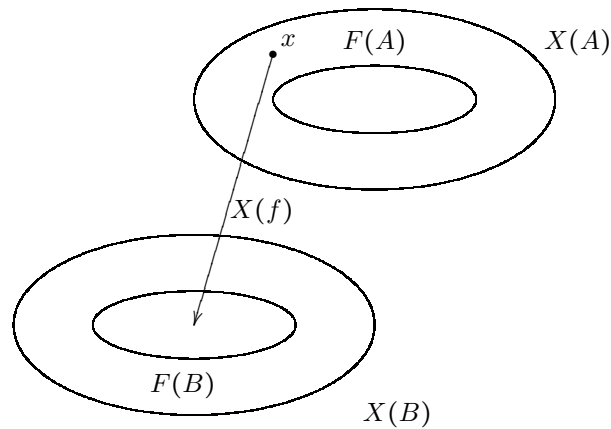
\[\begin{xy} 0;/r15mm/:
(0,0)="x",*@{},*+!U{}
,(0,0.2),*@{},*+!U{X(f)}
,(1,1.7),*@{},*+!U{F(A)}
,(2.8,1.7),*@{},*+!U{X(A)}
,(1.2,-1.6),*@{},*+!U{X(B)}
,(-0.7,-1.35),*@{},*+!U{F(B)}
,(0.1,1.4)="f"*@{*},*+!LD{x}
,(1,1)="o"*@{},*+!LD{}
,{\ellipse(1.6,0.8){}}
,{\ellipse(0.9,0.3){}}
,(-0.6,-1)="p"*@{},*+!R{}
,{\ellipse(1.6,0.8){}}
,{\ellipse(0.9,0.3){}}
,{\ar"f";"p"}
\end{xy}\]

i.e. $f$ belongs to $\chi^{\sigma}_A(x)$ iff $X(f)$ maps $ x$ into $F(B)$.                                           
\\
$\chi_{A}^F(x)$, as defined by equation \ref{eq:char}, 
represents a sieve on $A$.
\begin{proof}
Consider the following commuting diagram which represents sub-object $F$ of the presheaf $X$.
 \[\xymatrix{
F(A)\ar[rr]^{F(f)}\ar@{^{(}->}[dd]&&F(B)\ar[rr]^{F(g)}\ar@{^{(}->}[dd]&&F(C)\ar@{^{(}->}[dd]\\
&&&&\\
X(A)\ar[rr]^{X(f)}&&X(B)\ar[rr]^{X(g)}&&X(C)\\
}\]
If $f:B\rightarrow{A}$ belongs to $\chi_{A}^{\sigma}(x)$
then, given $g:C\rightarrow B$ it follows that $f \circ g$ belongs to $\chi_{A}^{\sigma}(x)$,
since from the above diagram it can be deduced that $X(f \circ g)(x)\in F(C)$. 
This is precisely the definition of a sieve, so we have proved that
$\chi_{A}^{\sigma}(x):=\{f:B\rightarrow A|X(f)(x)\in F(B)\}$  is a sieve.  
\end{proof}
As a consequence of axiom 3.1 in lecture 5 the condition of being a sub-object classifier can be 
restated in the following way.
\begin{Axiom}
\textbf{Omega Axiom}: $\Omega$  is a 
\textbf{sub-object classifier} iff there is a ``one to one" correspondence between sub-object of X and morphisms
from X to $\Omega$.
\end{Axiom}
Given this alternative definition of a sub-object classifier, it is easy to prove that $\Omega$,  as defined in \ref{def:subclas}
is a sub-object classifier. In fact, from equation \ref{eq:char}, 
we can see that indeed there is a 1:2:1 correspondence between sub-objects of $X$ and characteristic 
morphisms (character) $\chi$.\\
Moreover, for each morphism $\chi:X\rightarrow\Omega$ we have
\be
F^{\chi}(A):=\chi^{-1}_{A}\{1_{\Omega(A)}:=\downarrow A\}=\{x\in X(A)|\chi_{A}(x)=\downarrow A\}=\text{sub-object of} X(A) 
\ee 
\subsubsection{Elements of Sub-Object Classifier} 
The elements of the sub-object classifier $\Omega$ in a topos are derived from the following theorem:

\begin{Theorem} \label{the:2}
Given an object $A\in \c$ (where $\c$ is a locally small category), the a sieve on $A$ can be identified with representable functor (defined below) $y(C):=Hom_{\c}(-,A)\in\Sets^{\c^{\op}}$.
\end{Theorem}
In order to prove the above theorem we need the following lemma:
\begin{Lemma}
\textbf{Preliminary}: if $\c$ is a locally small category \footnote{ A Category $\c$ is said to be locally small iff its collection of morphisms form a proper set.}, then each object A of $\c$ induces a natural contravariant functor from $\c$ to $\mathbf{Sets}$  called a hom-functor $\mathbf{y}(A):=Hom_{\c}(-,A)$\footnote{We have already encounter this in lecture 4. }. Such a functor is defined on objects $C\in\c$ as 
\ba
\mathbf{y}(A):\c&\rightarrow& \Sets\\
C&\mapsto& Hom_{\c}(C,A)
\ea
on $\c$-morphisms $f:C\rightarrow B$ as 
\ba
\mathbf{y}(A)(f):Hom_{\c}(B,A)&\rightarrow& Hom_{\c}(C, A)\\
g&\mapsto& \mathbf{y}(A)(f)(g):=g\circ f
\ea
\textbf{Yoneda lemma}: Given an arbitrary presheaf P on $\c$ 
there exists a bijective correspondence between natural transformations $\mathbf{y}(A)\rightarrow P$ and elements of the set $P(A)$ ($A\in\c$) defined as an arrow
\ba
\theta:Nat_{\c}(\mathbf{y}(A),P)&\buildrel\simeq\over\rightarrow& P(A)\\
\Big(\alpha:\mathbf{y}(A)\rightarrow P\Big)&\mapsto&\theta(\alpha)=\alpha_A(1_A)
\ea
\end{Lemma} 
We have now the right tools to prove the above theorem. 
\begin{proof}
Let us consider $\Omega$ to be a sub-object classifier of $\hat{C}=\Sets^{C^{op}}$, i.e. we want $\Omega$ to classify sub-objects in $\Sets^{C^{op}}$. \\[2pt]
Consider now a presheaf 
$\mathbf{y}(C)=Hom_{\c}(-,C)\in\hat{C}$.
We know from axiom 3.1 lecture 5 that $$Sub_{\hat{C}}(Hom_\c(-,C))\cong Hom_{\hat{C}}(Hom_\c(-,C),\Omega)$$
From Yonedas lemma it follows that
$$Hom_{\hat{C}}(Hom_{\c}(-,C),\Omega)=\Omega(C)$$
therefore the sub-object classifier $\Omega$ must be a presheaf $\Omega:\c\rightarrow \Sets$, such that 
\begin{align*}
\Omega(C)=& Sub_{\hat{C}}(Hom_{\c}(-,C))\\
& =\{S|S\hspace{.1in}a\hspace{.1in}sub-functor\hspace{.1in}of\hspace{.1in}Hom_{\c}(-,C)\}
\end{align*}
Now if $Q\subset Hom_{\c}(-,C)$ is a sub-functor of $Hom_{\c}(-,C)$, then the set
$$S=\{f|\hspace{.1in}for\hspace{.1in}some\hspace{.1in}object\hspace{.1in}A,\hspace{.1in}f:A\rightarrow C
\hspace{.1in}and\hspace{.1in}f\in Q(A)\}$$
is a sieve on $C$. Conversely, given a sieve $S$ on $C$ we define
$$Q(A)=\{f|f:A\rightarrow C\hspace{.1in}and\hspace{.1in}f\in S\}\subseteq Hom_{\c}(A,C)$$
which determines a presheaf $Q:\c\rightarrow \Sets$ which is a sub-functor of $Hom_{\c}(-,C)$, i.e to each object $A\in\c$, Q assigns the set $Q(A)\subseteq Hom_{\c}(A,C)$.
The above discussion shows that there exist a bijective correspondence between sub-functors $Q\subseteq Hom_{\c}(-,C)$ and sieves $S$ on $C$. Therefore 
\begin{center}$Sieve\hspace{.05in}on\hspace{.05in}C\simeq \hspace{.05in}sub-functor\hspace{.05in}of\hspace{.05in}Hom_{\c}(-,C)$ \end{center}
\end{proof}
For each object $A$ in a category $\c$ we can now define a presheaf $\mathbf{y}(A)$ such that:
\begin{itemize}
\item 
Given an object $D$ of $\c$ we have 
$$\mathbf{y}(A)D=Hom_{\c}(D,A)$$
\item 
Given morphisms $\alpha :B\rightarrow D$ and $\theta:D\rightarrow A$ we obtain:
\ba
\mathbf{y}(A)(\alpha):Hom_{\c}(D,A)&\rightarrow& Hom_{\c}(B,A)\\
\theta&\mapsto&\mathbf{y}(A)(\alpha)(\theta)=\theta \circ \alpha
\ea
\end{itemize}
A very simple graphical example of the above is the following: 
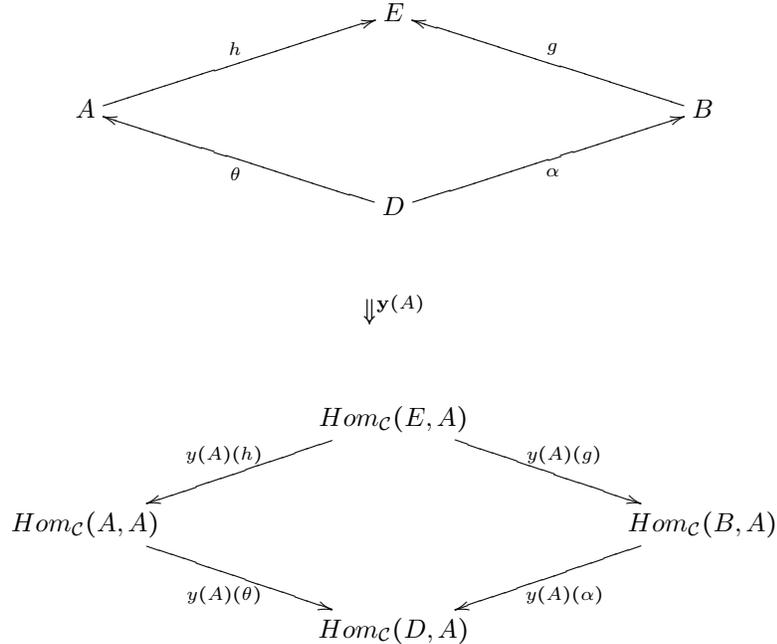
\[\xymatrix{
&&E&&\\
A\ar[rru]^h&&&&B\ar[llu]_g\\
&&D\ar[rru]_{\alpha}\ar[llu]^{\theta}&&\\
&&\Downarrow^{\mathbf{y}(A)}&&\\
&&Hom_{\c}(E,A)\ar[rrd]^{y(A)(g)}\ar[lld]_{y(A)(h)}&&\\
Hom_{\c}(A,A)\ar[rrd]_{y(A)(\theta)}&&&&Hom_{\c}(B,A)\ar[lld]^{y(A)(\alpha)}\\
&&Hom_{\c}(D,A)&&\\
}\]

Considering the above construction it follows that, for any morphism on $\c$ of the form $f:A\rightarrow A_1$ there exists a natural transformation $\mathbf{y}(A)\rightarrow \mathbf{y}(A_1)$ between the respective presheaves as constructed above.

\noindent
We can therefore deduce that $\mathbf{y}$ is actually a functor from the category $\c$ to the set of presheaves defined on $\c$, i.e 
$\mathbf{y}:\c\rightarrow \Sets^{\c^{op}}$, 
such that to each object A in $\c$, $\mathbf{y}$ assigns the Hom-functor $Hom_{\c}(-,A)$, i.e.
\ba
\mathbf{y}:\c&\rightarrow&\Sets^{\c^{op}}\\
A&\mapsto& \mathbf{y}(A):=Hom_{\c}(-,A)
\ea
where $Hom(-,A)$ corresponds to a Presheaf on $\c$. 

In this setting, given a $\c$-arrow $f:C\rightarrow D$ and a $\c$-object$ A$ the induced morphisms is
\be
\mathbf{y}(f):Hom_{\c}(-,C)\rightarrow Hom_{\c}(-,D)
\ee
whose natural components, for any $A\in \c$  are 
\ba
\mathbf{y}(f)(A)=Hom_{\c}(A, f):Hom_{\c}(A,C)\rightarrow Hom_{\c}(A,D)
\ea
i.e., they correspond to $\mathbf{y}(A)(f)$.

The importance of Yoneda's Lemma is really that it enables us to identify the elements of a sub-object classifier for the topos $\Sets^{\c^{\op}}$ as sieves. Since elements of the sub-object classifier are identified with truth values, it follows that in the topos $\Sets^{\c^{\op}}$ we will end up with a multivalued logic differently from the logic we obtain in $\Sets$, where $\uom=\{0,1\}$, i.e. where the only truth values are true and false. We will talk more about this in coming lectures. 

\subsection{Global sections}
Other important features of topos theory are the global and local sections which we will define below.
\begin{Definition} \label{def:glo}
A \textbf{global section} or \textbf{global element} of a presheaf X in $Sets^{\c^{op}}$ is an arrow 
$k:1\rightarrow X$ from the terminal object 1 to the presheaf X.
\end{Definition} 
What $k$ does is to assign to 
each object A in $\c$ an element $k_A\in X(A)$ in the corresponding set of the presheaf X.  
The assignment is such that, given an arrow $f:B\rightarrow A$ the following relation holds
\begin{equation}
X(f)(k_A)=k_B     \label{eq:k}
\end{equation}
What \ref{eq:k} uncovers, is that the elements of $X(A)$ assigned by the global section $k$, are mapped 
into each other by the morphisms in $X$.
Presheaves with a local or partial section can exist even if they do not have a global section.\\
\\
A particular important type of global sections are the global elements of $\uom$. In fact the collection $\Gamma(\uom)$ of all such global section forms a Heyting algebra and represents the collection of all truth values in a topos logic. We thus obtain as an internal logic in a topos a multivalued logic which is of an intuitionistic type.
\subsection{Local sections} \label{sloc}
\begin{Definition} \label{def:lo}
A \textbf{local} or \textbf{partial section} of a presheaf X in $Sets^{\c^{op}}$ is an arrow
$\rho:U\rightarrow X$ where U is a subobject of the terminal object 1.
\end{Definition}
In a presheaf, a subobject U of 1 can 
either be the empty set $\emptyset$, or a singleton $\{*\}$. Thus for each object $C\in \c^{op}$ we either obtain the empty set $U(C)=\emptyset$ or a singleton $U(C)=\{*\}$, to each such singleton we then assign an element of $X(C)$. This assignment is said to be 
``closed
downwards", i.e. given a subobject U(A)=$\{*\}$ of 1 and a $\c$-morphisms $f:B\rightarrow A$
then we have U(B)=$\{*\}$, therefore $\rho_a(\{*\})=:\rho_a\in X(A)$ and $X(f)(\rho_a)=\rho_b$. 

\noindent
To better explain the above let us  
consider a category with 4 elements $\{A,B,C,D\}$ such that the following relations hold between 
the elements:
\[\xymatrix{
 A\ar[rr]^f\ar[dd]_i&&B\ar[dd]^g\\
 &&\\
 D\ar[rr]_p&&C\\
 }\] 
Given a subobject U of 1 we then have the following relations
\[\xymatrix{
 U(A)\ar[rr]^{U(f)}\ar[dd]_{U(i)}&&U(B)\ar[dd]^{U(g)}\\
 &&\\
 U(D)\ar[rr]_{U(p)}&&U(C)\\
 }\]
If U(A)=$\emptyset$ then U(f) is either the unique function $\emptyset\rightarrow\{*\}$ iff $U(B)=\{*\}$
or $\emptyset\rightarrow\emptyset$ iff $U(B)=\emptyset$.
If instead $U(A)=\{*\}$ then the only possibility is that $U(B)=\{*\}$ since there does not exist a
function 
$\{*\}\rightarrow\emptyset$.
Therefore $\rho$ assigns to particular subsets of objects $A\in \c$, elements
$\rho_A$. namely those objects $A\in\c$ for which $U(A)=\{*\}$. These objects A are called the domain of $\rho$ $(dom\hspace{.05in}\rho)$ and are such 
that the following conditions are satisfied:
\begin{itemize}
\item The domain is closed downwards i.e. if $A\in dom\hspace{.05in}\rho$ and if there exists a map 
$f:B\rightarrow A$ then $B\in dom\hspace{.05in}\rho$\vspace{.3in}
\item If $A\in dom\rho$ and if there exists a map $f:B\rightarrow A$, then the following condition 
is satisfied:
\begin{equation*}
X(f)(\rho_A)=\rho_B 
\end{equation*}
\end{itemize}  

\section{Exponential}

In $\Sets^{\c^{op}}$ the exponentiation can be defined as follows:
Consider a functor $F\in \Sets^{\c^{op}}$ such that, given an object $a\in\c$, F defines a functor\footnote{Recall from lecture 2/3 that $\c\downarrow a$ is the comma category described in example 2.3} $F_a:\c\downarrow a\rightarrow \Sets$. $F_a$ assigns to each object $f:b\rightarrow a\in\c\downarrow a$ an object $F_a(f):=F(b)$, and to each arrow
$h:(b,f)\rightarrow (b,g)$, such that 
\[\xymatrix{
b\ar[rr]^h\ar[ddr]_f&&c\ar[ddl]^g\\
&&\\
&a& \\
}\] commutes in $\c$, it assigns
the arrow $F(h):F(c)\rightarrow F(b)$.\\
Given this context we define the exponential $G^F:\c\rightarrow \Sets$ between the contravariant functors $F, G\in \Sets^{\c^{op}}$ as the functor with
\begin{itemize}
\item Objects:
\be
G^F(a)=Nat[F_a,G_a]
\ee
i.e. the elements of $G^F(a)$ are the collection of all natural transformations from 
$F_a$ to $G_a$\footnote{Here the functor $G_a:\c\downarrow a\rightarrow \Sets$ is the induced functor from $G$ thus, for a given element $f:b\rightarrow a$ it assigns the element $G_a(f):=G(b)$, and to each morphism $i:f\rightarrow g$ ( $g:c\rightarrow a$) it assigns the morphism $G(i):G(c)\rightarrow G(b)$.}.
\item Morphisms:
given an arrow $k:a\rightarrow d$ we get
\be
G^F(k):Nat[F_d,G_d]\rightarrow Nat[F_a,G_a]
\ee
\end{itemize}
To better understand this definition consider 
the natural transformation $\alpha\in Nat[F_d,G_d]$ and $\theta\in Nat[F_a,G_a]$. The action of $G^F(k)$ can then be illustrated as 
follows:
\[\xymatrix{
F_d\ar[dd]_{\alpha}&&\\
&&\\
G_d&&
}
\xymatrix{
&&\\
&\ar[r]^{G^F(k)}&}
\xymatrix{&&F_a\ar[dd]^{\theta}\\
&&\\
&&G_a\\
}\]
i.e an arrow in $G^F(k)$ assigns, to each natural transformation $\alpha$ from 
$F_d$ to $G_d$ a natural transformation $\theta$ from $F_a$ to $G_a$. The way in which the natural transformation $\theta$ is picked given $\alpha$ can be understood by considering the individual components. In particular, an element $h:c\rightarrow d\in\c\downarrow d$, then
\ba
\alpha_h:F_d(h)&\rightarrow &G_d(h)\\
F(d)&\rightarrow& G(d)
\ea
On the other hand an element $f:c\rightarrow a\in \c\downarrow a$ gives
\ba
\theta_f:F_a(f)&\rightarrow &G_a(f)\\
F(c)&\rightarrow &G(c)
\ea
Now, if we consider a map $k:d\rightarrow a\in \c$, such that the following diagram commutes
\[\xymatrix{
a&&d\ar[ll]^k\\
&&\\
&&c\ar[lluu]^f\ar[uu]_{h}\\
}\]
we require that $\theta_f=\alpha_{k\circ h}$.\\
In this formulation the evaluation function is the map: $ev:G^F\times F\rightarrow G$ in $\Sets^{\c^{op}}$, such that 
\ba
ev_a:G^F(a)\times F(a)&\rightarrow& G(a)\\
\langle\theta,x\rangle&\mapsto& ev_a(\langle\theta,x\rangle)=\theta_{1_a}(x)
\ea
where $\theta\in Nat[F_a,G_a]$ and $x\in F(a)$.

\chapter{Lecture 8}
In this lecture I will describe how topos quantum theory can be seen as a contextual quantum theory, in the sense that each element is defined as a collection of `context dependent' descriptions. Such context dependent descriptions will turn out to be classical snapshots.

I will then describe the above mentioned contexts which are  Abelian Von Neumann sub-algebras, the collection of which forms a category. What this implies is that, although locally quantum theory can be defined in terms of local classical snapshots, the global/quantum information is put back into the picture by the categorical structure of the collection of all such classical snapshots.
I will give an example of the category of Abelian Von Neumann sub-algebras for a 4 dimensional Hilbert space.

Given the definition of our base category we  will then define the topos analogue of the state space. This is the \emph{spectral presheaf}. I will end with a specific example on how such a presheaf is constructed in the case of a $4$ dimensional complex Hilbert space.

\section{The Notion of Contextuality in the Topos Approach}
In previous lectures we have seen how the Knochen-Specker theorem seems to imply that quantum theory is contextual, since values of quantities depend on which other quantities are being measured at the same time. However, that is not the notion of contextuality that we want to address here. In fact, in the topos approach to quantum theory there is another type of contextuality arising which is fundamental for the formulation of the theory. Surprisingly enough also this notion of contextuality is derived from the Kochen-Speker theorem, but in a very different fashion. In particular, although the K-S theorem prohibited us to define values for all quantities at the same time in a consistent way, it nonetheless allowed for the possibility of assigning values to commuting subsets of quantities. These commuting subsets can be considered as classical snapshots since all the peculiarities of quantum theory arise from non-commuting operators. Thus, with respect to these classical snapshots (contexts), quantum theory behaves like classical theory.

\noindent
The idea is then to define quantum theory locally with respect to these classical snapshots but, then, on has to consider all the information coming from the collection of these classical snapshots all at the same time. Thus, in this way, quantum theory could be seen as a collection of local classical approximations. 

Although it seems like one is cheating by doing this, it turns out that this is not the case. The reason being that the collection of all the above mentioned classical snapshots actually forms a category which means that it is always possible two relate (compare) any two contexts. 

What is happening is the following: we first consider different contexts which represent classical snapshots, we then define our quantum theory locally in terms of such classical snapshots, therefore in a way performing a classical approximation. The quantum information, which is lost at the local level is, however, put back into the picture by the categorical structure of the collection of all the classical contexts. In this way no information is lost and we therefore did not cheat. 

The category of classical snapshots we will be utilising is the category of abelian von Neuman sub-algebras of the algebra $\mb(\mh)$ of bounded operators on the Hilbert space. 
\subsection{Category of Abelian von Neumann Sub-Algebras}
In what follows we will first give the axiomatic definition of what this category is and, then, explain through an example what exactly these von Neumann algebras are.

In particular, consider the algebra of bounded operators on a Hilbert space which we denote as $\mb(\mh)$. A quantum system can be represented by a von Neumann algebra $N$ which is identified with a sub-algebra of $\mb(\mh)$. A von Neumann algebra is a *-algebra of bounded operators. We will give the technical definition of a von Neumann algebra below, however we will also give a concrete example on how it is constructed, which might be much more clear to understand. For all practical purposes it is not necessary to understand in details what a von Neumann algebra is, since the study of these algebras is quite complex and would be behind the scope of this lecture course. All that is needed is to understand roughly what they are, how they can be formed and the philosophical implications of  their usage in topos quantum theory.

In order to give the technical definition of a von Neumann algebra we need to make a regression in ring theory and give a few definitions (you are not required to know these definitions or learn them since they will not be examinable, but they might help to get a general understanding. )
\begin{Definition}
A  ring is a set $X$ on which two binary operations are defined :
\ba
+:X\times X&\rightarrow& X\\
(x_1,x_2)&\mapsto& x_1+x_2
\ea 
and 
\ba
\cdot:X\times X&\rightarrow& X\\
(x_1,x_2)&\mapsto& x_1\cdot x_2
\ea
called addition and multiplication. Generally a ring is denoted as $(X, +, \cdot)$ and it has to satisfy the following axioms:
\begin{itemize}
\item $(X, +)$ must be an abelian group under addition. 
\item $(X, \cdot)$ must be a monoid under multiplication.
\end{itemize}
\end{Definition}
In the above definition only the addition operation is required to be commutative, while the multiplication is not. However both operations are required to be associative. For this reason rings are often also called associative rings to distinguish them from non-associative rings, which are a  subsequent generalisation of the concept of a ring in which $(X, \cdot)$ is not a monoid but all that is required is that the multiplication operation be linear in each variables. 

Of particular importance to us is the concept of a *-ring which is defined as follows
\begin{Definition}
A *-ring is an associative ring with a map $ * : A \rightarrow  A$ s.t.
\ba
(x + y)^* &=& x^* + y^*\\
(x\cdot y)^* &=& y^*\cdot x^*\\
1^* &= &1\\
(x^* )^* &=& x
\ea
for all x,y in A. We say that $*$ is an anti-automorphism and an involution.
Elements such that $x^* = x$ are called self-adjoint or Hermitian.

\end{Definition}
Given all the above definition we are now ready to define what a von Neumann algebra is 
\begin{Definition}
A von Neumann algebra is a *-algebra\footnote{A *-algebra A is a *-ring that is a module over a commutative *-ring R, with the * agreeing.
on $R\subseteq A$} of bounded operators on a Hilbert space that is closed in the weak operator topology and contains the identity operator.
\end{Definition}
The above definition can be trivially extended to the notion of abelian von Neumann sub-algebras.

\noindent
The way in which von Neumann algebras are generated given a Hilbert space is through the double commutant theorem.
In particular, given an algebra $B\subset \mb(\mh)$ of bounded operators on a Hilbert space $\mh$, which contains the identity and is closed under taking the adjoint, then the commutant of such an algebra is 
\be
B^{'}:=\{\hat{A}\in \mb(\mh)|[\hat{A},\hat{B}]=0\;\forall\; \hat{B}\in B\}\subset\mb(\mh)
\ee
The double commutant is then the commutant of $B^{'}$: $(B^{'})^{'}=B^{''}$. This algebra $B^{''}$ is the von- Neumann algebra generated by $B$ iff $B=B^{''}$. In the example below we will give a concrete example of how such algebras are generated.
\\
\\
Given a Hilbert space $\mh$, the collection of all the abelian von Neumann sub-algebras, denoted as $\mv(\mh)$, forms a category. The importance of this lies in the fact that although each algebra only gives a partial classical information of the system, the collection of all such algebras retains the full quantum information, since the categorical structure relates information coming from different contexts. In particular, let us consider two contexts $V_1$ and $V_2$. If they have a non-trivial intersection $V_1\cap V_2$ then we have the following relation-arrows: $$V_1\leftarrow  V_1\cap V_2\rightarrow V_2$$

Now, given any self adjoint operator $\hat{A}$ in $ V_1\cap V_2$ it can be written as $g(\hat{B})$ for a self adjoint operator $\hat{B}\in V_1$ and a Borel function $g:\Rl\rightarrow\Rl$. On the other hand $\hat{A}=f(\hat{C})$ for $\hat{C}\in V_2$ and $f$ is another Borel function. It follows that $[\hat{A}, \hat{B}]=[\hat{A}, \hat{C}]=0$, however it is not necessarily the case that $[\hat{B}, \hat{C}]=0$.
Thus, although the elements in $\mv(\mh)$ are abelian the categorical structure knows about the relation of non commutative operators.

The formal definition of the category $\mv(\mh)$ is as follows:
\begin{Definition}
The category $\mv(\mh)$ of abelian von Neumann sub-algebras has 
\begin{itemize}
\item [1.] \emph{Objects}: $V\in \mv(\mh)$ abelian von Neumann sub-algebras.
\item [2.] \emph{Morphisms}: given two sub-algebras $V_1$ and $V_2$ there exists an arrow between them $i:V_1\rightarrow V_2$ iff $V_1\subseteq V_2$
\end{itemize}
\end{Definition}
From the definition it is easy to understand that $\mv(\mh)$ is a poset, whose ordering is given by subset inclusion.

It is interesting to understand what this poset structure actually means from a physics perspective. In particular, if we consider an algebra $V^{'}$ such that $V^{'}\subseteq V$, then the set of self-adjoint operators present in $V^{'}$, which we denote $V^{'}_{sa}$, will be smaller than the set of self-adjoint operators in $V$, i.e. $V^{'}_{sa}\subseteq V_{sa}$. Since self-adjoint operators represent physical quantities, the context $V^{'}$ contains less physical information, so that, by viewing the system from the context $V^{'}$, we know less about it then when viewing it form the context $V$. This idea represents a type of coarse graining which takes place when going from a context with more information $V$ to a context with less information $V^{'}$. If we went the revers direction we would instead have a process of fine graining. 

\noindent 
This idea of coarse graining is central in the formulation of the topos quantum theory. We will see later in the course how it is actually implemented in detail.
\subsection{Example}
Let us consider a four dimensional Hilbert space $\mh=\Cl^4$. The first step is to identify the poset of abelian von Neumann sub-algebras $\mv(\Cl^4)$. Such algebras are sub-algebras of the algebra $\mathcal{B}(\Cl^4)$ of all bounded operators on $\mh$. Since the Hilbert space is $\Cl^4$, $\mathcal{B}(\Cl^4)$ is the algebra of all $4\times 4$ matrices with complex entries which act as linear transformations on $\Cl^4$. \\
In order to form the abelian von Neumann sub-algebras one considers an orthonormal basis $(\psi_1,\psi_2,\psi_3,\psi_4)$ and projection operators $(\hat{P}_1, \hat{P}_2, \hat{P}_3, \hat{P}_4)$ which project on the one-dimensional sub-spaces $\Cl\psi_1$, $\Cl\psi_2$, $\Cl\psi_3$, $\Cl\psi_4$, respectively. One possible von Neumann sub-algebra $V$ is, then, generated by the double commutant
of collections of the above projection operators, i.e. $V=lin_{\Cl}(\hat{P}_1, \hat{P}_2, \hat{P}_3, \hat{P}_4)$.\\
In matrix notation possible representatives for the projection operators are
\[\hat{P}_1=\begin{pmatrix} 1& 0& 0&0\\	
0&0&0 &0\\
0&0&0&0\\
0&0&0&0	
  \end{pmatrix}\;\;\;\;
\hat{P}_2=\begin{pmatrix} 0& 0& 0&0\\	
0&1&0 &0\\
0&0&0&0\\
0&0&0&0	
  \end{pmatrix}\;\;\;\;
\hat{P}_3=\begin{pmatrix} 0& 0& 0&0\\	
0&0&0 &0\\
0&0&1&0\\
0&0&0&0	
  \end{pmatrix}\;\;\;\;
  \hat{P}_4=\begin{pmatrix} 0& 0& 0&0\\	
0&0&0 &0\\
0&0&0&0\\
0&0&0&1	
  \end{pmatrix}
\]
The largest abelian von Neuamnn sub-algebra generated by the above projectors is $V=lin_{\Cl}(\hat{P}_1, \hat{P}_2, \hat{P}_3, \hat{P}_4)$, i.e. the algebra consisting of all $4\times 4$ diagonal matrices with complex entries on the diagonal. Since this algebra is the largest, i.e. not contained in any other abelian sub-algebra of $\mathcal{B}(\Cl^4)$ it is called \emph{maximal}.\\
Any change of bases $(\psi_1,\psi_2,\psi_3,\psi_4)\rightarrow (\rho_1, \rho_2, \rho_3, \rho_4)$ would give another maximal von Neumann sub-algebra $V^{'}=lin_{\Cl}(\rho_1, \rho_2, \rho_3, \rho_4)$. In fact, there are uncountably many such maximal algebras. If two basis are related by a simple permutation or phase factor then the abelian von Neumann sub-algebras they generate are the same.\\
Now considering again our example, the algebra $V$ will have many non maximal sub-algebras which however can be divided into two kinds as follows:
\be
V_{\hat{P}_i\hat{P}_j}=lin_{\Cl}(\hat{P}_i, \hat{P}_j, \hat{P}_k+\hat{P}_l)=\Cl\hat{P}_i+\Cl \hat{P}_j+\Cl (\hat{P}_k+\hat{P}_l)=\Cl\hat{P}_i+\Cl \hat{P}_j+\Cl(\hat{1}- \hat{P}_i+\hat{P}_j)
\ee
for $i\neq j\neq k \neq l\in\{1,2,3,4\}$\\
and
\be
V_{\hat{P}_i}=lin_{\Cl}(\hat{P}_i, \hat{P}_j+\hat{P}_k+\hat{P}_l)=\Cl\hat{P}_i+\Cl(\hat{P}_j+\hat{P}_k+\hat{P}_l)=\Cl\hat{P}_i+\Cl(\hat{1}-\hat{P}_i)
\ee
for $i\neq j\neq k \neq l=1,2,3,4$\\
Again there are uncountably many non maximal abelian sub-algebras. It is the case though that different maximal sub-algebras have common non-maximal sub-algebras, as could be the case that non-maximal abelian sub-algebras contain the same non-maximal abelian sub-algebra. \\
Thus, for example, the context $V$ above contains all the sub-algebras $V_{ij}$ and $V_i$ for $i,j\in\{1,2,3,4\}$. Now consider other 4 pair wise orthogonal projection operators $\hat{P}_1$,  $\hat{P}_2$,  $\hat{Q}_3$,  $\hat{Q}_4$, such that the maximal abelian von Neumann algebra $V^{'}=lin_{\Cl}(\hat{P}_1,\hat{P}_2, \hat{Q}_3,\hat{Q}_4)\neq V$. We then have that
\be
V\cap V^{'}=\{V_{\hat{P}_1}, V_{\hat{P}_2}, V_{\hat{P}_1,\hat{P}_2}\}
\ee
From the above discussion it is easy to deduce that the sub-algebras $V_{\hat{P}_i}$ are contained in all the other sub-algebras which contain the projection operator $\hat{P}_i$ and all the sub-algebras $V_{\hat{P}_i\hat{P}_j}$ are contained in all the sub-algebras which contain both projection operators $\hat{P}_i$ and $\hat{P}_j$.\\
We should mention that there is also the trivial algebra $V^{''}=\Cl\hat{1}$ but we will not consider such algebra when considering the category $\mv(\mh)$ since, otherwise, as will be clear later on, we will never end up with a proposition being false, but the minimal truth value we would end up would be the trivially true one.

\section{Topos Analogue of the State Space}
We would now like to define the topos analogue of the state space. The way in which we would like to construct such a state space is in analogy with how it is constructed in classical physics. In particular we would like a state space which allows a definition of physical quantities in terms of maps from the state space to the reals, as is the case in classical physics. 

\noindent Since we are in the realm of presheaves on $\mv(\mh)$, the state space will itself be a presheaf, thus it will be defined context wise, i.e. for each abelian von Neumann algebra $V\in \mv(\mh)$. It is precisely of such algebras that we will take advantage of when trying to define the state space. In fact, each algebra $V$ has associated to it its Gelf'and spectrum which is the topological space $\us_V$ of all multiplicative linear functionals of norm $1$ on $V$, i.e. $\us_V:=\{\lambda:V\rightarrow \Cl|\lambda(\hat{1})=1\}$. The property of being multiplicative means that
\be
\lambda(\hat{A}\hat{B})=\lambda(\hat{A})\lambda(\hat{B});\;\forall \hat{A},\hat{B}\in V
\ee  
So the elements $\lambda$ of the spectrum $\us_V$ are, in essence, algebra homomorphisms from $V$ to $\Cl$. The topology on $\us_V$ is that of a compact Hausdorff space in the weak *-topology.

Now, what is interesting is the action of such homomorphisms $\lambda\in\us_V$ on self adjoint operators $\hat{A}\in V$. In fact it turns out that such maps $\lambda$ actually represent valuations which respect the FUNC principle. To see this consider an operator $\hat{A}\in V$, then, for each element $\lambda$ of the spectrum $\us_V$ we obtain a value $\lambda(\hat{A})\in sp(\hat{A})$ of $A$. On the other hand, for each element $a$ of the spectrum of $\hat{A}\in V$, i.e., $a\in sp(\hat{A})$ there exists a corresponding element $\lambda_i$ such that $a$ is defined as $a=\lambda_i(\hat{A})$. 

Moreover, given a Borel function $g : \Rl\rightarrow \Rl$, then 
\be
\lambda(g(\hat{A})) = g(\lambda(\hat{A}))
\ee
This is precisely the FUNC principle. Therefore the elements of the Gel'fand spectrum can each be interpreted as (different) valuations, i.e. maps which send each self adjoint operator to an element of its spectrum such that FUNC holds.

Given the topological space $\us_V$ it is possible to represent a self-adjoint operator $\hat{A}\in V$ as a map from $\us_V$ to $\Cl$. This is because of the existence of the \emph{GelÕfand representation theorem} which states that each von Neumann algebra $V$ is isomorphic \footnote{ Technically it is an isometrically *-isomorphic (i.e., isomorphic as a $C^*$-algebra), but this precise definition does not really matter here, we will simply call it isomorphisms.} to the algebra of continuous, complex-valued functions (denoted as $C(\us_V)$ ) on its GelÕfand spectrum $\us_V$. That is to say the following map is an isomorphisms
\ba
V&\rightarrow&C(\us_V)\\
\hat{A}&\mapsto&(\bar{A}:\us_V\rightarrow \Cl)
\ea
where $\bar{A}$ is the Gel'fand transform of the operator $\hat{A}$ and is defined as $\bar{A}(\lambda) :=\lambda(\hat{A})$ for all $\lambda\in\us_V$. If $\hat{A}$ is self adjoint then $\bar{A}$ is real valued and is such that $\bar{A}(\us_V)=sp(\hat{A})$. 

Thus, for each context $V\in \mv(\mh)$ we have managed to reproduce a situation analogous to classical physics in which self-adjoint operators are identified with functions from a space to the reals. In this sense the topological space $\us_V$ can be interpreted as a local state space, one for each $V\in\mv(\mh)$. Obviously the complete quantum picture is only given when we consider the collection of all the local state spaces, since not all operators are contained in a single algebra $V$. It is precisely such a collection of local state spaces which will define the topos analogue of the state space. This will be called the spectral presheaf and it is defined as follows:
\begin{Definition}
The spectral presheaf, $\Sig$, is the covariant functor from the
category $\V(\Hi)^{op}$ to $\Sets$ (equivalently, the
contravariant functor from $\mathcal{V(H)}$ to $\Sets$) defined
by:
\begin{itemize}
\item \textbf{Objects}: Given an object $V$ in $\V(\Hi)^{op}$, the associated set $\Sig(V)=\us_V$ is defined to be the Gel'fand spectrum of the (unital) commutative von Neumann sub-algebra $V$, i.e. the set of all multiplicative linear functionals $\lambda:V\rightarrow \Cl$, such that $\lambda(\hat{1})=1$.
\item\textbf{Morphisms}: Given a morphism $i_{V^{'}V}:V^{'}\rightarrow V$ ($V^{'}\subseteq V$) in $\V(\Hi)^{op}$, the associated function $\Sig(i_{V^{'}V}):\Sig(V)\rightarrow
\Sig(V^{'})$ is defined for all $\lambda\in\Sig(V)$ to be the
restriction of the functional $\lambda:V\rightarrow\Cl$ to the
sub-algebra $V^{'}\subseteq V$, i.e.
$\Sig(i_{V^{'}V})(\lambda):=\lambda_{|V^{'}}$.
\end{itemize}
\end{Definition}
\subsection{Example}
Given the category $\mc(\Cl^4)$ defined in the previous example will define the spectral presheaf.
Let us first consider the maximal abelian sub-algebra $V=lin_{\Cl}(\hat{P}_1, \hat{P}_2, \hat{P}_3, \hat{P}_4)$, the Gel'fand spectrum $\us_V$ (which has a discrete topology) of this algebra contains 4 elements
\be
\lambda_i(\hat{P}_j)=\delta_{ij}\;\;\;\; (i=1,2,3,4)
\ee 
We then consider the sub-algebra $V_{\hat{P}_1\hat{P}_2}=lin_{\Cl}(\hat{P}_1,\hat{P}_2, \hat{P}_3+\hat{P}_4)$. Its Gel'fand spectrum $\us_{V_{\hat{P}_1\hat{P}_2}}$ will contain the elements
\be
\us_{V_{\hat{P}_1\hat{P}_2}}=\{\lambda^{'}_1, \lambda^{'}_2, \lambda^{'}_3 \}
\ee
such that $\lambda^{'}_1(\hat{P}_1)=1$, $\lambda^{'}_2(\hat{P}_2)=1$ and $\lambda^{'}_3(\hat{P}_3+\hat{P}_4)=1$, while all the rest will be zero.\\
Since $V_{\hat{P}_1\hat{P}_2}\subseteq V$, there exists a morphisms between the respective spectra as follows ( for notational simplicity we will denote $V_{\hat{P}_1\hat{P}_2}$ as $V^{'}$):
\ba
\us_{VV^{'}}:\us_V&\rightarrow &\us_{V^{'}}\\
\lambda&\mapsto&\lambda_{|V^{'}}
\ea
Such that we obtain the following:
\ba
\us_{VV^{'}}(\lambda_1)&=&\lambda_1^{'}\\
\us_{VV^{'}}(\lambda_2)&=&\lambda_2^{'}\nonumber\\
\us_{VV^{'}}(\lambda_3)&=&\lambda_3^{'}\nonumber\\
\us_{VV^{'}}(\lambda_4)&=&\lambda_3^{'}\nonumber
\ea
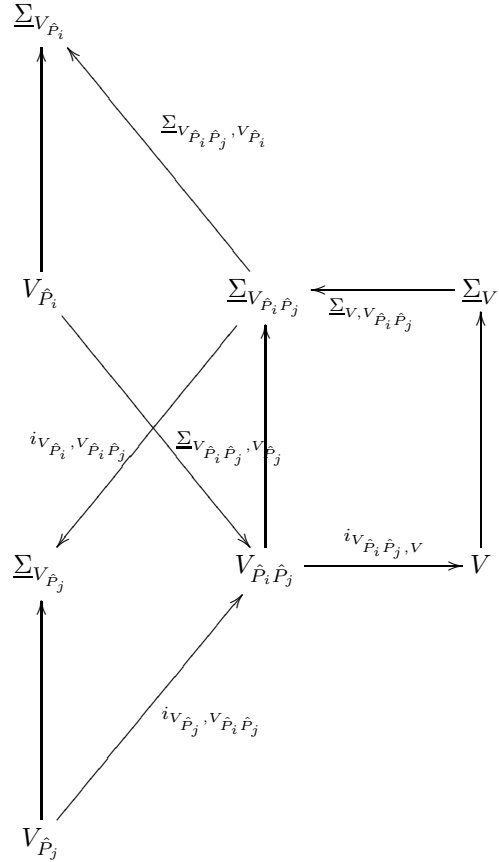
\[\xymatrix{
\us_{V_{\hat{P}_i}}&&&&\\
&&&&\\
&&&&\\
V_{\hat{P}_i}\ar[uuu]\ar[dddrr]_{i_{V_{\hat{P}_i},V_{\hat{P}_i\hat{P}_j}}\;\;\;}&&\us_{V_{\hat{P}_i\hat{P}_j}}\ar[dddll]^{\;\;\;\us_{V_{\hat{P}_i\hat{P}_j},V_{\hat{P}_j}}}\ar[lluuu]_{\us_{V_{\hat{P}_i\hat{P}_j},V_{\hat{P}_i}}}&&\us_V\ar[ll]^{\us_{V,V_{\hat{P}_i\hat{P}_j}}}\\
&&&&\\
&&&&\\
\us_{V_{\hat{P}_j}}&&V_{\hat{P}_i\hat{P}_j}\ar[uuu]\ar[rr]^{\;\;\;\ i_{V_{\hat{P}_i\hat{P}_j,V}}}&&V\ar[uuu]\\\
&&&&\\
&&&&\\
V_{\hat{P}_j}\ar[uuu]\ar[rruuu]_{i_{V_{\hat{P}_j},V_{\hat{P}_i\hat{P}_j}}}&&&&\\
}\]

From the above simple example we can generalise the definition of the spectrum for all sub-algebras $V^{'}\subseteq V$. In particular we get 
\be
\us_{V_{\hat{P}_i,\hat{P}_j}}=\{\lambda_i,\lambda_j,\lambda_{kl}\}
\ee 
where
\ba
\lambda_i(\hat{P}_j)&=&\delta_{ij}\\
\lambda_{kl}(\hat{P}_k+\hat{P}_l)&=&1\\
\ea
and all the rest equals zero. On the other hand for contexts $V_{\hat{P}_i}$ we obtain:
\be
\us_{V_{\hat{P}_i}}=\{\lambda_i,\lambda_{jkl}\}
\ee
where
\ba
\lambda_i(\hat{P}_i)&=&1\\
\lambda_{jkl}(\hat{P}_j+\hat{P}_k+\hat{P}_k)=1
\ea
and all the rest is equal to zero.
\chapter{Lecture 9}
In this lecture I will do the following:
\begin{enumerate}
\item[i)] I will first of all introduce the very important concept of daseinisation;
\item[ii)] I will then give examples of daseinisation;
\item[iii)] Such a concept will be used to define the topos analogue of a proposition.
\end{enumerate}
I will then give a concrete example on how a proposition regarding the value of the spin of a particle is defined, for the case of a $4$ dimensional Hilbert space.

\section{Propositions}
We will now describe how certain terms of type $P(\Sigma)$ (sub-objects of the state object) are represented in $\Sets^{\mv(\mh)^{op}}$, namely propositions. These, represented by projection operators in quantum theory, are identified with clopen (both open and closed) sub-objects of the spectral presheaf.
 A \emph{clopen} subobject
$\ps{S}\subseteq\Sig$ is an object such that, for each context
$V\in \mathcal{V(H)}$, the set $\ps{S}(V)$ is a clopen subset of $\Sig(V)$, where the latter is equipped
with the usual compact and Hausdorff spectral topology. 
We will now show, explicitly, how propositions are defined.

As a first step we have to introduce the concept of
`daseinization'. Roughly speaking, what daseinization does is to
approximate operators so as to `fit' into any given context $V$.
In fact, because the formalism defined so far is
contextual, any  proposition  one wants to consider has to be
studied within (with respect to) each context $V\in\V(\Hi)$.

To see  how this works consider the case in which we would like
to analyse the projection operator $\hat{P}$, which corresponds via the
spectral theorem to the proposition ``$A\in\Delta$''\footnote{It should be noted that different propositions correspond to the same projection operator, i.e. the mapping from propositions to projection operators is many to one. Thus, to account for this, one is really associating equivalence class of propositions to each projection operator. The reason why von Neumann algebras were chosen instead of general C$^*$ algebras is precisely because all projections representing propositions are contained in the former, but not necessarily in the latter.}. In
particular, let us take a context $V$ such that $\hat{P}\notin
P(V)$ (the lattice of projection operators in $V$). We, somehow need to define a
projection operator which does belong to $V$ and which is related,
in some way, to our original projection operator $\hat{P}$. This
can be achieved by approximating $\hat{P}$
from above in $V$, with the `smallest' projection operator in $V$,
greater than or equal to $\hat{P}$. More precisely, the
\emph{outer daseinization},
 $\delta^o(\hat P)$, of $\hat P$ is defined at each context $V$ by
 \begin{equation}
\delta^o(\hat{P})_V:=\bigwedge\{\hat{R}\in
P(V)|\hat{R}\geq\hat{P}\}
\end{equation}
Since projection operators represent propositions, $\delta^o(\hat{P})_V$ is a coarse graining of the proposition $``A\in \Delta"$.

This process of outer daseinization takes place for all contexts
and hence gives, for each projection operator $\hat{P}$, a
collection of daseinized projection operators, one for each
context V, i.e.,
\begin{align}
\hat{P}\mapsto\{\delta^o(\hat{P})_V|V\in\V(\Hi)\}
\end{align}
Because of the Gel'fand transform, to each operator $\hat{P}\in
P(V)$ there is associated the map $\bar{P}:\Sig_V\rightarrow\Cl$,
which takes values in $\{0,1\}\subset\Rl\subset\Cl$ since
$\hat{P}$ is a projection operator. Thus, $\bar{P}$ is a
characteristic function of the subset
$S_{\hat{P}}\subseteq\Sig(V)$ defined by
\begin{equation}
S_{\hat{P}}:=\{\lambda\in\Sig(V)|\bar{P}(\lambda):=\lambda(\hat{P})=1\}
\end{equation}
Since $\bar{P}$ is continuous with respect to the spectral
topology on $\underline\Sigma(V)$, then
$\bar{P}^{-1}(1)=S_{\hat{P}}$ is a clopen subset of
$\underline\Sigma(V)$, since both $\{0\}$ and $\{1\}$ are clopen
subsets of the Hausdorff space $\Cl$.

Through the Gel'fand transform it is then possible to define a
bijective map between projection operators, $\delta^o(\hat{P})_V\in
P(V)$, and clopen subsets of $\Sig_V$ where,
for each \emph{context} V,
\begin{equation}
S_{\delta^o(\hat{P})_V}:=\{\lambda\in\Sig_V|\lambda
(\delta^o(\hat{P})_V)=1\}
\end{equation}

This correspondence between projection operators and clopen
sub-objects of the spectral presheaf $\underline\Sigma$, which we denote as $Sub_{cl}(\us)$,  implies the
existence of a lattice homeomorphism for each $V$
\begin{equation}\label{equ:smap}
\mathfrak{S}:P(V)\rightarrow \Sub_{cl}(\Sig)_V\hspace{.2in}
\end{equation}
such that
\begin{equation}
\delta^o(\hat{P})_V\mapsto
\mathfrak{S}(\delta^o(\hat{P})_V):=S_{\delta^o(\hat{P})_V}
\end{equation}
where $Sub_{cl}(\us)_V$ is the lattice of subsets of the spectrum $\us_V$ with lattice operations given by intersection and union while the lattice ordering is given by subset inclusion.

It can be shown that the collection of subsets
$\{S_{\delta(\hat{P})_V}\}$, $V\in\mathcal{V(H)}$, induces a subobject
of $\Sig$. 

\noindent
In order to understand how this is done let us first give the definition of what a general sub-object of the topos analogue of the state space actually is.
\begin{Definition}
A sub-object $\ps{S}$ of the spectral presheaf $\us$ is a contravariant functor $\ps{S} :\mv(\mh)\rightarrow \Sets$ such that:
\begin{itemize}
\item $\ps{S}_V$ is a subset of $\us_V$ for all $V\in\mv(\mh)$ .
\item Given a map $i_{V^{'}V}:V^{'}\subseteq V$ , then $\ps{S}(i_{V^{'}V} ) : \ps{S}_V \rightarrow \ps{S}_V^{'}$ is simply the restriction of the map $\us(i_{V^{'}V} )$ to the subset $\ps{S}_V\subseteq \us_V$, thus it is given by $\lambda\mapsto\lambda_{|V^{'}}$.
\end{itemize}
\end{Definition}

Obviously, for clopen sub-objects we simply require that $\ps{S}_V$ be clopen in the above definition.

\begin{Theorem}
For each projection operator $\hat{P}\in P(\mh)$, the collection
\be
\ps{\delta(\hat{P})}:=\{S_{\delta(\hat{P})_V}|V\in\mathcal{V(H)}\}
\ee
forms a (clopen) sub-object of the spectral presheaf .
\end{Theorem}

\begin{proof}
We already know that for each $V\in\mv(\mh)$, $S_{\delta(\hat{P})_V}\subseteq \us_V$. Therefore, what we need to show is that these clopen subsets get mapped one to another by the presheaf morphisms. To see that this is the case consider an element $\lambda\in S_{\delta(\hat{P})_V}$. Given any $V^{'}\subseteq V$, then by the definition of daseininsation we get $\delta^o(\delta^o(\hat{P})_{V})_{V^{'}}=\bigwedge\{\hat{\alpha}\in P(V^{'})|\hat{\alpha}\geq  \delta^o(\hat{P})_V\}\geq \delta^o(\hat{P})_V$. Therefore, if $\delta^o(\hat{P})_{V^{'}}-\delta^o(\hat{P})_V=\hat{\beta}$, then 
$\lambda(\delta(\hat{P})_{V^{'}})=\lambda(\delta(\hat{P})_V)+\lambda(\hat{\beta})=1$ since $\lambda(\delta(\hat{P})_V)=1$ and $\lambda(\hat{\beta})\in\{0,1\}$. Therefore  
\be
\{\lambda_{|V^{'}}| \lambda\in S_{\delta(\hat{P})_V}\} \subseteq\{ \lambda\in S_{\delta(\hat{P})_{V^{'}}} \}
\ee
however $\lambda_{|V^{'}}$ is precisely $\us(i_{V^{'}V})\lambda$ therefore 
\be
\{\lambda_{|V^{'}}| \lambda\in S_{\delta(\hat{P})_V}\}=\us(i_{V^{'}V})S_{\delta(\hat{P})_V}
\ee
It follows that $\ps{\delta(\hat{P})}$ is a sub-object of $\us$.
\end{proof}

We can now define the (outer) daseinization as
a mapping from the projection operators to the subobject of the
spectral presheaf given by
\begin{align}\label{ali:glob}
\delta:&P(\Hi)\rightarrow \Sub_{cl}(\Sig)\\
&\hat{P}\mapsto(\mathfrak{S}(\delta^o(\hat{P})_V))_{V\in\V(\Hi)}=:\ps{\delta(\hat{P})}
\end{align}
We will sometimes denote $\mathfrak{S}(\delta^o(\hat{P})_V)$ as
$\ps{\delta(\hat{P})}_V$.
\\
Since the sub-objects of the spectral presheaf form a Heyting
algebra, the above map associates propositions to a distributive
lattice. Actually, it is first necessary to show that the
collection of \emph{clopen} sub-objects of $\underline\Sigma$ is a
Heyting algebra. We will report the proof below.
\begin{Theorem}
The collection, $Sub_{cl}(\us)$, of all clopen sub-objects of $\us$ is a Heyting algebra.
\end{Theorem}
\begin{proof}
First of all let us consider how the logical connectives are defined. 

\noindent
\textbf{The Ô$\wedge$Õ- and Ô$\vee$Õ-operations}. 
Given two sub-objects $\ps{T}$ and $\ps{S}$ of $\us$, then the Ô$\wedge$Õ- and Ô$\vee$Õ-operations are defined by
\ba
(\ps{S}\wedge \ps{T})_V &:=& \ps{S}_V \cap \ps{T}_V\\ 
(\ps{S}\vee \ps{T})_V &:=& \ps{S}_V\cup \ps{T}_V 
\ea
for all contexts $V\in\mv(\mh)$ . From the properties of open and closed subsets, it follows that if $\ps{S}_V$ and $ \ps{T}_V$ are clopen as subsets then so are $\ps{S}_V \cap \ps{T}_V$ and $  \ps{S}_V\cup \ps{T}_V$.\\

\noindent
\textbf{The zero and unit elements}. The zero element in $Sub_{cl}(\us)$ is the empty sub-object
\be
\ps{0}:=\{\emptyset_V|V\in Ob(\mv(\mh))\}
\ee 
Where $\emptyset_V$ is the empty subset of $\us_V$ and $Ob(\mv(\mh)$ simply indicates the objects in the category $\mv(\mh)$ . \\

\noindent 
The unit element in $Sub_{cl}(\us)$ is the unit sub-object. 
\be
\us:=\{\us_V|V\in Ob(\mh(\mv))\}
\ee
Clearly both $\ps{0}$ and $ \us$ are clopen sub-objects of $\us$.\\

\noindent
\textbf{ The `$\Rightarrow$'-operation}. We have seen in previous lectures that the negation operation in a Heyting algebra is given by the relative pseudo complement. In particular $\neg \ps{S}:=\ps{S}\Rightarrow \ps{0}$. To understand exactly how such an operation is defined let us first describe $\ps{S}\Rightarrow \ps{T}$. This is 
\ba
(\ps{S}\Rightarrow \ps{T})_V&:=&\{\lambda\in\us_V|\forall V^{'}\subseteq V\text{ if }\us(i_{V^{'}V})(\lambda)\in\ps{S}_{V^{'}}\text{ then }\us(i_{V^{'}V})(\lambda)\in\ps{T}_{V^{'}}\}\\
&=&\{\lambda\in\us_V|\forall V^{'}\subseteq V\text{ if }\lambda_{|V^{'}}\in\ps{S}_{V^{'}}\text{ then }\lambda_{|V^{'}}\in\ps{T}_{V^{'}}\}
\ea
From the above it follows that the negation operation in the Heyting algebra $Sub_{cl}(\us)$ is defined as follows:
\ba
(\neg\ps{S})_V&:=&(\ps{S}\Rightarrow \ps{0})_V=\{\lambda\in \us_V|\forall V^{'}\subseteq V,\us(i_{V^{'}V})\lambda\notin \ps{S}_{V^{'}}\}\\
&=&\{\lambda\in \us_V|\forall V^{'}\subseteq V,\lambda_{|V^{'}}\notin \ps{S}_{V^{'}}\}
\ea
It is also possible to write the negation in terms of the complement of sets as follows:
\be
(\neg\ps{S})_V=\bigcap_{V^{'}\subseteq V}\{\lambda\in\us_V|\lambda_{|V^{'}}\in\ps{S}^c_{V^{'}}\}
\ee
where $\ps{S}^c_{V^{'}}$ represents the standard complement of the set $\ps{S}_{V^{'}}$. Since $\ps{S}_{V^{'}}$ is clopen so will $\ps{S}^c_{V^{'}}$. The map $\us(i_{V^{'}V}):\us_V\rightarrow\us_V^{'}$ is continuous and surjective, thus $\us(i_{V^{'}V})^{-1}(\ps{S}^c_{V^{'}})$ is clopen. Such a subset is defined as
\be
\us(i_{V^{'}V})^{-1}(\ps{S}^c_{V^{'}}) =\{\lambda\in \us_V|\lambda_{|V^{'}}\in\ps{S}^c_{V^{'}}\} 
\ee
Substituting for the formula of the negation operation we obtain
\be\label{equ:net}
(\neg\ps{S})_V=\bigcap_{V^{'}\subseteq V}\us(i_{V^{'}V})^{-1}(\ps{S}^c_{V^{'}})
\ee
However, the right hand side of the above formula is not guaranteed to be clopen, in fact it is closed and it would only be clopen if the set $\{V^{'}|V^{'}\subseteq V\}$ over which the intersection ranges is actually finite. 

\noindent
Now we know that the collection of all clopen subsets for each $\us_V$ is a complete lattice, thus given a family of decreasing subsets there will exist a limiting point of such subsets which will belong to the lattice. 

\noindent
In our case the collection of $ \us(i_{V^{'}V})^{-1}(\ps{S}^c_{V^{'}})$ for all $\{V^{'}|V^{'}\subseteq V\}$ is a decreasing net of clopen subsets of $\us_V$. This means that if $V^{''}\subseteq V^{'}$ and $\lambda_{V^{''}}\in S_{V^{''}}^c$ then $\lambda_{|V}\in S_{V^{'}}^c$. That this is the case can be proved by contradiction, in fact if $\lambda_{|V}\in S_{V^{'}}$ then $\us_{V^{''}V^{'}}\lambda_{|V}=\lambda_{|V^{''}}\in S_{V^{''}}$ which would be a contradiction. Therefore $\us(i_{V^{''}V^{'}})^{-1}(\ps{S}^c_{V^{''}})\subseteq \us(i_{V^{'}V})^{-1}(\ps{S}^c_{V^{'}})$. Therefore the right hand side of \ref{equ:net} represents a decreasing net of clopen sub-sets of $\us_V$. If we now define the limit point of such a net and call it $(\neg S)_V $ we have a definition of the negation of an element as a clopen subset. Thus we define
\ba
(\neg S)_V&:=&int\bigcap_{V^{'}\subseteq V}\us(i_{V^{'}V})^{-1}(\ps{S}^c_{V^{'}})\\
&=&int\bigcap_{V^{'}\subseteq V}\{\lambda\in\us_V|\lambda_{|V^{'}}\in (\ps{S}^c_{V^{'}})\}
\ea
\end{proof}

Particular properties of the daseinization map worth
mentioning are:
\begin{enumerate}
\item $\ps{\delta(P\vee Q)}=\ps{\delta(P)}\vee\ps{\delta(Q)}$, i.e. it preserves the ``or" operation.
\item $\ps{\delta(P\wedge Q)}\leq\ps{\delta(P)}\wedge\ps{\delta(Q)}$, i.e. it does not preserve the ``and" operation.
\item If $\hat{P}\leq\hat{Q}$, then $\ps{\delta(P)}\leq\ps{\delta(Q)}$.
\item The daseinisation map is injective but not surjective.
\item $\ps{\delta(\hat{0})}=\{\emptyset_V|V\in Ob(\mv(\mh))\}$.
\item $\ps{\delta(\hat{1})}=\{\us_V|V\in Ob(\mv(\mh))\}$.
\end{enumerate}

\subsection{Physical Interpretation of Daseinisation}
What exactly does it mean to daseinise a projection? Let us consider a projection $\hat{P}$ which represents the proposition $A\in\Delta$. We now consider a context $V$ such that $\hat{P}\notin V$, thus we approximate this projection so as to be in $V$ obtaining $\delta^o(\hat{P})_V$. If the projection $\delta^o(\hat{P})_V$ is a spectral projector of the operator, $\hat{A}$, representing the quantity, $A$, then it represents the proposition $A\in\Gamma$ where $\Delta\subseteq \Gamma$.
Therefore, the mapping
\ba
\delta^o_V:P(\mh)&\rightarrow& P(\mv)\\
\hat{P}&\mapsto&\delta^o(\hat{P})_V
\ea
is the mathematical implementation of the idea of coarse graining of propositions, i.e. of generalizing a proposition.\\
If, on the other hand, $\delta^o(\hat{P})_V$ is not a spectral projector of the operator $\hat{A}$ representing the quantity, $A$, then $\delta^o(\hat{P})_V$ represents the proposition $B\in\Delta^{'}$. The physical quantity $B$ is now represented by the projection operator $\hat{B}\in P(V)$. Given the fact that $\hat{P}\leq\delta^o(\hat{P})_V$, the proposition $B\in\Delta^{'}$ is a coarse graining of $A\in\Gamma$, in fact a general form of $B\in\Delta^{'}$ could be $f(A)\in\Gamma$, for some Borel function $f: sp(\hat{A})\rightarrow \Rl$.\\
Obviously, for many contexts it is the case that $\delta^o(\hat{P})_V=\hat{1}$, which is the most general proposition of all.\\
From the analysis above we can deduce that, in this framework, there are two types of propositions:
\begin{enumerate}
\item[i)]\emph{Global propositions}, which are the propositions we start with and which we want to represent in  various contexts, i.e. $(A\in\Delta)$.
\item[ii)]\emph{Local propositions}, which are the individual coarse graining of the global propositions, as referred to individual contexts $V$.
\end{enumerate}
Thus, for every global proposition we obtain a collection of local propositions
\be
\hat{P}\rightarrow (\delta^o(\hat{P})_V)_{V\in\V(\Hi)}
\ee
In the topos perspective we consider the collection of all these local propositions at the same time, as exemplified by equation \ref{ali:glob}.
\subsection{Example}
To illustrate the concept of daseinisation of propositions let us consider a 2 spin system. We are interested in the spin in the $z$-direction, which is represented by the physical quantity $S_z$. In particular, we want to consider the following proposition $S_z\in [1.3,2.3]$. Since the total spin in the $z$ direction can only have values $-2$, $0$, $2$, the only value in the interval $ [1.3,2.3]$ which $S_z$ can take is $2$.\\
The self-adjoint operator representing $S_z$ is 
\[\hat{S}_z=\begin{pmatrix} 2& 0& 0&0\\	
0&0&0 &0\\
0&0&0&0\\
0&0&0&-2	
  \end{pmatrix}
  \]
 The eigenstate with eigenvalue 2 would be $\psi=(1,0,0,0)$, whose associated projector $\hat{P}:=\hat{E}[S_z\in [1.3,2.3]]=|\psi\rangle\langle\psi|$ would be 
  \[\hat{P}_1=\begin{pmatrix} 1& 0& 0&0\\	
0&0&0 &0\\
0&0&0&0\\
0&0&0&0	
  \end{pmatrix}
  \]
From our definition of $\mv(\Cl^4)$ we know that the operator $\hat{S}_z$ is contained in all algebras which contain the projector operators $\hat{P}_1$ and $\hat{P}_4$. These algebras are: i) the maximal algebra $V$ and ii) the non maximal sub-algebra $V_{\hat{P}_1\hat{P}_4}$. We will now analyse how the proposition $S_z\in [1.3,2.3]$, represented by the projection operator $\hat{P}_1$, gets represented in the various abelian von Neumann algebra in $\mv(\Cl^4)$.
\begin{enumerate}
\item \emph{Context $V$ and its sub-algebras}.\\\\
Since $V$, $V_{\hat{P}_1\hat{P}_i}$ ( $i\in\{2,3,4\}$) and $V_{\hat{P}_1}$ contain the projection operator $P_1$, then, for all these contexts we have
\be
\delta^o(\hat{P}_1)_V=\delta^o(\hat{P}_1)_{V_{\hat{P}_1\hat{P}_i}}=\delta^o(\hat{P}_1)_{V_{\hat{P}_1}}=\hat{P}_1
\ee
Instead for context $V_{\hat{P}_i}$ for $i\neq 1$ we have
\be
\delta^o(\hat{P}_1)_{V_{\hat{P}_i}}=\hat{P}_1+\hat{P}_j+\hat{P}_k\;\;\; j\neq i\neq k\in\{2,3,4\}
\ee
For contexts of the form $V_{\hat{P}_i\hat{P}_j}$, where $i\neq j\neq1$, we have
\be
\delta^o(\hat{P}_1)_{V_{\hat{P}_i\hat{P}_j}}=\hat{P}_1+\hat{P}_k\;\;\; j\neq i\neq k\in\{2,3,4\}
\ee
\item\emph{Other maximal algebras which contain $\hat{P}_1$ and their sub-algebras}.\\\\
Let us consider 4 pairwise orthogonal projection operators $\hat{P_1},\hat{Q}_2,\hat{Q}_3,\hat{Q}_4$, such that the maximal abelian von Neumann algebra generated by such projections is different from $V$, i.e. 
$$V^{'}=lin_{\Cl}( \hat{P_1},\hat{Q}_2,\hat{Q}_3,\hat{Q}_4)\neq V$$
We then have the following dasenised propositions:\\\\
For contexts $V^{'}$ and $V^{'}_{\hat{P}_1}$, as before, we have
\be
\delta^o(\hat{P}_1)_{V^{'}}=\delta^o(\hat{P}_1)_{V^{'}_{\hat{P}_1}}=\hat{P}_1
\ee
for contexts $V_{\hat{Q}_i}$ we have
\be
\delta^o(\hat{P}_1)_{V_{\hat{Q}_i}}=\hat{P}_1+\hat{Q}_j+\hat{Q}_k\;\;i\neq j\neq k\in\{2,3,4\}
\ee
Instead, for context $V_{\hat{Q}_i\hat{Q}_j}$, we have
\be
\delta^o(\hat{P}_1)_{V_{\hat{Q}_i\hat{Q}_j}}=\hat{P}_1+\hat{Q}_k\;\;i\neq j\neq k\in\{2,3,4\}
\ee
\item \emph{Contexts which contain a projection operator which is implied by $\hat{P}_1$}.\\\\
Let us consider contexts $\tilde{V}$ which contain the projection operator $\hat{Q}$, such that $\hat{Q}\geq \hat{P}_1$, but do not contain $\hat{P}_1$ (if they did contain $\hat{P}_1$, we would be in exactly the same situation as above). In this situation the daseinisated propositions will be
\be
\delta^o(\hat{P}_1)_{\tilde{V}}=\hat{Q}
\ee
\item\emph{Context which neither contain $\hat{P}_1$ or a projection operator implied by it}.\\\\
In these contexts $V{''}$ the only coarse grained proposition related to $\hat{P}_1$ is the unity operator, therefore we have
\be
\delta^o(\hat{P}_1)_{V^{''}}=\hat{1}
\ee
\end{enumerate}
Now that we have defined all the possible coarse grainings of the proposition $\hat{P}_1$, for all possible contexts, we can define the presheaf $\ps{\delta(\hat{P}_1)}$ which is the topos analogue of the proposition $S_z\in [1.3,2.3]$. As explained in the previous section, in order to obtain the presheaf $\ps{\delta(\hat{P}_1)}$ from the projection $\hat{P}_1$, we must apply the daseinisation map\footnote{Note that so far we have only used the outer daseinisation.} defined in \ref{ali:glob}, so as to obtain
\ba
\delta:P(\Cl^4)&\rightarrow& \Sub_{cl}(\Sig)\\
\hat{P}_1&\mapsto&(\mathfrak{S}(\delta^o(\hat{P}_1)_V))_{V\in\V(\Cl^4)}=:\ps{\delta(\hat{P})}
\ea
where the map $\mathfrak{S}$ was defined in \ref{equ:smap}, in particular 
\be
\mathfrak{S}(\delta^o(\hat{P}_1)_V):=S_{\delta^o(\hat{P}_1)_V}:=\{\lambda\in\Sig_V|\lambda
(\delta^o(\hat{P}_1)_V)=1\}
\ee
We now want to define the $\ps{\delta(\hat{P}_1)}$-morphisms. In order to do so we will again subdivide our analysis in different cases, as above.
\begin{enumerate}
\item \emph{Maximal algebra $V$ and its sub-algebras}.\\\\
The sub-algebras of $V$ are of two kinds: $V_{\hat{P}_i,\hat{P}_j}$ and $V_{\hat{P}_k}$ for $i,j,k\in\{1,2,3,4\}$, such that in $\mv(\Cl^4)$ we obtain the morphisms $i_{V_{\hat{P}_i\hat{P}_j},V}:V_{\hat{P}_i\hat{P}_j}\subseteq V$ and $i_{V_{\hat{P}_k},V}:V_{\hat{P}_k}\subseteq V$. Correspondingly $\ps{\delta(\hat{P}_1)}$-morphisms with domain $\ps{\delta(\hat{P}_1)}_V$ will be of two kinds. We will analyse one at the time. First we analyse the morphism
\be
\ps{\delta(\hat{P}_1)}(i_{V_{\hat{P}_i\hat{P}_j},V}):\ps{\delta(\hat{P}_1)}_{V}\rightarrow \ps{\delta(\hat{P}_1)}_{V_{\hat{P}_i,\hat{P}_j}}
\ee 
In this context we have
\be
\ps{\delta(\hat{P}_1)}_{V}=\{\lambda\in \us_V|\lambda(\delta(\hat{P}_1)_{V})=\lambda(\hat{P}_1)=1\}=\{\lambda_1\}
\ee
This is the case since, as we saw in the previous lecture $\us_V=\{\lambda_1, \lambda_2,\lambda_3,\lambda_4\}$ where $\lambda_i\hat{P}_j=\delta_{ij}$.\\
On the other hand for the contexts $V_{\hat{P}_i,\hat{P}_j}$ $i,j\in\{1,2,3,4\}$ we have the following:
\ba
\ps{\delta(\hat{P}_1)}_{V_{\hat{P}_1,\hat{P}_j}}&=&\{\lambda_1\} \text{ where }\lambda_1\Big(\delta(\hat{P}_1)_{V_{\hat{P}_1,\hat{P}_j}}=\hat{P}_1\Big)=1;\;\;j\in\{2,3,4\}\\
\ps{\delta(\hat{P}_1)}_{V_{\hat{P}_i,\hat{P}_j}}&=&\{\lambda_{1k}\} \text{ where }\lambda_{1k}\Big(\delta(\hat{P}_1)_{V_{\hat{P}_i,\hat{P}_j}}=(\hat{P}_1+\hat{P}_k)\Big)=1;\;\; i\neq j \neq k\neq 1
\ea
The $\ps{\delta(\hat{P}_1)}$-morphisms for the above contexts would be 
\ba
\ps{\delta(\hat{P}_1)}(i_{V_{\hat{P}_1,\hat{P}_j},V})(\lambda_1)&:=&\lambda_1\\
\ps{\delta(\hat{P}_1)}(i_{V_{\hat{P}_i,\hat{P}_j},V})(\lambda_1)&:=&\lambda_{1k}
\ea
The remaining $\ps{\delta(\hat{P}_1)}$-morphisms with domain $\ps{\delta(\hat{P}_1)}_V$ are
\be
\ps{\delta(\hat{P}_1)})i_{V_{\hat{P}_i},V}:\ps{\delta(\hat{P}_1)}_{V}\rightarrow \ps{\delta(\hat{P}_1)}_{V_{\hat{P}_i}}
\ee 
In this case the local propositions $\ps{\delta(\hat{P}_1)}_{V_{\hat{P}_i}}$, $i\in\{1,2,3,4\}$ are
\ba
\ps{\delta(\hat{P}_1)}_{V_{\hat{P}_1}}&=&\{\lambda_1\}\\
\ps{\delta(\hat{P}_1)}_{V_{\hat{P}_i}}&=&\{\lambda_{1jk}\}\;\;i,j,k\in\{2,3,4\}
\ea
The $\ps{\delta(\hat{P}_1)}$-morphisms are then 
\ba
\ps{\delta(\hat{P}_1)}(i_{V_{\hat{P}_1},V}(\lambda_1)&:=&\lambda_1\\
\ps{\delta(\hat{P}_1)}(i_{V_{\hat{P}_i},V})(\lambda_1)&:=&\lambda_{1kl}
\ea
\item\emph{Other maximal algebras which contain $\hat{P}_1$ and their sub-algebras}.\\\\
As before we consider 4 pairwise orthogonal projection operators $\hat{P_1},\hat{Q}_2,\hat{Q}_3,\hat{Q}_4$, such that the maximal abelian von Neumann algebra generated by such projections is different from $V$, i.e. \\
$V^{'}=lin_{\Cl}( \hat{P_1},\hat{Q}_2,\hat{Q}_3,\hat{Q}_4)\neq V$.\\
We then obtain the following morphisms with domain $\ps{\delta(\hat{P}_1)}_{V}$:
\ba
\ps{\delta(\hat{P}_1)}(i_{V_{\hat{P}_1},V}):\ps{\delta(\hat{P}_1)}_{V}&\rightarrow& \ps{\delta(\hat{P}_1)}_{V_{\hat{P}_1}}\\
\lambda_1&\mapsto&\lambda_1
\ea
\ba
\ps{\delta(\hat{P}_1)}(i_{V_{\hat{Q}_i},V}):\ps{\delta(\hat{P}_1)}_{V}&\rightarrow& \ps{\delta(\hat{P}_1)}_{V_{\hat{Q}_i}}\\
\lambda_1&\mapsto&\rho_{1jk}
\ea
where $\us_{\hat{Q}_i}:=\{\rho_i,\rho_{1jk}\}$, such that $\rho_i(\hat{Q}_i)=1$ and $\rho_{1jk}(\hat{P}_1+\hat{Q}_j+\hat{Q}_k)=1$.
\ba
\ps{\delta(\hat{P}_1)}(i_{V_{\hat{Q}_i,\hat{Q}_j},V}):\ps{\delta(\hat{P}_1)}_{V}&\rightarrow& \ps{\delta(\hat{P}_1)}_{V_{\hat{Q}_i},\hat{Q}_j}\\
\lambda_1&\mapsto&\rho_{1k}
\ea
where $\us_{\hat{Q}_i,\hat{Q}_j}:=\{\rho_i,\rho_j,\rho_{1k}\}$, such that $\rho_i(\hat{Q}_i)=1$, $\rho_j(\hat{Q}_j)=1$ and $\rho_{1k}(\hat{P}_1+\hat{Q}_k)=1$.
The computation of the remaining maps is left as an exercise.
\item\emph{Contexts which contain a projection operator which is implied by $\hat{P}_1$}.\\\\
We now consider a context $\tilde{V}$ which contains an operator $\hat{Q}$, such that $\hat{Q}\geq \hat{P}$.\\
For such a context we have $\ps{\delta(\hat{P}_1)}_{\tilde{V}}=\{\lambda|\lambda(\hat{Q})=1\}$. Therefore, for sub-algebras which contain the operator $\hat{Q}$ the morphisms will simply map $\lambda$ to itself. The rest of the maps are easily derivable.
\item\emph{Context which neither contain $\hat{P}_1$ or a projection operator implied by it}.\\\\
In such a context $V^{''}$ , whatever its spectrum is, each of the multiplicative linear functionals $\lambda_i\in\us_{V^{''}}$ will assign value 1 to $\ps{\delta(\hat{P}_1)}=\hat{1}$. And so will the elements of the spectrum of the sub-algebras $\bar{V}$ of $V^{''}$. Thus, all the maps $\ps{\delta(\hat{P}_1)}(i_{\bar{V},V^{''}}$ will simply be equivalent to spectral presheaf maps.
\end{enumerate}

\section{The Spectral Preshaef and the Kochen-Specker Theorem} 
We will now show how the non existence of global elements of the spectral presheaf $\us$ is equivalent to the Kochen-Specker theorem.

Let us consider the topos $\Sets^{\mv(\mh)^{\op}}$. In order to define a global element we first of all have to define what a terminal object looks like. This is identified as the preheaf $$\underline{1}:\mv(\mh)\rightarrow \Sets$$ such that for each $V\in \mv(\mh)$
\be
\underline{1}_V:=\{*\}
\ee
Given a map $i_{V^{'}V}:V^{'}\rightarrow V$ ($V^{'}\subseteq V$), the corresponding morphisms is simply the constant map
\ba
\underline{1}(i_{V^{'}V}):\{*\}\rightarrow \{*\}
\ea

We now want to define a global element of the presheaf $\us$. Recall that this is defined as a map
\ba
\gamma:\underline{1}\rightarrow \us
\ea 
such that for each context we get
\ba
\gamma_V:\underline{1}_V&\rightarrow& \us_V\\
\{*\}&\rightarrow  &\us_V\\
\{*\}&\mapsto&\gamma_V(\{*\}):=\lambda
\ea
Thus at the level of the stalks we retrieve the usual set definition of global element. \\\\

The connection between global sections and the Kochen-Specker theorem is given by the following theorem:

\begin{Theorem}
The spectral presheaf $\us$ has no global elements iff FUNC does not hold, i.e. $V(f(\hat{A}))\neq f(V(\hat{A}))$ for some Borel function $f:\Rl\rightarrow \Rl$, such that $\hat{B}=f(\hat{A})$.
\end{Theorem}
\begin{proof}
Let us assume that $\us$ did have global sections. This would imply that there existed maps $\gamma:\underline{1}\rightarrow \us$, such that to each element $V\in\mv(\mh)$, $\gamma_V\in\us_V$, i.e. $\gamma_V=\lambda:V\rightarrow \Cl$. In particular, for each self adjoint operator $\hat{A}\in V_{sa}$,  $\gamma_V(\hat{A})=\lambda(\hat{A})\in\sigma(\hat{A})$ is an element of the spectrum of $\hat{A}$. Given a map $i_{V^{'},V}:V^{'}\rightarrow V$ ($V^{'}\subseteq V$), then from the properties of global sections we have that 
\be\label{equ:globsec}
\us(i_{V^{'}, V})\gamma_V=\gamma_{V^{'}}
\ee
Now consider a self-adjoint operator $\hat{A}$, such that $\hat{A}\in V$ but $\hat{A}\notin V^{'}$. Given the fact that $V^{'}\subseteq V$ it is always possible to find an operator $\hat{B}\in V^{'}$ such that $f(\hat{A})=\hat{B}$ for some Borel function $f:\Rl\rightarrow\Rl$. Since $\gamma_V(\hat{A})\in\sigma(\hat{A})$, by applying equation \ref{equ:globsec} to $\hat{A}$ we obtain
\be
f(\gamma_V(\hat{A}))=\gamma_{V^{'}}\hat{B}
\ee
which is precisely FUNC.

\end{proof}

The above theorem leads immediately to the following statement
\begin{Corollary}
The Kochen-Specker theorem is equivalent to the
statement that, if $dim\mh > 2$, the spectral presheaf $\us$ 
has no global elements.
\end{Corollary}
\begin{proof}
We now want to show that the K-S theorem is equivalent to the statement that $\tilde{\us}$ has no global elements. So let us assume it does, it then follows that 
there exists a function $\gamma:\underline{1}\rightarrow\us$ which assigns to each (bounded, discrete spectrum) self-adjoint operator $\hat{A}$, a real number\footnote{Here for notational simplicity we simply wirte $\gamma(*)$ as $\gamma$.} $\gamma(\hat{A})\in sp(\hat{A})$. From the definition of a section it follow that if $\hat{A} = f(\hat{A})$ then $f(\gamma(\hat{A})) =\gamma(\hat{B})$. This is precisely the FUNC condition. However the K-S theorem tells us that this can not be the case, thus $\tilde{\us}$ has no global sections.
\end{proof}
\chapter{Lecture 10}
In this lecture I will describe what a sub-object classifier looks like in our quantum topos. I will then give a concrete example for the case of a $4$ dimensional Hilbert space. I will then define the topos analogue of a state, and give a concrete example of such a state.

\section{Representation of Sub-object Classifier}
We will now describe how the sub-object classifier is defined in the topos $\Sets^{\mv(\mh)^{op}}$. Such an object represents the truth value object whose elements (global sections) are truth values, which get assigned to propositions (clopen sub-objects of $\us$). As we will see, we end up with a multi valued logic.
In the topos $\Sets^{\mv(\mh)^{op}}$ the sub-object classifier $\Omega$ is identified with the following presheaf.
\begin{Definition}
The presheaf $\Om\in \Sets^{\V(\Hi)^{op}}$ is defined as follows:
\begin{enumerate}
\item For any $V\in\mathcal{V(H)}$, the set $\Om_V$ is defined as the set of all sieves on $V$.

\item Given a morphism $i_{V^{'}V}:V^{'}\rightarrow V$ $(V^{'}\subseteq V)$, the associated function in $\underline\Omega$ is
\begin{align}
\Om(i_{V^{'}V}):
&\Om_V\rightarrow \Om_{V^{'}}\\
&S \mapsto \Om((i_{V^{'}V}))(S):=\{V^{''}\subseteq V^{'}|V^{''}\in
S\}
\end{align}
\end{enumerate}
\end{Definition}
We have seen, in previous lectures, what a sieve is, however, for the particular case in which we are interested, namely sieves
defined on the poset $\V(\Hi)$, the definition of a sieve can be
simplified as follows:
\begin{Definition}
For all $V\in\V(\Hi)$, a sieve $S$ on $V$ is a collection of
sub-alebras $(V^{'}\subseteq V)$ such that, if $V^{'}\in S$ and
$(V^{''}\subseteq V^{'})$, then $V^{''}\in S$. Thus $S$ is a
downward closed set.
\end{Definition}
In this case a maximal sieve on $V$ is
\begin{equation}
\downarrow\! V:=\{V^{'}\in\V(\Hi)|V^{'}\subseteq V\}
\end{equation}

In order for $\Om$ to be a well defined presheaf, we need to show
that indeed $\Om((i_{V^{'}V}))(S):=\{V^{''}\subseteq
V^{'}|V^{''}\in S\}$ defines a sieve on $V^{'}$. Thus we
need to show that $\Om((i_{V^{'}V}))(S):=\{V^{''}\subseteq
V^{'}|V^{''}\in S\}$ is a downward closed set with respect to
$V^{'}$. It is straightforward to deduce this from the definition.

As previously stated, truth values are identified with global
section of the presheaf $\Om$. For each context, such global sections assign the `local' truth value. Therefore the picture we obtain is the following: \\each proposition and each state is defined as a collection of `local' representations, one for each $V\in\mv(\mh)$. Such `local' representations, are glued together by the categorical structure of $\mv(\mh)$. Now for each context we obtain a `local' truth value of the `local' proposition given the `local' state. Such `local' truth values are represented by the global element computed at that particular context. All such `local' truth values are `glued' together by the global section which, in turn, follows the categorical structure of the base category $\mv(\mh)$.
Thus, again, we obtain the quantum picture by considering a collection of `local' representatives. However, it is only the collection that corresponds to a well defined object in our theory, each local representative on its own is meaningless.

Coming back to truth values, 
the global section that consists
entirely of principal sieves, is interpreted as representing
`totally true'. In classical Boolean logic this is just  `true'.
Similarly, the global section that consists of empty sieves is
interpreted as `totally false'. In classical Boolean logic this
is just `false'.
\\
A very important property of sieves is that the set $\uom_V$  of sieves on $V$ has the
structure of a Heyting algebra, where the unit element $\underline{1}_{\uom_V}\in \uom_V$ is represented by the principal sieve $\downarrow V$ and, the null element $\underline{0}_{\uom_V}\in \uom_V$, is represented by the empty set $\emptyset$. \\
Moreover $\uom_V$ is equipped with a partial ordering given by subset inclusion, such that $S_i\leq S_j$ iff $S_i\subseteq S_j$. In this context the logical connectives are given by 
\ba
S_i\wedge S_j&:=& S_i \cap S_j\\
S_i\vee S_j&:=& S_i \cup S_j\\
S_i\Rightarrow S_j&:=&\{V^{'}\subseteq V |\forall V{''}\subseteq V^{'} \text{ if } V{''}\in S_i \text{ then } V{''}\in S_j\}
\ea
Being a Heyting algebra, the negation is given by the pseudo-complement. In particualr, given an element $S$, its pseudo-complement (negation) is the element 
\ba
\neg S&:=&S\Rightarrow 0\\
\neg S &:= &\{V^{'}\subseteq V |\forall V{''}\subseteq V^{'},  V{''}\notin S\}
\ea
\subsection{Example}
We will now describe an example of the truth object for the case of our 4 dimensional Hilbert space $\mh=\Cl^4$. Let us start with the \emph{maximal algebras} $V=lin_{\Cl}(\hat{P}_1,\hat{P}_2, \hat{P}_3, \hat{P}_4)$. What follows can be generalised to any maximal sub-alebra, not just $V$.\\
The collection of sieves on $V$ will be
\be
\uom_V:=\{\underline{0}_{\uom_V},S, S_{12}, S_{13}, S_{14}, S_{23}, S_{24}, S_{34}, S_{1},S_{2}, S_{3}, S_{4}, \cdots\}
\ee 
where the sieves in $\uom_V$ are defined as follows:
\ba
S&=&\{V, V_{\hat{P}_i,\hat{P}_j}, V_{\hat{P}_k}, | i,j,k\in\{1,2,3,4,\} \}=\text{ principal sieve on V }\\
S_{ij}&=&\{V_{\hat{P}_i,\hat{P}_j}, V_{\hat{P}_i}, V_{\hat{P}_j}\}=\text{ principal sieve on }V_{\hat{P}_i,\hat{P}_j}\\
S_{i}&=&\{V_{\hat{P}_i}\}\\
\underline{0}_{\uom_V}&=&\{\emptyset\}
\ea
We now consider a non maximal algebra $V_{\hat{P}_i,\hat{P}_j}$, the collection of sieves on such an algebras is
\be
\uom_{V_{\hat{P}_i,\hat{P}_j}}:=\{\underline{0}_{V_{\hat{P}_i,\hat{P}_j}},S_{ij}, S_i, S_j, \cdots\}
\ee
where the definitions of the individual sieves are the same as before.\\
Similarly, for the context $V_{\hat{P}_i}$ we have
\be
\uom_{V_{\hat{P}_i}}:=\{S_{i}\}
\ee
We now want to define the $\uom$-morphisms. To this end, let us first consider the $\uom$-morphism with domain $\uom_V$. There are various such morphisms, one for each pair $i,j\in\{1,2,3,4\}$, as follows:
\ba
\uom(i_{V_{\hat{P}_i,\hat{P}_j},V}):\uom_V\rightarrow \uom_{V_{\hat{P}_i,\hat{P}_j}}
\ea 
where $i_{V_{\hat{P}_i,\hat{P}_j},V}:V_{\hat{P}_i,\hat{P}_j}\subseteq V$. $\uom(i_{V_{\hat{P}_i,\hat{P}_j},V})$ is defined component wise as follows:
\ba
S\mapsto S_{ij}&\;\;\;&
S_{ij}\mapsto S_{ij}\\
S_{ik}\mapsto S_i&\;\;\;&
S_{kj}\mapsto S_j\\
S_i \mapsto S_i&\;\;\;&
S_{kl}\mapsto\underline{0}_{V_{\hat{P}_i,\hat{P}_j}}\\
S_j\mapsto S_j&\;\;\;&
S_k \mapsto \underline{0}_{V_{\hat{P}_i,\hat{P}_j}}
\ea
Moreover, for each $k\in \{1,2,3,4\}$ such that $i_{V_{\hat{P}_k},V}:V_{\hat{P}_k}\subseteq V$, we have the following $\uom$-morphisms:
\be
\uom_{V,V_{\hat{P}_k}}:\uom_V\rightarrow \uom_{V_{\hat{P}_k}}
\ee
which component wise are defined as follows:
\ba
S\mapsto S_{k}&\;\;\;&
S_{ij}\mapsto \underline{0}_{V_{\hat{P}_k}}\\
S_{ik}\mapsto S_k&\;\;\;&
S_{l}\mapsto\underline{0}_{V_{\hat{P}_k}}\\
S_k \mapsto S_k&\;\;\;&
S_j\mapsto \underline{0}_{V_{\hat{P}_k}}\\
&\;\;\;&S_i\mapsto \underline{0}_{V_{\hat{P}_k}}
\ea
It is straightforward to extend the definition of $\uom$-morphisms for all contexts $V_i\in \mv(\Cl^4)$.
\section{States}
 In classical physics a pure state, $s$, is a point in the state
space. It is the smallest subset of the state space which has
measure one with respect to the Dirac measure $\delta_s$. 

\noindent
Recall that a Dirac measure $\delta_s$ on some set $S$ is defined by 
\be
\delta_s(A)=\begin{cases} 1\text{ if }s\in A\\
0\text{ if } s\notin A
\end{cases}
\ee
for any $s\in S$ and any measurable subset $A\subseteq S$.

Identifying states with subsets which have measure one is a consequence of the one-to-one correspondence which subsists
between pure states and Dirac measure. In particular, for each
pure state, $s$, there corresponds a unique Dirac measure
$\delta_s$. Moreover, propositions which are true in a pure state
$s$ are given by subsets of the state space which have measure one,
with respect to the Dirac $\delta_s$, i.e. those subsets which
contain $s$. The smallest such subset is the one-element set
$\{s\}$. Thus, a pure state can be identified with a single point
in the state space.

In classical physics, more general states are represented by more
general probability measures on the state space. This is the
mathematical framework that underpins classical statistical
physics.

However, the spectral presheaf $\Sig$ has \emph{no}
points\footnote{Recall that in a topos $\tau$, a `point' (or `global element'
or just `element') of an object $O$ is defined to be a morphism
from the terminal object, $1_\tau$, to $O$.}. Indeed, this is
equivalent to the Kochen-Specker theorem! Thus the analogue of a
pure state must  be identified with some other construction. There
are two (ultimately equivalent)\ possibilities:  a `state' can be
identified with (i) an element of $P(P(\Sig))$ (the set of all possible sub-objects of the set of all possible subsets of $\us$); or (ii) an element
of $P(\Sig)$ (the set of all sub-object of $\us$). The first choice is called the \emph{truth-object}
option, the second is  the \emph{pseudo-state} option. In what
follows we will concentrate on the second option. The first option will be described later.

The second choice is the one that most resembles the notion of a point state since it represents the smallest sub-object of the state space $\us$. Since $\us$ is a presheaf, a sub-object of it will be itself a presheaf, thus the \emph{pseudo state} is a presheaf, i.e. an object in $\Sets^{\mv(\mh)^{\op}}$.

Specifically, given a pure quantum state $\psi\in\Hi$, we define
the presheaf
\begin{equation}
\ps{\w}^{\ket\psi}:= \ps{\delta(\ket\psi\langle\psi|)}
\end{equation}
such that for each context V we have
\begin{equation}
\ps{\delta(\ket\psi\langle\psi|)}_V:=
\mathfrak{S}(\bigwedge\{\hat{\alpha}\in
P(V)|\ket\psi\langle\psi|\leq\hat{\alpha}\})= \mathfrak{S}(\delta^o(\ket\psi\langle\psi|))\subseteq\Sig(V)
\end{equation}
Where the map $\mathfrak{S}$ was defined in equation
(1.5) lecture 9 but we will report it below for the sake of completeness:
\begin{equation}\label{equ:smap}
\mathfrak{S}:P(V)\rightarrow \Sub_{cl}(\Sig)_V\hspace{.2in}
\end{equation}
such that
\begin{equation}
\delta^o(\ket\psi\langle\psi|)_V\mapsto
\mathfrak{S}(\delta^o(\ket\psi\langle\psi|)_V):=S_{\delta^o(\ket\psi\langle\psi|)_V}
\end{equation}

Thus, for each context $V\in\mv(\mh)$, the projection operator $\ps{\delta(\ket\psi\langle\psi|)}_V$ is the smallest projection operator implied by $\ket\psi\langle\psi|$. Since $\ket\psi\langle\psi|$ projects on a 1-dimensional sub-space of the Hilbert space, i.e. it projects on a state, $\ps{\delta(\ket\psi\langle\psi|)}_V$ identifies the smallest sub-space of $\mh$ equal or bigger than the one dimensional sub-space $|\psi\rangle$. It is in this sense that $\ps{\w}^{\ket\psi}:= \ps{\delta(\ket\psi\langle\psi|)}$ represents the closest one can get to a point in $\us$.

The map
\begin{equation}
\ket\psi\rightarrow \ps{\w}^{\ket\psi}
\end{equation}
is injective. 
\begin{proof}
We want to show that if $\ps{\mathbb{T}}^{|\psi\rangle}=\ps{\mathbb{T}}^{|\phi\rangle}$ then $|\phi\rangle=e^{i\lambda}|\psi\rangle$. Now applying the definitions for each $V\in \mv(\mh)$ we have
\ba
\ps{\mathbb{T}}^{|\psi\rangle}_V&=&\{\hat{P}\in P(V)|\hat{P}\geq |\psi\rangle\langle\psi|\}\\
&=&\{\hat{P}\in P(V)|\langle\psi|\hat{P} |\psi\rangle=1\}\\
&=&\ps{\mathbb{T}}^{|\phi\rangle}_V\\
&=&\{\hat{P}\in P(V)|\langle\phi|\hat{P} |\phi\rangle=1\}
\ea
However if $\langle\psi|\hat{P} |\psi\rangle=1$ then $\langle\psi e^{-i\lambda}|\hat{P} |e^{i\lambda}\psi\rangle=\langle\phi|\hat{P} |\phi\rangle=1$ for some $\lambda$. Moreover that will be the only stat would satisfy this for all $\hat{P}$ therefore $|\phi\rangle=e^{i\lambda}|\psi\rangle$, i.e. $|\psi\rangle\rightarrow\ps{\mathbb{T}}^{|\psi\rangle}$ is injective. Since the association $\ps{\mathbb{T}}^{|\psi\rangle}\rightarrow \ps{\w}^{\ket\psi}$ is injective the result follows.

\end{proof}
Thus, for each state $\ket\psi$, there is associated a
topos pseudo-state, $\ps{\w}^{\ket\psi}$, which is defined as a
subobject of the spectral presheaf $\Sig$.

This presheaf $\ps{\w}^{\ket\psi}$ is interpreted as the smallest
clopen subobject of $\Sig$, which represents the proposition\footnote{Recall that in the topos framework propositions are identified with clopen sub-objects of the state space} which
is totally true in the state $|\psi\rangle$, namely the proposition $\ps{\delta(|\psi\rangle\langle\psi|)}$. Roughly speaking, it is the
closest one can get to defining a point in $\Sig$. The formal definition of the pseudo state is as follows:

\begin{Definition}
For each state $|\psi\rangle\in \mh$ we obtain the pseudo state $\ps{\w}^{\ket\psi}\in \Sets^{\mv(\mh)^{\op}}$ which is defined on 
\begin{itemize}
\item Objects: For each context $V\in\mv(\mh)$ we obtain 
\be
\ps{\delta(\ket\psi\langle\psi|)}_V:=\{\lambda\in\us_V|\lambda(\delta^o(|\psi\rangle\langle\psi|)_V)=1\}
\ee
\item Morphisms: For each $i_{V^{'}V}:V^{'}\subseteq V$ the corresponding map is simply the spectral presheaf map restricted to $\ps{\w}^{\ket\psi}$, i.e.
\ba
\ps{\w}^{\ket\psi}(i_{V^{'}V}):\ps{\w}^{\ket\psi}_V&\rightarrow &\ps{\w}^{\ket\psi}_{V^{'}}\\
\lambda&\mapsto&\lambda_{|V^{'}}
\ea
\end{itemize}

\end{Definition}

\subsection{Example}
We will now give an example of how to define pseudo states in our 4 dimensional Hilbert space $\Cl^4$. This is very similar to the example for propositions, since also for pseudo states the concept of daseinisation is utilised. However, for pedagogical reasons we will, nonetheless, report it below.\\
Let us consider a state $\psi=(0,1,0,0)$. The respective projection operator is 
\[\hat{P}_2=|\psi\rangle\langle\psi|=\begin{pmatrix} 0& 0& 0&0\\	
0&1&0 &0\\
0&0&0&0\\
0&0&0&0	
  \end{pmatrix}
  \]
We now want to compute the outer daseinisation of such a projection operator for various contexts $V\in \mv(\Cl^4)$. As was done for the proposition, we will subdivide our analysis in different cases:
\begin{enumerate}
\item \emph{Context $V$ and its sub-algebras}.\\\\
Since the maximal algebra $V$ is such that $|\psi\rangle\langle\psi|\in P(V)$, it follows that: 
\be
\delta^o(|\psi\rangle\langle\psi|)_V=|\psi\rangle\langle\psi|
\ee
This also holds for any sub-algebra of $V$ containing $|\psi\rangle\langle\psi|$, i.e. $V_{\hat{P}_2, \hat{P}_i}$ for $i=\{1,3,4\}$ and $V_{\hat{P}_2}$. \\Instead, for the algebras $V_{\hat{P}_i,\hat{P}_j}$, where $i, j\in\{1,3,4\}$, we have
\be
\delta^o(|\psi\rangle\langle\psi|)_{V_{\hat{P}_i,\hat{P}_j}}=\hat{P}_2+\hat{P}_k\;\; \text{for } k\neq i \neq j
\ee
On the other hand, for contexts $V_{\hat{P}_i}$ for $i\in\{1,3,4\}$ we have
\be
\delta^o(|\psi\rangle\langle\psi|)_{V_{\hat{P}_i}}=\hat{P}_2+\hat{P}_k+\hat{P}_j\;\; \text{for } k\neq i \neq j
\ee
\item\emph{Other maximal algebras which contain $|\psi\rangle\langle\psi|$ and their sub-algebras}.

Let us consider the 4 pairwise orthogonal operators $(\hat{Q}_1,\hat{P}_2, \hat{Q}_3,\hat{Q}_4)$, such that \\
$V^{'}:=lin_{\Cl}( \hat{Q}_1,\hat{P}_2, \hat{Q}_3,\hat{Q}_4)\neq V$. For these contexts we obtain
\ba
\delta^o(|\psi\rangle\langle\psi|)_{V^{'}}&=&\hat{P}_2\\
\delta^o(|\psi\rangle\langle\psi|)_{V_{\hat{Q_i}}}&=&\hat{P}_2+\hat{Q}_j+\hat{Q}_k\\
\delta^o(|\psi\rangle\langle\psi|)_{V_{\hat{Q_i}.\hat{Q}_j}}&=&\hat{P}_2+\hat{Q}_k
\ea
\item\emph{Contexts which contain a projection operator which is implied by $|\psi\rangle\langle\psi|$}.\\\\
If $V^{''}$ contains $\hat{Q}\geq|\psi\rangle\langle\psi|$ then 
\be
\delta^o(|\psi\rangle\langle\psi|)_{V^{'}}=\hat{Q}
\ee
\item\emph{Contexts which neither contain $|\psi\rangle\langle\psi|$ or a projection operator implied by it}.\\\\
\be
\delta^o(|\psi\rangle\langle\psi|)_{V^{'}}=\hat{1}
\ee
\end{enumerate}
We now would like to define the $\ps{\w}^{\ket\psi}$-morphisms. This is left as an exercise.

\chapter{Lecture 11}
This is a brief overview of the tight link between topos theory and logic. In particular it will be showed that to each topos there is associated a language with associated logic, but also the revers is true, given a language one can defined a corresponding topos. Given this tight connection,  it is also possible to view a theory of physics, as expressed within the mathematical formulation of topos theory, as a representation, in a topos, of an abstract language.

\section{Topos and Logic}
In this section we will try to describe the deep connection between topos and logic. The nature and scope of such a connection is deep and wide, so we will not be able to expose it in its full details. We will, however, try to give a general account of this intimate connection and, where possible, try to describe such a connection with explicit examples in physics.

We first of all need the definition of a language.
\subsection{First Order Languages $\l$}
A language, in its most raw definition, comprises a collection of atomic variables, and a collection of primitive operations called logical connectives, whose role is to combine together such primitive variables transforming them into formulas or sentences. Moreover, in order to reason with a given language one also requires rules of inference, i.e. rules which allow you to generate other valid sentences from the given ones.

The semantics or meaning of the logical connectives, however, is not given by the logical connectives themselves but it is defined through a so called evaluation map, which is a map from the set of atomic variables and sentences to a set of truth values.  Such a map enables one to determine when a formula is true and, thus, defines its semantics/meaning.

In this perspective it turns out that the meaning of the logical connectives is given in terms of some set of objects which represent the truth values. The logic that a given language will exhibit will depend on what the set of truth values is considered to be. In fact, the above is a very abstract characterisation of what a language is. To actually use it as a deductive system of reasoning one needs to define a mathematical context in which to represent this abstract language. In this way the elementary and compound propositions will be represented by certain mathematical objects, and the set of truth values will itself be identified with an algebra.

For example, in standard classical logic, the mathematical context used is $\mathbf{Sets}$ and the algebra of truth values is the Boolean algebras of subsets of a given set. However, as we will see, in a general topos the internal logic/algebra will not be Boolean but will be a generalisation of it, i.e. a Heyting algebra. We will explain, later on, the implications of this fact.

In order to get a better understanding of what has been said above we will start with a very simple language called propositional language $P(\l)$.
\subsection{Propositional Language}
The propositional language $P(\l)$ contains a set of symbols and a set of formation rules.\\
\textbf{Symbols of $P(\l)$}
\begin{itemize}
\item [i)] An infinite list of symbols $\alpha_0, \alpha_1, \alpha_2\cdots$ called \emph{primitive propositions}. 
\item [ii)] A set of symbols $\neg, \vee, \wedge, \Rightarrow$ which for now have no explicit meaning.
\item[iii)] Brakets $), ($.
\end{itemize} 
\textbf{Formation Rules}
\begin{itemize}
\item [i)] Each primitive proposition $\alpha_i\in P(\l)$ is a sentence.
\item[ii)] If $\alpha$ is a sentence, then so is $\neg\alpha$.
\item[iii)] If $\alpha_1$ and $\alpha_2$ are sentences, then so are $\alpha_1\wedge\alpha_2$, $\alpha_1\vee\alpha_2$ and $\alpha_1\Rightarrow \alpha_2$. 
\end{itemize}   
Note also that $P(\l)$ does not contain the quantifiers $\forall$ and $\exists$. This is because it is only a propositional language. To account for quantifiers one has to go to more complicated languages called higher-order languages, which will be described later.\\
The inference rule present in $P(\l)$ is the modus ponens (the Ôrule of detachmentÕ) which states that from $\alpha_i$ and $\alpha_i\Rightarrow \alpha_j$ the sentence $\alpha_j$ may be derived. Symbolically this is written as 
\be
\frac{\alpha_i, \alpha_i\Rightarrow \alpha_j}{\alpha_i}
\ee
We will see, later on, what exactly the above expression means.

In order to use the language $P(\l)$ one needs to represent it in a mathematical context. The choice of such context will depend on what type of system we want to reason about. For now we will consider a classical system, thus the mathematical context in which to represent the language $P(\l)$ will be \Sets. In $\mathbf{Sets}$ the truth object (object in which the truth values lie) will be the Boolean set $\{0,1\}$, thus the truth values will undergo a Boolean algebra. This, in turn, implies that the logic of the language $P(\l)$, as represented in $\mathbf{Sets}$, will be Boolean. 

The rigorous definition of a representation of the language $P(\l)$ in a mathematical context is a map $\phi$ from the set of primitive propositions  to elements in the algebra in question (in this case a Boolean algebra); $\alpha\rightarrow \phi(\alpha)$. The specification of the algebra, as we will see, will depend on what type of theory we are considering, i.e. classical or quantum. 

In the example (classical system) above the propositions are represented in the Boolean algebra of all (Borel) subsets of the classical state space (how this is done and why will be explained later on, for now we will just consider this statement as given)

Now that we have a representation of the abstract language we can also define the semantics of this language as follows:
\ba\label{ali:log}
\phi(\alpha_i\vee\alpha_j)&:=&\phi(\alpha_i)\vee\phi(\alpha_j)\\
\phi(\alpha_i\wedge\alpha_j)&:=&\phi(\alpha_i)\wedge\phi(\alpha_j)\nonumber\\
\phi(\alpha_i\Rightarrow\alpha_j)&:=&\phi(\alpha_i)\Rightarrow\phi(\alpha_j)\nonumber\\
\phi(\neg\alpha_i)&:=&\neg(\phi(\alpha_i))\nonumber
\ea
where, on the left hand side, the symbols $\{\neg, \wedge, \vee, \Rightarrow\}$ are elements of the language $P(\l)$, while on the right hand side they are the logical connectives in algebra, in which the representation takes place. It is in such an algebra that the logical connectives acquire meaning.\\
For the classical case, since the algebra of representations is the Boolean algebra of subsets, the logical connectives on the right hand side of \ref{ali:log} are defined in terms of set theoretic operations. In particular, we have the following associations:
\ba
\phi(\alpha_i)\vee\phi(\alpha_j)&:=& \phi(\alpha_i)\cup\phi(\alpha_j)\\
\phi(\alpha_i)\wedge\phi(\alpha_j)&:=&\phi(\alpha_i)\cap\phi(\alpha_j)\nonumber\\
\neg(\phi(\alpha_i))&:= &\phi(\alpha_i)^c\nonumber\\
\phi(\alpha_i)\Rightarrow\pi(\alpha_j)&:= &\phi(\alpha_i)^c \cup\phi(\alpha_j)
\ea

So far we have seen how logical connectives are represented in the topos $\mathbf{Sets}$ (since we have been considering the classical case). However, it is possible to give a general definition of logical connectives in terms of arrows. Such a definition would then be valid for any topos. To retrieve the logical connectives for the classical case, in which the topos is $\mathbf{Sets}$, we then simply replace, in the definitions that will follow, the general truth object $\Omega$ with the Boolean algebra $\{0,1\}=2$.

The way in which logical connectives are defined in a general topos is as follows:
\begin{itemize}
\item
\textbf{Negation}\\
We will now describe how to represent negation as an arrow in a given topos $\tau$. Let us assume that the $\tau$-arrow representing the value true is $\top:1\rightarrow \Omega$, which is the arrow used in the definition of the sub-object classifier. Given such an arrow \emph{true}, negation is identified with the unique arrow $\neg:\Omega\rightarrow \Omega $, such that the following diagram is a pullback
\[\xymatrix{
1\ar[rr]^{\perp}\ar[dd]&&\Omega\ar[dd]^{\neg}\\
&&\\
1\ar[rr]_{\top}&&\Omega\\
}\]
Where $\perp$ is the topos analogue of the arrow \textit{false} in $\mathbf{Sets}$, i.e.
$\perp$ is the character of $!_1:0\rightarrow1$
\[\xymatrix{
0\ar[rr]^{!_1}\ar[dd]_{!_1}&&1\ar[dd]^{\perp}\\
&&\\
1\ar[rr]_{\top}&&\Omega\\
}\]

\item
\textbf{Conjunction}\\
Conjunction is identified with the following arrow: 
\begin{equation*}
\cap:\Omega\times\Omega\rightarrow\Omega
\end{equation*}
which is the character of the product arrow 
$\langle \top,\top\rangle:1\rightarrow\Omega\times\Omega$,
such that the following diagram commutes
\[\xymatrix{
1\ar[rr]^{id_1}\ar[dd]_{\langle \top,\top\rangle}&&1\ar[dd]^\top\\
&&\\
\Omega\times\Omega\ar[rr]_{\cap}&&\Omega\\
}\]
where $\langle \top,\top\rangle$ is defined as follows:
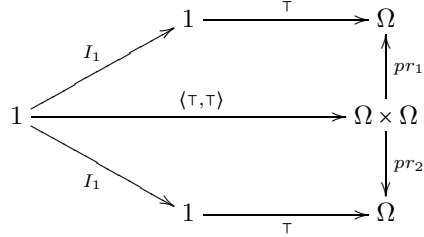
\[\xymatrix{
&&1\ar[rr]^{\top}&&\Omega\\
1\ar[rru]^{I_1}\ar[rrrr]^{\langle \top,\top\rangle}\ar[rrd]_{I_1}&&&&\Omega\times\Omega\ar[u]_{pr_1}\ar[d]^{pr_2}\\
&&1\ar[rr]_{\top}&&\Omega\\
}\]
\item 
\textbf{Disjunction}\\
Disjunction is identified with the arrow
\be
\cup:\Omega\times\Omega\rightarrow \Omega
\ee
which is the character of the image of the arrow 
\be
[\langle \top, 1_{\Omega}\rangle,\langle1_{\Omega},\top\rangle]:\Omega+\Omega\rightarrow \Omega\times\Omega
\ee
such that the following diagram commutes
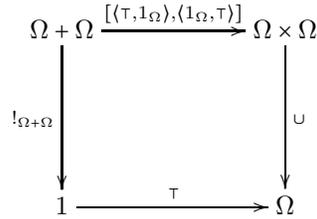
\[\xymatrix{
\Omega+\Omega\ar[rr]^{[\langle \top, 1_{\Omega}\rangle,\langle1_{\Omega},\top\rangle]}\ar[dd]_{!_{\Omega+\Omega}}&& \Omega\times\Omega\ar[dd]^{\cup}\\
&&\\
1\ar[rr]^{\top}&&\Omega\\
}\]
\item
\textbf{Implication}\\
Implication is identified with the arrow
\be
\Rightarrow:\Omega\times\Omega\rightarrow \Omega
\ee
which is the character of the equaliser map
\be
e:\leq\rightarrow\Omega\times\Omega
\ee
such that the following diagram commutes
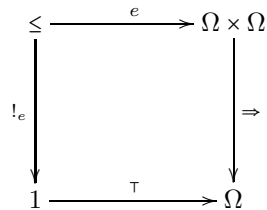
\[\xymatrix{
\leq\ar[rr]^{e}\ar[dd]_{!_{e}}&& \Omega\times\Omega\ar[dd]^{\Rightarrow}\\
&&\\
1\ar[rr]^{\top}&&\Omega\\
}\]
where $\leq:=\{\langle x,y\rangle|x\leq y \text{ in }\Omega\}$.\\
Now the above arrow $e$ is actually the equaliser of
 \[\xymatrix{
\Omega\times\Omega\ar@<3pt>[rr]^{\;\;\;\;\;\;\; \cap} \ar@<-3pt>[rr]_{\;\;\;\;\;\;\; pr_1}&&\Omega\\
}\]
i.e. $\cap\circ e=pr_1\circ e$.
\end{itemize}
In order to complete the definition of a propositional language in a given topos, we also need to define the valuation functions (which  gives us the semantics) in terms of arrows in that topos.\\
We recall from the definition of the sub-object classifier that a truth value in a general topos $\tau$ is given by a map $1\rightarrow \Omega$ (in $\Sets$ we have $1\rightarrow \{0,1\}=2=\Omega$). The collection of such functions $\tau(1,\Omega)$ represents the collection of all truth values. Thus, a valuation map in a general topos is defined to be the map $V:\{\pi(\alpha_i)\}\rightarrow \tau(1,\Omega)$.\\
It is, then, easy to show that the following equalities hold:
\ba
V(\neg(\pi(\alpha_i)))&=&\neg\circ V(\pi(\alpha_i))\\
V(\pi(\alpha_i)\vee \pi(\alpha_j))&=&\vee\circ \langle V(\pi(\alpha_i)), V(\pi(\alpha_j))\rangle\\
V(\pi(\alpha_i)\wedge \pi(\alpha_j))&=&\wedge\circ \langle V(\pi(\alpha_i)), V(\pi(\alpha_j))\rangle\\
V(\pi(\alpha_i)\Rightarrow \pi(\alpha_j))&=&\Rightarrow\circ \langle V(\pi(\alpha_i)), V(\pi(\alpha_j))\rangle
\ea
\subsubsection{Example In Classical Physics}
We have stated above that classical physics uses the topos $\mathbf{Sets}$. We now want to represent in $\Sets$ the propositional language $P(\l)$ as defined for a classical system $S$. Since such a language, in this case, will be used to talk about $S$, we will denote it $P\l(S)$ so as to make it explicit that we are talking about $S$. Now, since $S$ is a (classical) physical system, the standard propositions which it will contain will be of the form $A\in \Delta$ meaning `` the quantity $A$ which represents some physical observable, has value in a set $\Delta$". These are normally the types of propositions we deal with in classical physics and in physics in general.\\
We now define the representation map from this language to $\mathbf{Sets}$ as follows:
\ba
\pi_{cl}:P\l(S)&\rightarrow& \Sets\\
A\in \Delta&\mapsto&\{s\in S|\tilde{A}(s)\in \Delta\}=\tilde{A}^{-1}(\Delta)
\ea
where $S$ is the classical state space and $\tilde{A}:S\rightarrow \Rl $ is the map from the state space to the reals which identifies the physical quantity $A$.\\
We now define the truth values of such represented propositions. Normally, such truth values are state dependent, i.e. they depend on the state with respect to which we are preforming the evaluation. In classical physics states are simply identified with elements $s$ of the state space $S$. Thus, for all $s\in S$ we define the truth value of the proposition $\tilde{A}^{-1}(\Delta)$ as follows:
\be
v( A\in \Delta; s )=\begin{cases} 1\text{ iff } s\in \tilde{A}^{-1}(\Delta)\\
0\text{ otherwise }
\end{cases}
\ee
Thus the truth values lie in the Boolean algebra $\uom=\{0,1\}$.

It is interesting to note that the application of the propositional language $P(\l)$ for quantum theory fails. This is because in quantum theory propositions are identified with projection operators, thus the representation map would be 
\ba
\pi_q:\{\alpha_i\}&\rightarrow& P(\mh)\\
A\in\Delta&\mapsto&\pi_q(A\in\Delta):= \hat{E}[A\in \Delta]
\ea
where $ \hat{E}[A\in \Delta]$ is the projection operator which projects onto the subset $\Delta$ of the spectrum of $\hat{A}$. \\
Now the problem with this construction is that the set of all projection operators undergoes a logic which is not distributive, but the logic of the propositional language is distributive. Therefore, such a representation will not work. We will see later on how to fix this problem. However, to arrive at the solution we need to introduce a higher order language which we will examine in the next section.
\subsection{The Higher Order Type Language $\l$}
We now go a step higher and define a first order type language $\l$.  Such a language consists of a set of symbols and terms.\\
\textbf{Symbols}
\begin{enumerate}
\item A collection of ``sorts " or ``types". If $T_1,T_2,\cdots ,T_n$, $n\geq 1$, are type symbols, then so is $T_1\times T_2\times \cdots \times T_n$. If $n=0 $ then $T_1\times T_2\times \cdots \times T_n=1$.
\item If $T$ is a type symbol, then so is $PT$.
\item Given any type $T$ there are a countable set of variables of type $T$.
\item There is a special symbol $*$.
\item A set of function symbols for each pair of type symbols, together with a map which assigns to each such functions its type. This assignment consists of a finite non-empty list of types. Thus, for example, if we have the pair of type symbols $(T_1,T_2)$,  the associated set of function symbols would be $F_{\l}(T_1,T_2)$. A given $f\in F_{\l}(T_1,T_2)$ has type $T_1, T_2$, this is indicated by writing $f : T_1 \rightarrow T_2$.
\item A set of relation symbols $R_i$ together with a map which assigns the type of the arguments of the relation. This consists of a list of types. Thus, for example, a relation taking an argument $x_1\in T_1$ of type $T_1$ to an argument $x_2\in T_2$ of type $T_2$ is denoted as $R=R(x_1,x_2)\subseteq T_1\times T_2$.
\end{enumerate}
\textbf{Terms}
\begin{enumerate}
\item The variables of type $T$ are terms of type $T$, $\forall T$.
\item The symbol $*$ is a term of type $1$. 
\item A term of type $\Omega$ is called a formula. If the formula has no free variables then we call it a sentence.
\item Given a function symbol $f: T_1 \rightarrow T_2$ and $t$ a term of type $T_1$, then
$f(t)$ is term of type $T_2$. 
\item	Given $t_1, t_2,\cdots , t_n$ which are terms of type $T_1, T_2,\cdots,T_n$ respectively, then $\langle t_1,t_2,\cdots ,t_n\rangle$ is a term of type $T_1 \times T_2 \times \cdots\times T_n$. 
\item If $x$ is a term of type $T_1\times T_2\times \cdots \times T_n$, then for $1\leq i\leq n$, $t_i$ is a term of type $T_i$.
\item If $\omega$ is a term of type $\Omega$ and $\alpha$ is a variable of type $T$, then $\{\alpha|\omega\} $ is a term of type $PT$. 
\item If $x_1,x_2$ are terms of the same type, then $x_1 = x_2$ is a term of type $\Omega$.
\item  If $x_1, x_2$ are terms of type $T$ and $PT$ respectively, then $x_1 \in x_2$ is a term of type $\Omega$.
\end{enumerate}
The entire set of formulas in the language $\l$ are defined recursively through repeated applications of formation rules, which are the analogue of the standard logical connectives. In particular, we have \emph{atomic formulas} and  \emph{composite formulas}
The former are:
\begin{enumerate}
\item The terms of relation.
\item Equality terms defined above.
\item Truth $\top$ is an atomic formula with empty set of free variables.
\item False $\perp$ is an atomic formula with empty set of free variables.
\end{enumerate}
We can now built more complicated formulas through the use of the logical connectives $\vee$, $\wedge$, $\Rightarrow$ and $\neg$.
These are the \emph{composite formulas}:
\begin{enumerate}
\item Given two formulas $\alpha$ and $\beta$ then $\alpha\vee \beta$  is a formula such that, the set of free variables is defined to be the union of the free variables in $\alpha$ and $\beta$.
\item Given two formulas $\alpha$ and $\beta$ then $\alpha\wedge \beta$  is a formula such that, the set of free variables is defined to be the union of the free variables in $\alpha$ and $\beta$.
\item Given a formula $\alpha$ its negation $\neg\alpha$ is still a formula with the same amount of free variables.
\item Given two formulas  $\alpha$ and $\beta$, then $\alpha\Rightarrow\beta$ is a formula with free variables given by the union of the free variables in $\alpha$ and $\beta$.
\end{enumerate}
It is interesting to note that the logical operations just defined can actually be expressed in terms of the primitive symbols as follows:
\begin{enumerate}
\item \emph{true}$:=*=*$.
\item $\alpha\wedge\beta:=\langle\alpha, \beta\rangle=\langle\text{ \emph{true}},\text{ \emph{true}}\rangle=\langle*=*,*=*\rangle$.
\item $\alpha\Leftrightarrow\beta:=\alpha=\beta$.
\item $\alpha\Rightarrow \beta:=\Big((\alpha\wedge\beta)\Leftrightarrow \alpha\Big):=\langle\alpha, \beta\rangle=\langle\text{ \emph{true}},\text{ \emph{true}}\rangle=\alpha$.
\item $\forall x\alpha:=\{x:\alpha\}=\{x:\text{true}\}$.
\item \emph{false}$:=\forall ww:=\{w: w\}=\{w: \text{true}\}$.
\item $\neg \alpha:=\alpha\Rightarrow\text{false}$.
\item $\alpha\wedge\beta:=\forall w[(\alpha\Rightarrow w\wedge\beta\Rightarrow w)\Rightarrow]$.
\item $\exists x \alpha:=\forall w[\forall x(\alpha\Rightarrow w)\Rightarrow w]$.
\end{enumerate}
In the above the notation $\{x:y\}$ indicates the set of all $x$, such that $y$.

\subsection{Representation of $\l$ in a Topos}
We now want to show how a representation of the first order language $\l$ takes place in a topos. The main idea is that of identifying each of the terms in $\l$ with arrows in a topos. In particular we have:
\begin{Definition}
Given a topos $\tau$ the interpretation/representation (M) of the language $\l$ in $\tau$ consists of the following associations:
\begin{enumerate}
\item To each type $T\in \l$ an object $T^{\tau_M}\in\tau$.
\item To each relation symbol $R\subseteq T_1\times T_2\times \cdots\times T_n$ a sub-object $R^{\tau_M}\subseteq T^{\tau_M}_1\times T^{\tau_M}_2\times \cdots\times T^{\tau_M}_n$.
\item To each function symbol $f:T_1\times T_2\times \cdots\times T_n\rightarrow X$ a $\tau$-arrow $f^{\tau_M}:T^{\tau_M}_1\times T^{\tau_M}_2\times \cdots\times T^{\tau_M}_n\rightarrow X$.
\item To each constant $c$ of type $T$ a $\tau$-arrow $c:1^{\tau_M}\rightarrow T^{\tau_M}$.
\item To each variabe $x$ of type $T$ a $\tau$-arrow $x:T^{\tau_M}\rightarrow T^{\tau_M}$.
\item The symbol $\Omega$ is represented by the sub-object classifier $\Omega^{\tau_M}$.
\item The symbol $1$ is represented by the terminal object $1^{\tau_M}$.
\end{enumerate}
\end{Definition}
Now that we understand how the basic symbols of the abstract language $\l$ are represented in a topos we can proceed to understand also how the various terms and formulas are represented. Needless to say these are all defined in recursive manner.\\
Given a term $t(x_1, x_2,\cdots, x_n)$ of type $Y$ with free variables $x_i$ of type $T_i$, i.e. $t(x_1, x_2,\cdots, x_n):T_1\times\cdots \times T_n\rightarrow Y$, then the representative in a topos of this term would be a $\tau$-map
\be
t(x_1, x_2,\cdots, x_n):T^{\tau_M}_1\times\cdots \times T^{\tau_M}_n\rightarrow Y^{\tau_M}
\ee
Formulas in the language are interpreted with terms of type $\Omega$. In the topos $\tau$ this object $\Omega$ is identified with the sub-object classifier $\Omega$.\\
In particular, a term of type $\Omega$ of the form $\phi(t_1, t_2,\cdots t_n)$ with free variables $t_i$ of type $T_i$ is represented by an arrow $$\phi(t_1, t_2,\cdots t_n)^{\tau_M} :T^{\tau_M}_1\times\cdots \times T^{\tau_M}_n\rightarrow \Omega^{\tau_M}$$
On the other hand, a term, $\phi$ of type $\Omega$ with no free variables is represented by a global element $\phi:1^{\tau_M}\rightarrow \Omega^{\tau_M}$ As we will see these arrows will represent the truth values.\\
The reason that in a topos formulas are identified with arrows with codomain $\Omega$ rests in the fact that sub-objects, of a given object in a topos, are in 1:2:1 correspondence with maps from that object to the sub-object classifier. In fact, by construction, formulas single out sub-objects of a given object in terms of a particular relation which they satisfy, i.e. they define elements of $Sub(X)$. Such sub-objects are in 1:2:1 correspondence with maps $X\rightarrow \Omega$ .\\
In particular, given a formula $\phi(x_1,\cdots, x_n)$ with free variables $x_i$ of type $T_i$, which in the language $\l$ is associated with the subset $\{x_i|\phi\}\subseteq \prod_i T_i$, we obtain the topos representation 
\be
\{(x_1,\cdots, x_n)|\phi\}^{\tau_M}\subseteq T^{\tau_M}_1\times\cdots\times T^{\tau_M}_n
\ee
which, through the \emph{Omega Axiom}, gets identified with the map
\be
\{(x_1,\cdots, x_n)|\phi\}^{\tau_M}\rightarrow T^{\tau_M}_1\times\cdots\times T^{\tau_M}_n\xrightarrow{\chi_{\{(x_1,\cdots, x_n)|\phi\}^{\tau_M}}}\Omega^{\tau_M}
\ee
To illustrate this correspondence let us consider the formula stating that two terms are the same, i.e. $t(x_1, x_2,\cdots, x_n)=t^{'}(x_1, x_2,\cdots, x_n)$. The representation of such a formula in a topos $\tau$ is identified with the equalizer of the two $\tau$-arrows representing the terms $t(x_1, x_2,\cdots, x_n)$ and $t^{'}(x_1, x_2,\cdots, x_n)$. In particular we have
\[\xymatrix{ 
\{x_1, x_2,\cdots, x_n|t=t^{'}\}^{\tau_M} \ar@{{>}->}[r]
& T^{\tau_M}_1\times\cdots\times T^{\tau_M}_n\ar@<3pt>[r]^{\;\;\;\;\;\;\; t^{\tau_M}} \ar@<-3pt>[r]_{\;\;\;\;\;\;\;t^{'{\tau_M}}}
& Y^{\tau_M} 
}\]
Instead, if we consider a relation $R(t_1,\cdots t_n)$ of terms $t_i$ of type $Y_i$ with variables $x_j$ of type $T_j$, then the formula pertaining this relation $\{x_1\cdots x_n|R(t_1,\cdots t_n)\}$ is represented in $\tau$ by pulling back the sub-object $R^{\tau_M}\subseteq Y^{\tau_M}_1\times\cdots \times Y_n^{\tau_M}$ (representing the relation $R(t_1,\cdots t_n)$) along the term arrow $\langle t_1^{\tau_M}, \cdots t_n^{\tau_M}\rangle:T_1^{\tau_M}\times \cdots \times T^{\tau_M}_n\rightarrow Y^{\tau_M}_1\times\cdots \times Y_n^{\tau_M}$:
\[
\xymatrix{ 
\{x_1, x_2,\cdots, x_n|R(t_1,\cdots t_n)\}^{\tau_M}\ar[dd] \ar[rr]&&R^{\tau_M}\ar[dd]\\
&&\\
T^{\tau_M}_1\times\cdots\times T^{\tau_M}_n\ar[rr]^{\langle t_1^{\tau_M}, \cdots t_n^{\tau_M}\rangle}&&Y^{\tau_M}_1\times\cdots \times Y_n^{\tau_M}
}\]
The atomic formulas meaning truth and false ($\top$ and $\perp$ respectively) will be represented in a topos $\tau$ by the greatest and lowest elements of the Heyting algebra of the sub-objects of any object in the topos. Thus, for example, we have that 
\ba
\{x_1\cdots x_n|\top\}^{\tau_M}&=&T_1^{\tau_M}\times T_2^{\tau_M}\times \cdots\times T_n^{\tau_M}\\
\{x_1\cdots x_n|\perp\}^{\tau_M}&=&\emptyset^{\tau_M}
\ea
So far we have established how to define formulas in a topos. In the following, we will delineate how to represent logical connectives between formulas in a topos. In particular, given a collection of formulas represented as sub-objects of the type object $\prod_iT_i$, the logical connectives between these are represented by the corresponding operations in the Heyting algebra of sub-objects of the object $\prod_iT^{\tau_M}_i$ in $\tau$. As before, since we are dealing with sub-object we can also represent the logical connective with $\tau$-arrow with codomain $\Omega$ as follows:\\
Consider two formulas $\phi$, $\rho$, of type $\Omega$ with free variables $x_1$ and $x_2$ of type $T_1^{\tau_M}$ and $T^{\tau_M}_2$, respectively. The conjunction $\phi\wedge \rho$ is the map
\be
\phi\wedge \rho:T_1^{\tau_M}\times T_2^{\tau_M}\xrightarrow{\langle \phi, \rho\rangle} \Omega^{\tau_M}\times \Omega^{\tau_M}\xrightarrow{\wedge} \Omega^{\tau_M}
\ee
Similarly, we have
\ba
\phi\vee\rho&:&T_1^{\tau_M}\times T_2^{\tau_M}\xrightarrow{\langle \phi, \rho\rangle} \Omega^{\tau_M}\times \Omega^{\tau_M}\xrightarrow{\vee} \Omega^{\tau_M}\\
\phi\Rightarrow\rho&:&T_1^{\tau_M}\times T_2^{\tau_M}\xrightarrow{\langle \phi, \rho\rangle} \Omega^{\tau_M}\times \Omega^{\tau_M}\xrightarrow{\Rightarrow} \Omega^{\tau_M}\\
\neg\rho&:&T_1^{\tau_M}\times T_2^{\tau_M}\xrightarrow{ \rho} \Omega^{\tau_M}\xrightarrow{\neg} \Omega^{\tau_M}
\ea
Given a language $\l$ a theory $\mathcal{T}$ in $\l$ is a set of formulas which are called the axioms of $\mathcal{T}$. A model of such a theory is then a representation $M$ in which all the axioms of $\mathcal{T}$ are valid. Such axioms are then represented by the arrow $true : 1\rightarrow \Omega$.\\
An example of this is given by the theory of abelian groups which can be seen as model of a theory in a given language as follows.\\ The language required will only contain one type of elements $G$, no relations, two function symbols $$+:G\times G\rightarrow G$$ $$-:G\rightarrow G$$ and a constant $0$. An interpretation of this language, which will lead us to the theory of groups, will be defined in the topos $\Sets$. Such a representation of $G$ will be identified as a set $G^M$, on which the function symbols 
\ba
+:G^M\times G^M&\rightarrow& G^M\\
\langle g_1, g_2\rangle&\mapsto&g_1g_2
\ea
and 
\ba
-:G^M&\rightarrow& G^M\\
g&\mapsto&-g
\ea
act upon. The constant $0$ will be an element of the set $0^M\in G^M$. Such an interpretation will be a model for the theory of abelian groups if the function symbols satisfy the axioms of abelian groups, i.e. the following hold
\ba
(g_1+g_2)+g_3&=&g_1+(g_2+g_3)\\
g_1+g_2&=&g_2+g_1\\
g_1+0&=&g_1\\
g_1+(-g_1)&=&0
\ea
Given two models $M$ and $M^{'}$ of a theory $\mathcal{T}$ in a language $\mathcal{L}$, we say that these two models are homomorphic if there is a homomorphism of the respective interpretations of the model, i.e. for each symbol type $X$ in $\l$, the following maps are homomorphisms:
\be
H_X:X^M\rightarrow X^{M^{'}}
\ee
where $X^M$ and $X^{M^{'}}$ are the representations of the symbol type $X$ of $\mathcal{L}$ in the representation $M$ and $M{'}$, respectively.\\
Such a homomorphism has to respect every relation symbols, function symbols and constants.\\
In the example of abelian groups, model homomorphisms would simply be group homomorphisms.\\

The definition of homomorphic representations gives rise to a category $\mathcal{I}$, whose objects are all possible representations of a given language $\l$ in a topos $\tau$, and whose morphisms are the above mentioned homomorphisms of representations. Given such a category, each theory $\mathcal{T}$ gives rise to a full subcategory of $\mathcal{I}$ called $Mod(\mathcal{T}, \tau)$, whose objects are models of the theory $\mathcal{T}$ in the topos $\tau$, and whose morphisms are homomorphisms of models.\\

In this section we have seen how, given a first order type language $\l$ it is possible to represent such a language in a topos $\tau$. However, interestingly enough the converse is also true, namely, given a topos $\tau$, it has associated to it an internal first order language $\l$, which enables one to reason about $\tau$ in a set theoretic way, i.e. using the notion of elements.  
Thus, we have:
\begin{Definition}
Given a topos $\tau$, its internal language $\l(\tau)$ has as type symbol $^\lefthalfcap  A^\righthalfcap$ for each object $A\in \tau$. A function symbol $^\lefthalfcap f^\righthalfcap:\;^\lefthalfcap A^\righthalfcap_1\times \;^\lefthalfcap A^\righthalfcap_2\times\cdots\times\;^\lefthalfcap  A^\righthalfcap_n\rightarrow \;^\lefthalfcap B^\righthalfcap$ for each map $f:A_1\times A_2\times\cdots\times A_n\rightarrow B$ in $\tau$. And a relation $^\lefthalfcap R^\righthalfcap\subseteq  \;^\lefthalfcap A^\righthalfcap_1\times\; ^\lefthalfcap A^\righthalfcap_2\times\cdots\times \;^\lefthalfcap A^\righthalfcap_n$ for each sub-object $R\subseteq  A_1\times A_2\times\cdots\times A_n$ in $\tau$.
\end{Definition}

\subsection{ A Theory of Physics in the Language $\l$ Represented in a Topos $\tau$}
We will now try to construct a physics theory for a system $S$. The construction of such a theory is defined by an interplay between a language $\l(S)$, associated to the system $S$, a topos and the representation of the theory in the topos. In particular we can say that a theory of the system $S$ is defined by choosing a representation/model, $M$, of the language $\l(S)$ in a topos $\tau_M$. The choice of both topos and representation depend on the theory-type being used, i.e. if it is classical or quantum theory.\\
As we have seen above, since each topos $\tau$ has an internal language $\l(\tau)$ associated to it, constructing a theory of physics consists in translating the language, $\l(S)$, of the system in the local language $\l(\tau)$ of the topos.\\
For now we will not specify what the theory type is, but we will analyse what ground type terms and formulas would be present in a first order language $\l$, which wants to describe and talk about a physical system. However, if the theory type one is utilising is classical physics, than the topos in which to represent your model will be $\mathbf{Sets}$. For a quantum theory, as will be explained in details later on, the topos utilised will be of the form $\Sets^{\c^{op}}$ for an appropriate category $\c$. \\\\
The minimum set of type symbols and formulas, which are needed for a language to be able to talk about a physical system $S$, are the following:
\begin{enumerate}
\item The state space object and the quantity value object are represented in $\l(S)$ by the ground type symbols $\Sigma$ and $\mathcal{R}$. Given a representation $M$ of $\l$ in a topos $\tau$, these objects are represented by the objects $\Sigma_M$ and $\mathcal{R}_M$ in $\tau$. 
\item Given a physical quantity $A$, it is standard practice to represent such a quantity in terms of a function from the state space to the quantity value object. Thus, we require $\l(S)$ to contain the set function symbols $F_{\l(S)}(\Sigma, \mathcal{R})$ of signature $\Sigma\rightarrow \mathcal{R}$, such that the physical quantity is $A:\Sigma\rightarrow \mathcal{R}$. Given a topos $\tau$, these physical quantities are defined in terms of $\tau$-arrows between the $\tau$-objects $\Sigma_M$ and $\mathcal{R}_M$.\\
We will generally require the representation to be faithful, i.e. the map $A\rightarrow A_M$ is one-to-one.

\item We would like to have values of physical quantities. These are defined in $\l(S)$ as terms of type $\mathcal{R}$ with free variables $s$ of type $\Sigma$, i.e. $A(s)$, where $A:\Sigma\rightarrow \mathcal{R}\in F_{\l(S)}(\Sigma, \mathcal{R})$.
Such terms are represented in the topos by terms of type $\mathcal{R}_M$, i.e. $A_M:\Sigma_M\rightarrow \mathcal{R}_M$.

\item We generally would like to talk about values of physical quantities for a given state of the system, thus we require the presence of formulas of the type $A(s)\in \Delta$, where $\Delta$ is a variable of type $P\mathcal{R}$\footnote{By $P\mathcal{R}$ we mean the power set (collection of all subsets) of $\mathcal{R}$.} and $S$ is a variable of type $\Sigma$, which represents a state (being an element of the state space).\\
Since $A(s)\in \Delta$ is a formula, i.e. a term of type $\Omega$ it is represented in a topos $\tau$ by an arrow 
$$ [ A(s)\in \Delta]:\Sigma\times P\mathcal{R}\rightarrow \Omega$$ Such an arrow gets factored as follows:
\be
[ A(s)\in \Delta]=e_{\mathcal{R}}\circ \langle[A(s)], [\Delta]\rangle
\ee 
where $e_{\mathcal{R}}:\mathcal{R}\times P\mathcal{R}\rightarrow\Omega$ is the evaluation map, $[A(s)]:\Sigma\rightarrow \mathcal{R}$ is the arrow representing the physical quantity $A$ and $\Delta:\mathcal{R}\rightarrow\mathcal{R}$ is simply the identity arrow. Putting the two results together we have
\be
\Sigma\times P\mathcal{R}\xrightarrow{A\times id}\mathcal{R}\times P\mathcal{R}\xrightarrow{e_R}\Omega
\ee
\item We would also like to talk about collections of states of the system with a particular property. Such a collection is represented in terms of sub-objects of the state space, which comprises the states with that particular property in question. Thus we have terms $\{s| A(s) \in \Delta\}$ which are of type $P\Sigma$ with a free variable $\Delta$ of type $P\mathcal{R}$. Such a term is represented in a topos by an arrow 
\be
[\{s| A(s) \in \Delta\}]:P\mathcal{R}\rightarrow P\Sigma
\ee

Using this term of type $P(\Sigma)$ a proposition $A\in\Delta $ can be represented as follows:
\be
[\{s | A(s)\in\Xi \}]\circ [\Delta]:1\rightarrow P(\Sigma)
\ee
\item A formula $w$ with no free variables, which we denoted as a sentence, is a special element of $\Omega$ which is represented in a topos by a global element of $\Omega$, i.e.
\be
[w]:1\rightarrow\Omega
\ee
These, as we will see later on, will represent truth values for propositions about the system.
\item Any axioms added to the language have to be represented by the arrow true $T : 1\rightarrow \Omega$.
\end{enumerate}
\subsection{Deductive System of Reasoning for First Order Logic}
Once we have defined the symbols and formation rules for the first order language $\l$, in order to actually use it as a language that enables us to talk about things, we also require rules of inference. Such rules will allow us to derive true statements from other true statements. \\
In order to describe this better we need to introduce the notion of a sequent. 
\begin{Definition}
Given two formulae $\psi$ and $\phi$ a sequent is an expression $\psi\vdash_{\Vec{x}}\phi$ which indicates that $\phi$ is a logical consequence of $\psi$ in the context $\Vec{x}$\footnote{A context $\Vec{x}$ is a list of distinct variables. When applied to a formula it indicates the fact that, that formula, has free variables only within that context, i.e. a formula in a context.}.
\end{Definition}

What this means is that any assignment of values of the variables in $\Vec{x}$, which makes $\psi$ true will also make $\phi$ true.\\
The deduction system will then be defined as a sequent calculus, i.e. a set of inference rules which will allow us to infer a sequent from other sequents. Symbolically, the rule of inference is written as follows:
\be
\frac{\Gamma}{\psi \vdash_{\Vec{x}}\phi}
\ee
which means that the sequent $\psi \vdash_{\Vec{x}}\phi$ can be inferred by the collection of sequents $\Gamma$. We can also have a double inference as follows:
\[
\infer=[ ]
     {\psi \vdash_{\Vec{x}}\phi }
     {\Gamma}
  \]

This can be read in both directions, thus it means that $\psi \vdash_{\Vec{x}}\phi$ can be inferred by the collection of sequents $\Gamma$, but also that the collection of sequents $\Gamma$ can be inferred by $\psi \vdash_{\Vec{x}}\phi$.  \\
We will now define a list of inference rules. In the following, the symbol $\Gamma$ will represent a collection of sequents, the letters $\gamma, \beta,\alpha$ will represent formulae while the letters $\sigma,\tau,\cdots$ will represent terms of some type and $\alpha\cup\Gamma$ represent the collections of formulas in both $\Gamma$ and the formula $\alpha$.
\begin{itemize}
\item \emph{Thinning}
\[
\infer=[ ]
     {\Gamma \vdash_{\Vec{x}}\alpha }
     {\beta\cap\Gamma\vdash_{\Vec{x}}\alpha}
  \]

\item \emph{Cut}
\[
\infer=[ ]
     {\Gamma\vdash_{\Vec{x}}\beta}
     {\Gamma \vdash_{\Vec{x}}\alpha\;\;,\alpha\cup\Gamma \vdash_{\Vec{x}}\beta}
  \]
For any free variable of $\gamma$ free in $\Gamma$ or $\beta$.
\item \emph{Substitution}
\[
\infer=[ ]
     {\Gamma (x/\sigma)\vdash_{\Vec{x}}\alpha(x/\sigma)}
     {\Gamma\vdash_{\Vec{x}}\alpha}
  \]
  where $\Gamma (x/\sigma)$ indicates the term obtained from $\Gamma$ by substituting $\sigma$ (which is a term of some type) for each occurrence of $x$ and $\sigma$ is free for $x$ in $\Gamma$ and $\alpha$.
  \item\emph{Extentionality}
 \[
\infer=[ ]
     {\Gamma \vdash_{\Vec{x}}\sigma=\rho}
     {\Gamma\vdash_{\Vec{x}}x\in\sigma\Leftrightarrow x\in\rho}
  \] 
  where $x$ is not free in either $\Gamma$, $\sigma$ or $\rho$.
  \item\emph{Equivalence}
  \[
\infer=[ ]
     {\Gamma\vdash_{\Vec{x}}\alpha\Leftrightarrow \beta}
     {\alpha\cup\Gamma \vdash_{\Vec{x}}\beta\;\;\;\beta\cup\Gamma\vdash_{\Vec{x}}\alpha}
  \]  
\item \emph{Finite Conjuction}\\
The rules for finite conjunction consist of the following axioms:
\be
\alpha\vdash_{\Vec{x}}\alpha=\top\;\;, \alpha\wedge \beta\vdash_{\Vec{x}} \beta\;\;,\alpha\wedge\beta\vdash_{\Vec{x}}\alpha
\ee
Note that we have used part of the definition of the logical connective `\emph{if then}'.\\
The rule of inference is
\be
\frac{(\alpha\vdash_{\Vec{x}}\beta)(\alpha\vdash_{\Vec{x}}\gamma)}{\alpha\vdash_{\Vec{x}}\gamma\wedge\beta}
\ee
\begin{proof}
\[\infer[^{(7)}]
     {\alpha\vdash_{\Vec{x}}\gamma\wedge\beta}
     {\alpha\vdash_{\Vec{x}}\beta &
     {\infer[^{(6)}]
        {\alpha\cup\beta\vdash_{\Vec{x}}\alpha\wedge\beta}{\alpha\vdash_{\Vec{x}}\gamma & \infer[^{(5)}]
        {\alpha\cup\beta\vdash_{\Vec{x}}\gamma\wedge\beta}{\infer[^{(3)}]{\beta\vdash_{\Vec{x}}\beta=\text{true}}{}&\infer[^{(4)}]{\gamma\cup\beta=\text{true}\vdash_{\Vec{x}}\gamma\wedge\beta}{\infer[^{(1)}]{\gamma\vdash_{\Vec{x}}\gamma=\text{true}}{}&\infer[^{(2)}]{\gamma=\text{true}\cup\beta=\text{true}\vdash_{\Vec{x}}\gamma\wedge\beta}{}}}}}}
        \]
        \end{proof}
        This proof should be read from top to bottom and consists, as one can see, with a finite collection of sequents called a finite \emph{tree}, in which the bottom vertex represents the \emph{conclusion} of the proof. All the sequents of the proof are correlated to each other in the following way:
        \begin{enumerate}
        \item A sequent belonging to a node\footnote{A node is an inference step:$\frac{\Gamma_1}{\Gamma_2}$.} which has nodes above it is derived by applying a rule of inference to the sequents belonging to the above nodes.
        \item Every top most node is either a basic axiom or a premise of the proof.       
        \end{enumerate}  
        
        In the proof above we have that \[\infer[]{\gamma\vdash_{\Vec{x}}\gamma=\text{true}}{}\] is derived by the thinning axiom, the equivalence axiom and the axioms $\gamma\vdash_{\Vec{x}}\gamma$ and $\vdash_{\Vec{x}}\text{true}$ as follows:
        \begin{proof}
        
         \[
         \infer[^{(4)}]{\gamma\vdash_{\Vec{x}}\gamma=\text{true}}{\infer[^{(3)}]{\gamma\cup\gamma\vdash_{\Vec{x}}\text{true}}{\infer[^{(1)}]{\gamma\vdash_{\Vec{x}}\text{true}}{\vdash_{\Vec{x}}\text{true}}}&\infer[^{(2)}]{\text{true}\cup\gamma\vdash_{\Vec{x}}\gamma}{\infer[]{\gamma\vdash_{\Vec{x}}\gamma}{}}
         }
        \]
        
        \end{proof}
        where the lines $(1)$, $(2)$ and $(3)$ are an application of the thinning axiom, while line $(4)$ is the application of the equivalence axiom where the equivalence $\alpha\Leftrightarrow\beta:=\alpha=\beta$ was used. \\
        Going back to the proof of the \emph{conjunction axiom} the remaining lines are derived as follows:  \\\\
        i) Line $(2)$ is the definition of the logical connective $\wedge$.\\\\
        ii) All the other lines are derived from applications of the \emph{cut axiom}.\\\\
        It should be noted that it is also possible to form a more general version of the \emph{conjunction axiom} by replacing the single sequent $\alpha$ by a collection of sequents $\Gamma$ as follows:
  \[\infer[]{\Gamma \vdash_{\Vec{x}}\beta\wedge\gamma}{\Gamma\vdash_{\Vec{x}}\beta & \Gamma\vdash_{\Vec{x}}\gamma}
  \]
\item\emph{Finite Disjunction}\\
The rules for finite conjunction consist of the following axioms:
\be
\bot\vdash_{\Vec{x}}\alpha\;\;\alpha\vdash_{\Vec{x}}\alpha\vee\beta\;\;\beta\vdash_{\Vec{x}}\alpha\vee\beta
\ee
and the following rule of inference:
\be
\frac{(\alpha\vdash_{\Vec{x}}\gamma)(\beta\vdash_{\Vec{x}}\gamma)}{\alpha\vee\beta\vdash_{\Vec{x}}\gamma}
\ee
whose generalization is
\[\infer[]{\alpha\vee\beta\cup\Gamma\vdash_{\Vec{x}}\gamma}{\alpha\cup\Gamma\vdash_{\Vec{x}}\gamma & \beta\cup\Gamma\vdash_{\Vec{x}}\gamma}
\]
\item\emph{Implication}\\
For implication we have the double inference rule
\[
\infer=[ ]
     {\alpha\vdash_{\Vec{x}}\beta\Rightarrow \gamma }
     {\beta\wedge\alpha\vdash_{\Vec{x}}\gamma}
  \]
  Again the general form of which the above is a specification is 
  \[
\infer=[ ]
     {\Gamma\vdash_{\Vec{x}}\beta\Rightarrow \gamma }
     {\beta\cup\Gamma\vdash_{\Vec{x}}\gamma}
  \]
  To see why that is the case we will prove the above generalisation, but only one way: 
  \begin{proof}
  \[\infer[^{(4)}]{\Gamma\vdash_{\Vec{x}}\beta\Rightarrow \gamma}{\infer[^{(2)}] {\beta\wedge\gamma\cup\Gamma\vdash_{\Vec{x}}\beta}{\infer[^{(1)}]{\beta\cup\Gamma\vdash_{\Vec{x}}\beta}{\beta\vdash_{\Vec{x}}\beta}}& \infer[^{(3)}]{\beta\cup\Gamma\vdash_{\Vec{x}}\beta\wedge\gamma}{\beta\cup\Gamma\vdash_{\Vec{x}}\beta & \beta\cup\Gamma\vdash_{\Vec{x}}\gamma}
  }
  \]
  \end{proof}
\item \emph{Negation}  \\
  For negation we only have one axiom
  \be
  \bot\vdash_{\Vec{x}}\alpha
  \ee
  while the inference rules are
  \[
\infer[ ]
     {\Gamma\vdash_{\Vec{x}}\neg\alpha}
     {(\alpha\cup\Gamma)\vdash_{\Vec{x}}\bot}
  \]
  and
  \[
\infer[ ]
     {(\neg\alpha\cup\Gamma)\vdash_{\Vec{x}}\bot}
     {\Gamma\vdash_{\Vec{x}}\alpha}
  \]
\item \emph{Universal Quantification}\\
We have the following double inference rule
\[
\infer=[ ]
     {\alpha\vdash_{\Vec{x}}\forall y\beta}
     {\alpha\vdash_{\Vec{x}y}\beta}
  \]
where $y$ is a free variable \footnote{A variable $x$ in a term $\alpha$ is said to be \emph{bounded} if it appears in a context of the form $x\vdash \alpha$, otherwise it is said to be \emph{free}.} in $\beta$. \\
Again the generalization is
\[
\infer=[ ]
     {\Gamma\vdash_{\Vec{x}}\forall y\beta}
     {\Gamma\vdash_{\Vec{x}y}\beta}
  \]
  \item\emph{Existential Quantifier}\\
 We have the double inference rule
 \[
\infer=[ ]
     {(\exists y)\alpha\vdash_{\Vec{x}}\beta}
     {\alpha\vdash_{\Vec{x}y}\beta}
  \]
  where $y$ is a free variable in $\beta$.\\
Again the generalization would be
\[
\infer=[ ]
     {(\exists y)\alpha\cup\Gamma\vdash_{\Vec{x}}\beta}
     {\alpha\cup\Gamma\vdash_{\Vec{x}y}\beta}
  \]
\item\emph{Distributive Axiom}
\be
(\alpha\wedge(\beta\vee\gamma))\vdash_{\Vec{x}}((\alpha\wedge\beta)\vee(\alpha\wedge\gamma))
\ee
\item\emph{Frobenious Axiom}
\be
((\alpha\wedge(\exists y)\beta)\vdash_{\Vec{x}}(\exists y)(\alpha\wedge\beta)
\ee
where $y\notin \Vec{x}$.
\item\emph{Law of Excluded Middle}
\be
\top\vdash_{\Vec{x}}\alpha\vee\beta
\ee
It should be noted that for intuitionistic type of first order languages, the law of excluded middle does not hold. All the rest does.\\
With this we end our definition of the first order language $\mathcal{L}$ which is comprised of a set of term types, a set of logical connectives and a set of rules of inference which determine the logic.
\end{itemize}

\chapter{Lecture 12}
In this lecture I will describe how in the quantum topos it is possible to define the truth value of a proposition given a state. The collection of such truth values will be a Heyting algebra thus leading to a multivalued logic (intuitionistic logic). I will then give specific examples. I will also define an analogue of the pseudo state called the truth object. This will have the same role as the pseudo sate but can be generalised to represent density matrices not only pure states. We will then compute the truth values of certain propositions with respect to this truth object and check if they reproduce the truth values computes in terms of the pseudo state.

\section{Truth Values Using the Pseudo-State Object }
We are now ready to turn to the question of how truth values are assigned to
propositions which, in this case, are represented by daseinized
operators $\delta(\hat{P})$. For this purpose it is worth thinking
again about classical physics. There, as previously stated, we know that a proposition
$\hat{A}\in\Delta$ is true for a given state $s$ if $s\in
f_{\hat{A}}^{-1}(\Delta)$, i.e. if $s$ belongs to those subsets
$f_{\hat{A}}^{-1}(\Delta)$ of the state space for which the
proposition $\hat{A}\in\Delta$ is true. Therefore, given a state
$s$, all true propositions of $s$ are represented by those
measurable subsets which contain $s$, i.e. those subsets which
have measure $1$ with respect to the Dirac measure $\delta_s$.

In the quantum case, a proposition of the form ``$A\in\Delta$'' is
represented by the presheaf $\ps{\delta(\hat{E}[A\in\Delta])}$
where $\hat E[A\in\Delta]$ is the spectral projector for the
self-adjoint operator $\hat A$, which projects onto the subset $\Delta$ of the
spectrum of $\hat A$. On the other hand, states are represented by
the presheaves $\ps{\w}^{\ket\psi}$. As described above, these
identifications are obtained using the maps
$\mathfrak{S}:P(V)\rightarrow \Sub_{{\rm cl}}(\Sig_V)$,
$V\in\mathcal{V(H)}$, and the daseinization map
$\delta:P(\Hi)\rightarrow \Sub_{\rm {cl}}(\Sig)$, with the
properties that
\begin{eqnarray}
\{{\mathfrak{S}}(\delta(\hat{P})_V)\mid{V\in\V(\Hi)}\}
&:=&\ps{\delta(\hat{P})}\subseteq \Sig\nonumber\\
\{ {\mathfrak{S}}(\w^{\ket\psi}_V) \mid V\in\V(\Hi)\} &:=&
\ps{\w}^{\ket\psi}\subseteq \Sig
\end{eqnarray}
As a consequence, within the structure of formal, typed languages,
both presheaves $\ps{\w}^{\ket\psi}$ and $\ps{\delta(\hat{P})}$
are terms of type $P\Sig$, i.e. they are sub-objects of the spectral presheaf.

We now want to define the condition by which, for each context
$V$, the proposition $(\ps{\delta(\hat{P})})_V$ is true given
$\ps{\w}^{\ket\psi}_V$\footnote{Recall that $\w^{\ket\psi}_V$ represents the projection operator while $\ps{\w}^{\ket\psi}_V$ indicated the subset of $\us_V$. Although ultimately they are equivalent, it is always worth specifying what specific role $\w^{\ket\psi}$ has. }. To this end we recall that, for each
context $V$, the projection operator $\w^{\ket\psi}_V$ can be
written as follows:
\begin{align}
\w^{\ket\psi}_V&=\bigwedge\{\hat{\alpha}\in
P(V)|\ket\psi\langle\psi|\leq\hat{\alpha}\}\nonumber\\
&=\bigwedge\{\hat{\alpha}\in
P(V)|\langle\psi|\hat{\alpha}\ket\psi=1\}\nonumber\\
&=\delta^o(\ket\psi\langle\psi|)_V
\end{align}
This represents the smallest projection in P(V) which has
expectation value equal to one with respect to the state
$\ket\psi$. The associated subset of the Gel'fand spectrum is
defined as
\ba
\ps{\w}^{\ket\psi}_V&=&\mathfrak{S}(\bigwedge\{\hat{\alpha}\in
P(V)|\langle\psi|\hat{\alpha}\ket\psi=1\})\\
&=&\{\lambda\in\us_V|\lambda(\delta^o(\ket\psi\langle\psi|)_V)=1\}
\ea
It follows that
$\ps{\w}^{\ket\psi}= \{\ps{\w}^{\ket\psi}_V \mid{V\in\V(\Hi)}\}$
is the sub-object of the spectral presheaf $\Sig$, such that at each
context $V\in\V(\Hi)$ it identifies those subsets of the Gel'fand
spectrum which correspond (through the map $\mathfrak{S}$) to the
smallest projections of that context, which have expectation value
equal to one with respect to the state $\ket\psi$, i.e. which are
true in $\ket\psi$.

On the other hand, as previously defined, at a given context $V$, the operator
$\delta(\hat{P})_V$ is:
\begin{equation}
\delta^o(\hat{P})_V:=\bigwedge\{\hat{\alpha}\in
P(V)|\hat{P}\leq\hat{\alpha}\}
\end{equation}
Thus the sub-presheaf $\ps{\delta(\hat{P})}$ is defined as the
sub-object of $\Sig$, such that at each context $V$ it defines the
subset $\ps{\delta(\hat{P})}_V$ of the Gel'fand spectrum $\Sig_V$,
which represents (through the map $\mathfrak{S}$) the projection
operator $\delta(\hat{P})_V$.

We are interested in defining the condition by which the
proposition represented by the sub-object $\ps{\delta(\hat{P})}$ is
true given the state $\ps{\w}^{\ket\psi}$. Let us analyse this
condition for each context V. In this case, we need to define the
condition by which the projection operator $\delta(\hat{P})_V$,
associated to the proposition $\ps{\delta(\hat{P})}$ is true, given
the pseudo state $\ps{\w}^{\ket\psi}$. Since at each context $V$
the pseudo-state defines the smallest projection in that context
which is true with probability one, i.e. $(\w^{\ket\psi})_V$, for
any other projection to be true given this pseudo-state, this
projection must be a coarse-graining of $(\w^{\ket\psi})_V$, i.e.
it must be implied by $(\w^{\ket\psi})_V$. Thus, if
$(\w^{\ket\psi})_V$ is the smallest projection in $P(V)$, which is
true with probability one, then the projector $\delta^o(\hat{P})_V$
will be true if and only if $\delta^o(\hat{P})_V\geq
(\w^{\ket\psi})_V$. This condition is a consequence of the fact
that, if $\langle\psi|\hat{\alpha}\ket\psi=1$, then for all
$\hat{\beta}\geq\hat{\alpha}$ it follows that
$\langle\psi|\hat{\beta}\ket\psi=1$.

So far we have defined a `truthfulness' relation at the level of
projection operators, namely $\delta^o(\hat{P})_V\geq
(\w^{\ket\psi})_V$. Through the map $\mathfrak{S}$ it is
possible to shift this relation to the level of sub-objects of the
Gel'fand spectrum:
\begin{align}
\mathfrak{S}((\w^{\ket\psi})_V)&\subseteq
\mathfrak{S}(\delta^o(\hat{P})_V)\label{equ:truthvalue}\\
\ps{\w}^{\ket\psi}_V&\subseteq
\ps{\delta(\hat{P})}_V\nonumber\\
\{\lambda\in\Sig(V)|\lambda
((\delta^o(\ket\psi\langle\psi|)_V)=1\}&\subseteq
\{\lambda\in\Sig(V)|\lambda((\delta^o(\hat{P}))_V)=1\}
\end{align}
What the above equation reveals is that, at the level of
sub-objects of the Gel'fand spectrum, for each context $V$, a
`proposition' can be said to be (totally) true for given a
pseudo-state if, and only if, the sub-objects of the Gel'fand
spectrum, associated to the pseudo-state, are subsets of the
corresponding subsets of the Gel'fand spectrum associated to the
proposition. It is straightforward to see that if
$\delta(\hat{P})_V\geq (\w^{\ket\psi})_V$, then
$\mathfrak{S}((\w^{\ket\psi})_V)\subseteq
\mathfrak{S}(\delta(\hat{P})_V)$ since for projection operators
the map $\lambda$ takes the values 0,1 only.

We still need a further abstraction in order to work directly with
the presheaves $\ps{\w}^{\ket\psi}$ and $\ps{\delta(\hat{P})}$.
Thus we want the analogue of equation (\ref{equ:truthvalue}) at
the level of sub-objects of the spectral presheaf, $\Sig$. This
relation is easily derived to be
\begin{equation}\label{equ:tpre}
\ps{\w}^{\ket\psi}\subseteq\ps{\delta(\hat{P})}
\end{equation}

Equation (\ref{equ:tpre})  shows that, whether or not a proposition
$\ps{\delta(\hat{P})}$ is `totally true' given a pseudo state
$\ps{\w}^{\ket\psi}$ is determined by whether or not the
pseudo-state is a sub-presheaf of the presheaf
$\ps{\delta(\hat{P})}$. With this motivation, we can now define the
\emph{generalised truth value} of the proposition ``$A\in\Delta$''
at stage $V$, given the state $\ps{\w}^{\ket\psi}$, as:
\begin{align}\label{ali:true}
v(A\in\Delta;\ket\psi)_V&= v(\ps{\w}^{\ket\psi}
\subseteq\ps{\delta(\hat{E}[A\in\Delta])})_V             \\
&:=\{V^{'}\subseteq V|(\ps{\w}^{\ket\psi})_{V^{'}}\subseteq
\ps{\delta(\hat{E}[A\in\Delta])}_{V^{'}}\}\\ \nonumber
&=\{V^{'}\subseteq
V|\langle\psi|\delta(\hat{E}[A\in\Delta])_V\ket\psi=1\}
\end{align}
The last equality is derived by the fact that
$(\ps{\w}^{\ket\psi})_V\subseteq \ps{\delta(\hat{P})}_V $ is a
consequence that at the level of projection operator
$\delta^o(\hat{P})_V\geq (\w^{\ket\psi})_V$. But, since
$(\w^{\ket\psi})_V$ is the smallest projection operator such that
$\langle\psi|(\w^{\ket\psi})_V\ket\psi=1$, then
$\delta^o(\hat{P})_V\geq (\w^{\ket\psi})_V$ implies that
$\langle\psi|\delta^o(\hat{P})\ket\psi=1$.

The right hand side of equation (\ref{ali:true}) means that the
truth value, defined at $V$ of the proposition ``$A\in\Delta$'',
given the state $\ps{\w}^{\ket\psi}$, is given in terms of all
those sub-contexts $V^{'}\subseteq V$ for which the projection
operator $\delta(\hat{E}[A\in\Delta])_V$ has expectation value
equal to one with respect to the state $\ket\psi$. In other words,
this \emph{partial} truth value is defined to be the set of all
those sub-contexts for which the proposition is totally true.

So we can see how a local truth value can be seen as a measure of how for the proposition is from being true. In fact what the sieve $v(A\in\Delta;\ket\psi)_V$ tells you is how much you have to generalise your original proposition for it to be true.

We now need to show that indeed $v(A\in\Delta;\ket\psi)_V=S$ is a sieve. To this end we simply need to show that it is closed under left composition. 
\begin{proof}
Consider an algebra $V^{'}\in S$ then given any other algebra $V^{''}\subseteq V$ then we want to show that $V^{''}\in S$. Now since $V^{'}\in S $, then $\delta^o(\hat{P})_{V^{'}}\geq (\ps{\w}^{\ket\psi})_{V^{'}}$, however from the definition of daseinisation we have that $\delta^o(\hat{P})_{V^{''}}\geq \delta^o(\hat{P})_{V^{'}}$ therefore $\delta^o(\hat{P})_{V^{''}}\geq (\ps{\w}^{\ket\psi})_{V^{'}}$, which implies that $V^{''}\in S$.

\end{proof}
Thus pictorially we have the following situation
 \[\xymatrix{
\ps{\w}^{\ket\psi}\ar@{^{(}->}[rr]\ar[dd]&&\ps{\delta(\hat{P})}\ar[dd]\\
&&\\
\underline{1}\ar[rr]^{\gamma}&&\uom\\
}\]
where to each sub-object relations which identifies the mathematical concept of evaluation we associate a global element $\gamma:1\rightarrow\uom$ which represents the global truth value. Such a global truth value will have local components defined as follows

 \[\xymatrix{
(\ps{\w}^{\ket\psi})_V\ar@{^{(}->}[rr]\ar[dd]&&\delta^o(\hat{P})_V\ar[dd]\\
&&\\
\underline{1}_V\ar[rr]^{\gamma_V}&&\uom_V\\
}\]
Which pick out a particular sieve for each context.

The reason why all this works is that generalised truth values defined
in this way form  a \emph{sieve} on $V$; and the set of all of
these is a Heyting algebra. Specifically:
$v(\ps{\w}^{\ket\psi}\subseteq \ps{\delta(\hat{P})})$ is a
global element defined at stage V of the sub-object classifier
$\Om:=(\Om_V)_{V\in\V(\Hi)}$, where $\Om_V$ represents the set of
all sieves defined at stage V. \\
The set of truth values is defined as $\Gamma(\Omega)$ and it forms a Heyting algebra.\\\\
%
\section{Example}
We will consider again a 4 dimensional Hilbert space $\Cl^4$ whose category is $\mv(\Cl^4)$, the state space $\us$ we have already computed in previous examples. We now want to define the truth values of the proposition  $S_z\in[-3, -1]$ which has corresponding projector operator $\hat{P}:=\hat{E}[S_z\in[-3, -1]]=|\phi\rangle\langle\phi|$. This proposition is equivalent to the projection operator $\hat{P}_4$.\\
The state we will consider will be $\psi=(1,0,0,0)$ with respective operator $|\psi\rangle\langle\psi|=\hat{P}_1$. First of all we will consider the context $V=lin_{\Cl}(\hat{P}_1, \hat{P}_2, \hat{P}_3,\hat{P}_4)$. Since $\hat{P}_4\in V$ it follows that $\delta^o(\hat{P}_4)_V=\hat{P}_4$, but $\langle\psi|\hat{P}_4|\psi\rangle=0$, thus we need to go to smaller contexts. In particular, only contexts in which $\delta^o(\hat{P}_4))$ is an operator which is implied by $\hat{P}_1$ will contribute to the truth value. We thus obtain  
\ba
v(\ps{\w}^{\ket\psi}\subseteq\ps{\delta(\hat{P}_4)})_V&:=&\{V^{'}\subseteq
V|\langle\psi|\delta(\hat{E}[A\in\Delta])\ket\psi=1\}\\
&=&\{V_{\hat{P}_{i\in\{2,3\}}}, V_{\hat{P}_{2}\hat{P}_{3}}\}
\ea
On the other hand for sub contexts we obtain
\ba
v(\ps{\w}^{\ket\psi}\subseteq\ps{\delta(\hat{P}_4)})_{V_{\hat{P}_1}}&:=&\{\emptyset\}\\
v(\ps{\w}^{\ket\psi}\subseteq\ps{\delta(\hat{P}_4)})_{V_{\hat{P}_2}}&:=&\{V_{\hat{P}_2}\}\\
v(\ps{\w}^{\ket\psi}\subseteq\ps{\delta(\hat{P}_4)})_{V_{\hat{P}_3}}&:=&\{V_{\hat{P}_3}\}\\
v(\ps{\w}^{\ket\psi}\subseteq\ps{\delta(\hat{P}_4)})_{V_{\hat{P}_4}}&:=&\{\emptyset\}\\
v(\ps{\w}^{\ket\psi}\subseteq\ps{\delta(\hat{P}_4)})_{V_{\hat{P}_1, \hat{P}_2}}&:=&\{V_{\hat{P}_2}\}\\
v(\ps{\w}^{\ket\psi}\subseteq\ps{\delta(\hat{P}_4)})_{V_{\hat{P}_1, \hat{P}_3}}&:=&\{V_{\hat{P}_3}\}\\
v(\ps{\w}^{\ket\psi}\subseteq\ps{\delta(\hat{P}_4)})_{V_{\hat{P}_1, \hat{P}_4}}&:=&\{\emptyset\}\\
v(\ps{\w}^{\ket\psi}\subseteq\ps{\delta(\hat{P}_4)})_{V_{\hat{P}_2, \hat{P}_3}}&:=&\{V_{\hat{P}_2, \hat{P}_3},V_{\hat{P}_2},V_{\hat{P}_3}\}\\
v(\ps{\w}^{\ket\psi}\subseteq\ps{\delta(\hat{P}_4)})_{V_{\hat{P}_2, \hat{P}_4}}&:=&\{V_{\hat{P}_2}\}\\
v(\ps{\w}^{\ket\psi}\subseteq\ps{\delta(\hat{P}_4)})_{V_{\hat{P}_3, \hat{P}_4}}&:=&\{V_{\hat{P}_3}\}\\
\ea
As can be seen, the truth values obtained by using the pseudo state object coincide with the truth values obtained using the truth object.

\section{Truth Object}
We will now analyse the \emph{truth object} option\footnote{We have seen in previous lectures how to represent propositions, which are linguistic objects of type $P(\Sigma)$ in a topos $\tau$. Moreover, in this lecture we have defined such a representation for $\tau=\Sets^{\mv(\mh)^{op}}$. We are now interested in understanding how propositions get assigned truth values. We know that at a linguistic level a truth value is defined as an element of the Heyting algebra $\Gamma(\Omega)$, therefore we need to associate to the elements $\{s|A(s)\in\Delta\}$ or $A(s)\in\Delta$ of type $P(\Sigma)$, which represent propositions, an element of type $\Gamma(\Omega)$. \\
Since in a representation $M$ an element of $\Gamma(\Omega)$ is defined as a term of type $\Omega$ with no free variables, finding the truth value of a proposition would be equivalent in finding a way of transforming a term of type $P(\Sigma)$ to a term of type $\Omega$, with no free variables.
This can be done in two ways:
\begin{enumerate}
\item \emph{Truth Object} \\
This method consists in defining a term $\mathbb{T}$ of type $P(P(\Sigma))$, then given a term $t$ of type $P(\Sigma)$, the term $t\in \mathbb{T}$ would be the desired term of type $\Omega$.

\item \emph{Psuedo-State Object}\\
This method consists in defining a term $\w$ of type $P(\Sigma)$, which was defined above and represents the topos analogue of a state of the system. Then the truth value of a proposition $t$ (term of type $P(\Sigma)$) would be $\w\subseteq t$, which is a term of type $\Omega$.

\end{enumerate}
We will now analyse the \emph{truth object} option. The reason for introducing such an object is because it enables one to also define the topos analogue of density matrices, which was not possible when only considering the pseudo-state option.\\\\
Our aim is to define a term $T$ of type $P(P(\Sigma))$ such that, given the representation $\{s|A(s)\in\Delta\}$ of a proposition, the term $(\{s|A(s)\in\Delta\}\in \mathbb{T})$ is a term of type $\Omega$. Such a term has free variables $\Delta$ of type $P(\mathcal{R})$ and $\mathbb{T}$ of type $P(P(\Sigma))$. Therefore, its representation in a topos $\tau$ would be
\be
[\{s | A(s)\in \Delta\} \in \mathbb{T} ]_{\tau_{M}} : P(\mathcal{R})\times P(P(\Sigma))\rightarrow \Omega 
\ee   
which can be factored as follows:
\be
[\{s | A(s)\in \Delta\} \in  \mathbb{T} ]_{\tau_{M}}= e_{P(\Sigma)}\circ [ \{s | A(s)\in \Delta\}]\times [ \mathbb{T} ]
\ee
where $e_{P(\Sigma)}:P(\sigma)\times P(P(\Sigma))\rightarrow\Omega$ is the evaluation map and 
\ba
 [ \{s | A(s)\in \Delta\}]: P(\mathcal{R})&\rightarrow &P(\Sigma)\\
^\lefthalfcap\mathbb{T}^\righthalfcap:P(P(\Sigma))&\xrightarrow{id}&P(P(\Sigma))
\ea
Given the above, the truth value of the proposition $A\in\Delta$ is represented as
\be
v(A\in\Delta;\mathbb{T}):[\{s | A(s)\in \Delta\} \in  \mathbb{T} ]_{\tau_{M}}\circ\langle ^\lefthalfcap  \Delta^\righthalfcap, ^\lefthalfcap   \mathbb{T}^\righthalfcap  \rangle
\ee
where $\langle ^\lefthalfcap  \Delta^\righthalfcap, ^\lefthalfcap \mathbb{T} ^\righthalfcap  \rangle:1\rightarrow P(\mathcal{R})\times P(P(\Sigma))$}. The reason for introducing such an object is because it enables one to also define the topos analogue of density matrices, which was not possible when only considering the pseudo-state option which allowed only a topos representation for pure states.

In order to define the truth-object we need to first introduce a new presheaf called the \emph{outer presheaf} whose definition is as follows
\begin{Definition}
The outer presheaf $\underline{O}:\mv(\mh)\rightarrow Sets$ is defined on 
\begin{enumerate}
\item Objects: for each $V\in \mv(\mh)$ we obtain $\uo_V:=P(V)$, i.e. the collection of all projection operators in $V$.
\item Morphisms: given a map $i:V^{'}\subseteq V$ in $\mv(\mh)$ the corresponding presheaf map is 
\ba
\uo(i_V^{'}V):\uo_V&\rightarrow&\uo_{V^{'}}\\
\hat{\alpha}&\mapsto&\delta^o(\hat{\alpha})_{V^{'}}
\ea 

\end{enumerate}

\end{Definition}
The above is a well defined presheaf. To see this all we need to show is that, given another inclusion map $j:V^{''}\subseteq V^{'}$ then the following holds
\be
\uo(i\circ j)=\uo(j)\circ \uo(i)
\ee
Computing the left hand side we get
\ba
\uo(i\circ j):\uo_{V}&\rightarrow& \uo_{V^{''}}\\
\hat{\alpha}&\mapsto&\delta^o(\hat{\alpha})_{V^{''}}
\ea
Computing the right hand side we get
\ba
\uo_{V}&\xrightarrow{\uo(i)}&\uo_{V^{'}}\xrightarrow{\uo(j)}\uo_{V^{''}}\\
\hat{\alpha}&\mapsto&\delta^o(\hat{\alpha})_{V^{'}}\mapsto\delta^o(\delta^o(\hat{\alpha})_{V^{'}})_{V^{''}}
\ea
where $\delta^o(\delta^o(\hat{\alpha})_{V^{'}})_{V^{''}}\geq \delta^o(\hat{\alpha})_{V^{'}}$. In particular, applying the definition of daseinisation recursively it follows trivially that 
$\delta^o(\delta^o(\hat{\alpha})_{V^{'}})_{V^{''}}=\delta^o(\hat{\alpha})_{V^{'}_{|V^{''}}}=\delta^o(\hat{\alpha})_{V^{'}}$.

From the above definition it follows that for each projection operator $\hat{P}$, the assignment $V\rightarrow \delta^o(\hat{P})_V$ defines a global element of $\uo$. We thus arrive at an alternative, but ultimately equivalent definition of the daseinisation map
\ba
\delta: P(\mh)&\rightarrow &\Gamma(\uo)\\
\hat{P}&\mapsto&\{\delta^{o}(\hat{P})_V|V\in \mv(\mh)\}
\ea
\subsubsection{Property of the daseinisation map}
We will now state some properties of the daseinisation map as defined above
\begin{enumerate}
\item For all $V\in \mv(\mh)$ we obtain $\delta^o(\hat{0})_V=\hat{0}$. Thus the null projection operator represents the proposition of the form $A\in \Delta$ such that $sp(\hat{A})\cap\Delta=\emptyset$ (false proposition).
\item For all $V\in \mv(\mh)$ we obtain $\delta^o(\hat{1})_V=\hat{1}$. Thus the unit projection operator represents the proposition of the form $A\in \Delta$ such that $sp(\hat{A})\cap\Delta=sp(\hat{A})$ (true proposition).
\item The map $\delta$ is not surjective but it is injective.
\begin{proof}
Given two global elements which we denote $\delta(\hat{P})$ and $\delta(\hat{Q})$ such that $\delta(\hat{P})=\delta(\hat{Q})$, we want to show that $\hat{P}=\hat{Q}$. Now is is easy to see that $\hat{P}=\bigwedge_{V\in \mv(\mh)}\delta^o(\hat{P})_V$ since there will exist a context $V$ such that $\hat{P}\in V$. It follows that
\be
\hat{P}=\bigwedge_{V\in \mv(\mh)}\delta^o(\hat{P})_V=\bigwedge_{V\in \mv(\mh)}\delta^o(\hat{Q})_V=\hat{Q}
\ee

\end{proof}
\end{enumerate}

\subsubsection{Properties of the Outer-Daseinisation Presheaf }
In order to show that $\uo\subseteq P_{cl}(\us)$ we need to show that there exists a monic arrow $i:\uo\rightarrow P_{cl}(\us)$. To this end recall from previous lectures (exponential) that there exists, in any topos $\tau$, a bijection 
\be\label{equ:exp}
Hom_{\tau}(A, C^B)\rightarrow Hom_{\tau}(A\times B, C)
\ee
We would like to utilise this bijection to define the map $i$, to do so we need to utilise another result of topos theory which states that sub-objects of a given object are in bijective correspondence with maps from the object in question to the sub-object classifier\footnote{This was shown when we defined the subobject classifier.}. Thus for the case at hand $P(\us)\simeq \uom^{\us}$. Substituting this in equation \ref{equ:exp} we get
\be\label{equ:exp2}
Hom_{Sets^{\mv(\mh)^{op}}}(\uo, P(\us))\rightarrow Hom_{Sets^{\mv(\mh)^{op}}}(\uo\times \us, \uom)
\ee
Consider a map $j\in Hom_{Sets^{\mv(\mh)^{op}}}(\uo\times \us, \uom)$ which, for each $V\in \mv(\mh)$ we define as follows
\ba
j_V:(\uo\times \us)_V&\rightarrow& \uom_V\\
(\hat{\alpha}, \lambda)&\mapsto&j_V(\hat{\alpha}, \lambda):=\{V^{'}\subseteq V|\us(i_{V^{'}V})\lambda\in S_{\uo(i_{V^{'}V})\hat{\alpha}}\}
\ea
where $\hat{\alpha}\in P(V)$, $\uo(i_{V^{'}V})\hat{\alpha}=\delta^o(\hat{\alpha})_{V^{'}}$ and  $S_{\uo(i_{V^{'}V})\hat{\alpha}}:=\{\lambda\in \us_{V^{'}}|\langle\lambda,  \delta^o(\hat{\alpha})_{V^{'}}\rangle=1\}$

However we know that 
\be
S_{\us(i_{V^{'}V})\hat{\alpha}}=\us(i_{V^{'}V})S_{\hat{\alpha}}
\ee
for all $V^{'}\subseteq V$ and $\hat{\alpha}\in \uo_V$. Therefore we can write $j_V(\hat{\alpha}, \lambda)$ as follows
\be
j_V(\hat{\alpha}, \lambda):= \{V^{'}\subseteq V|\us(i_{V^{'}V})\lambda\in \us(i_{V^{'}V})S_{\hat{\alpha}}\}
\ee
for all $(\hat{\alpha}, \lambda)\in \uo(V)\times \us_V$. As defined above $j_V(\hat{\alpha}, \lambda)$ is a sieve in $\uom_V$. 
\begin{proof}
We want to show that 
\be
j_V(\hat{\alpha}, \lambda):= \{V^{'}\subseteq V|\us(i_{V^{'}V})\lambda\in \us(i_{V^{'}V})S_{\hat{\alpha}}\}
\ee
is a sieve.

To this end we need to show that if $V^{'}\in j_V(\hat{\alpha}, \lambda)$ then for all $V^{''}\subseteq V^{'}$ $V^{''}\in j_V(\hat{\alpha}, \lambda)$. 
So let us assume that $V^{'}\in j_V(\hat{\alpha}, \lambda)$ then it follows that $\lambda_{|V^{'}}\in S_{\delta^o(\hat{\alpha})_{V^{'}}}$. 
Therefore $\lambda_{|V^{'}}(\delta^o(\hat{\alpha})_{V^{'}})=1$ Now let us consider any $V^{''}\subseteq V^{'}$. 
It follows that $\us(i_{V^{''}V^{'}})\lambda_{|V^{'}}\in\us_{V^{''}}$ and $\us(i_{V^{''}V^{'}})S_{\delta^o(\hat{\alpha})_{V^{'}}}=S_{\delta^o(\hat{\alpha})_{V^{''}}}$. 
Now since $\delta^o(\hat{\alpha})_{V^{''}}\geq \delta^o(\hat{\alpha})_{V^{'}}$ then $\lambda_{|V^{'}}(\delta^o(\hat{\alpha})_{V^{''}})=1$ 
in particular $(\lambda_{|V^{'}})_{|V^{''}}(\delta^o(\hat{\alpha})_{V^{''}})=1$. Therefore $\us(i_{V^{''}V^{'}})\lambda_{|V^{'}}\in \us(i_{V^{''}V^{'}})S_{\delta^o(\hat{\alpha})_{V^{'}}}=S_{\delta^o(\hat{\alpha})_{V^{''}}}$, i.e. $V^{''}\in j_V(\hat{\alpha}, \lambda)$.

\end{proof}
We now need to show that the collection of maps $j_V:(\uo\times \us)_V\rightarrow \uom_V$ for each  $V\in \mv(\mh)$ as defined above, combine together to form a natural transformation $j:\uo\times \us\rightarrow \uom$. 
\begin{proof}
We want to show that $j:\uo\times \us\rightarrow \uom$ is indeed a natural transformation. Thus we need to show that for all pairs $V^{'}\subseteq V$ the following diagram commutes
\[\xymatrix{
\uo_V\times \us_V\ar[rr]^{j_V}\ar[dd]^{f(i_{V^{'}V})}&&\uom_V\ar[dd]^{\uom(i_{V^{'}V})}\\
&&\\
uo_{V^{'}}\times \us_{V^{'}}\ar[rr]^{j_{V^{'}}}&&\uom_{V^{'}}\\
}\]
If we chaise the diagram around we obtain
\ba
\uom(i_{V^{'}V})\circ j_{V}(\hat{\alpha}, \lambda)&=&( j_{V}(\hat{\alpha}, \lambda))_{|V^{'}}=j_{V}(\hat{\alpha}, \lambda)\cap \downarrow V^{'}\\
&=&\{V^{''}\subseteq V^{'}|V^{''}\in j_{V}(\hat{\alpha}, \lambda)\}\\
&=&\{V^{''}\subseteq V^{'}|\us(i_{V^{''}V})\lambda\in\us(i_{V^{'}V})(S_{\hat{\alpha}})\}
\ea
In order to define the left path of the diagram we need to define what the maps $f(i_{V^{'}V})$ are.

\ba
f(i_{V^{'}V}):(\uo_V\times \us_V)&\rightarrow& (\uo_{V^{'}}\times \us_{V^{'}})\\
(\hat{\alpha}, \lambda)&\mapsto&f(i_{V^{'}V})(\hat{\alpha}, \lambda):=(\uo(i_{V^{'}V})\hat{\alpha}, \us(i_{V^{"}V})\lambda)
\ea
Therefore going around the diagram we obtain
\ba
j_{V^{'}}\circ f(i_{V^{'}V})(\hat{\alpha}), \lambda)&=&j_{V^{'}}(\uo(i_{V^{'}V})\hat{\alpha}, \us(i_{V^{"}V})\lambda)\\
&=&j_{V^{'}}(\delta^o(\hat{\alpha})_{V^{'}}, \lambda_{|V^{'}})\\
&=&\{V^{''}\subseteq V^{'}|\us(i_{V^{''}V^{'}})(\lambda_{|V^{'}})\in \us(i_{V^{''}V^{'}})S_{\delta^o(\hat{\alpha})_{V^{'}}}\}
\ea

Now since
\ba
\us(i_{V^{''}V})\alpha&=&\lambda_{V^{''}}=\us(i_{V^{''}V^{'}})(\lambda_{|V^{'}})\\
\us(i_{V^{'}V})(S_{\hat{\alpha}})&=&S_{\delta^o(\hat{\alpha})_{V^{''}}}=\us(i_{V^{''}V^{'}})S_{\delta^o(\hat{\alpha})_{V^{'}}}
\ea
It follows that the above diagram commutes.

\end{proof} 

Because of the equivalence in equation \ref{equ:exp2} to the map $j$ there corresponds a map $i:\uo\rightarrow  P(\us)$. However we are interested in $i:\uo\rightarrow P_{cl}(\us)$, but this restriction poses no problems.
\begin{proof}
We are interested in restriction our attention only to clopen sub-objects of $\us$, i.e. we want $i:\uo\rightarrow P_{cl}(\us)$ rather than $i:\uo\rightarrow P(\us)$. Now let us consider $\uo$. For each context we have $\uo_V=P(V)$ which is a lattice of operator with the usual lattice ordering. Now it is possible to put some topology on this set. Whatever topology we choose the entire set $\uo_V$ will be both open and closed. Thus, for each clopen sub-objects $\ps{A}\subseteq \us$ we now form clopen sub-objects of $\uo\times\us$ as follows
\be
\uo\times\ps{A}\subseteq \uo\times \us
\ee
is a clopen sub-object iff for each $V\in \mv(\mh)$, $\uo_V\times\ps{A}_V$ is a clopen subset of $\uo_V\times \us_V$. We then have that $Sub_{cl}( \uo\times \us)\subseteq Sub( \uo\times \us)$.

However we know that
\be
Sub(X)\simeq Hom(X\rightarrow \uom)
\ee
Therefore
\be
Sub_{cl}( \uo\times \us)\subseteq Sub( \uo\times \us)\simeq Hom(\uo\times \us, \uom)\simeq Hom(\uo, P(\us))
\ee 
But
\be
Sub_{cl}( \uo\times \us)\simeq Hom_{cl}(\uo\times \us, \uom)\subseteq  Hom(\uo\times \us, \uom)
\ee
A moment of thought reveals that 
\be
Hom_{cl}(\uo\times \us, \uom)\simeq  Hom(\uo, P_{cl}(\us))
\ee
\end{proof}
The map $i:\uo\rightarrow P_{cl}(\us)$ is the desired map, which is injective: for each $V\in \mv(\mh)$, $i_V:\uo_V\rightarrow P_{cl}(\us)_V$ is injective.
\begin{proof}
To this end let us consider first the definition of $P_{cl}(\us)$. This is given as follows
\begin{Definition}
The power object $P\us$ of of $\us$ is the
presheaf given by 
\begin{enumerate}
\item(i) On objects $V\in \mv(\mh)$
$$ P\us_V:=\{\nu_V:\us_{\downarrow V}\rightarrow\uom_{\downarrow V}|\eta_V\text{ is a natural transformation}\}$$
Here $\us_{\downarrow V}$ is the restriction of $\us$ to a smaller poset namely $\downarrow V\subseteq \mv(\mh)$.
\item On morphisms : for $i:V^{'}\subseteq V$ the presheaf maps are
\ba
P\us(i_{V^{'}V}):P\us_V&\rightarrow& P\us_{V^{'}}\\
\eta&\mapsto&\eta_{|V^{'}}
\ea
where here $\eta_{|V^{'}}:\us_{\downarrow V^{'}}\rightarrow \uom_{\downarrow V^{'}}$\end{enumerate}
\end{Definition}

Form the isomorphism 
$$
Sub(X)\simeq Hom(X, \Omega)
$$

It follows that 
$$
P(\us)_V:=Sub(\us_{\downarrow V})
$$
If we then restrict to clopen sub-objects we get
\be
P_{cl}(\us)_V:=Sub_{cl}(\us_{\downarrow V})
\ee
Therefore, for each context $V\in \mv(\mh)$ we have that 
\ba
i_V:\uo_V&\rightarrow &P_{cl}(\us_V)\\
P(V)&\rightarrow&Sub_{cl}(\us_{\downarrow V})\\
\hat{P}&\mapsto&(\ps{S}_{\hat{P}})_{\downarrow V}
\ea
Thus the map is clearly monic.
\end{proof}

The above reasoning showed us that, since $ \mathbb{T}^{|\psi\rangle}$ is a subobject of $\uo$ and since $\uo\subseteq P_{cl}(\us)$ it follows that $ \mathbb{T}^{|\psi\rangle}\in P(P_{cl}(\us))$.

It is interesting to now compare the two definitions of daseinisation given so far in the lecture course.
\ba
\delta:P(\mh)&\rightarrow& \Gamma(\uo)\\
P(\mh)&\rightarrow& Sub_{cl}(\us)
\ea
These two definition, although seemingly different are exactly the same. In fact as we just showed, $\uo\subseteq P_{cl}(\us)$, thus a global element $\gamma$ of $\uo$ will pick out, for each $V\in \mv(\mh)$ an element $\gamma_V\in \uo_V\subseteq P_{cl}(\us)$. Clearly considering all the context together $\gamma$ will give rise to a (clopen due to how $\uo$ is defined)  sub-object of $\us$, hence $\Gamma(\uo)\in P_{cl}(\us)$. 
\subsection{Example of Truth Object in Classical Physics}
Let us consider a proposition $(A\in \Delta)$ meaning that the value of the quantity $A$ lies in $\Delta$. We want to define the truth value of such a proposition with respect to a given state $s\in X$ where $X$ is the state space. In classical theory,  the truth value of the above proposition in the state $s$ is given by
$$v(A\in \Delta;s) :=\begin{cases}1\text{ iff }s\in  f_{A\in \Delta}^{-1}(\Delta)\\0 \text{ otherwise }\end{cases}$$
where $f_{A\in \Delta}^{-1}(\Delta)\subseteq X$ is the subset of the state space for which the proposition $(A\in \Delta)$ is true.\\
Another way of defining truth values is through the truth object $\mathbb{T}^s$, which is state dependent. The definition of the truth object is as follows: \\
for each state $s$, we define the set
$$\mathbb{T}^s :=\{S\subseteq X|s\in S\}$$ 
Since $(s\in f_{A\in \Delta}^{-1}(\Delta))$ iff $ f_{A\in \Delta}^{-1}(\Delta)\in \mathbb{T}^s$, we can now write the truth value above in the following equivalent way:
$$v(A\in \Delta;\mathbb{T}):=\begin{cases} 1 \text{ iff }f_{A\in \Delta}^{-1}(\Delta)\in \mathbb{T}^s\\ 0\text{ otherwise }\end{cases}$$
In classical physics propositions $A\in \Delta$ are identified with subsets $S:\{s | A(s)\in \Delta\}$ of the state space $X$, for which that proposition is true. Therefore, the truth value $v(A\in \Delta;\mathbb{T}^s)$ is equivalent to the truth value of the mathematical statement $[\{s | A(s)\in \Delta\}\in \mathbb{T}^s]$\footnote{
In terms of arrows in a topos we have that the truth value term $(\{s | A(s)\in \Delta\}\in \mathbb{T}^s)$ with free variables $\Delta$ of type $P(\mathcal{R})$ and $\mathbb{T}^s$ of type $P(P(\Sigma))$, is defined by the function
\be
v(A\in \Delta;\mathbb{T}):=[\{s | A(s)\in \Delta\}\in \mathbb{T}^s]:P(\mathcal{R})\times P(P(\Sigma))\rightarrow \{0,1\}
\ee
such that
\ba
v(A\in \Delta;\mathbb{T}^s)(\Delta, \mathcal{T})&=&\begin{cases} 1\text{ if }\{s\in \Sigma|A(s)\in \Delta\}\in\mathbb{T}\\
0\text{ otherwise}
\end{cases}\\
&=&\begin{cases} 1\text{ if }A^{-1}(\Delta)\in\mathbb{T}\\
0\text{ otherwise}
\end{cases}
\ea}.

\subsection{ Truth Object in Quantum Theory}
We now want to define the state dependent truth object $\mathbb{T}^{|\psi\rangle}$ for quantum theory. We recall that the topos we utilise to express quantum theory is $\Sets^{\mv(\mh)^{op}}$, thus we need to define an object $\ps{\mathbb{T}}^{|\psi\rangle}$ of type $P(P(\us))$. However, since propositions are represented by clopen subobjects we actually need to restrict our attention to an element of type $P(P_{cl}(\us))$. 
Thus $\ps{\mathbb{T}}^{|\psi\rangle}$ has to be a subpresheaf of $P_{cl}(\us)$, i.e. $\ps{\mathbb{T}}^{|\psi\rangle}\subseteq P_{cl}(\us)$. Given a state $\psi$, the precise way in which this presheaf is defined is as follows:
\begin{Definition}
The presheaf $\ps{\mathbb{T}}^{|\psi\rangle}$ has as:
\begin{enumerate}
\item[i)] \emph{Objects}: For each context $V$ we get the set
\ba
\underline{\mathbb{T}}^{|\psi\rangle}_V&:=&\{\hat{\alpha}\in \uo_V|Prob(\hat{\alpha}, |\psi\rangle)=1\}\\
&=&\{\hat{\alpha}\in \uo_V|\langle\psi|\hat{\alpha}|\psi\rangle=1\}
\ea
\item [ii)] \emph{Morphisms}: Given two contexts $i:V^{'}\subseteq V$ the associated morphisms is
\ba
\ps{\mathbb{T}}^{|\psi\rangle}(i_{V^{'},V}):\ps{\mathbb{T}}^{|\psi\rangle}_V&\rightarrow& \ps{\mathbb{T}}^{|\psi\rangle}_{V^{'}}\\
\hat{\alpha}&\mapsto&\delta(\hat{\alpha})_{V^{'}}
\ea
\end{enumerate}
\end{Definition}
What about the truth object as defined for a density matrix? In that case the definition if as follows:

\begin{Definition}
Given a density matrix $\rho$, the truth-object $\ps{\mathbb{T}}^{\rho}$ associated to it is defined on
\begin{enumerate}
\item Objects: for each $V\in \mv(\mh)$ we obtain
\ba
\underline{\mathbb{T}}^{\rho}_V&:=&\{\hat{\alpha}\in \uo_V|Prob(\hat{\alpha}, \rho)=1\}\\
&=&\{\hat{\alpha}\in \uo_V|tr(\hat{\rho}\hat{\alpha})=1\}
\ea

\item Morphisms: Given two contexts $i:V^{'}\subseteq V$ the associated morphisms is
\ba
\ps{\mathbb{T}}^{\rho}(i_{V^{'},V}):\ps{\mathbb{T}}^{|\psi\rangle}_V&\rightarrow& \ps{\mathbb{T}}^{\rho}_{V^{'}}\\
\hat{\alpha}&\mapsto&\delta(\hat{\alpha})_{V^{'}}
\ea
\end{enumerate}
\end{Definition}
\subsection{Truth Values Using the Truth-Object}
Given the above definition of truth object we can deduce that the truth value of a given proposition, at a given context $V\in\mv(\mh)$, can be defined as follows:
\ba\label{ali:truth}
v(\ps{\delta(\hat{P})}\in \ps{\mathbb{T}}^{|\psi\rangle})_V&:=&\{V^{'}\subseteq V|\ps{\delta(\hat{P})}_{V^{'}}\in \ps{\mathbb{T}}^{|\psi\rangle}_{V^{'}}\}\\
&=&\{V^{'}\subseteq V|\langle\psi|\ps{\delta(\hat{P})}_{V^{'}}|\psi\rangle=1\}
\ea
and 
\ba
v(\ps{\delta(\hat{P})}\in \ps{\mathbb{T}}^{\rho})_V&:=&\{V^{'}\subseteq V|\ps{\delta(\hat{P})}_{V^{'}}\in \ps{\mathbb{T}}^{\rho}_{V^{'}}\}\\
&=&\{V^{'}\subseteq V|tr(\rho\ps{\delta(\hat{P})}_{V^{'}})=1\}
\ea
\subsection{Relation Between Pseudo-State Object and Truth Object}
We are now interested in understanding the relation between the two distinct ways in which a pure state is defined in the topos formulation of quantum theory. We expect that the two definitions turn out to be equivalent. We will first analyse the relation in classical physics then turn our attention to quantum theory. 

Thus we are trying to understand in classical physics what is the relation between $\{s\}$ and $\mathbb{T}^{s}$. Recall that $\mathbb{T}^{s}:\{X\subseteq  S|s\in X\}$. It then follows trivially that
\be
\{s\}=\bigcap\{X\in S| s\in X\}=\bigcap\{X\in \mathbb{T}^{s}\}
\ee

Following the example of classical physics we will now try to define $\ps{\w}^{\ket\psi}$ in terms of $\ps{\mathbb{T}}^{|\psi\rangle}$ and vice versa.

To this end, consider the assignment
\be
V\mapsto\bigwedge\{\hat{\alpha}\in\ps{\mathbb{T}}^{|\psi\rangle}_V\}=\bigwedge\{\hat{\alpha}\in\uo_V||\psi\rangle\langle\psi|\leq \hat{\alpha}\}
\ee
where the last equality follows from the application of the definition of $\ps{\mathbb{T}}^{|\psi\rangle}_V$. But 
\be
\ps{\w}^{\ket\psi}:V\mapsto\delta^o(|\psi\rangle\langle\psi|)_V=\bigwedge\{\hat{\alpha}\in\uo_V||\psi\rangle\langle\psi|\leq \hat{\alpha}\}
\ee
such that $\ps{\w}^{\ket\psi}_V$ is the smallest projection operator in $V$ such that $\ps{\w}^{\ket\psi}_V\geq |\psi\rangle\langle\psi|$. Therefore
\be
\ps{\w}^{\ket\psi}_V=\bigwedge\{\hat{\alpha}\in \ps{\mathbb{T}}^{|\psi\rangle}_V\}
\ee 

On the other hand, since 
\be
\ps{\mathbb{T}}^{|\psi\rangle}_V=\{\hat{\alpha}\in \uo_V|\hat{\alpha}\geq |\psi\rangle\langle\psi|\}
\ee
However, as said above,  $\ps{\w}^{\ket\psi}_V$ is the smallest projection operator in $V$ such that $\ps{\w}^{\ket\psi}_V\geq |\psi\rangle\langle\psi|$. Therefore
\be
\ps{\mathbb{T}}^{|\psi\rangle}_V=\{\hat{\alpha}\in \uo_V|\hat{\alpha}\geq \ps{\w}^{\ket\psi}_V\}
\ee
It follows that there is a one to one correspondence between $\ps{\w}^{\ket\psi}$ and $\ps{\mathbb{T}}^{|\psi\rangle}$.
\subsubsection{What About the Truth Values}
Let us reiterate the relation between $\ps{\mathbb{T}}^{|\psi\rangle}$ and $\ps{\w}^{\ket\psi}$ 
\begin{enumerate}
\item Since for each $V\in \mv(\mh)$,  $\ps{\w}^{\ket\psi}_V$ is the smallest projection operator in $V$ such that $\ps{\w}^{\ket\psi}_V\geq|\psi\rangle\langle\psi|$. 
Then if $\hat{\alpha}\in \ps{\mathbb{T}}^{|\psi\rangle}_V$ it follows that $\hat{\alpha}\geq \ps{\w}^{\ket\psi}_V$.
\item If $\hat{\alpha}\geq \ps{\w}^{\ket\psi}_V$, then from the definition of $\ps{\mathbb{T}}^{|\psi\rangle}$ it follows that $\hat{\alpha}\in \ps{\mathbb{T}}^{|\psi\rangle}_V$
\end{enumerate}
The above relations imply the following
\be
\hat{\alpha}\in \ps{\mathbb{T}}^{|\psi\rangle}_V\text{  iff  }\hat{\alpha}\geq \ps{\w}^{\ket\psi}_V
\ee

which, when applied to a daseinised proposition becomes
\be
\delta^o(\hat{P})_V\in \ps{\mathbb{T}}^{|\psi\rangle}_V\text{  iff  }\delta^o(\hat{P})_V\geq \ps{\w}^{\ket\psi}_V
\ee
Alternatively we can express this relation in terms of subsets of the state space $\us_V$ as follows:
\be
S_{\delta^o(\hat{P})_V}\in \ps{\mathbb{T}}^{|\psi\rangle}_V\text{  iff  }S_{\delta^o(\hat{P})_V}\supseteq S_{\ps{\w}^{\ket\psi}_V} 
\ee
Therefore equation \ref{ali:truth} can be written as 
\be
v(\ps{\delta(\hat{P})}\in \ps{\mathbb{T}}^{|\psi\rangle})_V=\{V^{'}\subseteq V|\delta^o(\hat{P})_V\geq \ps{\w}^{\ket\psi}_V\}
\ee
but this is precisely $v(\ps{\w}^{\ket\psi}\subseteq\ps{\delta(\hat{P})} )_V$, thus as presheaves
\be
\ps{\delta(\hat{P})}\in \ps{\mathbb{T}}^{|\psi\rangle}\text{  is equivalent to   }\ps{\w}^{\ket\psi}\subseteq\ps{\delta(\hat{P})} 
\ee
Therefore the truth values as computed with respect to the pseudo state or with respect to the truth object are exactly the same:
\be
v(\ps{\w}^{\ket\psi}\subseteq\ps{\delta(\hat{P})} )=v(\ps{\delta(\hat{P})}\in \ps{\mathbb{T}}^{|\psi\rangle})
\ee

\subsection{Example}
We will now construct the truth object for the 4 dimensional Hilbert space $\mh(\Cl^4)$ as defined with respect to the state $|\psi\rangle=(1,0,0,0)$. In our analysis we will only consider the maximal algebra $V:=lin_{\Cl}(\hat{P}_1,\hat{P}_2, \hat{P}_3, \hat{P}_4)$ and all its subalgebras. We then have:\\
\ba
\ps{\mathbb{T}}_V^{|\psi\rangle}&=&\{\hat{\alpha}\in \uo_V|\langle\psi|\hat{\alpha}|\psi\rangle=1\}\\
&=&\{\hat{P}_1, \hat{P}_1+\hat{P}_{j\in\{2,3,4\}},  \hat{P}_1+\hat{P}_{l\in\{2, 3\}}+\hat{P}_{k\in\{ 3,4\}},\hat{P}_1+\hat{P}_2+\hat{P}_3+\hat{P}_4\}\\
\mathbb{T}_{V_{\hat{P}_1}}^{|\psi\rangle}&=&\{\hat{P}_1,\hat{P}_1+\hat{P}_2+\hat{P}_3+\hat{P}_4 \}\\
\mathbb{T}^{|\psi\rangle}_{V_{\hat{P}_2}}&=&\{\hat{P}_1+\hat{P}_3+\hat{P}_4,\hat{P}_1+\hat{P}_2+\hat{P}_3+\hat{P}_4\}\\
\mathbb{T}^{|\psi\rangle}_{V_{\hat{P}_3}}&=&\{\hat{P}_1+\hat{P}_2+\hat{P}_4,\hat{P}_1+\hat{P}_2+\hat{P}_3+\hat{P}_4\}\\
\mathbb{T}^{|\psi\rangle}_{V_{\hat{P}_4}}&=&\{\hat{P}_1+\hat{P}_2+\hat{P}_3,\hat{P}_1+\hat{P}_2+\hat{P}_3+\hat{P}_4\}\\
\mathbb{T}_{V_{\hat{P}_1,\hat{P}_2}}^{|\psi\rangle}&=&\{\hat{P}_1,\hat{P}_1+\hat{P}_2,\hat{P}_1+\hat{P}_3+\hat{P}_4,\hat{P}_1+\hat{P}_2+\hat{P}_3+\hat{P}_4\}\\
\mathbb{T}_{V_{\hat{P}_1,\hat{P}_3}}^{|\psi\rangle}&=&\{\hat{P}_1, \hat{P}_1+\hat{P}_3, \hat{P}_1+\hat{P}_2+\hat{P}_4,\hat{P}_1+\hat{P}_2+\hat{P}_3+\hat{P}_4\}\\
\mathbb{T}_{V_{\hat{P}_1,\hat{P}_4}}^{|\psi\rangle}&=&\{\hat{P}_1, \hat{P}_1+\hat{P}_4,\hat{P}_1+\hat{P}_2+\hat{P}_3,\hat{P}_1+\hat{P}_2+\hat{P}_3+\hat{P}_4\}\\
\mathbb{T}_{V_{\hat{P}_2,\hat{P}_3}}^{|\psi\rangle}&=&\{ \hat{P}_1+\hat{P}_4,\hat{P}_1+\hat{P}_2+\hat{P}_4,\hat{P}_1+\hat{P}_3+\hat{P}_4,\hat{P}_1+\hat{P}_2+\hat{P}_3+\hat{P}_4\}\\
\mathbb{T}_{V_{\hat{P}_2,\hat{P}_4}}^{|\psi\rangle}&=&\{\hat{P}_1+\hat{P}_3,\hat{P}_1+\hat{P}_2+\hat{P}_3,\hat{P}_1+\hat{P}_3+\hat{P}_4,\hat{P}_1+\hat{P}_2+\hat{P}_3+\hat{P}_4\}\\
\mathbb{T}_{V_{\hat{P}_3,\hat{P}_4}}^{|\psi\rangle}&=&\{\hat{P}_1+\hat{P}_2,\hat{P}_1+\hat{P}_2+\hat{P}_3,\hat{P}_1+\hat{P}_2+\hat{P}_4,\hat{P}_1+\hat{P}_2+\hat{P}_3+\hat{P}_4\}\\
\ea
The maps between the different truth objects are:
\ba
\ps{\mathbb{T}}^{|\psi\rangle}(i_{V,V_1})(\hat{P}_1)&=&\hat{P}_1\\
\ps{\mathbb{T}}^{|\psi\rangle}(i_{V,V_1})(\hat{P}_1+\hat{P}_{j\in\{2,3,4\}})&=&\hat{P}_1+\hat{P}_2+\hat{P}_3+\hat{P}_4\\
\ps{\mathbb{T}}^{|\psi\rangle}(i_{V,V_1})(\hat{P}_1+\hat{P}_{l\in\{2, 3\}}+\hat{P}_{k\in\{ 2,4\}})&=&\hat{P}_1+\hat{P}_2+\hat{P}_3+\hat{P}_4\\
\ps{\mathbb{T}}^{|\psi\rangle}(i_{V,V_1})(\hat{P}_1+\hat{P}_2+\hat{P}_{3}+\hat{P}_4)&=&\hat{P}_1+\hat{P}_2+\hat{P}_3+\hat{P}_4
\ea
The remaining maps are left as an exercise.\\
We can now define the truth values of the proposition $S_z\in[-3, -1]$, which has corresponding projector operator $\hat{P}:=\hat{E}[S_z\in[-3, -1]]=|\phi\rangle\langle\phi|$. This proposition is equivalent to the projection operator $\hat{P}_4$.\\
We then obtain the following:
\be
v(\ps{\delta(\hat{P}_4)}\in\ps{\mathbb{T}}^{|\psi\rangle})_V:=\{V_{\hat{P}_{i\in\{2,3\}}},V_{\hat{P}_{2}\hat{P}_{3}}\}
\ee
\be
v(\ps{\delta(\hat{P}_4)}\in\ps{\mathbb{T}}^{|\psi\rangle})_{V_{\hat{P}_{i\in\{2,3\}}}}:=\{V_{\hat{P}_i}\}
\ee
\be
v(\ps{\delta(\hat{P}_4)}\in\ps{\mathbb{T}}^{|\psi\rangle})_{V_{\hat{P}_{i\in\{1,4\}}}}:=\{\emptyset\}
\ee

\ba
v(\ps{\delta(\hat{P}_4)}\in\ps{\mathbb{T}}^{|\psi\rangle})_{V_{\hat{P}_1,\hat{P}_2}}&:=&\{V_{\hat{P}_2}\}\\
v(\ps{\delta(\hat{P}_4)}\in\ps{\mathbb{T}}^{|\psi\rangle})_{V_{\hat{P}_1,\hat{P}_3}}&:=&\{V_{\hat{P}_3}\}\\
v(\ps{\delta(\hat{P}_4)}\in\ps{\mathbb{T}}^{|\psi\rangle})_{V_{\hat{P}_1,\hat{P}_4}}&:=&\{\emptyset\}\\
v(\ps{\delta(\hat{P}_4)}\in\ps{\mathbb{T}}^{|\psi\rangle})_{V_{\hat{P}_2,\hat{P}_3}}&:=&\{V_{\hat{P}_2},V_{\hat{P}_3},V_{\hat{P}_2,\hat{P}_3}\}\\
v(\ps{\delta(\hat{P}_4)}\in\ps{\mathbb{T}}^{|\psi\rangle})_{V_{\hat{P}_2,\hat{P}_4}}&:=&\{V_{\hat{P}_2}\}\\
v(\ps{\delta(\hat{P}_4)}\in\ps{\mathbb{T}}^{|\psi\rangle})_{V_{\hat{P}_3,\hat{P}_4}}&:=&\{V_{\hat{P}_3}\}\\
\ea
\subsection{Example for Density Matrix}
Let us consider again a $4$ dimensional Hilbert space $\Cl^4$ representing our 2-spin system. We would like to consider a density matrix
\be
\rho=\sum_ip_i|\psi\rangle\langle\psi|
\ee
In particular we consider a situation in which $1/2$ get $S_z=2$ while $1/2$ get $S_z=-2$, thus our density matrix is
\[\hat{\rho}=1/2\begin{pmatrix} 1& 0& 0&0\\	
0&0&0 &0\\
0&0&0&0\\
0&0&0&0	
  \end{pmatrix}+
1/2\begin{pmatrix} 0& 0& 0&0\\	
0&0&0 &0\\
0&0&0&0\\
0&0&0&1
  \end{pmatrix}
  \]
We now consider the context $V=lin_{\Cl}(\hat{P}_1, \hat{P}_2, \hat{P}_3, \hat{P}_4)$ and compute
\be
\ps{\mathbb{T}}^{\rho})_{V}=\{\hat{\alpha}\in \uo_V|tr(\hat{\rho}\hat{\alpha})=1\}
\ee
By considering all possible operators in $V$ we obtain
\be
\ps{\mathbb{T}}^{\rho})_{V}=\{(\hat{P}_1+\hat{P}_2), (\hat{P}_1+\hat{P}_2+\hat{P}_4),  (\hat{P}_1+\hat{P}_3+\hat{P}_4),  (\hat{P}_1+\hat{P}_2+\hat{P}_3+\hat{P}_4)\}
\ee
For context $V_{\hat{P}_1,\hat{P}_2}=lin_{\Cl}(\hat{P}_1, \hat{P}_2, \hat{P}_3+ \hat{P}_4)$ we obtain
\be
\ps{\mathbb{T}}^{\rho})_{V_{\hat{P}_1,\hat{P}_2}}=\{(\hat{P}_1+\hat{P}_3+\hat{P}_4),  (\hat{P}_1+\hat{P}_2+\hat{P}_3+\hat{P}_4)\}
\ee
For context $V_{\hat{P}_2,\hat{P}_3}=lin_{\Cl}(\hat{P}_2, \hat{P}_3, \hat{P}_1+ \hat{P}_4)$ we obtain
\be
\ps{\mathbb{T}}^{\rho}_{V_{\hat{P}_2,\hat{P}_3}}=\{(\hat{P}_1+\hat{P}_2), (\hat{P}_1+\hat{P}_2+\hat{P}_4),  (\hat{P}_1+\hat{P}_3+\hat{P}_4),  (\hat{P}_1+\hat{P}_2+\hat{P}_3+\hat{P}_4)\}
\ee
For context $V_{\hat{P}_3,\hat{P}_4}=lin_{\Cl}(\hat{P}_3, \hat{P}_4, \hat{P}_1+ \hat{P}_2)$ we obtain
\be
\ps{\mathbb{T}}^{\rho}_{V_{\hat{P}_3,\hat{P}_4}}=\{(\hat{P}_1+\hat{P}_2+\hat{P}_4), (\hat{P}_1+\hat{P}_2+\hat{P}_3+\hat{P}_4)\}
\ee
For context $V_{\hat{P}_1}=lin_{\Cl}(\hat{P}_1, \hat{P}_2+ \hat{P}_3+\hat{P}_4)$ we obtain
\be
\ps{\mathbb{T}}^{\rho}_{V_{\hat{P}_1}}=\{(\hat{P}_1+\hat{P}_2+\hat{P}_3+\hat{P}_4)\}
\ee
For context $V_{\hat{P}_2}=lin_{\Cl}(\hat{P}_2, \hat{P}_1+ \hat{P}_3+\hat{P}_4)$ we obtain
\be
\ps{\mathbb{T}}^{\rho}_{V_{\hat{P}_2}}=\{((\hat{P}_1+\hat{P}_3+\hat{P}_4), (\hat{P}_1+\hat{P}_2+\hat{P}_3+\hat{P}_4)\}
\ee
Morphisms:
\ba
\ps{\mathbb{T}}^{\rho}_V&\rightarrow&\ps{\mathbb{T}}^{\rho}_{V_{\hat{P}_1,\hat{P}_2}}\\
(\hat{P}_1+\hat{P}_2)&\mapsto&(\hat{P}_1+\hat{P}_3+\hat{P}_4)\\
(\hat{P}_1+\hat{P}_2+\hat{P}_4)&\mapsto&(\hat{P}_1+\hat{P}_2+\hat{P}_3+\hat{P}_4)\\
(\hat{P}_1+\hat{P}_3+\hat{P}_4) &\mapsto&(\hat{P}_1+\hat{P}_2+\hat{P}_3+\hat{P}_4)\\
(\hat{P}_1+\hat{P}_2+\hat{P}_3+\hat{P}_4)&\mapsto&(\hat{P}_1+\hat{P}_2+\hat{P}_3+\hat{P}_4)
\ea
\ba
\ps{\mathbb{T}}^{\rho}_V&\rightarrow&\ps{\mathbb{T}}^{\rho}_{V_{\hat{P}_3,\hat{P}_4}}\\
(\hat{P}_1+\hat{P}_2)&\mapsto&(\hat{P}_1+\hat{P}_2+\hat{P}_4)\\
(\hat{P}_1+\hat{P}_2+\hat{P}_4)&\mapsto&(\hat{P}_1+\hat{P}_2+\hat{P}_4)\\
(\hat{P}_1+\hat{P}_3+\hat{P}_4) &\mapsto&(\hat{P}_1+\hat{P}_2+\hat{P}_3+\hat{P}_4)\\
(\hat{P}_1+\hat{P}_2+\hat{P}_3+\hat{P}_4)&\mapsto&(\hat{P}_1+\hat{P}_2+\hat{P}_3+\hat{P}_4)
\ea
Let us now consider the proposition $S_z\in [1.3, 2.3]$ represented by the projection operator $\hat{P}_1$ we now want to compute

\ba
v(\ps{\delta(\hat{P})}\in \ps{\mathbb{T}}^{\rho})_V&:=&\{V^{'}\subseteq V|\ps{\delta(\hat{P})}_{V^{'}}\in \ps{\mathbb{T}}^{\rho}_{V^{'}}\}\\
&=&\{V^{'}\subseteq V|tr(\rho\ps{\delta(\hat{P})}_{V^{'}})=1\}
\ea
We consider context $V$, here the daseinised proposition is simply itself: $\delta^o(\hat{P}_1)=\hat{P}_1$ but $\hat{P}_1\notin\ps{\mathbb{T}}^{\rho}_{V}$ ( $tr(\hat{\rho}\hat{P}_1)\neq 1$) thus we need to go to a smaller algebra, i.e. we need to generalise our proposition. 

Lets consider $V_{\hat{P}_3,\hat{P}_4}$ in this context we get $\delta^o(\hat{P}_1)=\hat{P}_1+\hat{P}_2$ but again $\hat{P}_1+\hat{P}_2\notin\ps{\mathbb{T}}^{\rho}_{V_{\hat{P}_3,\hat{P}_4}}$ ($tr(\hat{\rho}(\hat{P}_1+\hat{P}_2)\neq 1$).

On the other hand for $V_{\hat{P}_2,\hat{P}_3}$ we obtain $\delta^o(\hat{P}_1)=\hat{P}_1+\hat{P}_4$. In this case $\hat{P}_1+\hat{P}_4\in \ps{\mathbb{T}}^{\rho}_{V_{\hat{P}_2,\hat{P}_3}}$ ($tr(\hat{\rho}(\hat{P}_1+\hat{P}_4)=1$).

In fact we get 
\be
v(\ps{\delta(\hat{P})}\in \ps{\mathbb{T}}^{\rho})_V=\{V_{\hat{P}_2,\hat{P}_3}, V_{\hat{P}_2}, V_{\hat{P}_3}\}
\ee

\chapter{Lecture 13}

I thins lecture I will describe the topos analogue of the real numbers. I will then introduce a new presheaf called inner presheaf which is related to the process of inner daseinisation. Such presheaves will be used to define physical quantities in topos quantum theory

\section{Topos Representation of the Type Symbol R}
In the topos $\Sets^{\mv(\mh)}$ the representation of the quantity value object $\mathcal{R}$ is given by the following presheaf:
\begin{Definition}
The presheaf $\ps{\Rl^{\leftrightarrow}}$ has as
\begin{enumerate}
\item [i)] Objects\footnote{A map $\mu:\downarrow V\rightarrow\Rl$ is said to be order reversing if $V^{'}\subseteq V$ implies that $\mu(V^{'})\leq\mu(V)$. A map $\nu:\downarrow V\rightarrow\Rl$ is order reversing if $V^{'}\subseteq V$ implies that $\nu(V^{'})\supseteq \nu(V)$.}: 
\be
\ps{\Rl^{\leftrightarrow}}_V:=\{(\mu,\nu)|\mu,\nu:\downarrow V\rightarrow\Rl|\mu\text{ is order preserving },\nu\text{ is order reversing }; \mu\leq\nu\}
\ee
\item[ii)] Arrows: given two contexts $V^{'}\subseteq V$ the corresponding morphism is
\ba
\ps{\Rl^{\leftrightarrow}}_{V,V^{'}}:\ps{\Rl^{\leftrightarrow}}_V&\rightarrow&\ps{\Rl^{\leftrightarrow}}_{V^{'}}\\
(\mu,\nu)&\mapsto&(\mu_{|V^{'}},\nu_{|V^{'}})
\ea
\end{enumerate}
\end{Definition}
This presheaf is where physical quantities take their values, thus it has the same role as the reals in classical physics.\\
The reason why the quantity value object is defined in terms of order reversing and order preserving functions is because, in general, in quantum theory one can only give approximate values to the quantities. In particular, in most cases, the best approximation to the value of a physical quantity one can give is the smallest interval of possible values of that quantity. \\
\\
Let us analyse the presheaf $\ps{\Rl^{\leftrightarrow}}$ in more depth. To this end we assume that we want to define the value of a physical quantity $A$ given a state $\psi$. If $\psi$ is an eigenstate of $A$, then we would get a sharp value of the quantity $A$ say $a$. If $\psi$ is not an eigenstate, then we would get a certain range $\Delta$ of values for $A$, where $\Delta\in sp(\hat{A})$. \\
Let us assume that $\Delta=[a,b]$, then what the presheaf $\ps{\Rl^{\leftrightarrow}}$ does is to single out this extreme points $a$ and $b$, so as to give a range (unsharp) of values for the physical quantity $A$. Obviously, since we are in the topos of presheaves, we have to define each object contextually, i.e. for each context $V\in\mv(\mh)$. 
It is precisely to accommodate this fact that the pair of order reversing and order preserving functions was chosen to define the extreme values of our intervals.\\
To understand this we consider a context $V$, such that the self-adjoint operator $\hat{A}$, which represents the physical quantity $A$, does belong to $V$ and such that the range of values of $A$ at $V$ is $[a,b]$. 
If we then consider the context $V^{'}\subseteq V$, such that $\hat{A}\notin V$, we will have to approximate $\hat{A}$ so as to fit $V^{'}$. The precise way in which self-adjoint  operators are approximated will be described later on, however, such an approximation will inevitably coarse-grain $\hat{A}$, i.e. it will make it more general. \\
It follows that the range of possible values of such an approximated operator $\hat{A}_1$ will be bigger. 
Therefore the range of values of $\hat{A}_1$ at $V^{'}$ will be, $[c,d]\supseteq [a,b]$ where $c\leq a$ and $d\geq b$. 
These relations between the extremal points can be achieved by the presheaf $\ps{\Rl^{\leftrightarrow}}$ through the order reversing and order preserving functions. 
Specifically, given that $a:=\mu(V)$, $b:=\nu(V)$ since $V^{'}\subseteq V$, it follows that $c:=\mu(V^{'})\leq\mu(V)$ ($\mu$ being order preserving) and $d:=\nu(V^{'})\geq \nu(V)$ ($\nu$ being order reversing). Moreover, the fact that, by definition, $\mu(V)\leq \nu(V)$, it implies that as one goes to smaller and smaller contexts the intervals $(\mu(V)_i,\nu(V)_i)$ keep getting bigger or stay the same.
\section{Inner Daseinisation}
We will now introduce a different kind of daseinisation called inner daseinisation. The role of such daseininsation is to approximate projection operators but from below. In fact while outer daseinisation would pick the smallest projection operator implied by the original projection operators (hence approximation from above), inner daseinisation picks the biggest projection operators which implies the original one (hence approximation form below). So how is such daseinisation defined?
\begin{Definition}
Given a projection operator $\hat{P}$, for each context $V\in \mv(\mh)$, inner daseinisation is defined as:
\be
\delta^i(\hat{P})_V=\bigvee\{\hat{\alpha}\in P(V)|\hat{\alpha}\leq \hat{P}\}
\ee

\end{Definition}
It follows that $\delta^i(\hat{P})_V$ is the best approximation in $V$ of $\hat{P}$ obtained by taking the `largest' projection operator in $V$ which implies $\hat{P}$. From the definition, given $V^{'}\subseteq V$ then 
\be
\delta^i(\delta^i(\hat{P}_{V}))_{V^{'}}=\delta^i(\hat{P})_{V^{'}}\leq \delta^i(\hat{P})_V
\ee
This implies that 
\be
\delta^i(\hat{P})_V\leq \delta^o(\hat{P})_V
\ee

Given this definition we can construct the analogue of the inner presheaf which is the analogue of the outer presheaf as follows
\begin{Definition}
The inner presheaf $\underline{I}$ is defined over the category $\mv(\mh)$ as
follows:
\begin{enumerate}
\item  Objects: for each $V\in \mv(\mh)$ we obtain $\underline{I}_V := P(V )$.
\item Morphisms: for each $ i_{V^{'}V} : V^{'}\subseteq V$ the corresponding presheaf map is 
\ba
\underline{I}(i_{V^{'}V }) : \underline{I}_V&\rightarrow& \underline{I}_{V^{'}}\\
\hat{\alpha}&\mapsto& \underline{I}(i_{V^{'}V })(\hat{\alpha}) := \delta^i(\hat{\alpha})_{V^{'}}
\ea
for all $\hat{\alpha}\in P(V)$.
\end{enumerate}

\end{Definition}

Similarly as was the case for the outer presheaf, the assignment 
\be
\hat{P}\mapsto\{\delta^i(\hat{P})_V|V\in \mv(\mh)\}
\ee
defines a global element of $\underline{I}$. Thus we can write inner daseinisation as
\ba
\delta^i:P(\mh)&\rightarrow &\Gamma(\underline{I})\\
\hat{P}&\mapsto&\{\delta^i(\hat{P})_V|V\in \mv(\mh)\}
\ea

Equivalently we can define inner daseinisation as a mapping from projection operators to sub-objects of the spectral presheaf 
\ba
\delta^i:P(\mh)&\rightarrow &Sub_{cl}(\us)\\
\hat{P}&\mapsto&\{\underline{T}_{\delta^i(\hat{P})_V}|V\in \mv(\mh)\}
\ea
Where, for each $V\in \mv(\mh)$, $\underline{T}_{\delta^i(\hat{P})_V}:=\{\lambda\in \us_V|\lambda(\delta^i(\hat{P})_V)=1\}$. The collection of all these clopen subsets forms a clopen sub-objects $\ps{\delta^i(\hat{P})}\subseteq \us$

It is interesting to see how the inner daseinisation helps in the definition of the negation operation. In particular we would like to understand the presheaf
\be
\ps{\delta(\neg \hat{P})} \text{  where  } \neg\hat{P}=\hat{1}-\hat{P}\text{ in the lattice }P(\mh)
\ee

In order to define $\ps{\delta(\neg \hat{P})}$ one makes use of both inner and outer daseinisation obtaining for each context $V\in \mv(\mh)$
\be
\uo(i_{V^{'}V })(\neg\hat{P}) = \neg\underline{I}(i_{V^{'}V} )(\hat{\alpha})
\ee
\begin{proof}
We want to show that 
\be\label{equ:neg}
\uo(i_{V^{'}V })(\neg\hat{P}) = \neg\underline{I}(i_{V^{'}V} )(\hat{P})
\ee
where $\hat{P}\in V$.

Applying the definition to the left hand side we obtain
\be
\uo(i_{V^{'}V })(\neg\hat{P})=\delta^o(\neg\hat{P})_V^{'}=\delta^o(\hat{1}-\hat{P})_V^{'}=\bigwedge\{\hat{\beta}\in P(V^{'})|\hat{\beta}\geq \hat{1}-\hat{P}\}
\ee
Let us assume that the right hand side of the above equation is the projection operator $\hat{\beta}$. Then such projection operators is the smallest projection operator in $V^{'}$ which is bigger than $\hat{P}$, i.e. $\hat{\beta}\geq\hat{P}$. Therefore $\hat{1}-\hat{\beta}=\hat{\beta}^c$ is the biggest projection operator in $P(V^{'})$ such that $\hat{1}-\hat{\beta}=\hat{\beta}^c\leq \hat{P}$. Thus $\hat{\beta}^c=\delta^i(\hat{P})_{V^{'}}$.

If we now consider the right hand side of \ref{equ:neg} we obtain
\be
\neg\underline{I}(i_{V^{'}V} )(\hat{P})=\hat{1}-\delta^i(\hat{P})_{V^{'}}=\hat{1}-\hat{\beta}^c=\hat{1}-\hat{1}+\hat{\beta}=\hat{\beta}
\ee
Hence \ref{equ:neg} is proved.
\end{proof}

It follows that for each $\hat{P}\in P(\mh)$ and each context $V\in \mv(\mh)$ we obtain
\be
\delta^o(\neg\hat{P})_V = \hat{1}-\delta^i(\hat{P})_V 
\ee

\section{Topos Representation of Physical Quantities}
We will now define the topos analogue of a physical quantity. Since we are trying to render quantum theory more realist we will mimic, in the context of the topos $\Sets^{\mv(\mh)^{op}}$, the way in which physical quantities are defined in classical theory. To this end we recall that in classical theory, physical quantities are represented by functions from the state space to the reals, i.e.  each physical quantity, $A$, is represented by a map $f_A:S\rightarrow \Rl$. Similarly we want to define, in topos quantum theory, physical quantities as a functor $A:\us\rightarrow \underline{\Rl}^{\leftrightarrow}$. Before attempting such a definition we have to do a small digression in measure theory.

\subsection{Spectral Decomposition}
In quantum theory, observables, i.e. things that we measure, are identified with self adjoint operators. Thus, the study of their eigenvalues and how these eigenvalues are measured is very important. Here is where the spectral theorem and consequently the spectral decomposition come into the picture. 

Given a self adjoint operator $\hat{A}$, the spectral theorem essentially tells us that it is possible to write $\hat{A}$ as
\be\label{equ:specdeco}
\hat{A}=\int_{\sigma(A)}\lambda d\hat{E}^{\hat{A}}_{\lambda}
\ee 
Such an expression is called the spectral decomposition of $\hat{A}$. Here $\sigma(A)\subseteq \Rl$ represents the spectrum of the operator $\hat{A}$ and $\{\hat{E}^{\hat{A}}_{\lambda}|\lambda\in\sigma(\hat{A})\}$ is the spectral family of $\hat{A}$. Such family determines the set of spectral projection operators of $\hat{A}$, namely $\hat{E}[A\in\Delta]$ as follows
\be
\hat{E}[A\in\Delta]:=\int_{\Delta}d\hat{E}^{\hat{A}}_{\lambda}
\ee
where $\Delta$ is a Borel subset of the spectrum $\sigma(\hat{A})$ of $A$. What such projection operators represent are subspaces of the Hilbert space for which the states $|\psi\rangle$ have a value of $A$ which lies in the interval $\Delta$. Therefore, if $a$ is a value in the discrete spectrum of $\hat{A}$, then
\be
\hat{E}[A\in \{a\}]
\ee
projects onto the eignestates of $\hat{A}$ with eignevalue $a$. 
 
In this setting, given a bounded Borel function $f:\Rl\rightarrow\Rl$, then the `transformed' operator $f(\hat{A})$ has spectral decomposition given by
\be
f(\hat{A})=\int_{\sigma(A)}f(\lambda) d\hat{E}^{\hat{A}}_{\lambda}
\ee

The formal theorem for the spectral decomposition for bounded operators is as follows
\begin{Theorem}
Given a bounded self adjoint operator $\hat{A}$ on $\mh$ there exists a family of projection operators $\{\hat{E}_{\lambda}|\lambda\in \Rl\}$ called the spectral family of $A$ such that the following conditions are satisfied:
\begin{enumerate}
\item $\hat{E}_{\lambda}\leq\hat{E}_{\lambda^{'}}$ for $\lambda\leq \lambda^{'}$
\item The net $\lambda\rightarrow \hat{E}_{\lambda}$ of projection operators in the lattice $P(\mh)$ is bounded
above by $\hat{1}$, and below by $\hat{0}$, i.e. 
\ba
\lim_{\lambda\rightarrow \infty}\hat{E}_{\lambda}&=&\hat{1}\\
\lim_{\lambda\rightarrow -\infty}\hat{E}_{\lambda}&=&\hat{0}
\ea
\item $\hat{E}_{\lambda+0}=\hat{E}_{\lambda}$
\item $\hat{A}=\int_{\sigma(A)\subseteq \Rl}\lambda d\hat{E}^{\hat{A}}_{\lambda}$
\item The map $\lambda\rightarrow \hat{E}_{\lambda}$ is right-continuous\footnote{Could equivalently require left continuity.}:
\be
\bigwedge_{\epsilon\downarrow 0}\hat{E}_{\lambda+\epsilon}=\hat{E}_{\lambda}
\ee
for all $\lambda\in \Rl$.
\end{enumerate}
\end{Theorem}
Given the spectral decomposition it is possible to define a different type of ordering\footnote{Recall that the standard operator ordering is given as follows: $\hat{A}\leq\hat{B}$ iff $\langle\psi|\hat{A}|\psi\rangle\leq\langle\psi|\hat{B}|\psi\rangle$ for all $\langle\psi|-|\psi\rangle:N_{sa}\rightarrow\Cl$, where $N_{sa}$ represents the collection of self adjoint operators on the Hilbert space $\mh$.} on operators called \emph{spectral ordering}
The reason why this order was chosen rather than the standard operator ordering is because the former preserves the relation between the spectrums of the operator, i.e. if $\hat{A}\leq_s\hat{B}$, then $sp(\hat{A})\subseteq sp(\hat{B})$. This feature will reveal itself very important when defining the values for physical quantitates. \\
We will now define what the spectral order is. Consider two self adjoint operators $\hat{A}$ and $\hat{B}$ with spectral families $(\hat{E}^{\hat{A}}_r)_{r\in\Rl}$ and $(\hat{E}^{\hat{B}}_r)_{r\in\Rl}$, respectively. The spectral order is then defined as follows:
\be
\hat{A}\leq_s\hat{B}\;\;\text{ iff }\;\;\forall r\in\Rl\;\;\;\hat{E}^{\hat{A}}_r\geq\hat{E}^{\hat{B}}_r
\ee
From the definition it follows that the spectral order implies the usual order between operators, i.e. if  $\hat{A}\leq_s\hat{B}$ then $\hat{A}\leq\hat{B}$, but the converse is not true.Thus the spectral order is a
partial order on $B(\mh)_{sa}$ (the
self-adjoint operators in $B(\mh)$) that is coarser than the usual one.

It is easy to see that the spectral ordering defines a genuine partial ordering on $B(\mh)_{sa}$. In fact, each bounded set $S$ of self-adjoint operators has a minimum $\bigwedge S\in B(\mh)_{sa}$ and a maximum  $\bigvee S\in B(\mh)_{sa}$ with respect to the spectral order, i.e. $B(\mh)_{sa}$ is a Ôboundedly completeÕ
lattice with respect to the spectral order.

If we defined the spectral oder for projection operators then we would obtain exactly the usual partial ordering. In fact if we consider two projection operators $\hat{P}$ and $\hat{Q}$ then their spectral decomposition is
\be
\hat{E}^{\hat{P}}_{\lambda}=\begin{cases}\hat{0}\text{ if }\lambda< 0\\
\hat{1}-\hat{P}\text{ if } 0\leq \lambda <1\\
\hat{P}\text{ if }1\leq\lambda
\end{cases}
\ee
It follows that 
\be
\hat{P}\leq_s\hat{Q}\text{ iff }\hat{P}\leq \hat{Q}
\ee

Thus the spectral order coincides with the usual partial order on $P(\mh)$. 

Moreover, if $\hat{A}$ and $\hat{B}$ are self-adjoint operators such that (i) either $\hat{A}$ or $\hat{B}$ is a projection,
or (ii) $[ \hat{A}, \hat{B}]=0$, then $\hat{A}\leq _s\hat{B}$ iff $\hat{A}\leq \hat{B}$.

\subsection {Daseinisation of Self Adjoint Operators}
Given the discussion above regarding the spectral order we are now ready to define the concept of both inner and outer daseinisation of self adjoint operators. 
To this end let us consider a self adjoint operator $\hat{A}$ and a context $V$, such that $\hat{A}\notin V_{sa}$ ($V_{sa}$ denotes the collection of self adjoint operators in $V$). We then need to approximate $\hat{A}$, so as to be in $V$. However, since we eventually want to define an interval of possible values of $\hat{A}$ at $V$, we will approximate $\hat{A}$ both from above and from below. In particular, we will consider the pair of operators 
\ba\label{ali:order}
\delta^o(\hat{A})_V&:=&\bigwedge\{\hat{B}\in V_{sa}|\hat{A}\leq_s\hat{B}\}\\
\delta^i(\hat{A})_V&:=&\bigvee\{\hat{B}\in V_{sa}|\hat{A}\geq_s\hat{B}\}\nonumber
\ea
In the above equation $\delta^o(\hat{A})_V$ represents the smallest self adjoint operator in $V$, which is spectrally larger or equal to $\hat{A}$, while $\delta^i(\hat{A})_V$ represents the biggest self adjoint operator in $V_{sa}$, that is spectrally smaller or equal to $\hat{A}$. The process represented by $\delta^i$ is what we defined as \emph{inner dasainisation}, while $\delta^o$ represents the \emph{outer daseinisation}.\\
From the definition of $\delta^i(\hat{A})_V$ it follows that if $V^{'}\subseteq V$ then $\delta^i(\hat{A})_{V^{'}}\leq_s\delta^i(\hat{A})_V$. Moreover, from \ref{ali:order} it follows that:
\be
sp(\hat{A})\subseteq sp(\delta^i(\hat{A})_V),\;\;\;\;\;sp(\delta^o(\hat{A})_V)\subseteq sp(\hat{A})
\ee
which, as mentioned above,  is precisely the reason why the spectral order was chosen.\\
GIven the spectral decomposition of $\hat{A}$: $\hat{A}=\int_{\Rl}\lambda d(\hat{E}^{\hat{A}}_{\lambda}$, if we apply the definition of spectral order to inner and outer daseinisation we obtain:
\be
\delta^i(\hat{A})_{V}\leq_s\delta^o(\hat{A})_V\;\;\text{ iff }\forall r\in\Rl\;\;\;\hat{E}^{\delta^i(\hat{A})_{V}}_r\geq\hat{E}^{\delta^o(\hat{A})_{V}}_r
\ee
Since $\delta^i(\hat{A})_{V}\leq\delta^o(\hat{A})_V$, it follows that for all $r\in\Rl$
\ba
\hat{E}^{\delta^i(\hat{A})_{V}}_r&:=&\delta^o(\hat{E}^{\hat{A}}_r)_V\\
\hat{E}^{\delta^o(\hat{A})_{V}}_r&:=&\delta^i(\hat{E}^{\hat{A}}_r)_V
\ea
The spectral family described by the second equation is right-continuous, while the first is not. To overcome this problem we define the following:
\be
\hat{E}^{\delta^i(\hat{A})_{V}}_r:=\bigwedge_{s> r}\delta^o(\hat{E}^{\hat{A}}_s)_V
\ee

What this amounts to is that given a spectral family $\lambda\rightarrow\hat{E}_{\lambda}$ then it is possible to construct, for each $V\in \mv(\mh)$ other two spectral families
\ba
\lambda&\rightarrow &\bigwedge_{s> r}\delta^o(\hat{E}^{\hat{A}}_s)\\
\lambda&\rightarrow&\delta^i(\hat{E}_{\lambda})\\
\ea

These will be precisely the spectral families we will utilise to define the inner and outer daseininsation of self-adjoint operators. In particular, we can now write inner and outer diaseinisation for self adjoint operators as follows:
\ba
\delta^o(\hat{A})_V&:=&\int_{\Rl}\lambda d\Big(\delta^i(\hat{E}_{\lambda}^{\hat{A}})\Big)\\
\delta^i(\hat{A})_V&:=&\int_{\Rl}\lambda d\Big(\bigwedge_{\mu>\lambda}\delta^o(\hat{E}_{\mu}^{\hat{A}})\Big)
\ea
Note that in these contexts the above integrals should be interpreted as Riemann Stieltjes integrals, which explains the condition of right-continuity: \be
\hat{A}=\int_{\Rl}\lambda d\hat{P}_{\lambda}\text{    means   } \langle\psi|\hat{A}|\psi\rangle=\int_{\Rl}\lambda \langle\psi|d\hat{P}_{\lambda}|\psi\rangle
\ee

We can now define the analogues of the inner and outer presheaves but for self adjoint operators rather than projection operators. 
\begin{Definition}
The outer de Groote presheaf $\underline{\mathbb{O}}:\mv(\mh)\rightarrow Sets$ is defined on:
\begin{enumerate}
\item Objects: for each $V\in \mv(\mh)$ we define $\od_V:=V_{sa}$ (the collection of self adjoin operators in $V$).
\item Morphisms: given a map $i:V^{'}\subseteq V$ the corresponding prehseaf map is $\od(i_{V^{'}V}):\od_V\rightarrow\od_{V^{'}}$ which is defined as follows
\ba
\od(i_{V^{'}V})(\hat{A})&:=&\delta^o(\hat{A})_{V^{'}}\\
&=&\int_{\Rl}\lambda d\big(\delta^i(\hat{E}_{\lambda}^{\hat{A}})_{V^{'}}\big)\\
&=&\int_{\Rl}\lambda d\big(\underline{I}(i_{V^{'}V})(\hat{E}^{\hat{A}}_{\lambda})\big)
\ea
for all $\hat{A}\in \od_V$. 
\end{enumerate}
\end{Definition}
In the above definition,  $\underline{I}$ is the inner presheaf for projeciton operators defined in previous sections. On the other hand the inner de Groote presheaf is:
\begin{Definition}
The inner de Groote presheaf $\ui:\mv(\mh)\rightarrow Sets$ is defined as follows:
\begin{enumerate}
\item Objects: for each $V\in \mv(\mh)$ we obtain $\ui_V:=V_{sa}$.
\item Morphisms: given a map $i:V^{'}\subseteq V$ then the corresponding presheaf map is $\ui(i_{V^{'}V}):\ui_V\rightarrow \ui_{V^{'}}$ such that
\ba
\ui(i_{V^{'}V})(\hat{A})&:=&\delta^i(\hat{A})_{V^{'}}\\
&=&\int_{\Rl}\lambda d\big(\bigwedge_{\mu>\lambda}(\delta^o(\hat{E}^{\hat{A}}_{\mu})_{V^{'}})\big)\\
&=&\int_{\Rl}\lambda d\big(\bigwedge_{\mu>\lambda}(\uo(i_{V^{'}V})(\hat{E}^{\hat{A}}_{\mu})\big)
\ea 
For all $\hat{A}\in \ui_V$ (here $\uo$ is the outer presheaf for projection operators).

\end{enumerate}

\end{Definition}
As can be expected, the inner and outer daseinisation map give rise to global elements of the inner and outer de Groote presheaves:
\ba
\delta^o(\hat{A}):V&\mapsto&\delta^o(\hat{A})_V\\
\delta^i(\hat{A}):V&\mapsto &\delta^i(\hat{A})_V
\ea
We then reach the following theorem
\begin{Theorem}
The maps
\ba
\delta^i:B(\mh)_{sa}&\rightarrow& \Gamma\ui\\
\hat{A}&\mapsto&\delta^i(\hat{A})
\ea
and 
\ba
\delta^o:B(\mh)_{sa}&\rightarrow& \Gamma\od\\
\hat{A}&\mapsto&\delta^o(\hat{A})
\ea
are injective.
\end{Theorem}
\begin{proof}
Given the definition of inner daseinisation we know that $\hat{A}\geq_s\delta^i(\hat{A})_V$ for each $V\in \mv(\mh)$. Therefore, since there exists a $V$ such that $\hat{A}\in V_{sa}$ it follows that
\be
\hat{A}=\bigvee _{V\in \mv(\mh)}\delta^i(\hat{A})_V
\ee
therefore, if $\delta^i(\hat{A})=\delta^i(\hat{B})$ it follows that
\be
\hat{A}=\bigvee _{V\in \mv(\mh)}\delta^i(\hat{A})_V=\bigvee _{V\in \mv(\mh)}\delta^i(\hat{B})_V=
\hat{B}
\ee
Similarly for outer daseinisation, if $\delta^o(\hat{A})=\delta^o(\hat{B})$ then 
\be
\hat{A}=\bigwedge _{V\in \mv(\mh)}\delta^i(\hat{A})_V=
\hat{B}=\bigwedge _{V\in \mv(\mh)}\delta^o(\hat{B})_V
\ee

\end{proof}
\subsection{Physical Quantities}
Given the definition of inner and outer daseinisation outlined in the previous section we can now define how physical quantities are represented in the topos framework of quantum theory.

A physical quantity $A$ with associated self adjoint operator $\hat{A}$ is represented by the map
\be
\breve{\delta}(\hat{A}):\us\rightarrow \ps{\Rl^{\leftrightarrow}}
\ee
which, at each context $V$, is defined as 
\ba
\breve{\delta}(\hat{A})_V:\us_V&\rightarrow& \ps{\Rl^{\leftrightarrow}}_V\\
\lambda&\mapsto&\breve{\delta}(\hat{A})_V(\lambda):=\big(\breve{\delta}^i(\hat{A})_V(\lambda), \breve{\delta}^o(\hat{A})_V(\lambda)\big)
\ea
Where $\breve{\delta}^o(\hat{A})_V$ is the order reversing function is defined by:

\be
\breve{\delta}^o(\hat{A})_V(\lambda):\downarrow V\rightarrow sp(\hat{A})
\ee
such that 
\ba
\big(\breve{\delta}^o(\hat{A})_V(\lambda)\big)(V^{'})&:=&\overline{\delta^o(\hat{A})_{V^{'}}}(\us(i_{V^{'}V})(\lambda))\\
&=&\overline{\delta^o(\hat{A})_{V^{'}}}(\lambda_{|V^{'}})\\
&=&\langle\lambda_{|V^{'}}, \delta^o(\hat{A})_{V^{'}}\rangle\\
&=&\langle\lambda, \delta^o(\hat{A})_{V^{'}}\rangle\\
&=&\lambda(\delta^o(\hat{A})_{V^{'}})
\ea
Here we have used the Gel'fand transform $\overline{\delta^o(\hat{A})_{V}}:\us_V\rightarrow \Rl$. 

The choice of order reversing functions was determined by the fact that, for all $V^{'}\subseteq V$ since $\delta^o(\hat{A})_{V^{'}}\geq \delta^o(\hat{A})_{V}$ then 
\be
\overline{\delta^o(\hat{A})_{V^{'}}}(\lambda_{|V^{'}})=\overline{\delta^o(\hat{A})_{V^{'}}}(\us(i_{V^{'}V})(\lambda))\geq \overline{\delta^o(\hat{A})_{V}}(\lambda)
\ee

On the other hand, the order preserving function is defined by:
\be
\breve{\delta}^i(\hat{A})_V(\lambda):\downarrow V\rightarrow sp(\hat{A})
\ee
such that 
\ba
\big(\breve{\delta}^i(\hat{A})_V(\lambda)\big)(V^{'})&:=&\overline{\delta^i(\hat{A})_{V^{'}}}(\us(i_{V^{'}V})(\lambda))\\
&=&\overline{\delta^i(\hat{A})_{V^{'}}}(\lambda_{|V^{'}})\\
&=&\langle\lambda_{|V^{'}}, \delta^i(\hat{A})_{V^{'}}\rangle\\
&=&\langle\lambda, \delta^i(\hat{A})_{V^{'}}\rangle\\
&=&=\lambda(\delta^i(\hat{A})_{V^{'}})
\ea
In this case the appropriate Gel'fand transform to use is $\overline{\delta^i(\hat{A})_{V^{'}}}:\us_V\rightarrow \Rl$. Here the choice of such an order preserving function is because, for $i:V^{'}\subseteq V$, since $\delta^i(\hat{A})_{V^{'}}\leq\delta^i(\hat{A})_V$ then
\be
\overline{\delta^i(\hat{A})_{V^{'}}}(\lambda_{|V^{'}})=\overline{\delta^o(\hat{A})_{V^{'}}}(\us(i_{V^{'}V})(\lambda))\leq \overline{\delta^i(\hat{A})_{V}}(\lambda)
\ee

The definition of $\breve{\delta}(\hat{A}):\us_V\rightarrow \Rl$ represents the mathematical implementation of the idea explained in the first section. There we saw that when going to smaller contexts the coarse graining of the self adjoint operators implied/induced an equivalent coarse graining in the interval of possible values for that operator. Therefore such an interval either becomes bigger or stays the same. This enlarging of the interval is precisely what is achieved by $\breve{\delta}(\hat{A})$. In fact as we go to smaller context $V^{'}\subseteq V$ the interval of possible values of the operator $\hat{A}$ which gets picked by $\breve{\delta}(\hat{A})$ becomes bigger or stays the same:
\be
[(\delta^i(\hat{A})_{V^{'}}(\lambda), \delta^o(\hat{A})_{V^{'}}(\lambda)]\geq[(\delta^i(\hat{A})_{V}(\lambda), \delta^o(\hat{A})_{V}(\lambda)]
\ee

We now need to show that indeed the map $\breve{\delta}(\hat{A}):\us\rightarrow \ps{\Rl^{\leftrightarrow}}$ is a well defined natural transformation. Therefore, given any $V^{'}\subseteq V$ we need to show that the following diagram commutes
\[\xymatrix{ 
\us_V\ar[rr]^{\breve{\delta}(\hat{A})_V}\ar[dd]_{\us(i_{V^{'}V})}&&\ps{\Rl^{\leftrightarrow}}_V\ar[dd]^{\ps{\Rl^{\leftrightarrow}}(i_{V^{'}V})}\\
&&\\
\us_{V^{'}}\ar[rr]^{\breve{\delta}(\hat{A})_{V^{'}}}&&\ps{\Rl^{\leftrightarrow}}_{V^{'}}
}
\]
Going one way round the diagram we obtain 
\ba
\ps{\Rl^{\leftrightarrow}}(i_{V^{'}V})\big(\breve{\delta}^i(\hat{A})_V(\lambda), \breve{\delta}^o(\hat{A})_V(\lambda)\big)&=&
\big(\breve{\delta}^i(\hat{A})_V(\lambda), \breve{\delta}^o(\hat{A})_V(\lambda)\big)_{V^{'}}\\
&=&\big((\breve{\delta}^i(\hat{A})_V(\lambda))_{|V^{'}}, (\breve{\delta}^o(\hat{A})_V(\lambda))_{|V^{'}}\big)\\
&=&\big(\breve{\delta}^i(\hat{A})_{V^{'}}(\lambda_{|V^{'}}), \breve{\delta}^o(\hat{A})_{V^{'}}(\lambda_{|V^{'}})\big)
\ea
On the other hand
\ba
\big(\breve{\delta}^i(\hat{A})_V(\cdot), \breve{\delta}^o(\hat{A})_V(\cdot)\big)\us(i_{V^{'}V})(\lambda)&=&\big(\breve{\delta}^i(\hat{A})_V(\cdot), \breve{\delta}^o(\hat{A})_V(\cdot)\big)(\lambda_{|V^{'}})\\
&=&\big(\breve{\delta}^i(\hat{A})_V(\lambda_{|V^{'}}), \breve{\delta}^o(\hat{A})_V(\lambda_{|V^{'}})\big)\\
&=&\big(\breve{\delta}^i(\hat{A})_{V^{'}}(\lambda_{|V^{'}}), \breve{\delta}^o(\hat{A})_{V^{'}}(\lambda_{|V^{'}})\big)
\ea
%
Thus $\big(\breve{\delta}(\hat{A})$ is a well defined natural transformation.

\chapter{Lecture 14}
In this Lecture we will start by giving a simple example of spectral order and outer daseinisation of self adjoint operators. We will then analyse the physical interpretation of the map $\breve{\delta}(\hat{A}):\us\rightarrow\ps{\Rl^{\leftrightarrow}}$ representing physical quantities and show how it can indeed be `used' to compute the value of a given physical quantity given a state.

\section{Example of Spectral Decomposition}
Let us consider again the $4$ dimensional Hilbert space $\Cl^4$. We would like to define the spectral decomposition of the self-adjoint operator 

\[\hat{S}_z=\begin{pmatrix} 2& 0& 0&0\\	
0&0&0 &0\\
0&0&0&0\\
0&0&0&-2	
  \end{pmatrix}
  \]
For such an operator the spectral family is 
\be
\hat{E}^{\hat{S}_z}_{\lambda}=\begin{cases}\hat{0}\text{ if }\lambda< -2\\
\hat{P}_4\text{ if } -2\leq \lambda <0\\
\hat{P}_4+\hat{P}_3+\hat{P}_2\text{ if }0\leq\lambda<2\\
\hat{P}_4+\hat{P}_3+\hat{P}_2+\hat{P}_1\text{ if }2\leq\lambda
\end{cases}
\ee

If we now considered the coarsed grained $\hat{S}^2_z$ whose matrix representation would be 
\[\hat{S}^2_z=\begin{pmatrix} 4& 0& 0&0\\	
0&0&0 &0\\
0&0&0&0\\
0&0&0&4	
  \end{pmatrix}
  \]

with corresponding spectral decomposition
\be
\hat{E}^{\hat{S}^2_z}_{\lambda}=\begin{cases}\hat{0}\text{ if }\lambda< 0\\
\hat{P}_2+\hat{P}_3\text{ if } 0\leq \lambda <4\\
\hat{P}_4+\hat{P}_3+\hat{P}_2+\hat{P}_1\text{ if }4\leq\lambda
\end{cases}
\ee
If we then utilised the spectral ordering to define which one is `bigger'  we would obtain the following
\be
\begin{cases}\hat{E}^{\hat{S}^2_z}_{\lambda}=\hat{E}^{\hat{S}_z}_{\lambda}=\hat{0}\text{ for }\lambda< -2 \\
\hat{E}^{\hat{S}^2_z}_{\lambda}\leq\hat{E}^{\hat{S}_z}_{\lambda}\text{ for } -2\leq \lambda <0\\
\hat{E}^{\hat{S}^2_z}_{\lambda}\leq\hat{E}^{\hat{S}_z}_{\lambda}\text{ for }0\leq\lambda < 2\\
\hat{E}^{\hat{S}^2_z}_{\lambda}\leq\hat{E}^{\hat{S}_z}_{\lambda}\text{ for }2\leq\lambda<4\\
\hat{E}^{\hat{S}^2_z}_{\lambda}=\hat{E}^{\hat{S}_z}_{\lambda}\text{ for }4\leq\lambda
\end{cases}
\ee
It follows that 
\be
\hat{S}^2_z\geq_s\hat{S}_z
\ee

We now would like to define the outer daseinisation of $\hat{S}_z$. To this end we note that $\hat{S}_z\in V$ thus $\delta^o(\hat{S}_z)_V=\hat{S}_z$. We then go to smaller sub-algebras and consider $V_{\hat{P}_2, \hat{P}_3}=lin_{\Cl}(\hat{P}_2, \hat{P}_3, \hat{P}_1+\hat{P}_4$ and compute
\be 
\delta^o(\hat{S}_z)_{V_{\hat{P}_2, \hat{P}_3}}:=\int_{\Rl}\lambda d\Big(\delta^i(\hat{E}_{\lambda}^{\hat{S}_z})\Big)
\ee
By looking at the spectral decomposition of $\hat{S}_z$ we can immediately define
\be
\delta^i(\hat{E}^{\hat{S}_z}_{\lambda})_{V_{\hat{P}_2, \hat{P}_3}}=\begin{cases}\hat{0}\text{ if }\lambda< -2\\
\hat{0}\text{ if } -2\leq \lambda <0\\
\hat{P}_3+\hat{P}_2\text{ if }0\leq\lambda<2\\
\hat{P}_4+\hat{P}_3+\hat{P}_2+\hat{P}_1\text{ if }2\leq\lambda
\end{cases}
\ee
Therefore in matrix representation we get
\[\delta^o(\hat{S}_z)_{V_{\hat{P}_2, \hat{P}_3}}=\begin{pmatrix} 2& 0& 0&0\\	
0&0&0 &0\\
0&0&0&0\\
0&0&0&2
  \end{pmatrix}
  \]
  
Thus 
\be
 \delta^o(\hat{S}_z)_{V_{\hat{P}_2, \hat{P}_3}}\geq_s\hat{S}_z
\ee
 implies that 
 \be
 sp( \delta^o(\hat{S}_z)_{V_{\hat{P}_2, \hat{P}_3}})\subseteq sp(\hat{S}_z)
 \ee 
  As can be easily deducible from the matrices them selves.

\section{Interpreting the Map Representing Physical Quantities}
In order to really understand the map $\breve{\delta}(\hat{A}):\us\rightarrow\ps{\Rl^{\leftrightarrow}}$ and how exactly once can use it in topos quantum theory, we need to analyse how expectation values are computed. To this end consider a vector $|\psi\rangle$ in a Hilbert space $\mh$. We are interested in computing the expectation value of the self adjoint operator $\hat{A}$. This is defiend (in standard quantum theory) as follows
\be\label{equ:expect}
\langle\psi|\hat{A}|\psi\rangle=\int^{||\hat{A}||}_{||\hat{A}||}\lambda d\langle\psi|\hat{E^{\hat{A}}}_{\lambda}|\psi\rangle
\ee
We can now re-write the above expression using the map $\breve{\delta}(\hat{A})$. To do so we note that for each $V\in \mv(\mh)$ the Gel'fand spectrum $\us_V$ of that algebra will contain a special element $\lambda^{|\psi\rangle}$ defined by $\lambda^{|\psi\rangle}(\hat{A}):=\langle\psi|\hat{A}|\psi\rangle$ for all $\hat{A}\in V$. Such an element of the spectrum is characterised by the properties i) $\lambda^{|\psi\rangle}(|\psi\rangle\langle\psi|)=1$ and ii) $\lambda^{|\psi\rangle}(\hat{P})=0$ for all $\hat{P}\in P(V)$ such that $\hat{P}|\psi\rangle\langle\psi|=0$.\\ 
Given such a definition, it follows that 
\be
\lambda^{|\psi\rangle}\in \us_V\text{ iff }|\psi\rangle\langle\psi|\in P(V)
\ee
We can now re-write equation \ref{equ:expect} as follows: for each context $V$ such that $|\psi\rangle\langle\psi|\in P(V)$
\ba
\langle\psi|\hat{A}|\psi\rangle&=&\int^{\delta^o(\hat{A})_V\lambda^{|\psi\rangle}}_{\delta^i(\hat{A})_V\lambda^{|\psi\rangle}}\lambda d\langle\psi|\hat{E}^{\hat{A}}_{\lambda}|\psi\rangle\\
&=&\int^{\langle\psi|\delta^o(\hat{A})_V|\psi\rangle}_{\langle\psi|\delta^i(\hat{A})_V|\psi\rangle}\lambda d\langle\psi|\hat{E}^{\hat{A}}_{\lambda}|\psi\rangle
\ea

Therefore it is possible to interpret (in the language of canonical quantum theory) $\delta^o(\hat{A})_V\lambda^{|\psi\rangle}$ and $\delta^i(\hat{A})_V\lambda^{|\psi\rangle}$ as the largest, respectively smallest result of measurements of a physical quantity $\hat{A}$ given the state $|\psi\rangle$. Obviously if $|\psi\rangle$ is an eigenstate of $\hat{A}$ then 
\be
\langle\psi|\hat{A}|\psi\rangle=\langle\psi|\delta^o(\hat{A})_V|\psi\rangle=\langle\psi|\delta^i(\hat{A})_V|\psi\rangle
\ee
However, if it is not an eigenstate then 
\be
\langle\psi|\delta^o(\hat{A})_V|\psi\rangle\geq \langle\psi|\hat{A}|\psi\rangle\geq \langle\psi|\delta^i(\hat{A})_V|\psi\rangle
\ee

If we now go to smaller context $i:V^{'}\subseteq V$, then the properties of inner and outer daseinisations imply that $\delta^o(\hat{A})_{V^{'}}\geq \delta^o(\hat{A})_V$ while $\delta^i(\hat{A})_{V^{'}}\leq \delta^i(\hat{A})_V$. It follows that 
\ba
\delta^o(\hat{A})_V\lambda^{|\psi\rangle}&=&\langle\psi|\delta^o(\hat{A})_V|\psi\rangle\leq\delta^o(\hat{A})_{V^{'}}\lambda^{|\psi\rangle}=\langle\psi|\delta^o(\hat{A})_{V^{'}}|\psi\rangle\\
\delta^i(\hat{A})_V\lambda^{|\psi\rangle}&=&\langle\psi|\delta^i(\hat{A})_V|\psi\rangle\geq\delta^i(\hat{A})_{V^{'}}\lambda^{|\psi\rangle}=\langle\psi|\delta^i(\hat{A})_{V^{'}}|\psi\rangle
\ea
Therefore the interval between the smallest and largest possible value for $\hat{A}$ becomes bigger and bigger as we go to smaller and smaller algebras. This is because we approximate our self adjoint operator more and more. 

In this setting, for each context $V\in \mv(\mh)$, the map
\ba
\breve{\delta}(\hat{A})_V:\us_V&\rightarrow&\ps{\Rl^{\leftrightarrow}}_V\\
\lambda&\mapsto&\breve{\delta}(\hat{A})_V(\lambda)=\Big(\breve{\delta}^i(\hat{A})_V(\lambda), \breve{\delta}^o(\hat{A})_V(\lambda)\Big)
\ea
defines the interval or range of possible values of the quantity $A$ at stages $V^{'}\subseteq V$. In particular, for each $\lambda\in \us_V$ we obtain a map$\breve{\delta}(\hat{A})_V(\lambda):\downarrow V\rightarrow sp(\hat{A})\times sp(\hat{A})$, such that, for each $V\in \mv(\mh)$ it picks the range of values
\be
[\breve{\delta}^i(\hat{A})_V(\lambda),\breve{\delta}^o(\hat{A})_V(\lambda)]\cap sp(\hat{A})
\ee
As we go to smaller and smaller contexts $V^{'}\subseteq V$ the range of values becomes bigger\footnote{Recall that $\delta^o(\hat{A})_V\geq \hat{A}$ and $\delta^i(\hat{A})_V\geq \hat{A}$ implied that $
sp(\delta^i(\hat{A})_V)\subseteq sp(\hat{A}) $ and $sp(\delta^o(\hat{A})_V)\subseteq sp(\hat{A})$ respectively.}

\be
[\breve{\delta}^i(\hat{A})_{V^{'}}(\lambda),\breve{\delta}^o(\hat{A})_{V^{'}}(\lambda)]\cap sp(\hat{A})
\ee

Thus reiterating, the map $\breve{\delta}(\hat{A}):\us\rightarrow\ps{\Rl^{\leftrightarrow}}$ is defined such that as we go down to smaller sub-algebras $V^{'}\subseteq V$, the range of possible values of our physical quantity becomes bigger. This is because $\hat{A}$ gets approximated both from above through the process of outer daseininsation and from below through the process of inner daseinisation. Such an approximation gets coarser as $V^{'}$ gets smaller, which basically means that $V^{'}$ contains less and less projections, i.e. less and less information. However,  such an interpretation, can only be local since the state space $\us$ has no global elements.

\subsection{Computing Values of Quantities Given a State}

We would now like to define the value of a physical quantity $\hat{A}$ given a state $\ps{\w}^{\ket\psi}$. Again as for all the constructs we have defined so far, we would like to mimic classical physics, i.e. we would like to define values for quantities given a state, in the same way is it is done in classical physics. To this end recall that in classical physics, given a state $s$, and a physical quantity $f_A$, the value of the latter given the former is $f_A(s)$.

Similarly we would like to define the value of $\breve{\delta}(\hat{A})$ given $\ps{\w}^{\ket\psi}$ as something like $\breve{\delta}(\hat{A})(\ps{\w}^{\ket\psi})$. Is this possible?

Let us recall how the pseudo-state $\ps{\w}^{\ket\psi}$ is defined. For each context $V$ we first define the approximate operator 
\be
\delta^o(|\psi\rangle\langle\psi|)_V:=\bigwedge\{\hat{\alpha}\in \uo_V||\psi\rangle\langle\psi|\leq \hat{\alpha}\}
\ee
We then associate a subset of the state space $\us_V$:
\be
\ps{\w}^{\ket\psi}_V :=\{\lambda\in \us_V|\lambda(\delta^o(|\psi\rangle\langle\psi|)_V)=1\}
\ee
Thus for each state $|\psi\rangle$, we get the presheaf $\ps{\w}^{\ket\psi}\subseteq \us_V$. Now since the physical quantity $\breve{\delta}(\hat{A})$ is defined in terms of a map whose codomain is $\us$ ($\breve{\delta}(\hat{A}):\us\rightarrow\ps{\Rl^{\leftrightarrow}}$), such a map has to be define also on any sub-object of $\us$. Thus we obtain the composite

\be
\ps{\w}^{\ket\psi}\rightarrow \us\xrightarrow{\breve{\delta}(\hat{A})}\ps{\Rl^{\leftrightarrow}}
\ee
In this setting one can indeed write the value of a physical quantity given a state as 
\be
\breve{\delta}(\hat{A})(\ps{\w}^{\ket\psi})
\ee
Such that for each context $V\in \mv(\mh)$ we obtain
\ba\label{ali:value}
\breve{\delta}(\hat{A})_V:\ps{\w}^{\ket\psi}_V&\rightarrow& \ps{\Rl^{\leftrightarrow}}_V\\
\lambda&\mapsto&\breve{\delta}(\hat{A})_V:=\Big(\breve{\delta}^i(\hat{A})_V(\lambda), \breve{\delta}^o(\hat{A})_V(\lambda)\Big)
\ea
\begin{Theorem}
$\breve{\delta}(\hat{A})(\ps{\w}^{\ket\psi})$ is a well defined subobject of $\ps{\Rl^{\leftrightarrow}}$
\end{Theorem}

\begin{proof}
We first need to show that $\breve{\delta}(\hat{A})_V(\ps{\w}^{\ket\psi})_V\subseteq \ps{\Rl^{\leftrightarrow}}_V$ for each $V\in\mv(\mh)$. This follows trivially from the definition in \ref{ali:value}. Thus what remains to show is that for each $V^{'}\subseteq V$
\be
\ps{\Rl^{\leftrightarrow}}(i_{V^{'}V})(\breve{\delta}(\hat{A})_V(\ps{\w}^{\ket\psi}_V))\subseteq \breve{\delta}(\hat{A})_{V^{'}}(\ps{\w}^{\ket\psi})_{V^{'}}
\ee
Given $\lambda\in \ps{\w}^{\ket\psi}_V$ then 
\be
\ps{\Rl^{\leftrightarrow}}(i_{V^{'}V})(\breve{\delta}(\hat{A})_V(\lambda))=(\breve{\delta}(\hat{A})_V(\lambda))_{|V^{'}}=\breve{\delta}(\hat{A})_{V^{'}}(\lambda_{|V^{'}})
\ee
However form the definition of pseudo-state we have
\be
\ps{\w}^{\ket\psi}_{V^{'}}=\us(i_{V^{'}V})(\ps{\w}^{\ket\psi}_V)=\{\lambda_{|V^{'}}|\lambda\in \ps{\w}^{\ket\psi}_V\}
\ee

Therefore, each element in $\ps{\w}^{\ket\psi}_{V^{'}}$ comes form restricting an element in $\ps{\w}^{\ket\psi}_V$. It follows that 
\be
\ps{\Rl^{\leftrightarrow}}(i_{V^{'}V})(\breve{\delta}(\hat{A})_V(\ps{\w}^{\ket\psi}_V))= \breve{\delta}(\hat{A})_{V^{'}}(\ps{\w}^{\ket\psi})_{V^{'}}
\ee
\end{proof}
We thus obtain a commuting diagram

\[\xymatrix{ 
\ps{\w}^{\ket\psi}_V\ar[rr]^{\breve{\delta}(\hat{A})_V}\ar[dd]_{\us(i_{V^{'}V})}&&\ps{\Rl^{\leftrightarrow}}_V\ar[dd]^{\ps{\Rl^{\leftrightarrow}}(i_{V^{'}V})}\\
&&\\
\ps{\w}^{\ket\psi}_{V^{'}}\ar[rr]^{\breve{\delta}(\hat{A})_{V^{'}}}&&\ps{\Rl^{\leftrightarrow}}_{V^{'}}
}
\]

\subsection{Examples}
We will now give two simle examples to show how the map representing self-adjoint operators actually works.
\subsubsection{Eigenvalue-Eigenstate Link}
We will first consider the case in which the state $|\psi\rangle$ is an eigenstate of $\hat{A}$.  We then consider abelian sub-algebras for which $\hat{A}\in V$ such that $\delta^o(\hat{A})=\delta^i(\hat{A})=\hat{A}$. 
The condition $\hat{A}\in V$ also implies that $|\psi\rangle\langle\psi|\in P(V)$ thus  $\ps{\w}^{\ket\psi}_{V}$ will contain the single element $\lambda_{|\psi\rangle\langle\psi|}\in \us_V$ such that $\lambda_{|\psi\rangle\langle\psi|}(|\psi\rangle\langle\psi|)=1$ while $\lambda_{|\psi\rangle\langle\psi|}(\hat{Q})=0$ for all $\hat{Q}|\psi\rangle\langle\psi|=0$. It follows that
\be
\breve{\delta}(\hat{A})_V(\ps{\w}^{\ket\psi}_V)=\Big( \breve{\delta}^i(\hat{A})_V (\lambda_{|\psi\rangle\langle\psi|}), \breve{\delta}^o(\hat{A})_V (\lambda_{|\psi\rangle\langle\psi|})\Big)
\ee

Since $\hat{A}\in V$, then $ \breve{\delta}^i(\hat{A})_V (\lambda_{|\psi\rangle\langle\psi|})= \breve{\delta}^o(\hat{A})_V (\lambda_{|\psi\rangle\langle\psi|})=\lambda_{|\psi\rangle\langle\psi|}(\hat{A})$ which is the eigenvalue of $\hat{A}$ given the state $|\psi\rangle$. Thus we get the usual eignevalue eigenstate link.

\subsection{Interval}
We now give both a simple example and a more complicated example on how to define values for quantities. \\\\
\textbf{Simple Example}\\\\
Let us consider the simple self adjoint projection operator $|\psi\rangle\langle\psi|$. Such an operator has $sp(|\psi\rangle\langle\psi|)=\{0,1\}$. For each context $V\in \mv(\mh)$ the map $\breve{\delta}(|\psi\rangle\langle\psi|):\us\rightarrow\ps{\Rl^{\leftrightarrow}}$ is
\ba
\breve{\delta}(|\psi\rangle\langle\psi|)_V:\us_V&\rightarrow&\ps{\Rl^{\leftrightarrow}}_V\\
\lambda&\mapsto&\Big(\breve{\delta}^i(|\psi\rangle\langle\psi|)_V(\lambda), \breve{\delta}^o(|\psi\rangle\langle\psi|)_V(\lambda)\Big)
\ea
However, given the spectrum of  $|\psi\rangle\langle\psi|$, we obtain
\ba
\breve{\delta}^i(|\psi\rangle\langle\psi|)_V(\lambda):\downarrow V&\rightarrow&\{0,1\}\\
\breve{\delta}^o(|\psi\rangle\langle\psi|)_V(\lambda):\downarrow V&\rightarrow&\{0,1\}\\
\ea
such that for all $V^{'}\subseteq V$ we have
\ba
\breve{\delta}^i(|\psi\rangle\langle\psi|)_V(\lambda)(V^{'})&:=&\langle\lambda, \delta^i(|\psi\rangle\langle\psi|)_{V^{'}}\rangle\\
\breve{\delta}^o(|\psi\rangle\langle\psi|)_V(\lambda)(V^{'})&:=&\langle\lambda, \delta^o(|\psi\rangle\langle\psi|)_{V^{'}}\rangle
\ea
We then consider the pseudo-state $\ps{\w}^{\ket\psi}$ and want to evaluate $\breve{\delta}(|\psi\rangle\langle\psi|)( \ps{\w}^{\ket\psi})$. From the definition of pseudo-state, given a context $V$, it follows that for all $\lambda\in  \ps{\w}^{\ket\psi}_V$ then, for all $V^{'}\subseteq V$
\be
\breve{\delta}^o(|\psi\rangle\langle\psi|)_V(\lambda)(V^{'})=\langle\lambda, \delta^o(|\psi\rangle\langle\psi|)_{V^{'}}\rangle=1
\ee

We call the constant function with value $1$ on all $V^{'}\subseteq V$ as $1_{\downarrow V}$. 

On the other hand, $\breve{\delta}^i(|\psi\rangle\langle\psi|)_V(\lambda)$ is such that for all $V^{'}\subseteq V$ we obtain

\be
\breve{\delta}^i(|\psi\rangle\langle\psi|)_V(\lambda)(V^{'})=\begin{cases}1 \text{ if }|\psi\rangle\langle\psi|\in V^{'}\\
0\text{ if }|\psi\rangle\langle\psi|\notin V^{'}
\end{cases}
\ee

We then arrive at a complete description of $\breve{\delta}(|\psi\rangle\langle\psi|)( \ps{\w}^{\ket\psi})$ as follows, given any context $V\in\mv(\mh)$ if $\lambda\in \ps{\w}^{\ket\psi}_V$ then we obtain
\be
\breve{\delta}(|\psi\rangle\langle\psi|)_V(\lambda)=\big(\breve{\delta}^i(|\psi\rangle\langle\psi|)_V(\lambda), 1_{\downarrow V})
\ee
\textbf{More Complicated Example}\\\\
Let us consider the 2 spin system in $\Cl^4$ defined in previous examples.  We are interested in the spin in the $z$-direction, which is represented by the physical quantity $S_z$. 
The self-adjoint operator representing $S_z$ is 
\[\hat{S}_z=\begin{pmatrix} 2& 0& 0&0\\	
0&0&0 &0\\
0&0&0&0\\
0&0&0&-2	
  \end{pmatrix}
  \]
We want compute $\breve{\delta}(\hat{S}_z):\us\rightarrow \Rl^{\leftrightarrow}$. Since we are dealing with presheaves we need to compute this contextwise, for each $V$, i.e.
\ba
\breve{\delta}(\hat{S}_z):\us&\rightarrow& \Rl^{\leftrightarrow}\\
\lambda&\mapsto&\breve{\delta}(\hat{S}_z)(\lambda)=\Big(\breve{\delta}^i(\hat{S}_z)_V(\lambda), \breve{\delta}^o(\hat{S}_z)_V(\lambda)\Big)
\ea
Where
\ba
\breve{\delta}^i(\hat{S}_z)_V(\lambda):\downarrow V&\rightarrow& \Rl\\
V^{'}&\mapsto&\lambda(\delta^i(\hat{S}_z)_{V^{'}})
\ea
and similar for $\breve{\delta}^o(\hat{S}_z)_V(\lambda)$. Thus,  in order to compute $\breve{\delta}^i(\hat{S}_z)_V$ and $\breve{\delta}^o(\hat{S}_z)_V$ we need to find the inner and outer daseinisation of the spectral family of $\hat{S}_z$ since we want to apply the formulas
\ba
\delta^o(\hat{S}_z)_V&=&\int_{\Rl}\lambda d\Big(\delta^i(\hat{E}^{\hat{S}_z}_{\lambda}\Big)\\
\delta^i(\hat{S}_z)_V&=&\int_{\Rl}\lambda d\Big(\bigwedge_{\mu>\lambda}\delta^o(\hat{E}^{\hat{S}_z}_{\lambda}\Big)\\
\ea
We know from previous examples that the spectral family of $\hat{S}_z$ is \be
\hat{E}^{\hat{S}_z}_{\lambda}=\begin{cases}\hat{0}\text{ if }\lambda< -2\\
\hat{P}_4\text{ if } -2\leq \lambda <0\\
\hat{P}_4+\hat{P}_3+\hat{P}_2\text{ if }0\leq\lambda<2\\
\hat{P}_4+\hat{P}_3+\hat{P}_2+\hat{P}_1\text{ if }2\leq\lambda
\end{cases}
\ee
Therefore, $j=\{2,4\}$ we obtain for 
\be
\delta^i(\hat{E}^{\hat{S}_z}_{\lambda})_{V_{\hat{P}_j\hat{P}_3}}=\begin{cases}\hat{0}\text{ if }\lambda< -2\\
\hat{0}\text{ if } -2\leq \lambda <0\\
\hat{P}_j+\hat{P}_3\text{ if }0\leq\lambda<2\\
\hat{1}\text{ if }2\leq\lambda
\end{cases}
\ee
For $j=\{2,3\}$
\be
\delta^i(\hat{E}^{\hat{S}_z}_{\lambda})_{V_{\hat{P}_1\hat{P}_j}}=\begin{cases}\hat{0}\text{ if }\lambda< -2\\
\hat{0}\text{ if } -2\leq \lambda <0\\
\hat{P}_2+\hat{P}_3+\hat{P}_4\text{ if }0\leq\lambda<2\\
\hat{1}\text{ if }2\leq\lambda
\end{cases}
\ee
For $V_{\hat{P}_1,\hat{P}_4}$
\be
\delta^i(\hat{E}^{\hat{S}_z}_{\lambda})_{V_{\hat{P}_1\hat{P}_4}}=\begin{cases}\hat{0}\text{ if }\lambda< -2\\
\hat{P}_4\text{ if } -2\leq \lambda <0\\
\hat{P}_2+\hat{P}_3+\hat{P}_4\text{ if }0\leq\lambda<2\\
\hat{1}\text{ if }2\leq\lambda
\end{cases}
\ee
For $V_{\hat{P}_2,\hat{P}_4}$
\be
\delta^i(\hat{E}^{\hat{S}_z}_{\lambda})_{V_{\hat{P}_2\hat{P}_4}}=\begin{cases}\hat{0}\text{ if }\lambda< -2\\
\hat{P}_4\text{ if } -2\leq \lambda <0\\
\hat{P}_2+\hat{P}_4\text{ if }0\leq\lambda<2\\
\hat{1}\text{ if }2\leq\lambda
\end{cases}
\ee
For $j=\{2,3\}$
 \be
\delta^i(\hat{E}^{\hat{S}_z}_{\lambda})_{V_{\hat{P}_j}}=\begin{cases}\hat{0}\text{ if }\lambda< -2\\
\hat{0}\text{ if } -2\leq \lambda <0\\
\hat{P}_j\text{ if }0\leq\lambda<2\\
\hat{1}\text{ if }2\leq\lambda
\end{cases}
\ee
For $V_{\hat{P}_1}$
 \be
\delta^i(\hat{E}^{\hat{S}_z}_{\lambda})_{V_{\hat{P}_1}}=\begin{cases}\hat{0}\text{ if }\lambda< -2\\
\hat{0}\text{ if } -2\leq \lambda <0\\
\hat{P}_2+\hat{P}_3+\hat{P}_4\text{ if }0\leq\lambda<2\\
\hat{1}\text{ if }2\leq\lambda
\end{cases}
\ee
Finally for $V_{\hat{P}}$
\be
\delta^i(\hat{E}^{\hat{S}_z}_{\lambda})_{V_{\hat{P}_4}}=\begin{cases}\hat{0}\text{ if }\lambda< -2\\
\hat{P}_4\text{ if } -2\leq \lambda <0\\
\hat{P}_4\text{ if }0\leq\lambda<2\\
\hat{1}\text{ if }2\leq\lambda
\end{cases}
\ee
Similarly for outer daseinisation we obtain for $i, j=\{2,3\}$
\be
\delta^o(\hat{E}^{\hat{S}_z}_{\lambda})_{V_{\hat{P}_j\hat{P}_1}}=\begin{cases}\hat{0}\text{ if }\lambda< -2\\
\hat{P}_4+\hat{P}_i\text{ if } -2\leq \lambda <0\; i\neq j\\
\hat{P}_j+\hat{P}_4+\hat{P}_3\text{ if }0\leq\lambda<2\\
\hat{1}\text{ if }2\leq\lambda
\end{cases}
\ee
For $j=\{1,2,3\}$
\be
\delta^o(\hat{E}^{\hat{S}_z}_{\lambda})_{V_{\hat{P}_j}}=\begin{cases}\hat{0}\text{ if }\lambda< -2\\
\hat{1}-\hat{P}_j\text{ if } -2\leq \lambda <0\\
\hat{1} (-\hat{P_j}\text{ if } j=1)\text{ if }0\leq\lambda<2\\
\hat{1}\text{ if }2\leq\lambda
\end{cases}
\ee
For $j=\{2,3\}$
\be
\delta^o(\hat{E}^{\hat{S}_z}_{\lambda})_{V_{\hat{P}_j\hat{P}_4}}=\begin{cases}\hat{0}\text{ if }\lambda< -2\\
\hat{P}_4\text{ if } -2\leq \lambda <0\\
\hat{1}\text{ if }0\leq\lambda<2\\
\hat{1}\text{ if }2\leq\lambda
\end{cases}
\ee
For $V_{\hat{P}_1,\hat{P}_4}$
\be
\delta^o(\hat{E}^{\hat{S}_z}_{\lambda})_{V_{\hat{P}_4\hat{P}_1}}=\begin{cases}\hat{0}\text{ if }\lambda< -2\\
\hat{P}_4\text{ if } -2\leq \lambda <0\\
\hat{P}_2+\hat{P}_3+\hat{P}_4\text{ if }0\leq\lambda<2\\
\hat{1}\text{ if }2\leq\lambda
\end{cases}
\ee
For $V_{\hat{P}_2, \hat{P}_3}$
\be
\delta^o(\hat{E}^{\hat{S}_z}_{\lambda})_{V_{\hat{P}_2\hat{P}_3}}=\begin{cases}\hat{0}\text{ if }\lambda< -2\\
\hat{P}_4+\hat{P}_1\text{ if } -2\leq \lambda <0\\
\hat{1}\text{ if }0\leq\lambda<2\\
\hat{1}\text{ if }2\leq\lambda
\end{cases}
\ee
We can now compute the daseinisations of $\hat{S}_z$ for each contexts. These are for $V_{\hat{P}_4}$ 
\be
\delta^o(\hat{S}_z)_{V_{\hat{P}_4}}=\begin{pmatrix} 2& 0& 0&0\\	
0&2&0 &0\\
0&0&2&0\\
0&0&0&-2	
  \end{pmatrix}\text{ and }\;\;\delta^i(\hat{S}_Z)_{V_{\hat{P}_4}}=\begin{pmatrix} 0& 0& 0&0\\	
0&2&0 &0\\
0&0&2&0\\
0&0&0&-2	
  \end{pmatrix}
\ee
For $V_{\hat{P}_1}$

\be
\delta^o(\hat{S}_z)_{V_{\hat{P}_4}}=\begin{pmatrix} 2& 0& 0&0\\	
0&0&0 &0\\
0&0&0&0\\
0&0&0&0	
  \end{pmatrix}\text{ and }\;\;\delta^i(\hat{S}_Z)_{V_{\hat{P}_1}}=\begin{pmatrix} 2& 0& 0&0\\	
0&-2&0 &0\\
0&0&-2&0\\
0&0&0&-2	
  \end{pmatrix}
\ee
For $V_{\hat{P}_1}$

\be
\delta^o(\hat{S}_z)_{V_{\hat{P}_2}}=\begin{pmatrix} 2& 0& 0&0\\	
0&0&0 &0\\
0&0&2&0\\
0&0&0&2
  \end{pmatrix}\text{ and }\;\;\delta^i(\hat{S}_Z)_{V_{\hat{P}_2}}=\begin{pmatrix} 2& 0& 0&0\\	
0&0&0 &0\\
0&0&-2&0\\
0&0&0&-2	
  \end{pmatrix}
\ee
For $V_{\hat{P}_3}$

\be
\delta^o(\hat{S}_z)_{V_{\hat{P}_3}}=\begin{pmatrix} 2& 0& 0&0\\	
0&2&0 &0\\
0&0&0&0\\
0&0&0&2	
  \end{pmatrix}\text{ and }\;\;\delta^i(\hat{S}_Z)_{V_{\hat{P}_3}}=\begin{pmatrix} -2& 0& 0&0\\	
0&-2&0 &0\\
0&0&0&0\\
0&0&0&-2	
  \end{pmatrix}
\ee
For $V_{\hat{P}_1,\hat{P}_2}$

\be
\delta^o(\hat{S}_z)_{V_{\hat{P}_1,\hat{P}_2}}=\begin{pmatrix} 2& 0& 0&0\\	
0&0&0 &0\\
0&0&0&0\\
0&0&0&0	
  \end{pmatrix}\text{ and }\;\;\delta^i(\hat{S}_Z)_{V_{\hat{P}_1,\hat{P}_2}}=\begin{pmatrix} 2& 0& 0&0\\	
0&0&0 &0\\
0&0&-2&0\\
0&0&0&-2	
  \end{pmatrix}
\ee
For $V_{\hat{P}_1,\hat{P}_3}$

\be
\delta^o(\hat{S}_z)_{V_{\hat{P}_1,\hat{P}_3}}=\begin{pmatrix} 2& 0& 0&0\\	
0&0&0 &0\\
0&0&0&0\\
0&0&0&0	
  \end{pmatrix}\text{  and  }\;\;\delta^i(\hat{S}_Z)_{V_{\hat{P}_1,\hat{P}_3}}=\begin{pmatrix} 2& 0& 0&0\\	
0&-2&0 &0\\
0&0&0&0\\
0&0&0&-2	
  \end{pmatrix}
\ee
For $V_{\hat{P}_1,\hat{P}_4}$

\be
\delta^o(\hat{S}_z)_{V_{\hat{P}_1,\hat{P}_4}}=\begin{pmatrix} 2& 0& 0&0\\	
0&0&0 &0\\
0&0&0&0\\
0&0&0&-2
  \end{pmatrix}\text{  and  }\;\;\delta^i(\hat{S}_Z)_{V_{\hat{P}_1,\hat{P}_4}}=\begin{pmatrix} 2& 0& 0&0\\	
0&0&0 &0\\
0&0&0&0\\
0&0&0&-2	
  \end{pmatrix}
\ee
For $V_{\hat{P}_2,\hat{P}_3}$

\be
\delta^o(\hat{S}_z)_{V_{\hat{P}_2,\hat{P}_3}}=\begin{pmatrix} 2& 0& 0&0\\	
0&0&0 &0\\
0&0&0&0\\
0&0&0&2	
  \end{pmatrix}\text{ and }\;\;\delta^i(\hat{S}_Z)_{V_{\hat{P}_2,\hat{P}_3}}=\begin{pmatrix} -2& 0& 0&0\\	
0&0&0 &0\\
0&0&0&0\\
0&0&0&-2	
  \end{pmatrix}
\ee
For $V_{\hat{P}_3,\hat{P}_4}$

\be
\delta^o(\hat{S}_z)_{V_{\hat{P}_3,\hat{P}_4}}=\begin{pmatrix} 2& 0& 0&0\\	
0&2&0 &0\\
0&0&0&0\\
0&0&0&0	
  \end{pmatrix}\text{ and }\;\;\delta^i(\hat{S}_Z)_{V_{\hat{P}_3,\hat{P}_4}}=\begin{pmatrix} 0& 0& 0&0\\	
0&0&0 &0\\
0&0&0&0\\
0&0&0&-2	
  \end{pmatrix}
\ee
Now given a state $|\psi\rangle=(1,0,0,0)$ we want to compute the physical quantity $\breve{\delta}(\hat{S}_z)$. Thus for each context $V$ we need to compute the pair $(\breve{\delta}^i(\hat{S}_z)_V(\cdot),\breve{\delta}^o(\hat{S}_z)_V(\cdot))$ which will then act on $\lambda\in \ps{\w}^{\ket\psi}_V$ \\
For $V_{\hat{P}_4}$, $\ps{\w}^{\ket\psi}_{V_{\hat{P}_4}}=\{\lambda\}$ is such that $\lambda(\delta^o(|\psi\rangle\langle\psi|)_{V_{\hat{P}_4}})=1$ where 
\be
\delta^o(|\psi\rangle\langle\psi|)_{V_{\hat{P}_4}}=\begin{pmatrix} 1& 0& 0&0\\	
0&1&0 &0\\
0&0&1&0\\
0&0&0&0	
  \end{pmatrix}
\ee
Hence
\ba
\breve{\delta}^o(\hat{S}_z)_{V_{\hat{P}_4}}(\lambda)(V_{\hat{P}_4})&=&2\\
\breve{\delta}^i(\hat{S}_z)_{V_{\hat{P}_4}}(\lambda)(\lambda)(V_{\hat{P}_4})&=&0
\ea

Note that this is equivalent to 
\be
\langle\psi|\delta^o(\hat{S}_z)_{V_{\hat{P}_4}}|\psi\rangle=\begin{pmatrix}1&0&0&0\end{pmatrix}\cdot\begin{pmatrix} 2& 0& 0&0\\	
0&2&0 &0\\
0&0&2&0\\
0&0&0&-2	
  \end{pmatrix}\begin{pmatrix} 1\\	
0\\
0\\
0	
  \end{pmatrix}=2
\ee
and
\be
\langle\psi|\delta^i(\hat{S}_z)_{V_{\hat{P}_4}}|\psi\rangle=\begin{pmatrix}1&0&0&0\end{pmatrix}\cdot\begin{pmatrix} 0& 0& 0&0\\	
0&0&0 &0\\
0&0&0&0\\
0&0&0&-2	
  \end{pmatrix}\begin{pmatrix} 1\\	
0\\
0\\
0	
\end{pmatrix}=0
\ee

Similarly, for context $V_{\hat{P}_2,\hat{P}_4}$ we obtain
\be
\langle\psi|\delta^o(\hat{S}_z)_{V_{\hat{P}_2,\hat{P}_4}}|\psi\rangle=\begin{pmatrix}1&0&0&0\end{pmatrix}\cdot\begin{pmatrix} 2& 0& 0&0\\	
0&0&0 &0\\
0&0&2&0\\
0&0&0&-2	
\end{pmatrix}\begin{pmatrix} 1\\	
0\\
0\\
0	
  \end{pmatrix}=2
\ee
and
\be
\langle\psi|\delta^i(\hat{S}_z)_{V_{\hat{P}_2}}|\psi\rangle=\begin{pmatrix}1&0&0&0\end{pmatrix}\cdot\begin{pmatrix} 0& 0& 0&0\\	
0&0&0 &0\\
0&0&0&0\\
0&0&0&-2	
  \end{pmatrix}\cdot\begin{pmatrix} 1\\	
0\\
0\\
0	
  \end{pmatrix}=0
\ee
For $V_{\hat{P}_2}$ we obtain 
\be
\langle\psi|\delta^o(\hat{S}_z)_{V_{\hat{P}_2,\hat{P}_4}}|\psi\rangle=\begin{pmatrix}1&0&0&0\end{pmatrix}\cdot\begin{pmatrix} 2& 0& 0&0\\	
0&0&0 &0\\
0&0&2&0\\
0&0&0&2	
  \end{pmatrix}\begin{pmatrix} 1\\	
0\\
0\\
0	
  \end{pmatrix}=2
\ee
and
\be
\langle\psi|\delta^i(\hat{S}_z)_{V_{\hat{P}_2}}|\psi\rangle=\begin{pmatrix}1&0&0&0\end{pmatrix}\cdot\begin{pmatrix} -2& 0& 0&0\\	
0&0&0 &0\\
0&0&-2&0\\
0&0&0&-2	
  \end{pmatrix}\cdot\begin{pmatrix} 1\\	
0\\
0\\
0	
  \end{pmatrix}=-2
\ee
Given the above results we have that for $V_{\hat{P}_2,\hat{P}_4}$ then the pair of order preserving, order reversing functions for the physical quantity $\breve{\delta}(\hat{S}_z)$ is $\big(\breve{\delta}^i(\hat{S}_z)_{V_{\hat{P}_2,\hat{P}_4}}(\lambda), \breve{\delta}^o(\hat{S}_z)_{V_{\hat{P}_2,\hat{P}_4}}(\lambda)\big)$ where
\ba
\breve{\delta}^i(\hat{S}_z)_{V_{\hat{P}_2,\hat{P}_4}}(\lambda):\downarrow V_{\hat{P}_2,\hat{P}_4}&\rightarrow&sp(\hat{S}_z)\\
\breve{\delta}^o(\hat{S}_z)_{V_{\hat{P}_2,\hat{P}_4}}(\lambda):\downarrow V_{\hat{P}_2,\hat{P}_4}&\rightarrow&sp(\hat{S}_z)\\
\ea
Since the latter as we have seen as constant value 2 we can write 
\be
\big(\breve{\delta}^i(\hat{S}_z)_{V_{\hat{P}_2,\hat{P}_4}}(\lambda), \breve{\delta}^o(\hat{S}_z)_{V_{\hat{P}_2,\hat{P}_4}}(\lambda)\big)=\big(\breve{\delta}^i(\hat{S}_z)_{V_{\hat{P}_2,\hat{P}_4}}(\lambda), 2_{\downarrow V_{\hat{P}_2, \hat{P}_4}}\big)
\ee
Similar analysis can be performed for all the remaining contexts.
\chapter{Lecture 14}
In this lecture we will describe the concept of a sheaf and its relation to presheaves and certain bundles called etal\'e bundles. We will also introduce the very important concept of adjoint functors which will then be utilised to define geometric morphisms which are maps between topoi. The reason we will need all this is because we will eventually change the topos we are working with from a topos of presheaves over the category $\mv(\mh)$ to the topos of sheaves over the same category but now seen as a topological space. Such a change is needed to define the topos analogue of probabilities and the concept of a group and group transformation in a topos. 

\section{Sheaves}

We will now describe what a sheaf is. In order to do this we will first give the bundle theoretical definition and then give the categorical definition. This equivalence of descriptions is possible since there is a 1:1 (one two one) correspondence between sheaves and a special type of budles namely: \emph{etal\'e bundles}.

So what is an etal\'e bundle?
\begin{Definition}
Given a topological space $X$, a bundle $p_{E}:E\rightarrow X$ is said to be etal\'e iff $p_A$ is a local homeomorphism. By this we mean that, for each $e\in E$ there exists an open set $V$ with $e\in V\subseteq E$, such that $pV$ is open in $X$ and $p_{|V}$ is a homeomorphism $V\rightarrow pV$.
\end{Definition}
If for example $X=R^2$ then for each point of a fibre there will be an open disc isomorphic to an open disc in $R^2$. It is not necessary that these discs have the sane size. Such a collection of open discs on each fibre are glued together by the topology on $E$. 

Another example of etal\'e bundles are covering spaces. However, although all covering spaces are etal\'e, it is not the case that all etal\'e bundles are covering spaces.

Given an etal\'e bundle $p_{E}:E\rightarrow X$ and an open subset $U\subseteq X$, then the pullback of $p_E$ via $i:U\subseteq X$ is etal\'e:

\[\xymatrix{
E_U\ar[rr]^{}\ar[dd]_{\mu}&&E\ar[dd]^{p_E}\\
&&\\
U\ar[rr]_{i}&&X\\
}\]

i.e. $E_U\rightarrow U$ is etal\'e.

This result generalises as follows
\begin{Lemma}
Given any continuous map $f:X\rightarrow Y$ and an etal\'e bundle $p_E:E\rightarrow Y$ then $g:f^*E\rightarrow X$ is etal\'e over $X$.

\end{Lemma}
\begin{proof}
We want to show that $g$ is a local homeomorphism. From the definition of pullback $f^*E:=\{(x,e)|f(x)=p_E(e)\}\subseteq X\times E$. Therefore, given an element $(x,e)\in f^*E$, we want to show that there exists an open neighbourhood $(V, U)\ni (x,e)$ which is mapped homeomorphically onto $U$ via $g$. Since $p_E$ is etal\'e, then there exists an open neighbourhood $U\ni e$ which is mapped homeomorphically into an open set $p_E(U)$ in $Y$. Then since $f$ is continuous, it follows that $f^{-1}(p_E(U))\times U$ is open in $ X\times E$ and is a neighbourhood of some $(x,e)$. If we then define $\big(f^{-1}(p_E(U))\times U\big)\cap f^*E $ this will be an open set since its the intersection of two opens, and it will be a neighbourhood of $(x,e)$ in $ f^*E$. Given the definition of pullback, then $g(f^{-1}(p_E(U))\times U)=f^{-1}(p_E(U))$, i.e. $f^{-1}(p_E(U))\times U$ will be mapped homeomorphically onto $f^{-1}(p_E(U))$. Thus $g$ is etal\'e.
\end{proof}
Each etal\'e bundle is equipped with an etal\'e topology on the stalk space. Such a topology is defined in terms of sections of the bundle as follows
\begin{Definition}
Given an etal\'e bundle $p_E:E\rightarrow X$, both $p$ and any section $s:X\rightarrow E$ of $p_E$ are open maps. Through every point $e\in E$ there is at least one section $s:U\rightarrow X$, and the images of $s(U)$ for all sections form a base for the topology of $E$. If $s$ and $t$ are two sections, the set $W=\{x|s(x)=t(x)\}$ where both sections are defined and agree on, is open in $X$. 

\end{Definition}

From the definition it follows that each stalk has a discrete topology, since by definition of a section, $s$ will pick an element in each stalk $p^{-1}(x)$ for all $x\in U$

Summarising, a sheaf can essentially be thought of as a bundle with some extra
topological properties. In particular,
given a topological space $X$, a sheaf over $X$ is a pair $(A,p)$ consisting of a topological space $A$ and a continuous map $p:A\rightarrow X$, which is a local homeomorphism. 

Thus, pictorially, one can imagine that to each point, in each
fibre, one associates an open disk (each of which will have a
different size) thus obtaining a stack of open disks for each
fibre. These different open discs are then glued together by the
topology on $A$.

The above is the more intuitive definition of
what a sheaf is. Now we come to the technical definition which is
the following:
\begin{Definition}
A sheaf of sets $F$ on a topological space $X$ is a functor\footnote{Here $\mathcal{O}(X)$ indicates the category of open sets of $X$ ordered by inclusions. }
$F:\mathcal{O}(X)^{op}\rightarrow Sets$, such that each open
covering $U=\bigcup_{i}U_i$, $i\in I$ of an open set $U$ of $X$
determines an equaliser
\[\xymatrix{
F(U)\ar@{{>}->}[r]^e & \prod_i F(U_i)\ar@<3pt>[r]^{p}
\ar@<-3pt>[r]_{q} & \prod_{i,j}F(U_i\cap U_j) }\] where for $t\in
F(U)$ we have $e(t)=\{t|_{U_i}|i\in I\}$ and for a family $t_i\in
F(U_i)$ we obtain \be p\{t_i\}=\{t_i|_{U_i\cap
U_j}\},\;\;\;q\{t_i\}=\{t_j|_{U_i\cap U_j}\} \ee
\end{Definition}

Given the definition of product it follows that the maps $e$, $p$, and $q$ above are dermined though the diagram
\[\xymatrix{
&&F(U_i)\ar[rr]^{F(U_i\cap U_j\subseteq U_i)}&&F(U_1\cap U_j)\\
&&&&\\
F(U)\ar@{{>}->}[rr]^e\ar[rruu]\ar[rrdd] && \prod_i F(U_i)\ar@<3pt>[rr]^{p}
\ar@<-3pt>[rr]_{q} \ar[uu]\ar[dd]&& \prod_{i,j}F(U_i\cap U_j) \ar[uu]\ar[dd]\\
&&&&\\
&&F(U_j)\ar[rr]^{F(U_i\cap U_j\subseteq U_j)}&&F(U_i\cap U_j)
}\]

It is clear from the above definition that a sheaf is a special type of presheaf. In fact, given a topological space $X$, $Sh(X)$ is a full subcategory of $Sets^{\mathcal{O}(X)^{op}}$. Similarly as $Sets^{\mathcal{O}(X)^{op}}$, also $Sh(X)$ forms a topos.
\subsubsection{Simple Example}
A very simple example of a sheaf is the following. Consider a presheaf 
\ba
C:\mathcal{O}(X):&\rightarrow& Sets\\
U&\mapsto& C(U):=\{f|f:U\rightarrow \Rl\text{ if continuous }\}
\ea
This is definitely a well defined presheaf in fact, for $U^{'}\subseteq U$ the presheaf maps can be defined through restriction
\ba
C(U)&\rightarrow& C(U^{'})\\
f&\mapsto& f_{|U^{'}}
\ea
The reason $C$ is also a sheaf follows from the continuity properties of the maps $f$. In fact if we consider a covering $U_i$, $i\in I$ of $U$, such that we have continuous functions $f_i:U_i\rightarrow \Rl$ for all $i\in I$. Because of continuity if follows that there exists at most one map $f:U\rightarrow\Rl$ such that $f_{|U_i}=f_i$. Moreover, such a map exists iff
\be
f_i(x)=f_j(x)\;\forall\; x\in U_i\cap U_j\;, (i,j \in I)
\ee 
It follows that the requirement of the map $e:C(U)\rightarrow\coprod_{i\in I} C(U_i)$ being an equaliser is satisfied, thus $C$ is a sheaf.
\subsubsection{Connection Between Sheaves and Etal\'e Bundles}
From the definitions  given above it seems hard to understand what the connection between sheaves and an etal\'e bundles might be. In order to understand this  connection we need to introduce the notion of \emph{germ} of a function. Once we have introduced such a notion that it can be shown that each sheaf is a sheaf of cross sections of a suitable bundle. All this will become clear as we proceed. So first thing what is a \emph{germ}? \emph{Germs} represent constructions which define local properties of functions. In particular they indicate how similar two functions are locally. Because of this locality requirement, \emph{germs} are generally defined on functions acting on topological spaces such that the word local acquires meaning. For example one can consider measure of `locality' to be a power series expansion of a function around some fixed point. Thus, one can say that two holomorphic functions $f, g: U\rightarrow \Cl$ have the same germ at a point $a\in U$ iff the power series expansions around that point are the same. Thus $f, g$ agree on some neighbourhood of $a$, i.e., with respect to that neighbourhood they ``look" the same.

This definition obviously holds only if a power series expansion exists, however it is possible to generalise such a definition in a way that it only requires topological properties of the spaces involve. Thus for example two functions $f, g:X\rightarrow E$ have the same germ at $x\in X$ if there exist some neighbourhood of $x$ on which they agree. In this case we write\footnote{This should be read as: the germ of f at x is the same as the germ of g at x.} $germ_x f=germ_x g$ which implies that $f(x)=g(x)$. But the converse is not true.\\

How do we generalise such a definition of germs in the case of presheaves? Let us consider a presheaf $P:\mathcal{O}(X)\rightarrow Sets\in Sets^{\mathcal{O}(X)^{op}}$ where $X$ is a topological space and $ \mathcal{O}(X)^{op}$ is the category of open sets with reverse ordering (to the inclusion ordering). Given a point $x\in X$ and two neighbourhoods $U$ and $V$ of $x$, the presheaf $P$ assigns two sets $P(U)$ and $P(V)$. Now consider two points $t\in P(V)$ and $s\in P(U)$. We then say that $t$ and $s$ have the same germ at $x$ iff there exists some open $W\subseteq U\cap V$ such that $x\in W$ and $s_{|W}=t_{|W}\in P(W)$. 

The condition of having the same germ at x defines an equivalence class which is denoted as $germ_x s$. Thus $t\in germ_x s$ iff, given two opens $U, V\ni x$ then there exists some $W\subseteq U\cap V$ such that $x\in W$ and $t_{|W}=s_{|W}\in P(W)$, where $s\in P(U)$ and $t\in P(V)$. It follows that the set of all elements obtained though the $P$ presheaf get `quotient' through the equivalence relation of \emph{``belonging to the same germ"}. Therefore, for each point $x\in X$ there will exist a collection of germs at $x$, i.e., a collection of equivalence classes:
\be
P_x:=\{germ_x s|s\in P(U), x\in U\text{ open in } X\}
\ee
We can now collect all these set of germs for all points $x\in X$ defining
\be
\Lambda_P=\coprod_{x\in X}P_x=\{\text{all } germ_x s|s\in X, s\in P(U)\}
\ee
Thus what we have done so far is basically divide the preheaf space in equivalence classes. We can now define the map
\ba
p:\Lambda_P&\rightarrow& X\\
germ_x s&\mapsto&x\\
germ_y s&\mapsto&y
\ea
which sends each germ to the point in which it is taken.
It follows that each $s\in P(U)$ defines a function 
\ba
\dot{s}:U&\rightarrow &\Lambda_P\\
x&\mapsto&germ _x s
\ea
It is straight forward to see that $\dot{s}$ is a section of $p:\Lambda_P\rightarrow X$. Since the assignments $s\rightarrow \dot{s}$ is unique, it is possible to replace each element $s$ in the original presheaf with a section $\dot{s}$ to the set of germs $\Lambda_P$.

We now define a topology on $\Lambda_P$ by considering as basis of open sets all the image sets $\dot{s}(U)\subseteq\Lambda_P$ for $U$ open in $X$, i.e. open sets are unions of images of sections. Such a topology obviously makes $p$ continuous. In fact, given an open set $U\subseteq X$ then $p^{-1}(U)$ is open by definition of the topology on $\Lambda_P$, since $p^{-1}(U)=\bigcup_{s_i\in P(U)}\dot{s}_i(U)$.

On the other hand it is also possible to show that the sections $\dot{s}$ as defined above are continuous with respect to the topology on $\Lambda_P$. To understand this consider two elements $t\in P(V) $ and $s\in P(U)$ such that $\dot{t}(x)=\dot{s}(x)$, i.e. $germ_x (t)=germ_x(s)$ where $x\in V\cap U$. It then follows that there exists an open set $W\ni x$ such that $W\subseteq V\cap U$. If we considered all those elements $y\in V\cap U\subseteq X$ for which $\dot{s}(y)=\dot{x}(y)$ then all such elements will comprise the open set $W\subseteq V\cap U$. Given this reasoning we want to show that given an open $\mathcal{O}\in \Lambda_P$, then $\dot{s}^{-1}(\mathcal{O})$ is open in $X$. Without loss of generality we can choose $\mathcal{O}$ to be a basis set, i.e. 
\be
\dot{s}(W)=\{germ_x(s)|\forall x\in W\}
\ee
Thus $W$ consists of all those points $x$ such that $\dot{s}(x)=\dot{t}$ for $t, s\in germs_x(s)$. It follows that $W$ is open.

One can also show that $\dot{s}$ is open and an injection \begin{proof}
We want to show that $\dot{s}$ is an injection and is open. The fact that it is open follows from the definition of topology on $\Lambda_p$ since the basis of open sets are all the image sets $\dot{s}(U)\subseteq\Lambda_P$ for $U$ open in $X$.  To show that it is injective we need to show that  if $germ_x s=germ_y s$ then $x=y$. This follows from  the definition of germs at a point.
\end{proof}
Putting all these results together we show that $\dot{s}:U\rightarrow \dot{s}(U)$ is a homeomorphism.

We have so managed to construct a bundle $p:\Lambda_P\rightarrow X$ which is a local homeomorphism since each point $germ_x(s)\in \Lambda_P$ has an open neighbourhood $\dot{s}(U)$ such that $p$ restricted to $\dot{s}(U)$ $p:\dot{s}(U)\rightarrow X$ has a two sided inverse $\dot{s}:U\rightarrow\dot{s}(U)$: 
\be
p\circ \dot{s}=id_X;\;\;\;\dot{s}\circ p=id_{\Lambda_P}
\ee
Hence $p$ is a local homeomorphism.\\

The above reasoning shows how, given a presheaf $P$ it is possible to construct a bundle $p:\Lambda_P\rightarrow X$ out of it. Given such a bundle, it is then possible to construct a sheaf in terms of it. In fact, consider the following sheaf
\ba
\Gamma(\Lambda_P):\mathcal{O}^{op}&\rightarrow& Sets\\
U&\mapsto&\{\dot{s}|s\in P(U)\}
\ea
\begin{proof}
We want to show that the presheaf
\ba
\Gamma(\Lambda_P):\mathcal{O}(X)^{op}&\rightarrow& Sets\\
U&\mapsto&\{\dot{s}|s\in P(U)\}
\ea
is actually a sheaf. To this end we should note that the maps are defined by restriction, i.e. given $U_i\subseteq U$ then
\ba
\Gamma(\Lambda_P):\mathcal{O}(X)^{op}&\rightarrow& Sets\\
U_i&\mapsto&\{\dot{s}_i|s_i\in P(U_i)\}
\ea
where $\dot{s}\mapsto\dot{s}_i$ id defined via $\dot{s}_i=P(i_{U_iU})s$. Now since
\ba
\dot{s}:U&\rightarrow&\Lambda_p(U)\\
x&\mapsto&germ_x s
\ea
while 
\ba
\dot{s}_i:U_i&\rightarrow&\Lambda_p(U_i)\\
y&\mapsto&germ_y s_i
\ea
Since $U_i\subseteq U$ then 
\ba
\dot{s}:U_i&\rightarrow&\Lambda_p(U_i)\\
y_i&\mapsto&germ_y s
\ea
Thus $\dot{s}_i=\dot{s}_{|U_i}$. 

In order to show that the above is indeed a sheaf we need to show that the diagram
\[\xymatrix{
\Gamma(\Lambda_p(U))\ar@{{>}->}[r]^e & \prod_i \Gamma(\Lambda_p(U_i))\ar@<3pt>[r]^{p}
\ar@<-3pt>[r]_{q} & \prod_{i,j}\Gamma(\Lambda_p(U_i\cap U_j)) }\] 
is an equaliser.
By applying the definition of the sheaf maps we obtain
\ba
e:\Gamma(\Lambda_p(U))&\rightarrow &\prod_i \Gamma(\Lambda_p(U_i))\\
\dot{s}&\rightarrow&e(\dot{s})=\{\dot{s}_{U_i}|i\in I\}=\{\dot{s}_i|i\in I\}
\ea
On the other hand 
\be
p(\dot{s}_i)=\{s_i|_{U_i\cap U_j}\}=\{s_{|U_i\cap U_j}\}
\ee
while
\be
q(\dot{s}_j)=\{s_j|_{U_i\cap U_j}\}=\{s_{|U_i\cap U_j}\}
\ee
\end{proof}

$\Gamma(\Lambda_P)$ is called the sheaf of cross sections of the bundle $p:\Lambda_P\rightarrow X$.

We can now define a map 
\be
\eta:P\rightarrow\Gamma\circ \Lambda_P
\ee
such that for each context $U\in \mathcal{O}(X)^{op}$ we obtain 
\ba
\eta_U:P_U&\rightarrow&\Gamma(\Lambda_P)(U)\\
s&\mapsto&\dot{s}
\ea
\begin{Theorem}
If $P$ is a sheaf then $\eta$ is an isomorphism.
\end{Theorem}
\begin{proof}
We need to show that $\eta$ is 1:1 and onto. 
\begin{enumerate}
\item One to one:\\
We want to show that if $\dot{s}=\dot{t}$ then $t=s$. Given $t, s\in P(U)$, $\dot{s}=\dot{t}$ means that $germ_x(s)=germ_x(t)$ for all $x\in U$. Therefore there exists opens $V_x\subseteq U$ such that $x\in V_x$ and $t_{|V_x}=s_{|V_x}$. The collection of these opens $V_x$ for all $x\in U$ form a cover of $U$ such that $s_{V_x}=t_{|V_x}$. This implies that $s,t$ agree on the map $P(U)\rightarrow\coprod_{x\in U}P(V_x)$. From the sheaf requirements it follows that $t=s$.

\item Onto:\\
We want to show that any section $h:U\rightarrow\Lambda_p$ is of the form $\eta_U(s)=\dot{s}$ for some $s\in P(U)$. Thus consider a section $h:U\rightarrow\Lambda_p$, this will pick for each $x\in U$ an element say $h(x)=germ_x(s_x)$. 
Therefore for each $x\in U$ there will exist an open $U_x\ni x$ such that $s_x\in P(U_x)$. By definition $germ_x(s_x)=\dot{s}_x(x)$ where $\dot{s}_x$ is a continuous section therefore for each open $U_x$ we get $\dot{s}_x(U_x)=\{germ_x(s_x)|\forall x\in U_x\}$ which is open by definition. 
It follows that for each $x\in U_x$ there will exist some $t, s\in germ_x(s)$ such that $\dot{s}(x)=\dot{t}(x)$. This implies that there exists some open set $W_x$ for which $x\in W_x\subseteq U_x\subseteq U$ and such that $t_{W_x}=s_{|W_x}$. 
These open sets $W_x$ form a covering  of $U$, i.e. $U=\coprod_{x\in U_x}W_x$ with $s_{|W_x}\in P(W_x)$ for each $P(W_x)$. Moreover, since $h(x)=germ_x(s_x)$ for $x\in U_x$ it follows that $h=\dot{s}$ for each $W_x$. 
Now consider two sections $\dot{s}_x$ and $\dot{s}_y$ for $x\in P(W_x)$ and $y\in P(W_y)$ then on the intersection $W_x\cap W_y$, $h$ agrees with both $\dot{s}_x$ and $\dot{s}_y$ therefore the latter agree in the intersection. This means that $germ_z(s_x)=germ_z(s_y)$ for $z\in W_x\cap W_y$, therefore $s_x|_{W_x\cap W_y}=s_y|_{W_x\cap W_y}$\\
We thus obtain a family of elements $s_x$ for each $x\in U_x$ such that they agree on both maps $P(U_x)\leftrightarrows\coprod_{x\in U_x}P(W_x)\cap P(W_y)$.
From the condition of being a sheaf it follows that there exists an $s\in P(U)$ such that $s_{V_x}=s_x$. Then at each $x\in U$ we have $h(x)=germ_x(s_x)=germ_x(s)=\dot{s}(x)$. Therefore $h=\dot{s}$

\end{enumerate}
\end{proof}
It follows that all sheaves are sheaves of cross sections of some bundle.\\
Moreover it is possible to generalise the above process and define a pair of functors
\be
Sets^{\mathcal{O}(X)^{op}}\xrightarrow{\Lambda}Bund(X)\xrightarrow{\Gamma}Sh(X)
\ee 
Which if we combine together we get the so called \emph{sheafification functor}:
\be
\Gamma\Lambda:Sets^{\mathcal{O}(X)^{op}}\rightarrow Sh(X)
\ee
Such a functor sends each presheaf $P$ on $X$ to the ``best approximation" $\Gamma\Lambda_P$ of $P$ by a sheaf.

In the case of etal\'e bundles we then obtain the following equivalence of categories
\[\xymatrix{
\text{Etal\'e}(X)\ar@/_/[r]_{\Gamma}
&Sh(X) \ar@/_/[l]_{\Lambda} 
}\]
The pair of functors $\Gamma$ and $\Lambda$ are an adjoint pair (see section 3.1).
Where we have restricetd the functors to act on $Sh(X)\subseteq Sets^{\mathcal{O}(X)^{op}}$
\section{Sheaves on a Partially Ordered Set}
In the case at hand, since our base category $\mv(\mh)$ is a poset
we have an interesting result. In particular, each poset $P$ is equipped with an Alexandroff topology whose basis is given by the collection of all lower
sets in the poset $P$, i.e., by sets of the form $\downarrow
p:=\{p^{'}\in P|p^{'}\leq p\}$, $p\in P$\footnote{ 
Note that a function $\alpha
: P_1\rightarrow P_2$ between posets
$P_1$ and $P_2$ is continuous with respect to the Alexandroff topologies on each poset, if and only if it is order preserving.}.

The dual of such a topology is the topology of upper sets, i.e.
the topology generated by the sets $\uparrow p:=\{p^{'}\in P|
p^{'}\leq p\}$. Given such a topology it is a standard result
that, for any poset $P$, \be \Sets^P\simeq Sh(P^+) \ee where $P^+$
denotes the complete Heyting algebra of upper sets, which are the
duals of lower sets. It follows that \be \Sets^{P^{op}}\simeq
Sh((P^{op})^+)\simeq Sh(P^-) \ee where $P^-$ denotes the set of
all lower sets in $P$. In particular, for the poset $\mv(\mh)$ we have 
\be\label{equ:correspondence}
\Sets^{\mv(\mh)^{op}}\simeq Sh(\mv(\mh)^-) \ee
Thus every presheaf in our theory is in fact a sheaf with respect to the topology $\mv(\mh)^-$. We will denote by $\underline{\bar{A}}$ the sheaves over $\mv(\mh)$, while the respective presheaf will denote by $\underline{A}$. Moreover, in order to simplify the notation we will write $Sh(\mv(\mh)^-)$ as just $Sh(\mv(\mh))$.

We shall frequently use the particular class of lower sets in
$\mv(\mh)$ of the form \be \downarrow V:= \{V^{'}|V^{'}\subseteq
V\} \ee where $V\in Ob(\mv(\mh))$. It is easy to see that the set
of all of these is a basis for the topology $\mv(\mh)^-$. Moreover
\be \downarrow V_ 1\cap\downarrow V_2 =\downarrow (V_1\cap V_2)
\ee
i.e., these basis elements are closed under finite intersections.

It should be noted that $\downarrow V$ is the `smallest' open set
containing $V$ , i.e., the intersection of all open neighbourhoods
of $V$ is $\downarrow V$ . The existence of such a smallest open
neighbourhood is typical of an Alexandroff space.

If we were to include the minimal algebra
$\Cl(\hat{1})$ in $\mv(\mh)$ then, for any $V_1$, $V_2$ the
intersection $V_1\cap V_2$ would be non-empty. This would imply
that $\mv(\mh)$ is non Hausdorff.
To avoid this, we will exclude the minimal algebra from $\mv(\mh)$. This means that, when $V_1\cap V_2$ equals $\Cl(\hat{1})$ we will not consider it. 

More precisely the semi-lattice operation $V_1, V_2\rightarrow V_1 \wedge V_2$ becomes a partial operation which is defined as $V_1\cap V_2$ only if $V_1 \cap V_2\neq \Cl(\hat{1})$, otherwise it is zero.

\noindent
This restriction implies that when considering the topology on the
poset $\mv(\mh)-\Cl(\hat{1})$ we obtain \be
\downarrow V_1\cap\downarrow V_2=\begin{cases}\downarrow (V_1\cap V_2)& {\rm if}\hspace{.1in}V_1 \cap V_2\neq \Cl(\hat{1});\\
\emptyset& \text{otherwise}.
\end{cases}
\ee There are a few properties regarding sheaves on a poset worth
mentioning:
\begin{enumerate}
\item When constructing sheaves it suffices to restrict attention to the basis elements of
the form $\downarrow V$, $V\in Ob(\mv(\mh)$. For a given presheaf
$\underline{A}$, a key relation between its associated sheaf,
$\underline{\bar{A}}$ is simply \be\label{equ:up}
\underline{\bar{A}}(\downarrow V ) := \underline{A}_V \ee where
the left hand side is the sheaf using the topology $\mv(\mh)^{-}$
and the right hand side
is the presheaf on $\mv(\mh)$. 
\begin{proof}
We want to show that 
\be
\underline{\overline{A}}(\downarrow V)\simeq\underline{A}_V
\ee
Consider the open $\downarrow V$, this forms a poset, thus we can define the presheaf $\ps{A}_{|\downarrow V}$. This has as objects $a_i\in \ps{V_i}$ for all $V_i\in \downarrow V$ and morphisms $\ps{A}_{V}\rightarrow\ps{A}_{V_j}$; $a\mapsto \ps{A}(i_{V, V_j})a$ where $V_j\in \downarrow V$. Moreover for $V_i\cap V_j\subseteq V$ we have 
\ba
\ps{A}_{V}&\rightarrow&\ps{A}_{V_j\cap V_j}\\
a&\mapsto &\ps{A}(i_{V, V_i\cap V_j})a\\
\ea
In order for such a presheaf to be a sheaf we require that 
\[\xymatrix{
\ps{A}(\downarrow V)\ar@{{>}->}[r]^e & \prod_i \ps{A}(\downarrow V_i)\ar@<3pt>[r]^{p}
\ar@<-3pt>[r]_{q} & \prod_{i,j}\ps{A}(\downarrow (V_i\cap V_j)) }\] 
Indeed we have that
\ba
\ps{A}_{V}&\rightarrow&\ps{A}_{V_i}\rightarrow\ps{A}_{(V_i\cap V_j)}\\
a&\mapsto&\ps{A}(i_{V_iV})(a)\mapsto \ps{A}(i_{V_i, V_i\cap V_j})\big(\ps{A}(i_{V_iV})(a)\big)=\ps{A}(i_{V, V_i\cap V_j})a
\ea
on the other hand
\ba
\ps{A}_{V}&\rightarrow&\ps{A}_{V_j}\rightarrow\ps{A}_{(V_i\cap V_j)}\\
a&\mapsto&\ps{A}(i_{V_jV})(a)\mapsto \ps{A}(i_{V_i, V_i\cap V_j})\big(\ps{A}(i_{V_jV})(a)\big)=\ps{A}(i_{V, V_i\cap V_j})a
\ea
Thus 
\be
\underline{\overline{A}}(\downarrow V)\simeq\ps{A}_{|\downarrow V}
\ee
However, given the initial algebra $V\in\downarrow V$, then from $\ps{A}_V$ and the presheaf maps, we can retrieve all of the elements in $\ps{A}_{|\downarrow V}$. Thus
\be
\underline{\overline{A}}(\downarrow V)\simeq\underline{A}_V
\ee
Given a presheaf map, there is an associated restriction map for
sheaves. In particular, given $i_{V_1V} : V_1\rightarrow V$ with
associated presheaf map $\underline{A}(i_{V_1V} ) :
\underline{A}_V\rightarrow \underline{A}_{V_1}$ , then the
restriction map $\rho_{V_1V} :\underline{\bar{A}}(\downarrow V
)\rightarrow\underline{\bar{A}}(\downarrow V_1)$ for the sheaf
$\underline{\bar{A}}$ is defined as \be a_{|\downarrow V_1} =
\rho_{V_1V} (a) := \underline{A}(i_{V_1V} )(a) \ee for all $a\in
\underline{\bar{A}}(\downarrow V )\simeq \underline{A}_V$.

\end{proof}
\item Given an open set $\mathcal{O}$ in $\mv(\mh)^-$ such a set is covered by the down set $\downarrow V$, $V\in Ob(\mv(\mh))$. Therefore we have 
\be 
\underline{\bar{A}}(\mathcal{O}) = {\lim_{\longleftarrow}}_{ V\subseteq
\mathcal{O}}\underline{\bar{A}}(\downarrow V) = {\lim_{\longleftarrow}}_{ V\subseteq
\mathcal{O}}\underline{A}_V \ee 
Where $\lim_{\longleftarrow}$ indicated the inverese limit.

 A direct
consequence of the above is that \be
\underline{\bar{A}}(\mathcal{O})=\Gamma\underline{A}_{|\mathcal{O}}
\ee The connection with \ref{equ:up} is given by the fact that
$\Gamma\underline{\bar{A}}_{|\downarrow V}\simeq\underline{A}_V$.

%
%
\item For presheaves on partially ordered sets the sub-object classifier $\uom^{\mv(\mh)}$ has some interesting properties. In particular, given the set $\uom^{\mv(\mh)}_V$ of sieves on $V$, there exists a bijection between sieves in $\uom^{\mv(\mh)}_V$ and lower sets of $V$. To understand this let us consider any sieve $S$ on $V$, we can then define the lower set of $V$
\be L_S:=\bigcup_{V_1\in S}\downarrow V_1 \ee Conversely, given a
lower set $L$ of $V$ we can construct a sieve on $V$ \be
S_L:=\{V_2\subseteq V|\downarrow V_2\subseteq L\} \ee However if
$\downarrow V_2\subseteq L_S$ then $\downarrow V_2\subseteq
\bigcup_{V_1\in S}\downarrow V_1$, therefore $V_2\in S$. On the
other hand if $V_2\in S_L$ ($S_L$ sieve on $V$), then
$V_2\subseteq V$ and $\downarrow V_2\subseteq L$, therefore
$V_2\in \bigcup_{V_1\in S}\downarrow V_1$, i.e. $V_2\in L_S$. This
implies that the above operations are inverse of each other.
Therefore \be \bar{\uom}^{\mv(\mh)}(\downarrow
V):=\uom^{\mv(\mh)}_V\simeq \Theta(V) \ee where $\Theta(V)$ is
the collection of lower subsets (i.e. open subsets in $\mv(\mh)$)
of $V$. This is equivalent to the fact that, in a topological
space $X$, we have that $\Omega^X(O)$ is the set of all open
subsets of $O\subseteq X$.
\end{enumerate}

\section{Geometric Morphisms}
We will now introduce a very important concept in topos theory, namely the idea of geometric morphisms. Such objects are very important because they allows to define maps between topoi in a way that a lot of internal relations are preserved. In order to fully understand what a geometric morphisms is we first have to introduce the concept of adjunction. This concept also is very important since adjunctions arise pretty much everywhere when one does topos theory,. 
\subsection{Adjunctions}
Consider two categories $\c$ and $\d$ and two functors between them going the opposite directions
\be
F:\d\rightarrow\c\;\text{ and } G:\c\rightarrow\d
\ee
$F$ is said to be \emph{left adjoint} to $G$ (or $G$ is \emph{right adjoint} to $F$) iff given any two objects $A\in \c$ and $X\in \d$ there exists a natural bijection between morphisms:

\be
\frac{X\xrightarrow{f}G(A)}{F(X)\xrightarrow{h}A}
\ee

What this means is that here is an exact correspondence between certain type of any maps, i.e. to each map from $X$ to $G(A)$ there uniquely corresponds a map form $F(X) $ to $A$. In other words $f$ uniquely determines $h$ and vice versa. Therefore we can write the following bijection
\be\label{equ:bijec}
\theta:Hom_{\d}(X, G(A))\xrightarrow{\sim}Hom_{\c}(F(X), A)
\ee
Such a bijection is said to be natural in the sense that, given any morphisms $\alpha:A\rightarrow A^{'}$ in $\c$ and $\beta: X^{'}\rightarrow X$ in $\d$, then the composition between these arrows and $f$ and $h$ above creates yet another correspondence:
\be
\frac{X^{'}\xrightarrow{\beta}X\xrightarrow{f}G(A)\xrightarrow{G(\alpha)}G(A^{'})}{F(X^{'})\xrightarrow{f(\beta)}FX\xrightarrow{h}A\xrightarrow{\alpha}A^{'}}
\ee
The symbol to indicate an adjunction relation between functors is $F\dashv G$.

An important consequence of adjunctions is the existence of unit and co-unit morphisms. These are defined as follows
\begin{Definition}
Given an adjunction $F\dashv G$ with corresponding bijection \ref{equ:bijec}, taking $A=F(X)$ then we obtain a unique map 
\be
\eta_X:X\rightarrow GF(X)
\ee
such that $\theta(\eta_X)=id_{F(X)}$. Such a map is called the unit of the adjunction. Moreover, given a map $f$, then $h$ is uniquely determined such that the following diagram commutes 
\[\xymatrix{
X\ar[rr]^{\eta}\ar[ddrr]_{f}&&GF(X)\ar[dd]^{G(h)}\\
&&\\
&&G(A)\\
}\;\;\;\xymatrix{&F(X)\ar[dd]^h&\\
&&\\
&A&\\
}\] 
$\eta$ is universal among the arrows which make the above diagram commute (i.e. any other such arrow uniquely factors though $\eta$)
\end{Definition}

Similarly we also have the notion of \emph{co-unit} of the adjuntion which is defined as follows
\begin{Definition}
By taking $G(A)=X$ in \ref{equ:bijec} and $f$ the identity on $G(A)$ then $h$ becomes
\be
\epsilon_A:FG(A)\rightarrow (A)
\ee
Therefore $\theta^{-1}(\epsilon_A)=1_{G(A)}$. Moreover, given any $g:F(X)\rightarrow A$ there exists a unique $f$ such that $\epsilon$ is universal amont the arrows which make the following diagram commute
 \[\xymatrix{
X\ar[rr]^{\eta}\ar[ddrr]_{f}&&GF(X)\ar[dd]^{G(h)}\\
&&\\
&&G(A)\\
}\;\;\;\xymatrix{&X\ar[dd]^{f}&\\
&&\\
&G(A)&\\
}\] 

\end{Definition}
\subsubsection{Example}
\textbf{Posets}\\
In a poset we can define an adjunction as follows.
Suppose $g :Q\rightarrow P$ is a monotone map between posets. Then, given any element $x\in P$, We call a $g$-approximation of $x$ (from above) an element $y\in Q$ such that $x \leq g(y)$. Moreover among all such approximations there will be the best one (is a bit similar to how one defined greatest lower bound):
a best $g$-approximation of $x$ is an element $y\in Q$ such that
\be
x \leq g(y)\text{ and } \forall z \in Q ( x \leq g(z)\Rightarrow y \leq z )
\ee
If a best g-approximation exists then it is clearly unique since we are in a poset (at most one arrow between any two elements). Moreover if it does exists, for all $x \in P$  then we have a function $f : P\rightarrow Q$  such that, for all $ x \in  P$,  $z \in Q$:
\be
x \leq g(z) \Leftrightarrow f(x) \leq z
\ee
We say that $f$ is the left adjoint of $g$, and $g$ is the right adjoint of $f$. Again it is trivial to see that the left adjoint of $g$, if it exists, is
uniquely determined by $g$.\\\\
\textbf{Exponential}\\
Consider a category $\c$ in which product are defined. For a given object $A\in \c$ one can define the functor 
\ba
A\times -:\c&\rightarrow &\c\\
B&\mapsto& A\times B
\ea
It is possible to define the right adjoint of such a functor, namely:
\ba
(-)^A:\c&\rightarrow& \c\\
B&\mapsto&B^A
\ea
which is simply the \emph{exponential} as defined in previous lectures. We then obtain the adjunction 
\be
\big(A\times -\big)\dashv (-)^A
\ee

The property of being an adjoint pair implies that there exists the bijection
\be
\frac{C\rightarrow B^A}{A\times C\rightarrow B}
\ee
In this context the \emph{co-unit} map is 
\be
\epsilon:A\times B^A\rightarrow B
\ee
such that, given any map $h:A\times C\rightarrow B$ there exists a unique $f:C\rightarrow B^A$ such that $\epsilon\circ (1\times f)=h$, i.e.

\[\xymatrix{
A\times B^A\ar[rr]^{\eta}&&B\\
&&\\
A\times C\ar[uurr]_h\ar[uu]^{1\times f}&&\\
}\]
commutes. Thus $\epsilon =ev$ and we get the usual definition of exponentiation.
\subsection{Geometric Morphisms}
Now that we have defined what an adjunction is we can define what a geometric morphism is.
\begin{Definition}
A \emph{geometric morphism} $\phi:\tau_1\rightarrow \tau_2$
between topoi $\tau_1$ and $ \tau_2$ is defined to be a pair of
functors $\phi_*:\tau_1\rightarrow \tau_2$ and
$\phi^*:\tau_2\rightarrow \tau_1$, called respectively the
\emph{direct image} and the \emph{inverse image} part of the
geometric morphism, such that
\begin{enumerate}
\item $\phi^*\dashv  \phi_*$ i.e., $\phi^*$ is the left adjoint of $\phi_*$
\item $\phi^*$ is left exact, i.e., it preserves all finite limits.
\end{enumerate}
\end{Definition}
In the case of presheaf topoi, an important  source of such
geometric morphisms arises from functors between the base
categories,  according to the following theorem.
\begin{Theorem} \label{GeomMo} 
A functor $\theta:A\rightarrow B$ between two categories $A$ and
$B$, induces a geometric morphism (also denoted $\theta$)
\begin{equation}
\theta:\Sets^{A^{op}}\rightarrow \Sets^{B^{op}}
\end{equation}
of which the inverse image part
$\theta^*:\Sets^{B^{op}}\rightarrow \Sets^{A^{op}}$ is such that
\begin{equation}
F\mapsto \theta^*(F):=F\circ \theta
\end{equation}

\end{Theorem}

\section{Twisted Presheaves}
In this section we will briefly analyse the problem of twisted
presheaves. If we have time we will see how it is possible to solve this problem by changing the topos we work with.

Given a group $G$, its action on the base category $\mv(\mh)$ is defined as $l_g(V):=\hat{U}_gV\hat{U}_g^{-1}:=\{\hat{U}_g\hat{A}\hat{U}_g^{-1}|\hat{A}\in
V\}$, $g\in G$. When considering the topos $\Sets^{\mv(\mh)^{\op}}$, for each $g$ we obtain the
functor $l_{\hat{U}_g}:\mv(\mh)\rightarrow \mv(\mh)$ with 
induces a geometric morphisms \be
l_{\hat{U}_g}:\Sets^{\mv(\mh)^{\op}}\rightarrow \Sets^{\mv(\mh)^{\op}} \ee whose
inverse image part is \ba
l^*_{\hat{U}_g}:\Sets^{\mv(\mh)^{\op}}&\rightarrow& \Sets^{\mv(\mh)^{\op}}\\
\underline{F}&\mapsto&l^*_{\hat{U}_g}(\underline{F}):=\underline{F}\circ
l_{\hat{U}_g} \ea The above
geometric morphism acted on the spectral presheaf
$\us^{\mv(\mh)}$, the quantity value object
$\underline{\Rl}^{\leftrightarrow}$, truth values and
daseinisation. Let us analyse each of such actions in detail.
\subsection{Group Action on the Presheaves}
In this section we will describe the group action on the presheaves in $\Sets^{\mv(\mh)}$ gives rise to the twisted presheaves.
\subsubsection{Spectral Presheaf}
Given the speactral presheaf $\us\in \Sets^{\mv(\mh)^{\op}}$, the action of each element of the group is
given by the following theorem:
\begin{Theorem}
For each $\hat{U}\in\mathcal{U}(\mh)$, there is a natural
isomorphism $\iota^{\hat{U}}:\us\rightarrow\us^{\hat{U}}$ which is
defined through the following diagram:
\[\xymatrix{
\us_V\ar[rr]^{\iota^{\hat{U}}_V}\ar[dd]_{\us_V(i_{V^{'}V})}&&\us^{\hat{U}}_V\ar[dd]^{\us^{\hat{U}}_V(i_{V^{'}V})}\\
&&\\
\us_{V^{'}}\ar[rr]_{\iota^{\hat{U}}_{V^{'}}}&&\us^{\hat{U}}_{V^{'}}\\
}\] where, at each stage $V$ \be
(\iota^{\hat{U}}_V(\lambda))(\hat{A}):=\langle \lambda,
\hat{U}\hat{A}\hat{U}^{-1}\rangle \ee for all $\lambda\in \us_V$
and $\hat{A}\in V_{sa}$.
\end{Theorem}
The presheaf $\us^{\hat{U}}$ is the twisted presheaf associated to
the unitary operator $\hat{U}$. Such a presheaf is defined as
follows:
\begin{Definition}
The twisted presheaf $\us^{\hat{U}}$ has as:
\begin{itemize}
\item[--] Objects: for each $V\in \mv(\mh)$ it assigns the Gel'fand spectrum of the algebra $\hat{U}V\hat{U}^{-1}$, i.e., $\us^{\hat{U}}_V:=\{\lambda:\hat{U}V\hat{U}^{-1}\rightarrow\Cl|\lambda(\hat{1})=1\}$.
\item[--] Morphisms: for each $i_{V^{'}V}:V^{'}\rightarrow V$ ($V^{'}\subseteq V$) it assigns the presheaf maps
\ba
\us^{\hat{U}}(i_{V^{'}V}):\us^{\hat{U}}_V&\rightarrow&\us^{\hat{U}}_{V^{'}}\\
\lambda&\mapsto&\lambda_{|\hat{U}V^{'}\hat{U}^{-1}} \ea
\end{itemize}
\end{Definition}
\subsubsection{Quantity Value Object}
Similarly, for the quantity value object we obtain the following
theorem:
\begin{Theorem}
For each $\hat{U}\in \mathcal{U}(\mh)$, there exists a natural
isomorphism
$k^{\hat{U}}:\underline{\Rl}^{\leftrightarrow}\rightarrow(\underline{\Rl}^{\leftrightarrow})^{\hat{U}}$,
such that for each $V\in \mv(\mh)$ we obtain the individual
components
$k^{\hat{U}}:\underline{\Rl}^{\leftrightarrow}_V\rightarrow(\underline{\Rl}^{\leftrightarrow})^{\hat{U}}_V$
defined as \be\label{grouponR} k^{\hat{U}}_V(\mu,
\nu)(l^{\hat{U}}(V^{'})):=(\mu(V^{'}),\nu(V^{'})) \ee for all
$V^{'}\subseteq V$
\end{Theorem}
Here, $\mu\in \mathcal{R}^{\leftrightarrow}_V$ is an order
preserving function
$\mu :\downarrow V\rightarrow \Rl$ such that, if $V_2\subseteq V_1\subseteq V$, then $\mu(V_2)\geq\mu(V_1)\geq\mu(V)$, while $\nu$ is an order reversing function $\nu :\downarrow V\rightarrow \Rl$ such that, if $V_2\subseteq V_1\subseteq V$, then $\nu(V_2)\leq\nu(V_1)\leq\nu(V)$.

\noindent
In the equation \ref{grouponR} we have used the bijection between
the sets $\downarrow l^{\hat{U}}(V )$ and $\downarrow V$ .
\subsubsection{Daseinisation}
We recall the concept of daseinisation: 
given a projection operator $\hat{P}$ its daseinisation with
respect to each context $V$ is \be
\delta^o(\hat{P})_V:=\bigwedge\{\hat{Q}\in\mathcal{P}(V)|\hat{Q}\geq\hat{P}\}
\ee 
where $P(V)$ represents the collection of projection operators in $V$.

If we then act upon it by any $\hat{U}$ we obtain \ba
\hat{U}\delta^o(\hat{P})_V\hat{U}^{-1}&:=&\hat{U}\bigwedge\{\hat{Q}\in\mathcal{P}(V)|\hat{Q}\geq\hat{P}\}\hat{U}^{-1}\\
&=&\bigwedge\{\hat{U}\hat{Q}\hat{U}^{-1}\in\mathcal{P}(l_{\hat{U}}(V))|\hat{Q}\geq\hat{P}\}\\
&=&\bigwedge\{\hat{U}\hat{Q}\hat{U}^{-1}\in\mathcal{P}(l_{\hat{U}}(V))|\hat{U}\hat{Q}\hat{U}^{-1}\geq\hat{U}\hat{P}\hat{U}^{-1}\}\\
&=&\delta^o(\hat{U}\hat{P}\hat{U}^{-1})_{l_{\hat{U}}(V)} \ea
where the second and third equation hold since the map $\hat{Q}\rightarrow \hat{U}\hat{Q}\hat{U}^{-1}$ is weakly continuous.

What this implies is that the clopen sub-objects which represent
propositions, i.e., $\underline{\delta(\hat{P})}$, get mapped to
one another by the action of the group.
\subsubsection{Truth Values}
Now that we have defined the group action on daseinisation we can
define the group action on the truth values. We recall that for
pure states the truth object at each stage $V$ is defined as
\ba\label{ali:truthob}
\underline{\mathbb{T}}^{|\psi\rangle}_V&:=&\{\hat{\alpha}\in\mathcal{P}(V)|Prob(\hat{\alpha};|\psi\rangle)=1\}\\
&=&\{\hat{\alpha}\in\mathcal{P}(V)|\langle\psi|\hat{\alpha}|\psi\rangle=1\}
\ea For each context $V\in \mv(\mh)$ the truth value is \ba
v(\underline{\delta(\hat{P})}\in \underline{\mathbb{T}}^{|\psi\rangle})_V&:=&\{V^{'}\subseteq V|\delta^o(\hat{ P})_{V^{'}}\in\mathbb{T}^{|\psi\rangle}_{V^{'}}\}\\
&=&\{V^{'}\subseteq V|\langle\psi|\delta^o(\hat{
P})_{V^{'}}|\psi\rangle=1\} \ea we now act upon it with a group
element $\hat{U}$ obtaining \ba
l_{\hat{U}}\Big(v(\delta^o(\hat{P})\in
\underline{\mathbb{T}}^{|\psi\rangle})_V\Big)
&:=&l_{\hat{U}} \{V^{'}\subseteq V|\langle\psi|\delta^o(\hat{ P})_{V^{'}}|\psi\rangle=1\}\\
&=&\{l_{\hat{U}} V^{'}\subseteq l_{\hat{U}} V|\langle\psi|\delta^o(\hat{ P})_{V^{'}}|\psi\rangle=1\}\\
&=&\{l_{\hat{U}} V^{'}\subseteq l_{\hat{U}} V|\langle\psi|\hat{U}^{-1}\hat{U}\delta^o(\hat{ P})_{V^{'}}\hat{U}^{-1}\hat{U}|\psi\rangle=1\}\\
&=&\{l_{\hat{U}} V^{'}\subseteq l_{\hat{U}} V|\langle\psi|\hat{U}^{-1}\delta^o(\hat{U}\hat{P}\hat{U}^{-1})_{l_{\hat{U}}(V)}\hat{U}|\psi\rangle=1\}\\
&=&v(\delta^o(\hat{U}\hat{P}\hat{U}^{-1})\in
\underline{\mathbb{T}}^{\hat{U}|\psi\rangle})_{l_{\hat{U}}(V)} \ea
We thus obtain the following equality: \be
l_{\hat{U}}\Big(v(\delta^o(\hat{P})\in
\underline{\mathbb{T}}^{|\psi\rangle})_V\Big)=v(\delta^o(\hat{U}\hat{P}\hat{U}^{-1})\in
\underline{\mathbb{T}}^{\hat{U}|\psi\rangle})_{l_{\hat{U}}(V)} \ee
Thus truth values are invariant under the group transformations. This is the topos analogue of Dirac covariance, i.e., given a state $|\psi\rangle$ and a physical quantity $\hat{A}$, we would obtain the same predictions if we replaced the state by $\hat{U}|\psi\rangle$ and the quantity by $\hat{U}\hat{A}\hat{U}^{-1} $

\chapter{Lecture 16}
In this lecture we will try understanding a possible way of solving the problem of twisted presheaves This was done in \cite{group}. In particular, we will change the base category to be the category of abelian von-Neumann sub-algebras on which we assume that no group acts upon. We will then define the topos of sheaves over such a category equipped with the Alexandroff topology. This will be the new sheaf we will work with. It turns out that by using such a topos we will be able to define the concept of a group and group transformations which does not lead to twisted presheaves.

\section{In Need of a Different Base Category}
When analysing the origin of the twisted presheaves, it is clear that the reason we do get a twist is because the group moves the abelian algebras around, i.e. the group action is defined on the base category itself. Thus a possible way of avoiding the occurrence of twists is by imposing that the group does not act on the base category. The category of abelian von-Neumann sub-algebras with no group acting on it will be denoted $\mv_f(\mh)$ where we have added the subscript $f$ (for fixed) to distinguish
this situation from the case in which the group does act. Obviously if one then just defined sheaves over $\mv_f(\mh)$, then there would be no group action at all. Therefore something extra is needed. As we will see this `extra' will be the introduction of an intermediate category which will be used as an intermediate base category. On such an intermediate category the group is allowed to act, thus the sheaves defined over it will admit a group action. Once this is done, everything is ``pushed down" to the fixed category $\mv_f(\mh)$. A we will see, the sheaves defined in this way will admit a group action which now takes place at an intermediate stage, but will not produce any twists since the final base category stays fixed. 

\noindent
Thus the first question to address is: what is this intermediate category?

\section{The Sheaf of Faithful Representations}
In our new approach we still use   the poset $\mv(\mh)$ as the
base category but now we `forget' the group action. 
We now  consider the collection, $Hom_{faithful}(\mv_f(\mh),
\mv(\mh))$, of all faithful poset representations of $\mv_f(\mh)$
in $\mv(\mh)$ that come from the action of some transformation group
$G$. Thus we have the collection of all homomorphisms
$\phi_g:\mv_f(\mh)\rightarrow \mv(\mh)$, $g\in G$,  such that
$$\phi_{g}(V):=\hat{U}_gV\hat{U}_{g^{-1}}$$
We can `localise' $Hom_{faithful}(\mv_f(\mh), \mv(\mh))$ by
considering for each $V$, the set $Hom_{faithful}(\downarrow
V,\mv(\mh))$. It is easy to see that this actually defines a
\emph{presheaf} over $\mh$; which we will denote
$\underline{Hom}_{faithful}(\mv_f(\mh), \mv(\mh))$

Now, for each algebra $V$ there exists the fixed point group
$$G_{FV}:=\{g\in G|\forall v\in V\;\hat{U}_gv\hat{U}_{g^-1}=v\}$$
This implies that the collection of all faithful representations
for each $V$ is actually the quotient space $G/G_{FV}$. This
follows from the fact that the group homomorphisms
$\phi:G\rightarrow GL(V)$ has to be injective, but that would not
be the case if we also considered the elements of $G_{FV}$, since
each such element would give the same homomorphism.

Thus for each $V$ we have that
$$\underline{Hom}_{faithful}(\mv_f(\mh), \mv(\mh))_V:=
Hom_{faithful}(\downarrow\! V,\mv(\mh))\cong G/G_{FV}$$ As we will
shortly see, there is a presheaf, $\ps{G/G_F}$, such that, as
presheaves,
$$ \underline{Hom}_{faithful}(\mv_f(\mh), \mv(\mh))\cong \ps{G/G_F}$$
whose local components are defined above. In the rest of this
paper, unless otherwise specified, $\underline{Hom}(\mv_f(\mh),
\mv(\mh))$ will mean $\underline{Hom}_{faithful}(\mv_f(\mh),
\mv(\mh))$
\begin{Lemma}
$G_{FV}$ is a normal subgroup of $G_V$.
\end{Lemma}

\begin{Proof}
Consider an element $g\in G_{FV}$, then given any other element
$g_i\in G_V$ we consider the element $g_igg_i^{-1}$. Such an
element acts on each $v\in V$ as follows: \ba
\hat{U}_{g_igg_i^{-1}}v\hat{U}_{(g_igg_i^{-1})^{-1}}&=&\hat{U}_{g_igg_i^{-1}}v\hat{U}_{g_ig^{-1}g_i^{-1}}\\
&=&\hat{U}_{g_i}\hat{U}_g\hat{U}_{g_i^{-1}}v\hat{U}_{g_i}\hat{U}_{g^{-1}}\hat{U}_{g_i^{-1}}\nonumber\\
&=&\hat{U}_{g_i}\hat{U}_gv^{'}\hat{U}_{g^{-1}}\hat{U}_{g_i^{-1}}\nonumber\\
&=&\hat{U}_{g_i}v^{'}\hat{U}_{g_i^{-1}}\nonumber\\
&=&\hat{U}_{g_i}\hat{U}_{g_i^{-1}}v\hat{U}_{g_i}\hat{U}_{g_i^{-1}}\nonumber\\
&=&v\nonumber \ea where $v^{'}\in V$ because $g_i\in G_V$.
\end{Proof}

We then have the standard result that if $G$ is a group and $N$ a
normal subgroup of $G$ then the coset space $G/N$ has a natural
group structure. In the Lie group case, $G/N$ would only have a
Lie group structure if $N$ is a \emph{closed} subgroup of $G$.
However, it is clear from the definition of $G_{FV}$ that it is
closed, and hence for each $V$ we have a Lie group, $G/G_{FV}$. We
note \emph{en passant} that $G$ is a principal fibre bundle over
$G/G_{FV}$ with fiber $G_{FV}$.

For us, the interesting aspect of the collection $G_{FV}$,
$V\in\mv(\mh)$ is that, unlike the collection of stability groups
$G_V$, $V\in\mv(\mh)$, form the components of a presheaf over
$\mv_f(\mh)$ (or $\mv(\mh)$) defined as follows:
\begin{Definition}
The presheaf $\ug_F$ over $\mv_f(\mh)$ has as
\begin{enumerate}
\item [--] Objects: for each $V\in \mv_f(\mh)$ we define set $\ug_{FV}:=G_{FV}=\{g\in G|\forall v\in V\;\hat{U}_gv\hat{U}_{g^-1}=v\}$
\item [--] Morphisms: given a map $i:V^{'}\rightarrow V$ in $\mv_f(\mh)$ ($V^{'}\subseteq V$) then we define the morphism $\ug_F(i):\ug_{FV}\rightarrow \ug_{FV^{'}}$, as subgroup inclusion.
\end{enumerate}
\end{Definition}

The morphisms $\ug_F(i):\ug_{FV}\rightarrow \ug_{FV^{'}}$ are well defined since if   $V^{'}\subseteq V$ then clearly $G_{FV}\subseteq G_{FV^{'}}$. Associativity is obvious.

We now define the  presheaf $\ps{G/G_F}$ as follows:
\begin{Definition}
The presheaf $\ps{G/G_F}$ is defined as the presheaf with
\begin{enumerate}
\item[--] Objects: for each $V\in \mv_f(\mh)$ we assign $(\ps{G/G_F})_V:=G/G_{FV}\cong Hom(\downarrow\!V, \mv(\mh))$. An element of $G/G_{FV}$ is an orbit $w^g_V:=\{g\cdot G_{FV}\}$ which corresponds to the unique homeomorphism $\phi^g$.
\item[--] Morphisms: Given a morphisms $i_{V^{'}V}:V^{'}\rightarrow V$ ($V^{'}\subseteq V$) in $\mv_f(\mh)$ we define
\ba
\ps{G/G_F}(i_{V^{'}V}):G/G_{FV}&\rightarrow& G/G_{FV^{'}}\\
w^g_V&\mapsto&\ps{G/G_F}(i_{V^{'}V})(w^g_V)
\ea as the projection maps $\pi_{V^{'}V}$ of the fibre bundles \be
G_{FV^{'}}/G_{FV}\rightarrow G/G_{FV}\rightarrow G/G_{FV^{'}} \ee
with fibre isomorphic to $G_{FV^{'}}/G_{FV}$.

What this means is that to each $w^g_{V^{'}}=g\cdot G_{FV^{'}}\in
G/G_{FV^{'}}$ one obtains in $G/G_{FV}$ the fibre \ba
\pi^{-1}_{V^{'}V}(g\cdot G_{FV^{'}})&:=&\sigma^g_{V}=\{g_i(g\cdot G_{FV})|\forall g_i\in G_{FV^{'}}\}\\
&=&\{l_{g_i}\cdot w^g_{V}|\forall g_i\in G_{FV^{'}}\}\\
&=&\{w^{g_ig}_V|g_i\in G_{FV^{'}}\} \ea In the above expression we
have used the usual action of the group $G$ on an orbit:  \be
l_{g_i}\cdot w^g_{V}=g_i\cdot(g\cdot G_{FV})=g_i\cdot g\cdot
G_{FV}=:w^{g_ig}_V \ee The fibre $\sigma^g_{V}$ is obviously
isomorphic to $G_{FV^{'}}/G_{FV}$. Thus the projection map $\pi_{V^{'}V}$
projects \be \pi_{V^{'}V}(\sigma^g_{V})=g\cdot G_{FV^{'}}=w^g_{V^{'}} \ee
such that for individual elements we have \be
\ps{G/G_F}(i_{V^{'}V})(w^{g}_V):=\pi_{V^{'}V}(\sigma^g_{V})=w^g_{V^{'}}
\ee Note that when $g_i\in G_{FV^{'}}$ but $g_i\notin G_{FV}$ then
$w^{g}_V=g\cdot G_{FV}$ and $w^{g}_{V^{'}}=g\cdot
G_{FV^{'}}=g_igG_{FV^{'}}=w^{g_ig}_{V^{'}}$. Therefore
$\ps{G/G_F}(i_{V^{'}V})w^{g}_V=w^{g}_{V^{'}}=w^{g_ig}_{V^{'}}$

\end{enumerate}
\end{Definition}

It should be noted that the morphisms in the presheaf $\ps{G/G_F}$
can also be defined in terms of the homeomorphisms
$Hom(\mv_f(\mh), \mv(\mh))$. Namely, given an element $g_j\in
w^g_V$ we obtain the associated homomorphisms $\phi_{g_j}$, such
that \be \ps{G/G_F}(i_{V^{'}V})\phi_{g_j}:=\phi_{g_j|V^{'}} \ee

We will now define another presheaf which we will then show to be isomorphic to $\ps{G/G_F}$. To this end we first of all have to introduce the constant presheaf $\underline{G}$. This is defined as follows
\begin{Definition}
The presheaf $\underline{G}$ over $\mv_f(\mh)$ is defined on 
\begin{itemize}
\item Objects: for each context $V$, $\ps{G}_V$ is simply the entire group, i.e. $\ps{G}_V=G$
\item Morphisms: given a morphisms $i:V^{'}\subseteq V$ in $\mv(\mh)$, the corresponding morphisms $\ps{G}_V\rightarrow \ps{G}_{V^{'}}$ is simply the identity map.
\end{itemize}
\end{Definition}
We are now ready to define the new presheaf.
\begin{Definition}
The presheaf $\ps{G}/\ps{G_F}$ over $\mv_f(\mh)$ is defined on 
\begin{itemize}
\item Objects. For each $V\in \mv_f(\mh)$ we obtain $(\ps{G}/\ps{G_F})_V:=G/G_{FV}$. Since as previously explained the equivalence relation is computed context wise.
\item Morphisms. For each map $i:V^{'}\subseteq V$ we obtain the morphisms 
\ba
(\ps{G}/\ps{G_F})_V&\rightarrow& (\ps{G}/\ps{G_F})_{V^{'}}\\
G/G_{FV}&\rightarrow&G/G_{FV^{'}}
\ea
These are defined to be the projection maps $\pi_{V^{'}V}$ of the fibre bundles \be
G_{FV^{'}}/G_{FV}\rightarrow G/G_{FV}\rightarrow G/G_{FV^{'}} \ee
with fibre isomorphic to $G_{FV^{'}}/G_{FV}$.

\end{itemize}
 
\end{Definition}
From the above definition it is trivial to show the following theorem.
\begin{Theorem}
\be
\ps{G/G_F}\simeq\ps{G}/\ps{G_F}
\ee
\end{Theorem}
\begin{Proof}
We construct the map $k:\ps{G/G_F}\rightarrow \ps{G}/\ps{G_F}$ such that, for each context $V$ we have
\ba
k_V:\ps{G/G_F}_V&\rightarrow& \ps{G}/\ps{G_F}_V\\
G/G_{FV}&\mapsto&G/G_{FV}
\ea
This follows from the definitions of the individual presheaves.
\end{Proof}

\section{Changing Base Category}
We know that given a sheaf over a poset we obtain the corresponding
etal\'e bundle. In our case the sheaf in question is $\ps{G/G_F}$ with corresponding etal\'e bundle $p:\Lambda \ps{G/G_F}\rightarrow \mv_f(\mh)$ where $\Lambda\ps{G/G_F}$ is the etal\'e space. We will now equip the etal\'e space $\Lambda(\ps{G/G_F})=\coprod_{V\in\mv_F(\mh)}(\ps{G/G_F})_V$ with a poset structure.

The most obvious poset structure to use would be the partial order
given by restriction, i.e., $w_{V}\leq w_{V^{'}}$ iff $V\subseteq V^{'}$ and 
$w_V^g=w^g_{V^{'}}|_V$ or equivalently $g\cdot G_{FV}=g\cdot
(G_{FV}\cap G_{FV^{'}})$. We could write this last condition as an
inclusion of sets as follows:
$w_{V}\subseteq w_V^{'}$ ($g\cdot G_{FV}\subseteq g\cdot G_{FV^{'}}$). 
However this poset structure would not give a presheaf if we were
to use it as the base category, rather it would give a covariant
functor. To solve this problem we adopt the \emph{order dual} of
the partially ordered set, which is the same set but equipped with
the inverse order which is itself a partial order. We thus define
the ordering on $\Lambda( \ps{G/G_F})$ as follows:

\begin{Lemma}\label{lem:ordering}
Given two orbits $w^g_V\in G/G_{FV}$ and $w^g_{V^{'}}\in
G/G_{FV^{'}}$ we define the partial ordering $\leq$, by
defining
$$w^g_{V^{'}}\leq w^g_V$$ iff
\ba
V^{'}&\subseteq &V\\
w^g_V&\subseteq& w^g_{V^{'}} \ea Note that the last condition is
equivalent to $w^g_V=w^g_{V^{'}}|_V$ ($g\cdot G_{FV}=g\cdot
(G_{FV}\cap G_{FV^{'}})$).
\end{Lemma}
It should be noted though that if $w^g_V=w^g_{V^{'}}|_V$ then $\ps{G/G_F}(i_{V^{'}V})(w^g_V)=
\ps{G/G_F}(i_{V^{'}V})(w^g_{V^{'}}|_{V})=w^g_{V^{'}}$. In other words it is also possible to define the partial ordering 
in terms of the presheaf maps defined above, i.e., \be w^g_V\geq
w^g_{V^{'}}\text{  iff  }w^g_{V^{'}}=\ps{G/G_F}(i_{V^{'}V})w^g_V
\ee

We now show that the ordering defined on $\Lambda(\ps{G/G_F})$ is
indeed a partial order.
\begin{Proof}
\noindent
\begin{enumerate}
\item \emph{Reflexivity}. Trivially $w^g_V\leq w^g_V$ for all $w^g_V\in \Lambda \ps{G/G_F}$.
\item \emph{Transitivity}. If $w^g_{V_1}\leq w^g_{V_2}$ and $w^g_{V_2}\leq w^g_{V_3}$ then $V_1\subseteq V_2$ and $V_2\subseteq V_3$. From the partial ordering on $\mv(\mh)$ it follows that $V_1\subseteq V_3$. Moreover from the definition of  ordering on $\Lambda (\ps{G/G_F})$ we have that $w^g_{V_2}=w^g_{V_1}|_{V_2}$ and $w^g_{V_3}=w^g_{V_2}|_{V_3}$ which implies that $w^g_{V_3}=w^g_{V_1}|_{V_3}$. It follows that $w_{V_1}\leq w_{V_3}$.
\item \emph{Antisymmetry}. If $w_{V_1}\leq w_{V_2}$ and $w_{V_2}\leq w_{V_1}$, it implies that $V_1\leq V_2$ and $V_2\leq V_1$ which, by the partial ordering on $\mv(\mh)$ implies that $V_1=V_2$. Moreover the above conditions imply that $w_{V_1}= w_{V_2}|_{V_1}$ and $w_{V_2}=w_{V_1}|_{V_2}$, which by the property of subsets implies that $w_{V_1}=w_{V_2}$.
\end{enumerate}
\end{Proof}
Given the previously defined isomorphisms, 
$Hom(\down{V}, \mv(\mh))\cong (\ps{G/G_F})_V$ for each $V$, then to each equivalence class $w_V^g$ there is associated a particular homeomorphism $\phi_g:\down V\rightarrow \mv(\mh)$. Even though $w_V^g$ is an equivalence class, each element in it will give the same $\phi_g$, i.e. it will pick out the same $V_i\in \mv(\mh)$. This is because the equivalence relation is defined in terms of the fixed point group for $V$. 

Therefore it is also possible to define the ordering relation on
$\Lambda(\ps{G/G_F})$ in terms of the homeomorphisms $\phi^g_i$.
First of all we introduce the bundle space $\Lambda J\simeq \Lambda(\ps{G/G_F})$ which is essentially the same as $\Lambda(\ps{G/G_F})$, but whose elements are now the maps $\phi^g_i$, i.e., $\Lambda J=\Lambda(\underline{Hom}(\mv_f(\mh), \mv(\mh)))$. The associated bundle map is $p_J:\Lambda J\rightarrow \mv_f(\mh)$.
We then define the ordering on $\Lambda J$ as 
$\phi^g_i\leq \phi^g_j$ iff
 \be
 p_J(\phi^g_i)\subseteq p_J(\phi^g_j)
\ee and \be \phi^g_i=\phi^g_j| _{p_J(\phi^g_i)} \ee We now need to
show that this does indeed define a partial
order on $\Lambda J$.
\begin{Proof}\hspace{.2in}
\begin{enumerate}
\item \emph{Reflexivity}. Trivially $\phi^g_i\leq \phi^g_i$ since  $p_J(\phi^g_i)\subseteq p_J(\phi^g_i)$ and $\phi^g_i=\phi^g_i$.
\item \emph{Transitivity}. If $\phi^g_i\leq \phi^g_j$ and $\phi^g_j\leq \phi^g_k$ then $p_J(\phi^g_i)\subseteq p_J(\phi^g_j)$ and  $p_J(\phi^g_j)\subseteq p_J(\phi^g_k)$, therefore  $p_J(\phi^g_i)\subseteq p_J(\phi^g_k)$. Moreover we have that $\phi^g_i=\phi^g_j| _{p_J(\phi^g_i)}$ and $\phi^g_j=\phi^g_k| _{p_J(\phi^g_j)}$, therefore $\phi^g_i=\phi^g_k| _{p_J(\phi^g_i)}$.
\item \emph{Antisymmetry}.  If $\phi^g_i\leq \phi^g_j$ and $\phi^g_j\leq \phi^g_i$ it implies that  $p_J(\phi^g_i)\subseteq p_J(\phi^g_j)$ and  $p_J(\phi^g_j)\subseteq p_J(\phi^g_i)$, thus  $p_J(\phi^g_i)= p_J(\phi^g_j)$. Moreover we have that $\phi^g_i=\phi^g_j| _{p_J(\phi^g_i)}$ and $\phi^g_j=\phi^g_i| _{p_J(\phi^g_j)}$, therefore $\phi^g_i=\phi^g_j$.
\end{enumerate}
\end{Proof}
Given this ordering we can now define the corresponding ordering on $\Lambda(\ps{G/G_F})$ as
$w^g_{V_i}\leq w^g_{V_j}$ iff
$\phi^g_i\leq \phi^g_j$. We  have again used the fact that to each $w^g_{V_i}$ there is
associated a unique homeomorphism $\phi^g_V:\downarrow V\rightarrow
\mv(\mh)$.
\section{From Sheaves on the Old Base Category to Sheaves on The New Base Category}
In what follows we will move freely between the language of presheaves and that of sheaves which we will  both denote as $\underline{A}$. Which of the two is being used should be clear from the context. The reason we are able to do this is because our base categories are posets (see discussion at the end of section 2).

We are now interested in `transforming' all the physically
relevant sheaves on $\mv(\mh)$ to sheaves over
$\Lambda(\ps{G/G_F})$. Therefore we are interested in finding a
functor $I: Sh(\mv(\mh))\rightarrow Sh(\Lambda(\ps{G/G_F}))$. As a
first attempt we define: \ba
I: Sh(\mv(\mh))&\rightarrow& Sh(\Lambda(\ps{G/G_F}))\\
\underline{A}&\mapsto&I(\underline{A}) \ea such that for each
context $w_V^g\simeq \phi^g$ we define \be
\big(I(\underline{A})\big)_{w_V^g}:=\underline{A}_{\phi^g(V)}=\Big((\phi^g)^*(\underline{A})\Big)(V)
\ee
where $\phi^g:\down V\rightarrow \mv(\mh)$ is the unique homeomorphism associated with the equivalence class $w^g_V=g\cdot G_{FV}$.

We then need to define the morphisms. Thus, given $i_{w_{V^{'}}^g,w_V^g}:w_{V^{'}}^g\rightarrow w_{V}^g$ ($w_{V^{'}}^g\leq
w_{V}^g$) with corresponding homeomorphisms $\phi_2^g\leq \phi_2^g$ ($\phi_1^g\in Hom(\down V, \mv(\mh))$ and $\phi_2^g\in Hom(\downarrow V^{'}, \mv(\mh))$) we have the associated morphisms
$I\underline{A}(i_{w_{V^{'}}^g,w_V^g}):\big(I(\underline{A})\Big)_{w_V^g}\rightarrow
\big(I(\underline{A})\Big)_{w_{V^{'}}^g}$ defined as \be
(I\underline{A}(i_{w_{V^{'}}^g,w_V^g}))(a)=(I\underline{A}(i_{\phi_2^g,\phi_1^g}))(a):=\underline{A}_{\phi^g_1(V),\phi^g_2(V^{'})}(a)
\ee
for all $a\in \underline{A}_{\phi^g(V)}$. In the above equation $V=p_J(\phi_1^g)$ and $V^{'}=p_J(\phi_2^g)$\footnote{Recall that $p_J:\Lambda J=\Lambda(\underline{Hom}(\mv_f(\mh),\mv(\mh))\rightarrow\mv_f(\mh)$.}. 
Moreover, since $\phi_2^g\leq \phi^g_1$ is equivalent to the condition $w_{V^{'}}^g\leq w_{V}^g$, then $\phi^g_2(V^{'})\subseteq \phi^g_1(V)$ and $\phi^g_2=\phi^g_1|_{V^{'}}$.
\begin{Theorem}
The map $I:Sh(\mv(\mh))\rightarrow Sh(\Lambda( \ps{G/G_F}))$ is a
functor defined as follows:
\begin{enumerate}
\item [(i)] Objects: $\big(I(\underline{A})\big)_{w_V^g}:=\underline{A}_{\phi^g_1(V)}=\Big((\phi^g)^*(\underline{A})\Big)(V)$. If $w_{V^{'}}^g\leq w_{V}^g$ with associated homeomorphisms $\phi_2^g\leq \phi^g_1$ ($\phi_1^g\in Hom(\down V, \mv(\mh))$ and $\phi_2^g\in Hom(\downarrow V^{'}, \mv(\mh))$), then
$$(I\underline{A}(i_{w_{V^{'}}^g,w_V^g}))=I\underline{A}(i_{\phi_2^g,\phi_1^g}):=\underline{A}_{\phi^g_1(V),\phi^g_2(V^{'})}:\underline{A}_{\phi^g_1(V)}\rightarrow \underline{A}_{\phi^g_2(V^{'})}$$
where $V=p_J(\phi_1^g)$ and $V^{'}=p_J(\phi_2^g)$.
\item [(ii)] Morphisms: if we have a morphisms $f:\underline{A}\rightarrow\underline{B}$ in $Sh(\mv(\mh))$ we then define the corresponding morphisms in $Sh(\Lambda (\ps{G/G_F}))$ as
\ba
I(f)_{w_V^g}:I(\underline{A})_{w_V^g}&\rightarrow& I(\underline{B})_{w_V^g}\\
f_{\phi_1^g}:\underline{A}_{\phi_1^g(p_J(\phi_1^g))}&\rightarrow
&\underline{B}_{\phi_1^g(p_J(\phi_1^g))} \ea
\end{enumerate}
\end{Theorem}
\begin{Proof}
Consider an arrow $f:\underline{A}\rightarrow\underline{B}$ in
$Sh(\mv(\mh))$ so that, for each $V\in \mv(\mh)$, the local
component is $f_V:\underline{A}_V\rightarrow\underline{B}_V$ with
commutative diagram
\[\xymatrix{
\underline{A}_{V_1}\ar[rr]^{f_{V_1}}\ar[dd]_{\underline{A}_{V_1V_2}}&&\underline{B}_{V_1}\ar[dd]^{\underline{B}_{V_1V_2}}\\
&&\\
\underline{A}_{V_2}\ar[rr]_{f_{V_2}}&&\underline{B}_{V_2}\\
}\] for all pairs $V_1$, $V_2$ with $V_2\leq V_1$. Now suppose
that $w_{V_2}^g\leq w_{V_1}^g$ with associated homeomorphisms $\phi^g_2\leq\phi^g_1$, such
that (i) $p_J(\phi^g_2)\subseteq p_J(\phi^g_1)$; and (ii) $\phi^g_2 =
\phi^g_1|_{p_J(\phi^g_2)}$. We want to show that the action of the $I$
functor gives the commutative diagram
\[\xymatrix{
I(\underline{A})_{w_{V_1}^g}\ar[rr]^{I(f)_{w_{V_1}^g}}\ar[dd]_{I(\underline{A})(i_{w_{V_1}^g,w_{V_2}^g})}&&I(\underline{B})_{w_{V_1}^g}\ar[dd]^{I(\underline{B})(i_{w_{V_1}^g,w_{V_2}^g})}\\
&&\\
I(\underline{A})_{w_{V_2}^g}\ar[rr]_{I(f)_{w_{V_2}^g}}&&I(\underline{B})_{w_{V_2}^g}\\
}\] for all $V_2\subseteq V_1$. By applying the definitions we get

\[\xymatrix{
\underline{A}_{\phi^g_1(p_J(\phi^g_1))}\ar[rr]^{f_{\phi^g_1(p_J(\phi^g_1))}}\ar[dd]_{\underline{A}_{\phi^g_1(p_J(\phi^g_1)),\phi^g_2(p_J(\phi^g_2))}}&&\underline{B}_{\phi^g_1(p_J(\phi^g_1))}\ar[dd]^{\underline{B}_{\phi^g_1(p_J(\phi^g_1)),\phi^g_2(p_J(\phi^g_2))}}\\
&&\\
\underline{A}_{\phi^g_2(p_J(\phi^g_2))}\ar[rr]_{f_{\phi^g_2(p_J(\phi^g_2))}}&&\underline{B}_{\phi^g_2(p_J(\phi^g_2))}\\
}\]
which is commutative. Therefore $I(f)$ is a well defined arrow in $Sh(\Lambda\ps{G/G_F})$ from $I(\underline{A})$ to $I(\underline{B})$.

Given two arrows $f,g$ in $Sh(\mv(\mh))$ then it follows that: \be
I(f\circ g)=I(f)\circ I(g) \ee This proves that $I$ is a functor
from $Sh(\mv(\mh))$ to $Sh(\Lambda(\ps{G/G_F}))$.
\end{Proof}
From the above definition of the functor $I$ we immediately have
the following corollary:
\begin{Corollary}
The functor $I$ preserves monic arrows.
\end{Corollary}
\begin{Proof}
Given a monic arrow $f:\underline{A}\rightarrow \underline{B}$ in
$Sh(\mv(\mh))$ then by definition \ba
I(f)_{w_{V_2}^g}:I(\underline{A})_{w_{V_2}^g}&\rightarrow& I(\underline{B})_{w_{V_2}^g}\\
f_{\phi_2^g(p_J(\phi^g_2))}:\underline{A}_{\phi_2^g(p_J(\phi^g_2))}&\rightarrow&\underline{B}_{\phi_2^g(p_J(\phi^g_2))}
\ea The fact that such a map is monic is straightforward.
\end{Proof}
Similarly we can show that
\begin{Corollary}
The functor $I$ preserves epic arrows.
\end{Corollary}
\begin{Proof}
Given an epic arrow $f:\underline{A}\rightarrow \underline{B}$ in
$Sh(\mv(\mh))$ then by definition \ba
I(f)_{w_{V_2}^g}:I(\underline{A})_{w_{V_2}^g}&\rightarrow& I(\underline{B})_{w_{V_2}^g}\\
f_{\phi_2^g(p_J(\phi^g_2))}:\underline{A}_{\phi_2^g(p_J(\phi^g_2))}&\rightarrow&\underline{B}_{\phi_2^g(p_J(\phi^g_2))}
\ea The fact that such a map is epic is straightforward.
\end{Proof}
We would now like to know how such a functor behaves with respect
to the terminal object. To this end we define the following
corollary:
\begin{Corollary}
The functor $I$ preserves the terminal object.
\end{Corollary}
\begin{Proof}
The terminal object in $Sh(\mv(\mh))$ is the objects
$\underline{1}_{Sh(\mv(\mh))}$ such that to each element $V\in
\mv(\mh)$ it associates the singleton set $\{*\}$. We now apply
the $I$ functor to such an object obtaining \be
I(\underline{1}_{Sh(\mv(\mh))})_{w_V^g}:=(\underline{1}_{Sh(\mv(\mh))})_{\phi^g(p_J(\phi^g))}=\{*\}
\ee
where $\phi^g$ is the unique homeomorphism associated to the coset $w^g_V$.\\
Thus it follows that
$I(\underline{1}_{Sh(\mv(\mh))})=\underline{1}_{Sh(\Lambda(\ps{G/G_F}))}$
\end{Proof}
We now check whether $I$ preserves the initial object. We recall that
the initial object in $Sh(\mv(\mh))$ is simply the sheaf
$\underline{O}_{Sh(\mv(\mh))}$ which assigns to each element $V$
the empty set $\{\emptyset\}$. We then have \be
I(\underline{O}_{Sh(\mv(\mh))})_{w_V^g}:=(\underline{O}_{Sh(\mv(\mh))})_{\phi^g(p_J(\phi^g))}=\{\emptyset\}
\ee
where $\phi^g\in Hom(\down V, \mv(\mh))$ is the unique homeomorphism associated with the coset $w^g_V$.\\
It follows that: \be
I(\underline{O}_{Sh(\mv(\mh))})=\underline{O}_{Sh(\Lambda(\ps{G/G_F}))}
\ee

From the above proof it transpires that the reason the functor $I$ preserves monic, epic, terminal object, and initial object is manly due to the fact that the action of $I$ is defined component-wise as $(I(\underline{A}))_{\phi}:=\underline{A}_{\phi(V)}$ for $\phi\in Hom(\downarrow V,\mv(\mh))$. In particular, it can be shown that $I$ preserves all limits and colimits.
\begin{Theorem}\label{the:lim}
The functor $I$ preserves limits.
\end{Theorem}

In order to prove the above theorem we first of all have to recall some general results and definitions. To this end consider two categories $\mc$ and $\md$, such that there exists a functor between them $F:\mc\rightarrow \md$. For a small index category $J$, we consider diagrams of type $J$ in both $\mc$ and $\md$, i.e. elements in $\mc^J$ and $\md^J$, respectively. The functor $F$ then induces a functor  between these diagrams as follows:
\ba
F^J:\mc^J&\rightarrow& \md^J\\
A&\mapsto&F^J(A)
\ea
such that $(F^J(A))(j):=F(A(j))$. Therefore, if limits of type $J$ exist in $\mc$ and $\md$ we obtain the diagram
\[\xymatrix{
\mc^J\ar[rr]^{\lim_{\leftarrow J}}\ar[dd]_{F^J}&&\mc\ar[dd]^{F}\\
&&\\
\md^J\ar[rr]_{\lim_{\leftarrow J}}&&\md\\
}\]
where the map
\ba
\lim_{\leftarrow J}:\mc^J&\rightarrow&\mc\\
A&\mapsto&\lim_{\leftarrow J}(A)
\ea
assigns, to each diagram $A$ of type $J$ in $\mc$, its limit $\lim_{\leftarrow J}(A)\in \mc$.
By the universal properties of limits we obtain the \emph{natural transformation} 
\be
\alpha_{J}:F\circ \lim_{\leftarrow J}\rightarrow \lim_{\leftarrow J}\circ F^J
\ee
We then say that $F$ preserves limits if $\alpha_J$ is a \emph{natural isomorphisms}. \\
For the case at hand, in order to show that the functor $I$ preserves limits we need to show that there exists a map
\be
\alpha_J:I\circ \lim_{\leftarrow J}\rightarrow \lim_{\leftarrow J}\circ I^J
\ee
which is a natural isomorphisms. Here $I^J$ represents the map
\ba
I^J:(Sh(\mv(\mh))\Big)^J&\rightarrow& \Big(Sh(\Lambda(\ps{G/G_F}))\Big)^J\\
A&\mapsto&I^J(A)
\ea
where $(I^J(A)(j))_{\phi}:=I(A(j))_{\phi}$.

\noindent
The proof of $\alpha_J$ being a natural isomorphisms will utilise a result derived in \cite{topos7} where it is shown that for any diagram $A:J\rightarrow \mc^{\md}$ of type $J$ in $\mc^{\md}$ the following isomorphisms holds
\be\label{equ:dual}
\Big( \lim_{\leftarrow J} A\Big)D\simeq \lim_{\leftarrow J}A_D\;\forall\; D\in \md
\ee
where $A_D:J\rightarrow \mc$ is a diagram in $\md$.
With these results in mind we are now ready to prove theorem \ref{the:lim}
\begin{Proof}
Let us consider a diagram $A:J\rightarrow Sets^{\mv(\mh)}$ of type $J$ in $Sets^{\mv(\mh)}$:
\ba
A:J&\rightarrow& Sets^{\mv(\mh)}\\
j&\mapsto&A(j)
\ea
where $A(j)(V):=A_V(j)$ for $A_V:j\rightarrow Sets$ a diagram in $Sets$. Assume that $L$ is a limit of type $J$ for $A$, i.e. $L:\mv(\mh)\rightarrow Sets$ such that $\lim_{\leftarrow J}A=J$. We then construct the diagram
\[\xymatrix{
\Big(Sets^{\mv(\mh)}\Big)^J\ar[rr]^{\lim_{\leftarrow J}}\ar[dd]_{I^J}&&Sets^{\mv(\mh)}\ar[dd]^{I}\\
&&\\
\Big(Sets^{\Lambda(\ps{G/G_F})}\Big)^J\ar[rr]_{\lim_{\leftarrow J}}&&Sets^{\Lambda(\ps{G/G_F})}\\
}\]
and the associated natural transformation
\be
\alpha_J:I\circ \lim_{\leftarrow J}\rightarrow \lim_{\leftarrow J}\circ I^J
\ee
For each diagram $A:J\rightarrow Sets^{\mv(\mh)}$ and $\phi\in\Lambda(\ps{G/G_F})$ we obtain
\be
\Big(I\circ \lim_{\leftarrow J}(A)\Big)_{\phi}=\Big(I\big(\lim_{\leftarrow J}A\big)\Big)_{\phi}:=\big(\lim_{\leftarrow J}A\big)_{\phi(V)}\simeq \lim_{\leftarrow J} A_{\phi(V)}
\ee
where $A_{\phi(V)}:J\rightarrow Sets$, such that $A_{\phi(V)}(j)=A(j)(\phi{V})$\footnote{Recall that $A:J\rightarrow Sets^{\mv(\mh)}$ is such that $A_{V}(j)=A(j)(V)$, therefore $\big(I(A(j))\big)_{\phi}:=A(j)_{\phi(V)}=A_{\phi(V)}(j)$}

On the other hand
\be
\Big(\big(\lim_{\leftarrow J}\circ I^J\big)A\Big)_{\phi}=\Big(\lim_{J}(I^J(A))\Big)_{\phi}\simeq \lim_{\leftarrow J}(I^J(A))_{\phi}=\lim_{\leftarrow J}A_{\phi(V)}
\ee
where 
\ba
I^J(A):J&\rightarrow& Sets^{\Lambda(\ps{G/G_F})}\\
j&\mapsto&I^J(A)(j)
\ea
such that for all $\phi\in \Lambda(\ps{G/G_F})$ we have $\big(I^J(A(j))\big)_{\phi}=\big(I(A(j))\big)_{\phi}=A(j)_{\phi(V)}$.

\noindent
It follows that 
\be
I\circ \lim_{\leftarrow J}\simeq \lim_{\leftarrow J}\circ I^J
\ee
\end{Proof}
Similarly one can show that
\begin{Theorem}
The functor $I$ preserves all colimits
\end{Theorem}
Since colimits are simply duals to the limits, the proof of this theorem is similar to the proof given above. However, for completeness sake we will, nonetheless, report it here.
\begin{Proof}
We first of all construct the analogue of the diagram above:
\[\xymatrix{
\Big(Sets^{\mv(\mh)}\Big)^J\ar[rr]^{\lim_{\rightarrow J}}\ar[dd]_{I^J}&&Sets^{\mv(\mh)}\ar[dd]^{I}\\
&&\\
\Big(Sets^{\Lambda(\ps{G/G_F})}\Big)^J\ar[rr]_{\lim_{\rightarrow J}}&&Sets^{\Lambda(\ps{G/G_F})}\\
}\]
where $\lim_{\rightarrow J}:\Big(Sets^{\mv(\mh)}\Big)^I\rightarrow Sets^{\mv(\mh)}$ represents the map which assigns colimits to all diagrams in $ \Big(Sets^{\mv(\mh)}\Big)^I$.

\noindent
We now need to show that the associated natural transformation
\be
\beta_J:I\circ \lim_{\rightarrow J} \rightarrow \lim_{\rightarrow J}\circ I^J
\ee
is a natural isomorphisms.

\noindent
For any diagram $A\in \Big(Sets^{\mv(\mh)}\Big)^I$ and $\phi\in \Lambda(\ps{G/G_F})$ we compute
\be
\Big(I\circ \lim_{\rightarrow J}(A)\Big)_{\phi}=\Big(I( \lim_{\rightarrow J}A)\Big)_{\phi}=\Big(\lim_{\rightarrow J}A\Big)_{\phi(V)}\simeq \lim_{\rightarrow J}A_{\phi(V)}
\ee
where $\Big(\lim_{\rightarrow J}A\Big)_{\phi(V)}\simeq \lim_{\rightarrow J}A_{\phi(V)}$ is the dual of \ref{equ:dual}.
On the other hand 
\be
\Big(( \lim_{\rightarrow J}\circ I^J)(A)\Big)_{\phi}=\Big( \lim_{\rightarrow J}(I^J(A))\Big)_{\phi}\simeq \lim_{\rightarrow J}(I^J(A))_{\phi}=\lim_{\rightarrow J}A_{\phi(V)}
\ee
It follows that indeed $\beta_J$ is a natural isomorphisms.
\end{Proof}

\section{The Adjoint Pair }
It is a standard result that, given a map $f : X\rightarrow Y$
between topological spaces $X$ and $Y$, we obtain a
geometric morphisms
 \ba
f^*: Sh(Y )&\rightarrow& Sh(X)\\
f_* : Sh(X)&\rightarrow& Sh(Y ) \ea and we know that $f^*\dashv
f_*$, i.e., $f^*$ is the left-adjoint of $f_*$. If $f$ is an
etal\'e map, however, there also exists the left adjoint $f!$ to
$f^*$, namely \be f! : Sh(X)\rightarrow Sh(Y ) \ee
with $f!\dashv f^*\dashv f_*$.\\
In the appendix we will show that \be f!(p_A : A\rightarrow X) =
f\circ p_A : A\rightarrow Y \ee so that we combine the etal\'e
bundle $p_A:A\rightarrow X$ with the etal\'e map $f : X\rightarrow
Y$ to give the
etal\'e bundle $f\circ p_A : A\rightarrow Y$. Here we have used the fact that sheaves can be defined in terms of etal\'e bundles. In fact in previous lectures it was shown that there exists an equivalence of categories $Sh(X)\simeq Etale(X)$ for any topological space $X$.

Given a map $\alpha:A\rightarrow B$ of etal\'e bundles over $X$,
we obtain the map $f!(\alpha) : f!(A)\rightarrow f!(B)$ which is
defined as follows. We start with the collection of fibre maps $\alpha_x :
A_x\rightarrow B_x$,  $x\in X$, where $A_x := p^{-1}A (\{x\})$.
Then, for each $y\in Y$ we want to define the maps $f!(\alpha)_y :
f!(A)_y\rightarrow f!(B)_y$,  i.e., $f!(\alpha)_y : p^{-1}\big( A
(f^{-1}\{y\})\big)\rightarrow p^{-1}\big( B (f^{-1}\{y\})\big)$.
This are defined as \be\label{equ:shreack} f!(\alpha)_y(a) :=
\alpha_{p_A(a)}(a) \ee
for all $a\in f!(A)_y = p^{-1}\big( A (f^{-1}\{a\})\big)$.

For the case of interest we obtain the left adjoint functor
$p_J!:Sh(\Lambda(\ps{G/G_F}))\rightarrow Sh(\mv_f(\mh))$ of
$p_J^*:Sh(\mv_f(\mh))\rightarrow Sh(\Lambda(\ps{G/G_F}))$. The
existence of such a functor enables us to define the composite
functor \be F:=p_J!\circ I:Sh(\mv(\mh))\rightarrow Sh(\mv_f(\mh))
\ee Such a functor sends all the original sheaves we had defined
over $\mv(\mh)$ to new sheaves over $\mv_f(\mh)$. Thus, denoting
the sheaves over $\mv_f(\mh)$ as $\underline{\breve{A}}$ we have
\be \breve{\us}:=F(\us)=p_J!\circ I(\us) \ee

What happens to the terminal object? Given
$\underline{1}_{\mv(\mh)}$ we obtain \be
F(\underline{1}_{\mv(\mh)})=p_J!\circ
I(\underline{1}_{\mv(\mh)})=p_J!(\underline{1}_{\Lambda(\ps{G/G_F})})
\ee Now the etal\'e bundle associated to the sheaf
$\underline{1}_{\Lambda(\ps{G/G_F})}$ is
$p_1:\Lambda(\{*\})\rightarrow \Lambda(\ps{G/G_F}))$ where
$\Lambda(\{*\})$ represents the collection of singletons, one for
each $w_V^g\in \Lambda(\ps{G/G_F})$. Obviously the etal\'e bundle
$p_1:\Lambda(\{*\})\rightarrow \Lambda(\ps{G/G_F}))$ is nothing
but $\Lambda(\ps{G/G_F})$. Thus by applying the definition of
$p_J!$ we then get \be
p_J!(\underline{1}_{\Lambda(\ps{G/G_F})})=\ps{G/G_F} \ee
It follows that the functor $F$ does not preserve the terminal object therefore it can not be a right adjoint. In fact we would like $F$ to be left adjoint. However so far that does not seem the case. 
%
We have seen above that the functor $I$ preserves \emph{colimits} (initial object) and \emph{limits}. Since $F=p_J!\circ I$ and $p_J!$ is left adjoint thus preserves \emph{colimits}, it follows that $F$ will preserve \emph{colimits}.

Of particular importance to us is the following: each object
$\underline{A}\in Sh(\mv(\mh))$ has associated to it the unique
arrow $!\underline{A}:\underline{A}\rightarrow
\underline{1}_{\mv(\mh)}$. This arrow is epic thus
$F(!\underline{A}):F(\underline{A})\rightarrow
F(\underline{1}_{\mv(\mh)})$ is also epic. In particular we obtain
\ba
F(!\underline{A}):F(\underline{A})&\rightarrow &F(\underline{1}_{\mv(\mh)})\\
\underline{\breve{A}}&\rightarrow&\ps{G/G_F} \ea such that for
each $V\in \mv(\mh)$ we get \ba
\underline{\breve{A}}_V&\rightarrow&(\ps{G/G_F})_V\\
\coprod_{w^g_V\in(\ps{
G/G_F})_V}\underline{A}_{w^g_V}&\rightarrow&G/G_{FV} \ea 

However,
since we are considering sub-objects of the state object presheaf
$\breve{\us}$ we would like the $F$ functor to also preserve monic
arrows. And indeed it does.
\begin{Lemma}
The functor $F:Sh(\mv(\mh))\rightarrow Sh(\mv_f(\mh))$ preserves
monics.
\end{Lemma}
\begin{Proof}
Let $i:\underline{A}\rightarrow \underline{B}$ be a monic arrow in
$Sh(\mv(\mh))$, then we have that \be F(i) =p_J!(I(i)) \ee
However, the $I$ functor preserves monics, as a consequence $I(i)$
is monic in $Sh(\Lambda(\ps{G/G_F}))$.
\\
Moreover, from the defining equation \ref{equ:shreack}, it follows
that if $f : X\rightarrow Y$ is etal\'e and $p_A:A\rightarrow X$
is etal\'e then, since $i : A\rightarrow B$ is monic then so is
$f!(i) : f!(A)\rightarrow f!(B)$. Therefore applying this
reasoning to our case it follows that $F(i) =p_J!(I(i))$ is monic.
\end{Proof}
\section{From Sheaves over $\mv(\mh)$ to Sheaves over $\mv(\mh_f)$}
Now that we have defined the functor $F$ we will map
all the sheaves in our original formalism ($Sh(\mv(\mh))$) to
sheaves over $\mv_f(\mh)$. We will then analyse how the truth
values behave under such mappings.
\subsection{Spectral Sheaf}
Given the spectral sheaf $\us\in Sh(\mv(\mh))$ we define the
following: \be \breve{\us}:=F(\us)=p_{I}!\circ I (\us) \ee This will
be our new spectral sheaf. The definition given below will be in
terms of the corresponding presheaf (which we will still denote
$\breve{\us)}$), where we have used the correspondence between
sheaves and presheaves induced by the fact that the base category
is a poset (see lecture 14)
\begin{Definition}
The spectral presheaf $\breve{\us}$ is defined on
\begin{itemize}
\item [--] Objects: For each $V\in\mv_f(\mh)$ we have
\be
\breve{\us}_V:=\coprod_{w^{g_i}_V\in\Lambda(\ps{G/G_F})_V}\us_{w^{g_i}_V}\simeq\coprod_{\phi_i\in
Hom(\down V,\mv(\mh))}\us_{\phi(V)} \ee which represents the disjoint
union of the Gel'fand spectrum of all algebras related to $V$ via
a group transformation
\item[--] Morphisms: Given a morphism $i:V^{'}\rightarrow V$, ($V^{'}\subseteq V)$ in $\mv_f(\mh)$ the corresponding spectral presheaf morphism is
\ba
\breve{\us}(i_{V^{'}V}):\breve{\us}_{V}&\rightarrow& \breve{\us}_{V^{'}}\\
\coprod_{\phi_i\in
Hom(\down V,\mv(\mh))}\us_{\phi_i(V)}&\rightarrow&\coprod_{\phi_j\in
Hom(\down V^{'},\mv(\mh))}\us_{\phi(V^{'})} \ea
such that given $\lambda\in\us_{\phi_i(V)}$ we obtain $\breve{\us}(i_{V^{'}V})(\lambda):=\us_{\phi_i(V),\phi_j(V^{'})}\lambda=\lambda_{|\phi_j(V^{'})}$\\
Thus in effect $\breve{\us}(i_{V^{'}V})$ is actually a co-product of morphisms $\us_{\phi_i(V),\phi_j(V^{'})}$, one for each 
$\phi_i \in Hom(\down V,\mv(\mh))$.
\end{itemize}
\end{Definition}
From the above definition it is clear that the new spectral sheaf
contains the information of all possible representations of a
given abelian von-Neumann algebra at the same time. It is such an
idea that will reveal itself fruitful when considering how
quantisation is defined in a topos.
\subsubsection{Topology on The State Space}
We would now like to analyse what kind of topology the sheaf
$\breve{\us}:=F(\us)$ has. We know that for each $V\in\mv_f(\mh)$
we obtain the collection $\coprod_{w^{g_i}_V\in
G/G_{FV}}\us_{w^{g_i}_V}$, where each
$\us_{w^{g_i}_V}:=\us_{\phi^{g_i}(V)}$ is equipped with the spectral
topology. Thus, similarly as was the case of the sheaf $\us\in
Sh(\mv(\mh))$, we could equip $\breve{\us}$ with the disjoint
union topology or with the spectral topology. In order to
understand the spectral topology we should recall that the functor
$F:Sh(\mv(\mh))\rightarrow Sh(\mv_f(\mh))$ preserves monics, thus
if $\underline{S}\subseteq \us$, then
$\breve{\underline{S}}:=F(\underline{S})\subseteq
\breve{\us}:=F(\us)$. We can then define the spectral topology on
$\breve{\us}$ as follows
\begin{Definition}
The spectral topology on $\breve{\us}$ has as basis the collection
of clopen sub-objects $\breve{\underline{S}}\subseteq \breve{\us}$
which are defined for each $V\in \mv_f(\mh)$ as \be
\breve{\underline{S}}_V:=\coprod_{w_{V}^{g_i}\in
G/G_{FV}}\underline{S}_{w_{V}^{g_i}}=\coprod_{\phi_i\in Hom(\down V,
\mv(\mh))}\underline{S}_{\phi_i(V)} \ee
\end{Definition}
From the definition it follows that on each element $\us_{w^{g_i}_V}$ of the stalks we retrieve the standard spectral topology.

It is easy to see that the map $p:\coprod_{w^{g_i}_V\in \Lambda(\ps{G/G_F})}\us_{w^{g_i}_V}\rightarrow \mv_f(\mh)$ is continuous since
$p^{-1}(\downarrow V):=\coprod_{w^{g_i}_{V^{'}}\in \downarrow w^{g_i}_V|\forall w^{g_i}_V\in G/G_{FV}}\us_{w^{g_i}_{V^{'}}}$ is the clopen sub-object which has value $\coprod_{w^{g_i}_{V^{'}}\in G/G_{FV^{'}}}\us_{w^{g_i}_{V^{'}}}$ at each context $V^{'}\in \downarrow V$ and $\emptyset $ everywhere else.

Similarly, as was the case for the topology on $\us\in
Sh(\mv(\mh))$, the spectral topology defined above is weaker than
the product topology and it has the advantage that if takes into account both the `vertical' topology on the fibres and the `horizontal' topology on the base space $\mv_f(\mh)$. 

A moment of thought will reveal that also with respect to the disjoint union topology the map $p$ is continuous, however because of the above argument, from now on we will use the spectral
topology on the spectral presheaf.

\subsection{Quantity Value Object}
We are now interested in mapping the quantity value objects
$\underline{\Rl}^{\leftrightarrow}\in Sh(\mv(\mh))$ to an object
in $Sh(\mv_f(\mh))$ via the $F$ functor. We thus define:
\begin{Definition}
The quantity value objects
$\breve{\underline{R}}^{\leftrightarrow}:=F(\underline{\Rl}^{\leftrightarrow})=p_{I}!\circ
I(\underline{\Rl}^{\leftrightarrow})$ is an $\Rl$-valued presheaf
of order-preserving and order-reversing functions on $\mv_f(\mh)$
defined as follows:
\begin{itemize}
\item[--] On objects $V\in\mv_f(\mh)$ we have
\be (F(\underline{\Rl}^{\leftrightarrow}))_V:=\coprod_{\phi_i\in
Hom(\down V,\mv(\mh))}\underline{\Rl}^{\leftrightarrow}_{\phi_i(V)} \ee
where each \be
\underline{\Rl}^{\leftrightarrow}_{\phi_i(V)}:=\{(\mu, \nu)|\mu\in
OP(\downarrow \phi_i(V), \Rl)\;,\; \mu\in OR(\downarrow \phi_i(V),
\Rl),\; \mu\leq\nu\} \ee The downward set $\downarrow \phi_i(V)$
comprises all the sub-algebras $V^{'}\subseteq \phi_i(V)$. The
condition $\mu\leq\nu$ implies that for all $V^{'}\in\down
\phi_i(V)$, $\mu(V^{'})\leq\nu(V^{'})$.
\item[--] On morphisms $i_{V^{'}V}:V^{'}\rightarrow V$ ($V^{'}\subseteq V)$ we get:
\ba
\breve{\underline{R}}^{\leftrightarrow}(i_{V^{'}V}):\breve{\underline{R}}^{\leftrightarrow}_V&\rightarrow& \breve{\underline{R}}^{\leftrightarrow}_{V^{'}}\\
\coprod_{\phi_i\in
Hom(\down V,\mv(\mh))}\underline{\Rl}^{\leftrightarrow}_{\phi_i(V)}&\rightarrow
&\coprod_{\phi_j\in
Hom(\down V^{'},\mv(\mh))}\underline{\Rl}^{\leftrightarrow}_{\phi_j(V^{'})}
\ea where for each element $(\mu,\nu)\in
\underline{\Rl}^{\leftrightarrow}_{\phi_i(V)}$ we obtain \ba
\breve{\underline{R}}^{\leftrightarrow}(i_{V^{'}V})(\mu,\nu)&:=&\underline{R}^{\leftrightarrow}(i_{\phi_i(V),\phi_j(V^{'})})(\mu,\nu)\\
&=&(\mu_{|\phi_i(V^{'})},\nu_{|\phi_j(V^{'})}) \ea where
$\mu_{|\phi_i(V^{'})}$ denotes the restriction of $\mu$ to
$\downarrow \phi_j(V^{'})\subseteq\downarrow\phi_i(V)$, and
analogously for $\nu_{|\phi_j(V^{'})}$.
\end{itemize}
\end{Definition}
\subsubsection{Topology on the Quantity Value Object}
We are now interested in defining a topology for our newly defined quantity value object $\breve{\underline{\Rl}}$. Similarly, as was done for the spectral sheaf, we define the set
\be
\mathcal{R}=\coprod_{V\in\mv_f(\mh)}\breve{\underline{\Rl}}^{\leftrightarrow}_V=\bigcup_{V\in\mv_f(\mh)}\{V\}\times\breve{\underline{\Rl}}^{\leftrightarrow}_V
\ee
where each $\breve{\underline{\Rl}}^{\leftrightarrow}_V:=\coprod_{\phi_i\in Hom(\down V, \mv(\mh))}\underline{\Rl}^{\leftrightarrow}_{\phi_i(V)}$.\\
The above represents a bundle over $\mv_f(\mh)$ with bundle map $p_{\mathcal{R}}:\mathcal{R}\rightarrow \mv_f(\mh)$ such that $p_{\mathcal{R}}(\mu, \nu)=V=p_J(\phi_i)$, where $V$ is the context such that $(\mu, \nu)\in\underline{\Rl}^{\leftrightarrow}_{\phi_i(V)}$. 
In this setting
$p^{-1}_{\mathcal{R}}(V)=\breve{\underline{\Rl}}^{\leftrightarrow}_{V}$
are the fibres of the map $p_{\mathcal{R}}$.

We would like to define a topology on $\mathcal{R}$ with the minimal require that the map $p_{\mathcal{R}}$ is continuous. We know that the category $\mv_f(\mh)$ has the Alexandroff topology whose basis open sets are of the form $\downarrow V$ for some $V\in\mv_f(\mh)$. Thus we are looking for a topology such that the pullback 
$p_{\underline{\Rl}}^{-1}(\downarrow V):=\coprod_{V^{'}\in\downarrow V}\underline{\breve{\Rl}}_{V^{'}}$ is open in $\mathcal{R}$. 

Following the discussion at the end of section 2.1 we know that each $\underline{\Rl}^{\leftrightarrow}$ is equipped with the discrete topology in which all sub-objects are open (in particular each $\underline{\Rl}^{\leftrightarrow}_V$ has the discrete topology).
Since the $F$ functor preserves monics, if $\underline{Q}\subseteq\underline{\Rl}^{\leftrightarrow}$ is open then $F(\underline{Q})\subseteq F(\underline{\Rl}^{\leftrightarrow})$ is open, where $F(\underline{Q}):=\coprod_{\phi_i\in Hom(\down V, \mv(\mh))}\underline{Q}_{\phi_i(V)}$. 

Therefore we define a sub-sheaf, $\underline{\breve{Q}}$, of 
$\breve{\underline{\Rl}}^{\leftrightarrow}$ to be \emph{open}
if for each $V\in\mv_f(\mh)$ the set $\underline{\breve{Q}}_V\subseteq \breve{\underline{\Rl}}_V$ is open, i.e., each $\underline{Q}_{\phi_i(V)}\subseteq \underline{\Rl}^{\leftrightarrow}_{\phi_i(V)}$ is open in the discrete topology on $\underline{\Rl}^{\leftrightarrow}_{\phi_i(V)}$. It follows that the sheaf $\breve{\underline{\Rl}}^{\leftrightarrow}$ gets induced the discrete topology in which all sub-objects are open. In this setting the `horizontal' topology on the base category $\mv_f(\mh)$ would be accounted for by the sheave maps.

For each $\downarrow V$ we then obtain the open set $p_{\underline{\Rl}}^{-1}(\downarrow V)$ which has value
$\underline{\breve{\Rl}}_{V^{'}}$ at contexts $V^{'}\in\downarrow
V$ and $\emptyset$ everywhere else
\subsection{Truth Values}
We now want to see what happens to the truth
values when they are mapped via the functor $F$. In particular,
given the sub-object classifier $\uom^{\mv(\mh)}\in Sh(\mv(\mh))$ we want to know
what $F(\uom^{\mv(\mh)})$ is. Since \ba
F(\uom^{\mv(\mh)})=p_J!\circ I (\uom^{\mv(\mh)}) \ea we first of
all need to analyse what $I (\uom^{\mv(\mh)})$ is. Applying the
definition for each $w^{g_i}_V\in \Lambda(\ps{G/G_F})$ we obtain 
\be (I(\uom^{\mv(\mh)}))_{w^{g_i}_V}:=\uom^{\mv(\mh)}_{\phi_i(V)} \ee Where
$\phi_i\in Hom(\down V, \mv(\mh))$ is the unique homeomorphism
associated to the equivalence class $w^{g_i}_V\in G/G_{FV}$. If we
then consider another element $w^{g_j}_V\in G/G_{FV}$, we then
have \be (I
(\uom^{\mv(\mh)}))_{w^{g_j}_V}:=\uom^{\mv(\mh)}_{\phi_j(V)} \ee
where now $\phi_i(V)\neq \phi_j(V))$. What this implies is
that once we apply the functor $p_J!$ to push everything down to
$\mv_f(\mh)$, the distinct elements $\uom^{\mv(\mh)}_{\phi_i(V)}$
and $\uom^{\mv(\mh)}_{\phi_j(V)}$ will be pushed down to the
same $V$, since both $\phi_i, \phi_j\in Hom(\down V, \mv(\mh))$. It
follows that, for every $V\in \mv_f(\mh)$, $F(\uom^{\mv(\mh)})$ is
defined as \be F(\uom^{\mv(\mh)})_V:=\coprod_{w^{g_i}_V\in
G/G_{FV}}\uom^{\mv(\mh)}_{w^{g_i}_V}\simeq\bigcup_{w^{g_i}_V\in
G/G_{FV}}\{w^{g_i}_V\}\times\uom^{\mv(\mh)}_{w^{g_i}_V}\simeq\bigcup_{\phi_i\in
Hom(\down V,\mv(\mh))}\{\phi_i\}\times \uom^{\mv(\mh)}_{\phi_i(V)}\ee

Thus it seems that for each $V\in \mv_f(\mh)$,
$F(\uom^{\mv(\mh)})_V$ assigns the disjoint union of the collection
of sieves for each algebra $V_{i}\in \mv(\mh)$ such that $V_{i}=\phi_i(V)$,
where $\phi_i$ are the unique homeomorphisms associated to each
$w^{g_i}_V\in G/G_{FV}$. This leads to the
following conjecture:
\begin{Conjecture}
$F(\uom^{\mv(\mh)})\simeq\ps{G/G_F}\times \uom^{\mv(\mh)}$
\end{Conjecture}
It should be noted that $\mv_f(\mh)\simeq \mv(\mh)$ since
$\mv_f(\mh)$ and  $\mv(\mh)$ are in fact the same categories only
that in the former there is no group action on it. Thus it also
follows trivially that $\uom^{\mv_f(\mh)}\simeq \uom^{\mv(\mh)}$. Having said that we can now prove the above conjecture
\begin{Proof}
For each $V\in \mv_f(\mh)$ we define the map \ba
i_V: F(\uom^{\mv(\mh)})_V&\rightarrow&G/G_{FV}\times \uom^{\mv(\mh)}_V\\
S&\mapsto&(w^{g_i}_V, l_{g_i^{-1}}S) \ea
where $S\in \uom^{\mv(\mh)}_{w^{g_i}_V}=\uom^{\mv(\mh)}_{\phi_i(V)}$ for $\phi_i\in Hom(\downarrow V, \mv(\mh))$ and $\phi_i(V):=l_{g_i}V$ while $l_{g_i^{-1}}S\in \uom^{\mv(\mh)}_V$.

Such a map is one to one since if $(w^{g_i}_V, l_{g_i^{-1}}S_1)=(w^{g_i}_V, l_{g_i^{-1}}S_2)$ then $l_{g_i^{-1}}S_1=l_{g_i^{-1}}S_2$ and $S_1=S_2$.
The fact that it is onto follows form the definition.

We now construct, for each $V\in \mv(\mh)$ the map \ba
j:G/G_{FV}\times \uom^{\mv(\mh)}_V&\rightarrow&F(\uom^{\mv(\mh)})_V\\
(w^{g_i}_V, S)&\mapsto&l_{g_i}(S) \ea 
where $S\in \uom^{\mv(\mh)}_V$ and $l_{g_i}S\in \uom^{\mv(\mh)}_{l_{g_i}V}$ for $l_{g_i}V=\phi_i(V)$ thus $l_{g_i}S\in \uom^{\mv(\mh)}_{w^{g_i}_V}$

A moment of thought reveals that
$j=i^{-1}$
\end{Proof}
From the above result we obtain the following conjecture:
\begin{Conjecture}\label{con:trueiso}
$\uom^{\mv_f(\mh)}\simeq F(\uom^{\mv(\mh)})/\ug$
\end{Conjecture}
Before proving the above conjecture we, first of all, need to define what a quotient presheaf is. This is simply a presheaf in which the quotient is computed context wise, thus, in the case at hand the quotient is computed for each $V\in\mv_f(\mh)$. In order to understand the definition of the quotient presheaf we will analyse what the equivalence classes look like.

We already know that for presheaves over $\mv_f(\mh)$ the group
action is at the level of the base category $\Lambda(\ps{G/G_F})$.
In particular for each $g\in G$ we have \be
(l^*_g(\uom^{\mv(\mh)}))_{\phi (V)}:=\uom^{\mv(\mh)}_{l_g(\phi (V))}
\ee where $\phi\in Hom(\down V, \mv(\mh))$. Therefore by
defining for each $V\in \mv(\mh)$ the equivalence relation on
$\coprod_{\phi_i\in
Hom(\down V,\mv(\mh))}\uom^{\mv(\mh)}_{\phi_i(V)}=:(F(\uom^{\mv(\mh)}))_V$ by
the action of $G$, the elements in $(F(\uom^{\mv(\mh)}))_V/G_V=\Big(\coprod_{\phi_i\in
Hom(\down V,\mv(\mh))}\uom^{\mv(\mh)}_{\phi_i(V)}\Big)/G$ will be
equivalence classes of sieves, i.e., \be [S_i]:=\{l_g(S_i)|g\in
G\} \ee
for each $S_i\in \uom^{\mv(\mh)}_{\phi_i(V)}\in \coprod_{{\phi_i\in Hom(\down V,\mv(\mh))}}\uom^{\mv(\mh)}_{\phi_i(V)}$. In the above we used the action of the group $G$ on sieves which is defined as $l_gS:=\{l_gV^{'}|V^{'}\in S\}$.
We are now ready to define the presheaf $F(\uom^{\mv(\mh)})/\ug$.

\begin{Definition}
The Presheaf $F(\uom^{\mv(\mh)})/\ug$ is defined:
\begin{itemize}
\item On objects: for each context $V\in \mv_f(\mh)$ we have the object 
\be
(F(\uom^{\mv(\mh)}))_V/G_V:=\Big(\coprod_{\phi_i\in Hom(\down V,\mv(\mh))}\uom^{\mv(\mh)}_{\phi_i(V)}\Big)/(G)\ee
whose elements are equivalence classes of sieves $[S_i]$, i.e., $S_1, S_2\in [S_i]$ iff $S_1:=\{l_gS_2|g\in G\}$ and $S_2\in\uom^{\mv(\mh)}_{\phi_i(V)}$ and $S_1=\uom^{\mv(\mh)}_{l_g\phi_i(V)}$, i.e. each equivalence class will contain only one sieve for each algebra. This definition of equivalence condition follows from the fact that the group action of $G$ moves each set $\uom^{\mv(\mh)}_{\phi_i(V)}$ to another set $\uom^{\mv(\mh)}_{l_g\phi_i(V)}$ in the same stork $F(\uom^{\mv(\mh)}))_V$, i.e. the group action is at the level of the base category $\Lambda(\ps{G/G_F})$.
\item On morphisms: for each $V^{'}\subseteq V$ we then have the corresponding morphisms
\ba
\alpha_{VV^{'}}:\Big(\coprod_{\phi_i\in Hom(\down V,\mv(\mh))}\uom^{\mv(\mh)}_{\phi_i(V)}\Big)/(G)&\rightarrow& \Big(\coprod_{\phi_j\in Hom(\down V^{'},\mv(\mh))}\uom^{\mv(\mh)}_{\phi_j(V^{'})}\Big)/(G)\\
\;[S]&\mapsto&\alpha_{VV^{'}}([S]):=[S\cap V^{'}] \ea 

where
$[S\cap V^{'}]:=\{l_g(S\cap V^{'})|g\in G\}$,  and we choose as the representative for the equivalence class $S\in\uom^{\mv(\mh)}_V$ for $V=\phi_i(V)$ where $\phi_i\in Hom(\downarrow V, \mv(\mh))$ is associated to some $g\in G_V$
\end{itemize}
\end{Definition}
We can now prove the above conjecture (\ref{con:trueiso}), i.e.,
we will show that the functor \be \beta:\uom^{\mv_f(\mh)}\rightarrow
F(\uom^{\mv(\mh)})/(\ug) \ee
is an isomorphism.

In particular for each context $V\in \mv(\mh)$ we define \ba
\beta_V:\uom^{\mv_f(\mh)}_V&\rightarrow& F(\uom^{\mv(\mh)})_V/(\ug)_{V}\\
S&\mapsto&[S] \ea
where $[S]$ denotes the equivalence class to which the sieve $S$ belongs to, i.e., $[S]:=\{l_g S|g\in G\}$.

First we need to show that $\beta$ is indeed a functor,
i.e., we need to show that the following diagram commutes
\[\xymatrix{
\uom^{\mv_f(\mh)}_V\ar[rr]^{\beta_V}\ar[dd]_{\uom^{\mv_f(\mh)}(i_{V^{'}V})}&&F(\uom^{\mv(\mh)})_V/G\ar[dd]^{\alpha_{VV^{'}}}\\
&&\\
\uom^{\mv_f(\mh)}_{V^{'}}\ar[rr]^{\beta_{V^{'}}}&&F(\uom^{\mv(\mh)})_V/G\\
}\] Thus for each $S$ we obtain for one direction
\be\big(\beta_{V^{'}}\circ\uom^{\mv_f(\mh)}(i_{V^{'}V})\big)(S)=\beta_{V^{'}}(S\cap
V^{'})=[S\cap V^{'}] \ee 
where the first equality follows from the definition of the sub-object classifier $\uom^{\mv_f(\mh)}$ \cite{andreas5}.

Going the opposite direction we get
\be \big(\alpha_{VV^{'}}\circ
\beta_V\big)S=\alpha_{VV^{'}}[S]=[S\cap V^{'}] \ee It follows that
indeed the above diagram commutes. Now that we have showed that
$\beta$ is a functor we need to show that it is an isomorphisms. We consider each individual component $\beta_V$, $V\in  \mv_f(\mh)$.
\begin{enumerate}
\item
\emph{The map $\beta_V$ is one-to-one}.

Given $S_1, S_2\in\uom^{\mv_f(\mh)}_{V}$, 
if $\beta_V(S_1)=\beta_V(S_2)$ then $[S_1]=[S_2]$, thus both $S_1$
and $S_2$ belong to the same equivalence class. Each equivalence
class is of the form $[S]=\{l_gS|g\in G\}$, therefore $S_1=l_gS_2$ for some $g\in G$.
However, the definition of the equivalence classes of sieves implied that for each equivalence class there is one and only one sieve for each algebra. Thus if $[S_1]=[S_2]$ and both $S_1, S_2\in\uom^{\mv_f(\mh)}_V$, then $S_1=S_2$.

%
%
%
\item
\emph{The map $\beta_V$ is onto}. This follows at once from the definition.
\item
\emph{The map $\beta_V$ has an inverse}.

We now need to define an inverse. We choose \be
\gamma:F(\uom^{\mv(\mh)})/G\rightarrow \uom^{\mv_f(\mh)} \ee such
that for each context we get \ba
\gamma_V:F(\uom^{\mv(\mh)})_V/G&\rightarrow& \uom^{\mv_f(\mh)}_V\\
\;[S]&\mapsto&[S]\cap V \ea where $[S]\cap V:=\{l_g(S)\cap V|g\in
G\}$ represents the only sieve in the equivalence class which belongs to $\uom^{\mv_f(\mh)}_V$ . We first of all have to show that this is indeed a functor.
Thus we need to show that, for each $V^{'}\subseteq V$ the following diagram commutes
\[\xymatrix{
F(\uom^{\mv(\mh)})_V/G\ar[dd]_{\alpha_{VV^{'}}}\ar[rr]^{\gamma_V}&&\uom^{\mv_f(\mh)}_V\ar[dd]^{\uom^{\mv_f(\mh)}(i_{V^{'}V})}\\
&&\\
F(\uom^{\mv(\mh)})_V/G\ar[rr]^{\gamma_{V^{'}}}&&\uom^{\mv_f(\mh)}{V^{'}}
}\] Chasing the diagram around for each $S$ we obtain \be
\uom^{\mv_f(\mh)}(i_{V^{'}V})\circ
\gamma_V([S])=\uom^{\mv_f(\mh)}(i_{V^{'}V})([S]\cap V)=([S]\cap
V)\cap V^{'}=[S]\cap V^{'} \ee

On the other hand we have \be \gamma_{V^{'}}\circ
\alpha_{VV^{'}}[S]=\gamma_{V^{'}}[S\cap V^{'}]=[S\cap V^{'}]\cap
V^{'}=[S]\cap V^{'} \ee
where the last equality follows since $[S\cap V^{'}]\cap V^{'}:=\{l_g(S\cap V^{'})|g\in G\}\cap V^{'}$ and the only sieve in $[S]$ belonging to $\uom^{\mv_f(\mh)}_{V^{'}}$ is $S\cap V^{'}$. Therefore the map $\gamma$ is a functor. 

It now remains to show that, for each $V\in \mv_f(\mh)$ and each
$S\in\uom^{\mv(\mh)}_V$, $\gamma_V$  is the inverse of $\beta_V$.
Thus 
\be \gamma_V\circ\beta_V(S)=\gamma_V([S])=[S]\cap V=S \ee
where the last equality follows from the fact that in each
equivalence class of sieves there is one and only one referred to
each context $l_gV$. On the other hand we have \be \beta_V\circ
\gamma_V([S])=\beta_V\circ ([S]\cap V)=\beta_V(S)=[S] \ee
\end{enumerate}
The functor $\beta$ is indeed an isomorphism.

\subsection{Group Action on the New Sheaves}
We would now like to analyse what the group action on the new
sheaves is. In particular we will show how the action of the group
$\ug$ on the sheaves define on $\mv_f(\mh)$ via the $F$ functor
will not induce twisted sheaves.
\subsubsection{Spectral Sheaf}
The action of the group $\ug$ on the new spectral sheaf
$\breve{\us}:=F(\us)$ is given by the following map: \be \ug\times
\breve{\us}\rightarrow\breve{\us} \ee defined for each context
$V\in\mv_f(\mh)$ as \ba\label{ali:gonvalue}
\ug_V\times \breve{\us}_V&\rightarrow&\breve{\us}_V\\
(g, \lambda)&\mapsto&l_g\lambda \ea where
$\breve{\us}_V:=\coprod_{\phi_i\in Hom(\down V,
\mv(\mh))}\us_{\phi_i(V)}$ such that if $\lambda\in \us_{\phi_i(V)}$
we define $l_g\lambda\in l_g\us_{\phi_i(V)}:=\us_{l_g(\phi_i(V))}$ by
\be (l_g(\lambda))\hat{A}:=\langle
\lambda,\hat{U}(g)^{-1}\hat{A}\hat{U}(g)\rangle \ee
for all $g\in G$, $\hat{A}\in V_{sa}$(self adjoint operators in $V$) and $V\in \mv(\mh)$.

However from the definition of $\breve{\us}$, both $ \us_{\phi_i(V)}$ and $\us_{l_g(\phi_i(V))}$ belong to the same stalk, i.e., belong to $\breve{\us}_V$.

\noindent
We thus obtain a well defined group action which does not induce twisted presheaves. 

We would now like to check whether such a group action is
continuous with respect to the spectral topology, i.e., if the map
\be \rho:\ug\times \breve{\us}\rightarrow\breve{\us} \ee is
continuous. In particular we want to check if for each
$V\in\mv_f(\mh)$ the local component \be \rho_V:\ug_V\times
\breve{\us}_V\rightarrow\breve{\us}_V \ee is continuous, i.e., if
$\rho^{-1}_V\breve{\underline{S}}_V=\rho^{-1}_V\Big(\coprod_{\phi_i\in
Hom(\down V, \mv(\mh))}\underline{S}_{\phi_i(V)}\Big)$ is open for
$\breve{\underline{S}}_V$ open. \ba
\rho^{-1}_V\Big(\coprod_{\phi_i\in Hom(\down V, \mv(\mh))}\underline{S}_{\phi_i(V)}\Big)&=&\{(g_j, \underline{S}_{\phi_i(V)})|l_{g_j}(\underline{S}_{\phi_i(V)})\in \underline{\breve{S}}_V\}\\
&=& (G, \underline{\breve{S}}_V) \ea where
$l_{g_j}(\underline{S}_{\phi_i(V)}):=\underline{S}_{l_{g_j}\phi_i(V)}=\underline{S}_{l_{g_j}(\phi_i(V))}$.
It follows that the action is continuous.

Moreover it seems that the sub-objects $\breve{\underline{S}}$ actually remain invariant under the group action. In fact, for each $V\in \mv_f(\mh)$, $\breve{\underline{S}}_V=\coprod_{\phi_i\in Hom(\down V, \mv(\mh))}\underline{S}_{\phi_i(V)}$ where the set $Hom(\down V, \mv(\mh))$ contains all $G$ related homeomorphisms, i.e. all $l_{g_j}(\phi_i)\;\forall \; g_j\in G$, ($l_{g_j}(\phi)(V):=l_{g_j}(\phi(V)))$.

\noindent
It follows that the sub-objects $\breve{\underline{S}}\subseteq\breve{\us}$ are invariant under the group action.

This is an important result when considering propositions which
are identified with clopen sub-objects coming from
daseinisation. In this context the group action is defined, for each
$V\in\mv_f(\mh)$, as: \ba
\underline{G}_V\times \underline{\delta\breve{P}}_V&\rightarrow& \underline{\delta\breve{P}}_V\\
\underline{G}_V\times \coprod_{\phi_i\in Hom(\down V, \mv(\mh))}\delta^o(\hat{P})_{\phi_i(V)}&\rightarrow&\coprod_{\phi_i\in Hom(\down V, \mv(\mh))}\delta^o(\hat{P})_{\phi_i(V)}\\
(g,
\delta^o(\hat{P})_{\phi_i(V)})&\mapsto&\delta^o(\hat{U}_g\hat{P}\hat{U}_g^{-1})_{l_g(\phi_i(V))}
\ea Thus for each $g\in G$ we get a collection of transformations each similar to those obtained in the original formalism. However, since the effect
of such a transformation is to move the objects around within a
stalk, when considering the action of the entire $G$, the
stalk, as an entire set,  remains invariant, i.e., the collection
of local component of the propositions stays the same.

Moreover the fact that individual sub-objects $\breve{\underline{S}}\subseteq\breve{\us}$ are invariant under the group action, implies that the action $\ps{G}\times\breve{\us}\rightarrow \breve{\us}$ is not transitive. In fact the transitivity of the action of a group sheaf is defined as follows
\begin{Definition}
Given a group $\ps{G}$, we say that the action of $\ps{G}$ on any other sheaf $\underline{A}$ is transitive iff there are no invariant sub-objects of $A$.
\end{Definition}

Thus although the group actions moves the elements around in each stalk, it never moves elements in between different stalks, thus each sub-object is left invariant.
\subsubsection{Sub-object Classifier}
We now are interested in defining the group action on the
sub-object classifier $\uom^{\mv_f(\mh)}$. However, by definition,
there is no action on such object. The only action which could be
defined would be the action on $\breve{\uom}:=F(\uom^{\mv(\mh)})$.
In this case, for each $V\in \mv_f(\mh)$, we have \ba
\alpha_V:\ug_V\times \breve{\uom}_V&\rightarrow&\breve{\uom}_V\\
\ug_V\times \coprod_{w^{g_i}_V\in G/G_{FV}}\uom_{w^{g_i}_V}&\rightarrow&\coprod_{w^{g_i}_V\in G/G_{FV}}\uom_{w^{g_i}_V}\nonumber\\
\ug_V\times \coprod_{\phi_i\in Hom(\down V,\mv(\mh))}\uom_{\phi_i(V)}&\rightarrow&\coprod_{\phi_i\in Hom(\down V,\mv(\mh))}\uom_{\phi_i(V)}\nonumber\\
(g, S)&\mapsto&l_g(S) \ea
where $l_g(S):=\{l_gV|V\in S\}$. 

If $S\in \uom_{\phi_i(V)}\in \coprod_{\phi_i\in Hom(\down V,\mv(\mh))}\uom_{\phi_i(V)}$, then $l_g(S)$ is a sieve on $l_g\phi_i(V)$, i.e., $l_g(S)\in
\uom_{l_g\phi_i(V)}\in \coprod_{\phi_i\in Hom(\down V,\mv(\mh))}\uom_{\phi_i(V)}$. 

\noindent 
It follows that the action of the group $\ug$ is to move sieves around in each stalk but never to move sieves to different stalks.

The next question is to define a topology on $\breve{\uom}$ and
check whether the action is continuous or not.

\noindent 
A possible topology would be the topology whose basis are the collection of open sub-sheaves of $\breve{\uom}$. If we assume that each $\uom_{\phi(V)}$ has the discrete topology, coming from the fact that it can be seen as an etal\'e bundle, then the topology on $\breve{\uom}$ will be the topology in which each sub-sheaf is open, i.e., the discrete topology. 

Given such a topology we would like to check if the group action
is continuous. To this end we need to show that
$\alpha_V^{-1}(\underline{\breve{S}}_V)$ is open for $\breve{S}_V$
open sub-object. We recall that $\underline{\breve{S}}_V=\coprod_{\phi_i\in
Hom(\down V, \mv(\mh)}\underline{S}_{\phi_i(V)}$. We then obtain \ba
\alpha_V^{-1}(\underline{\breve{S}}_V)&=&\{(g, S)|l_g(S)\in \underline{\breve{S}}_V\}\\
&=&(\underline{G}_V, \underline{\breve{S}}_V) \ea which is open.
\subsubsection{Quantity Value Object}
We would now like to analyse how the group acts on the new
quantity value object $\breve{\underline{\Rl}}^{\leftrightarrow}$.
This is defined via the map \ba \ug\times
\breve{\underline{\Rl}}^{\leftrightarrow}\rightarrow
\breve{\underline{\Rl}}^{\leftrightarrow} \ea which, for each
$V\in\mv_f(\mh)$, has local components \ba \ug_V\times
\breve{\underline{\Rl}}^{\leftrightarrow}_V&\rightarrow&
\breve{\underline{\Rl}}^{\leftrightarrow}_V\\\nonumber \ug_V\times
\coprod_{\phi_i\in
Hom(\down V,\mv(\mh))}\underline{\Rl}^{\leftrightarrow}_{\phi_i(V)}&\rightarrow&\coprod_{\phi_i\in
Hom(\down V,\mv(\mh))}\underline{\Rl}^{\leftrightarrow}_{\phi_i(V)}\\\nonumber
\Big(g, (\mu,\nu)\Big)&\mapsto&(l_g\mu, l_g\nu) \ea where
$(\mu,\nu)\in \underline{\Rl}^{\leftrightarrow}_{\phi_i(V)}$,
while $(l_g\mu,
l_g\nu)\in\underline{\Rl}^{\leftrightarrow}_{l_g(\phi_i(V))}$.
Therefore $l_g\mu:\downarrow l_g(\phi_i(V))\rightarrow\Rl$ and $l_g\nu:\downarrow l_g(\phi_i(\nu))\rightarrow \Rl$.

As it can be easily deduced, even in this case the action of the $\ug$ group is to map elements around in the same stalk but never to map elements between different stalks. Thus yet again we do not obtain twisted sheaves.

We would now like to check whether the group action is continuous
with respect to the discrete topology on $\breve{\underline{\Rl}}$ defined in section 12.2.1.
Thus  we have to check whether for $V\in \mv_f(\mh)$ the following
map is continuous \ba
\Phi_V:\underline{G}_V\times \breve{\underline{\Rl}}_V&\rightarrow &\breve{\underline{\Rl}}_V\\
(g, (\mu,\nu))&\rightarrow&(l_g\mu,l_g\nu)
\ea
A typical open set in $\breve{\underline{\Rl}}_V$ is of the form $\underline{\breve{Q}}_{V}:=\coprod_{\phi_i\in Hom(\down V, \mv(\mh)}\underline{Q}_{\phi_i(V)}$ where each $\underline{Q}_{\phi_i(V)}\subseteq \underline{\Rl^{\leftrightarrow}}_{\phi_i(V)}$ is open.
Therefore \ba
\Phi^{-1}_V(\underline{\breve{Q}}_{V})&=&\{g_i, (\mu, \nu)|(l_{g_i}\mu, l_{g_i}\nu)\in \underline{\breve{Q}}_{V}\}\\
&=&(G, \underline{\breve{Q}}_{V}) \ea 
Therefore the group action with respect to the discrete topology is continuous. 
\subsubsection{Truth Object}
The new truth value object for pure states obtained through the
action of the $F$ functor is \be
\underline{\breve{\mathbb{T}}}^{|\psi\rangle}:=F(\underline{\mathbb{T}}^{|\psi\rangle})
\ee  which is defined as follows:
\begin{Definition}
The truth object $F(\underline{\mathbb{T}}^{|\psi\rangle})$ is the
presheaf defined on
\begin{itemize}
\item [--] Objects: for each $V\in\mv_f(\mh)$  we get
\be F(\underline{\mathbb{T}}^{|\psi\rangle}):=\coprod_{\phi_i\in
Hom(\down V, \mv(\mh))}\underline{\mathbb{T}}^{|\psi\rangle}_{\phi_i(V)}
\ee where
$\underline{\mathbb{T}}^{|\psi\rangle}_{\phi_i(V)}:=\{\hat{\alpha}\in
P(\phi_i(V))|\langle\psi|\hat{\alpha}|\psi\rangle=1\}$ and $P(\phi_i(V))$ denotes the collection of all projection operators in $\phi_i(V)$.
\item[--] Morphisms: given $V^{'}\subseteq V$ the corresponding map is
\be
\underline{\breve{\mathbb{T}}}^{|\psi\rangle}(i_{V^{'}V}):\coprod_{\phi_i\in
Hom(\down V,
\mv(\mh))}\underline{\mathbb{T}}^{|\psi\rangle}_{\phi_i(V)}\rightarrow
\coprod_{\phi_j\in Hom(\down V^{'},
\mv(\mh))}\underline{\mathbb{T}}^{|\psi\rangle}_{\phi_j(V^{'})}
\ee such that, given $\underline{S}\in
\underline{\mathbb{T}}^{|\psi\rangle}_{\phi_i(V)}$, then \be
\underline{\breve{\mathbb{T}}}^{|\psi\rangle}(i_{V^{'}V})\underline{S}:=\underline{\mathbb{T}}^{|\psi\rangle}(i_{\phi_i(V),\phi_j(V^{'})})\underline{S}=\underline{S}_{|\phi_j(V^{'})}
\ee where $\phi_j\leq \phi_i$ thus $\phi_j(V^{'})\subseteq
\phi_i(V)$ and $\phi_j(V^{'})={\phi_i}_{|V^{'}}(V^{'})$.
\end{itemize}
\end{Definition}
\subsection{New Representation of Physical Quantities}
We are now interested in understanding the action of the $F$
functor on physical quantities. We thus define the following \be
F(\breve{\delta}(\hat{A})):\breve{\us}\rightarrow
\ps{\breve{\Rl}^{\leftrightarrow}} \ee which, at each context $V$,
is defined as \be F(\breve{\delta}(\hat{A}))_V:\coprod_{\phi_i\in
Hom(\down V,\mv(\mh))}\us_{\phi_i(V)}\rightarrow\coprod_{\phi_i\in
Hom(\down V,\mv(\mh))} \ps{\Rl^{\leftrightarrow}}_{\phi_i(V)} \ee such
that for a given $\lambda\in \us_{\phi_i(V)}$ we obtain \ba
F(\breve{\delta}(\hat{A}))_V(\lambda)&:=&\breve{\delta}(\hat{A})_{\phi_i(V)}(\lambda)\\\nonumber
&=&(\breve{\delta}^i(\hat{A})_{\phi_i(V)}(\cdot),
\breve{\delta}^o(\hat{A})_{\phi_i(V)}(\cdot))(\lambda)=(\mu_{\lambda},\nu_{\lambda})
\ea
Thus in effect the map $F(\breve{\delta}(\hat{A}))_V$ is a co-product of maps of the form $F(\breve{\delta}(\hat{A}))_{\phi_i(V)}$ for all $\phi_i\in Hom(\down V, \mv(\mh))$.

From this definition it is straightforward to understand how the
group acts on such physical quantities. In particular, for each context $V\in\mv_f(\mh)$ we obtain a collection of
maps \be F(\breve{\delta}(\hat{A}))_V:\coprod_{\phi_i\in
Hom(\down V,\mv(\mh))}\us_{\phi_i(V)}\rightarrow\coprod_{\phi_i\in
Hom(\down V,\mv(\mh))} \ps{\Rl^{\leftrightarrow}}_{\phi_i(V)} \ee and
the group action is to map individual maps in such a collection into one
another. Thus, for example, if we consider the component \be
\breve{\delta}(\hat{A})_{\phi_i(V)}:\us_{\phi_i(V)}\rightarrow
\ps{\Rl^{\leftrightarrow}}_{\phi_i(V)} \ee by acting on it by an
element of the group we would obtain \ba
l_g\Big(\breve{\delta}(\hat{A})_{\phi_i(V)}\Big):l_g\us_{\phi_i(V)}&\rightarrow& l_g \ps{\Rl^{\leftrightarrow}}_{\phi_i(V)}\\
\us_{l_g(\phi_i(V))}&\rightarrow&
\ps{\Rl^{\leftrightarrow}}_{l_g(\phi_i(V))} \ea
Let us now analyse what exactly is $l_g\breve{\delta}(\hat{A}))_{\phi_i(V)}$.\\
We know that it is comprised of two functions, namely \be
l_g\Big(\breve{\delta}(\hat{A})_{\phi_i(V)}\Big)=\Big(l_g\Big(\breve{\delta}^i(\hat{A}))_{\phi_i(V)}\Big)(\cdot),\Big(l_g\breve{\delta}^o(\hat{A}))_{\phi_i(V)}\Big)(\cdot)\Big)
\ee We will consider each of them separately. Given
$\lambda\in\us_{l_g(\phi_i(V))}$ we obtain \ba
l_g\Big(\breve{\delta}^i(\hat{A})_{\phi_i(V)}\Big)(\lambda)&=&\lambda\Big(l_g\big(\delta^i(\hat{A})_{\phi_i(V)}\big)\Big)\\\nonumber
&=&\lambda\Big(\hat{U}_g\big(\delta^i(\hat{A})_{\phi_i(V)}\big)\hat{U}_g^{-1}\Big)\\\nonumber
&=&\lambda\Big(\delta^i(\hat{U}_g\hat{A}\hat{U}_g^{-1})_{l_g(\phi_i(V))}\Big)\\\nonumber
&=&\breve{\delta}^i(\hat{U}_g\hat{A}\hat{U}_g^{-1})_{l_g(\phi_i(V))}(\lambda)
\ea Similarly for the order reversing function we obtain \be
l_g\Big(\breve{\delta}^o(\hat{A})_{\phi_i(V)}\Big)(\lambda)=\breve{\delta}^o(\hat{U}_g\hat{A}\hat{U}_g^{-1})_{l_g(\phi_i(V))}(\lambda)
\ee Thus putting the two results together we have \be
l_g\Big(\breve{\delta}(\hat{A})_{\phi_i(V)}\Big)=\Big(\breve{\delta}(\hat{U}_g\hat{A}\hat{U}_g^{-1})_{l_g(\phi_i(V))}\Big)
\ee This is the topos analogue of the standard transformation of
self adjoint operators in the canonical formalism of quantum
theory. In particular, given a self adjoint operator
$\breve{\delta}(\hat{A})$ its local component in the context $V$
is $\breve{\delta}(\hat{A})_{V}$.
 This `represents' the pair of self adjoint operators $(\delta^i(\hat{A})_{V}, \delta^o(\hat{A})_{V})$ which live in $V$. By acting with a unitary transformation we obtain the transformed quantity $l_g\Big(\breve{\delta}(\hat{A})\Big)$ with local components $\Big(\breve{\delta}(\hat{U}_g\hat{A}\hat{U}_g^{-1})_{l_gV}\Big)$, $V\in \mv_f(\mh)$. Such a quantity represents the pair $(\delta^i(\hat{U}_g\hat{A}\hat{U}_g^{-1})_{l_g(V)}, \delta^o(\hat{U}_g\hat{A}\hat{U}_g^{-1})_{l_g(V)})$ of self adjoint operators living in the transformed context $l_g(V)$.

\section{Appendix}
\begin{Theorem}
Given the etal\'e map $f : X\rightarrow Y$ the left adjoint functor $f!:Sh(X)\rightarrow Sh(Y)$ is defined as follows 
\be
f!(p_A : A\rightarrow X) = f\circ p_A : A\rightarrow Y 
\ee
for $p_A:A\rightarrow Y$ an etal\'e bundle
\end{Theorem}
\begin{Proof}
In the proof we will first define the functor $f!$ for general presheaf situation, then restrict our attention to the case of sheaves ($Sh(X)\subseteq \Sets^{X^{\op}}$) and $f$ etal\'e. 

Consider the map $f:X\rightarrow Y$, this gives rise to the functor $f!:\Sets^{X^{\op}}\rightarrow \Sets^{Y^{\op}}$. The standard definition of $f!$ is as follows:
\be
f!:=-\otimes_{X}(_{f}X^{\bullet})
\ee
which is defined on objects $A\in \Sets^{X^{\op}}$ as
\be
A\otimes_{X}(_{f}Y^{\bullet})
\ee
This is a presheaf in $\Sets^{Y^{\op}}$, thus for each element $y\in Y$ we obtain the set
\be
(A\otimes_{X}(_{f}Y^{\bullet}))y:=A\otimes_{X}(_{f}Y^{\bullet})(-, y)
\ee
where $(_{f}Y^{\bullet})$ is the presheaf 
\be
(_{f}Y^{\bullet}):X\times Y^{\op}\rightarrow \Sets
\ee
This presheaf derives from the composition of $f\times id_{Y^{\op}}:X\times Y^{op}\rightarrow Y\times Y^{op}$ ($(f\times id_{Y^{\op}})^*:\Sets^{Y\times Y^{op}}\rightarrow \Sets^{X\times Y^{op}}$) with $^{\bullet}Y^{\bullet}:Y\times Y^{\op}\rightarrow \Sets$, i.e.,
\be
(_{f}Y^{\bullet}):=(f\times id_{Y^{\op}})^*(^{\bullet}Y^{\bullet})=^{\bullet}Y^{\bullet}\circ (f\times id_{Y^{\op}})
\ee
where $^{\bullet}Y^{\bullet}$ is the bi-functor 
\ba
^{\bullet}Y^{\bullet}:Y\times Y^{\op}&\rightarrow&\Sets\\
(y,y^{'})&\mapsto&Hom_Y(y^{'},y)
\ea
Now coming back to our situation we then have the restricted functor
\ba
(_{f}Y^{\bullet})(-, y):(X, y)&\rightarrow& \Sets\\
(x, y)&\mapsto &(_{f}Y^{\bullet})(x,y)
\ea
which from the definition given above is
\ba
(_{f}Y^{\bullet})(x,y)=^{\bullet}Y^{\bullet}\circ (f\times id_{Y^{\op}})(x,y)=^{\bullet}Y^{\bullet}(f(x), y)=Hom_Y(y, f(x))
\ea
Therefore putting all the results together we have that for each $y\in Y$ we obtain $A\otimes_{X}(_{f}Y^{\bullet})(-, y)$, defined for each $x\in X$ as 
\be
A(-)\otimes_{X}(_{f}Y^{\bullet})(x,y):=A(x)\otimes_{X}Hom_Y(y, f(x))
\ee
This represents the presheaf $A$ defined over the element $x$, plus a collection of maps in $Y$ mapping the original $y$ to the image of $x$ via $f$. 

In particular $A(x)\otimes_{X}(_{f}X^{\bullet})=A(x)\otimes_{X}Hom_Y(y, f(-))$ represents the following equaliser: 
\[\xymatrix{
\coprod_{x, x^{'}}A(x)\times Hom_X(x^{'},x)\times Hom_Y(y, f(x^{'})\ar@<3pt>[rr]^{\;\;\;\;\;\;\; \tau} \ar@<-3pt>[rr]_{\;\;\;\;\;\;\; \theta}&&\coprod_x A(x)\times Hom_Y(y, f(x))\ar[dd]^{\sigma}\\
&&\\
&&A(-)\otimes_{X}Hom_Y(y, f(-)) \\
}\]
Such that given a triplet $(a, g, h)\in A(x)\times Hom_X(x^{'},x)\times Hom_Y(y, f(x^{'})$ we then obtain that 
\be\label{equ:equivalence}
\tau(a, g, h)=(ag, h)=\theta(a,g,h)=(a, gh)
\ee
Therefore $A(-)\otimes_{X}Hom_Y(y, f(-))$ is the quotient space of $\coprod_x A(x)\times Hom_Y(y, f(x))$ by the above equivalence conditions.

We now consider the situation in which $A$ is a sheaf on $X$, in particular it is an etal\'e bundle $p_A:A\rightarrow X$ and $f$ is an etal\'e map which means that it is a local homeomorphism, i.e. for each $x\in X$ there is an open set $V$ such that $x\in V$ and such that $f_{|V}:V\rightarrow f(V)$ is a homeomorphism. It follows that for each $x_i\in V$ there is a unique element $y_i$ such that $f_{|V}(x_i)=y_i$. In particular for each $V\subset X$ then $f_{|V}(V)=U$ for some $U\subset Y$.

\noindent
It can be the case that $f_{|V_i}(V_i)=f_{|V_j}(V_j)$ even if $V_i\neq V_j$, since the condition of being a homeomorphism is only local, however in these cases the restricted etal'e maps have to agree on the intersections, i.e. $f_{|V_i}(V_i\cap V_j)=f_{|V_j}(V_j\cap V_j)$

Let us now consider an open set $V$ with local homeomorphism $f_{|V}$. In this setting each element $y_i\in f_{|V}(V)$ will be of the form $f(x_i)$ for a unique $x_i$. Moreover, if we consider two open sets $V_1,V_2\subseteq V$, then to each map $V_1\rightarrow V_2$ in $X$, with associated bundle map $A(V_2)\rightarrow A(V_1)$, there corresponds a map $f_{|V}V_1\rightarrow f_{|V}(V_2)$ in $Y$. Therefore evaluating $A(-)\otimes_{X}Hom_Y(-, f(-))$ at the open set $f_{|V}(V)\subset Y$ we get, for each $V_i\subseteq V$ the equivalence classes $[A(V_i)\times_XHom_Y(f_{|V}(V), f_{|V}(V_i))]$ where $A(V_j)\times_XHom_Y(f_{|V}(V), f_{|V}(V_j))\simeq A(V_k)\times_XHom_Y(f_{|V}(V), f_{|V}(V_k))$ iff there exists a map $f_{|V}(V_j)\rightarrow f_{|V}(V_k)$ (which combines giving $f_{|V}(V)\rightarrow f_{|V}(V_k)$)  and corresponding bundle map $A(V_k)\rightarrow A(V_j)$ (which combine giving $A(V_k)\rightarrow A(V)$) given by the map $V_j\rightarrow V_k$ ( which combined gives $V\rightarrow V_k$) in $X$. A moment of thought reveals that such an equivalence class is nothing but $p_A^{-1}(V)$ (the fibre of $p_A$ at $V$) with associated fibre maps induced from the base maps.

We will now denote such an equivalence class by $[A(V)\times_XHom_Y(f_{|V}(V), f_{|V}(V))]$, since obviously in each equivalence class there will be the element $[A(V)\times_XHom_Y(f_{|V}(V), f_{|V}(V))]$

We apply the same procedure for each open set $V_i\subset X$. We can obtain two cases: \begin{enumerate}
\item [i)] $f_{|V_i}(V_i)=U\neq f_{V}(V)$. In that case we simply get an independent equivalence class for $U$.

\item [ii)] If $f_{|V_i}(V_i)=U=f_{V}(V)$ and there is no map $i:V\rightarrow V_i$ in $X$ then, in this case, we obtain for $U$ two distinct equivalence classes $[A(V_i)\times_XHom_Y(f_{|V_i}(V_i), f_{|V_i}(V_i))]$ and $[A(V)\times_XHom_Y(f_{|V}(V), f_{|V}(V))]$. 
\end{enumerate}
Thus the sheaf $A(-)\otimes_X(_fY^{\bullet})$ is defined for each open set $f_V(V)\subset Y$ as the set $[A(V)\times_XHom_Y(f_{|V}(V), f_{|V}(V))]\simeq A(V)))$, and for each map $f_{V^{'}}(V^{'})\rightarrow f_V(V)$ in $Y$ (with associated map $V^{'}\rightarrow V$ in $X$), the corresponding maps $[A(V)\times_XHom_Y(f_{|V}(V), f_{|V}(V))]\simeq A(V)\rightarrow [A(V^{'})\times_XHom_Y(f_{|V^{'}}(V)^{'}, f_{|V^{'}}(V^{'}))]\simeq A(V^{'})$

This is precisely what the etal\'e bundle $f\circ p_A:A\rightarrow Y$ is.

\end{Proof}

\chapter{Lecture 17}

In this lecture we will define the topos analogue of probabilities. To this end we will have to introduce a new topos namely the topos of sheaves over the category $\mv(\mh)\times (0,1)_L$

\section{Topos Reformulation of Probabilities}
We will now delineate how probabilities can be described in terms of truth values in a given topos. In such a formulation logical concepts are seen as fundamental, while probabilities become derived concepts. \\
This view is very useful in quantum theory, in as much that it overcomes the problems related to the relative frequency interpretation of probabilities. \\
In particular, in the topos approach truth values can be assigned to any proposition in quantum theory and, since probabilities are defined in terms of truth values, probabilities can always be assigned even in the context of closed systems.\\
For the interested reader, the full analysis of these ideas can be found in \cite{probabilities}.

\section{General Definition of Probabilities in the Language of Topos Theory}\label{sec:prob}
In this section we will outline the general way in which probabilities can be described in a topos. How this general definition will apply to classical physics and quantum physics will be described in subsequent sections.\\
What we will do in this section is to define a possible topos in which probabilities can be expressed in terms of truth values in that topos. What we are looking for is a way to combine on the one hand probabilities, which are described by numbers in the interval $[0,1]$ and, on the other hand, truth values expressed in a topos $\tau$, which are defined as global elements of the sub-object classifier $\Omega^{\tau}$. \\
What this amounts to is to find a topos $\tau$ such that the following holds
\be\label{eq:tau}
[0,1]\simeq \Gamma\Omega^{\tau}
\ee
If we are able to achieve this, it will mean that we have found a bijective correspondence between probabilities and truth values in a topos $\tau$.\\
The way in which the equivalence in \ref{eq:tau} is achieved is as follows:\\\\
the first step is to define the topological space $(0,1)$, whose open sets are the intervals $(0,r)$ for $0\leq r\leq 1$. This topological space is denoted as $(0,1)_L$ and the collection of open sets as $\mathcal{O}((0,1)_L)$. This is a category under inclusion. 
One can then define a bijection
\ba
\beta:[0,1]&\rightarrow& \mathcal{O}((0,1)_L)\\
r&\mapsto&(0,r)
\ea
The strategy is then to associate open sets $(0,r)$ with truth values in a certain topos. Such open sets will then, in turn, be associated with the probability $r$. \\
In this way we will not be loosing anything by considering open sets $(0,r)$ which don't actually contain $r$ instead of closed sets $(0, r]$ or $[0,r]$. Moreover, if we were to consider either $(0, r]$ or $[0,r]$ we would run into troubles. In fact, if we had chosen $[0,r]$ we would have obtained situations in which all propositions are totally true with probability zero, even the totally false proposition (more details later on). If instead we had chosen $(0, r]$ we wouldn't have obtained a topology, since these sets do not close under arbitrary unions.\\\\

Having chosen the topology on our space, we then know from topos theory that, for any topological space $X$, there is the following isomorphisms of Heyting algebras
\be
\mathcal{O}(X)\simeq\Gamma\Omega^{Sh(X)}
\ee
where $Sh(X)$ identifies the topos of sheaves over the topological space $X$. Therefore, in order to obtain equation \ref{eq:tau}, a possible topos which could be used is $\tau=Sh((0,1)_L)$.\\
In such a topos the isomorphisms we are after is 
\ba
\sigma:\mathcal{O}((0,1)_L)&\rightarrow& \Gamma\Omega^{Sh((0,1)_L)}\\
(0, p)&\mapsto&\sigma(0,p):=(l(p):\ps{1}\rightarrow \Omega^{Sh((0,1)_L)}
\ea
such that for each stage/context $(0,r)\in \mathcal{O}((0,1)_L)$ we have the truth value
\be
l(p)_{(0,r)}=\begin{cases}\{(0,r^{'})\in\mathcal{O}((0,1)_L)|r^{'}\leq r\}=\Omega^{Sh((0,1)_L)}_{(0,r)}&\text{ if } p\geq r\\
\{(0,r^{'})\in\mathcal{O}((0,1)_L)|r^{'}\leq p\}&\text{ if }0< p < r\\
\emptyset &\text{ if } p=0
\end{cases}
\ee
If we use the isomorphisms $\{(0,r^{'})\in \mathcal{O}(0,1)_L|r^{'}\leq r\}\simeq [0,r]$ the we can write the above as
\be
l(p)_{(0,r)}=\begin{cases}[0,r]=\Omega^{Sh((0,1)_L)}_{(0,r)}&\text{ if } p\geq r\\
[0,p]&\text{ if }0< p < r\\
\emptyset &\text{ if } p=0
\end{cases}
\ee

As can be seen from the definition, $l$ is nothing but the combination of $\beta$ and $\sigma$. Thus $l:[0,1]\rightarrow\Gamma\Omega^{Sh((0,1)_L)}$ is a bijection between probabilities and truth values:
\ba
[0,1]&\xrightarrow{\beta}&\mathcal{O}((0,1)_L)\xrightarrow{\sigma}\Gamma\Omega^{Sh((0,1)_L)}\\
p&\mapsto&(0,p)\mapsto l(p)
\ea
It can also be shown that $l$ is an order preserving isomorphisms. This implies that probabilities are faithfully represented as truth values in the topos $Sh((0,1)_L)$.

\section{Example for Classical Probability Theory} 
In this section we will apply the topos theoretic description of probabilities defined in the previous section to a classical system. 
As a first step we will define truth values of propositions regarding a classical system and, then, show how such truth values are related to classical probabilities. \\
Let us consider a proposition $(A\in \Delta)$ meaning that the value of the quantity $A$ lies in $\Delta$. We want to define the truth value of such a proposition with respect to a given state $s\in X$, where $X$ is the state space. Recall that in classical theory, the truth value of the above proposition in the state $s$ is given by
$$v(A\in \Delta;s) :=\begin{cases}1\text{ iff }s\in  A^{-1}(\Delta)\\0 \text{ otherwise }\end{cases}$$
where $A^{-1}(\Delta)\subseteq X$ is the subset of the state space for which the proposition $(A\in \Delta)$ is true.\\
Another way of defining truth values is through the truth object $\mathbb{T}^s$ which is state dependent. The definition of the truth object was defined in previous lecture, but we will report here for clarity reasons: \\\\
for each state $s$, we define the set
$$\mathbb{T}^s :=\{S\subseteq X|s\in S\}$$ 
Since $(s\in A^{-1}(\Delta))$ iff $ A^{-1}(\Delta)\in \mathbb{T}^s$, we can now write the truth value above in the following equivalent way
$$v(A\in \Delta;s):=\begin{cases} 1 \text{ iff }A^{-1}(\Delta)\in \mathbb{T}^s\\ 0\text{ otherwise }\end{cases}$$
It follows that the truth value $v(A\in \Delta;s)$ is equivalent to the truth value of the mathematical statement $[[A^{-1}(\Delta)\in \mathbb{T}^s]]$.\\
We would now  like to relate the above defined truth values to probability measures. To this end we associate to the space $X$ the probability measure 
\be
\mu:Sub(X)\rightarrow [0,1]
\ee
where $Sub(X)$ denotes the measurable subsets of $X$.\\
It is now possible to define a measure dependent truth object as follows: 
\be
\mathbb{T}^{\mu}_r:=\{S\subseteq X|\mu(S)\geq r\}
\ee
for all $r\in(0,1]$\footnote{Note that we will not include the value r=0. The reason, as will be explained later on, is to avoid obtaining situations in which all propositions are totally true with probability zero.}. What the above truth object defines, is all those proposition which are true with probability equal or greater than $r$.\\
So far, the objects we have defined are simply sets. However, we would like to find their analogue in the topos $Sh(\mathcal{O}(0,1)_L)$, which was used to define probabilities, so that truth values and probabilities have a common ground in which to be compared. Thus what we are looking for is a way of expressing both truth object $\mathbb{T}^{\mu}$ and the truth values $v(A\in \Delta;s)$ as objects in $Sh(\mathcal{O}(0,1)_L)$.\\
 To this end we perform the following ``translations":\\\\
 First of all we map the state space $X$ to the constant sheaf
\be
X\rightarrow \underline{X}
\ee
What this means is that for all $(0,r)\in\mathcal{O}((0,1)_L)$
\be
\underline{X}(0,r)=X
\ee
\\
For any measurable set $S\in Sub(X)$ we define the constant sheaf $S\rightarrow \underline{S}$; $\underline{S}(0,r)=S$. Thus obtaining the map
\ba
\Delta:Sub(X)&\rightarrow &Sub_{Sh(\mathcal{O}((0,1)_L))}(\underline{X})\\
S&\mapsto&\underline{S}
\ea
The analogue of the truth object in $Sh (\mathcal{O}((0,1)_L))$ is then: $\underline{\mathbb{T}}^{\mu}_{(0,r)}:=\{\underline{S}\subseteq \underline{X}|\mu(\underline{S})_{(0,r)}\geq r\}$ for all $(0,r)\in\mathcal{O}((0,1)_L)$.\\
Now that we have defined the analogue of the relevant object in the topos $Sh(\mathcal{O}((0,1)_L))$ we can define the truth value of a proposition $(A\in\Delta:=S)\subseteq X$ as a global section of $\uom^{Sh(\mathcal{O}((0,1)_L))}$:
\ba
[[\underline{S}\in \underline{\mathbb{T}}^{\mu}]](0,r)&:=&\{(0,r^{'})\leq (0,r)| \underline{S}_{(0,r^{'})}\in\ps{\mathbb{T}}^{\mu}_{(0,r^{'})}\}\\
&=& \{(0,r^{'})\leq (0,r)|\mu(S)\geq r^{'}\}\\
&=&[0,\mu(S)]\cap(0,r]\\
&=&[0,min\{\mu(S),r\}]
\ea
which is an element of $\underline{\Omega}^{Sh(\mathcal{O}((0,1)_L))}(0,r)$. Thus, globally, we get 
$$[[\underline{S}\in \underline{\mathbb{T}}^{\mu}]]\in\Gamma(\underline{\Omega}^{Sh(\mathcal{O}((0,1)_L})$$
It is easy to see that for any set $S\subseteq X$ the value $\mu(S)$ can be recovered by the valuation $[[\underline{S}\in \underline{\mathbb{T}}^{\mu}]]\in\Gamma(\underline{\Omega}^{Sh(\mathcal{O}((0,1)_L))})$. This means that probabilities can be replaced by  truth values without any information being lost. \\
In particular, in the topos $Sh(\mathcal{O}((0,1)_L))$ the probability measure $\mu:Sub(X)\rightarrow [0,1]$ can be uniquely expressed through the map $\epsilon^{\mu}$ (defined below) in the sense that there exists a bijective correspondence between $\mu$ and $\epsilon^{\mu}$, i.e. for each $\mu$ we get:
\ba
\epsilon^{\mu}:Sub_{Sh(\mathcal{O}((0,1)_L))}(\underline{X})&\rightarrow&  \Gamma(\underline{\Omega}^{Sh(\mathcal{O}((0,1)_L))})\\
\underline{S}&\mapsto&[[\underline{S}\in \underline{\mathbb{T}}^{\mu}]]
\ea 
\\\\
What this means is that we have effectively replaced the probability measure $\mu:Sub(X)\rightarrow [0,1]$ with the collection of truth values $\Gamma(\underline{\Omega}^{Sh(\mathcal{O}((0,1)_L))})$. \\
Since $\Gamma(\underline{\Omega}^{Sh(\mathcal{O}((0,1)_L))})$ is a Heyting algebra, probabilities are now interpreted in the context of intuitionistic logic.\\\\
What has been done so far can be summarised by the following commutative diagram:
\[\xymatrix{
Sub(X)\ar[rr]^{\mu}\ar[dd]^{\Delta}&&[0,1]\ar[dd]^l\\
&&\\
Sub_{Sh(\mathcal{O}((0,1)_L))}(\underline{X})\ar[rr]^{\epsilon^{\mu}}&&\Gamma\uom^{Sh(\mathcal{O}((0,1)_L))}\\
}\]
By chasing the diagram around we now have the following equalities
\ba
[l\circ \mu(S)](r)&=&l(\mu(S))(r)=\{(0,r^{'})\subseteq (0,r)|\mu(S)\geq r^{'}\}\\
&=&[[\ps{S}\in \ps{T}^{\mu}]](r)=\epsilon^{\mu}(\Delta(S))(r)
\ea
It was also shown in \cite{probabilities} that it is possible to define the logical analogue of the $\sigma$-additivity of the measure $\mu$. We recall that the $\sigma$-additivity of a measure $\mu$ is defined as follows:\\\\
for any countable family $(S_i)_{i\in N}$ of pairwise disjoint, measurable subsets of $X$ (the space where the measure is defined), we have
\be\label{equ:measurec}
\mu(\bigcup_i S_i)=\sum_i\mu(S_i)
\ee

Let us now consider a countable increasing family $(\tilde{S}_i)_{i\in N}$ of measurable subsets of $X$ defined as follows
\be
\tilde{S}_i:=\tilde{S}_{i-1}\cup \tilde{S}_i\text{ for } i>0\text{ while } \tilde{S}_0:=S_0
\ee
Then equation \ref{equ:measurec} becomes 
\be\label{equ:mc}
\mu(\bigvee_i\tilde{S}_i)=\mu(\bigcup_i\tilde{S}_i)=\sup_{i}\mu(\tilde{S}_i)=\bigvee_i\mu(\tilde{S}_i)
\ee
i.e., $\mu$ preserves countable joins (suprema).

The logical analogue of this is as follows: \\\\
Given a countable increasing family of measurable subsets $(\tilde{S}_i)_{i\in N}$ of $X$ we then have
\be
(\epsilon^{\mu}\circ\Delta)(\bigvee_i\tilde{S}_i)=(l\circ\mu)(\bigvee_i(\tilde{S}_i)
\ee
We know that $\mu$ preserves countable joins, we now want to show that $l$ does. To this end we will explicitly compute, for each context $(0,r)$ how $l$ acts on  countable unions. Since $\mu:sub(X)\rightarrow [0,1]$ then $\mu(\bigvee_i(\tilde{S}_i))\stackrel{\ref{equ:mc}}{=}\bigvee_i(\mu(\tilde{S}_i))=\bigvee_ip_i$ where $p_{i\in I}$ is some family of real numbers in the interval $[0,1]$. Thus the question is what is $l(\bigvee_ip_i)$ as computed for each context $(0,r)$? 

By applying the definition we obtain
\be
l(\bigvee_ip_i)_{(0,r)}:=\begin{cases} \uom_{(0,r)}^{Sh((0,1)_L)}\text{  if  }\bigvee_ip_i\geq r\\
\{(0,r^{'})\in\mathcal{O}((0,1)_L)|r^{'}\leq \bigvee_ip_i\}\text{  if  }0<\bigvee_ip_i<r\\
\emptyset \text{  if  } \bigvee_ip_i=0
\end{cases}
\ee
Or equivalently
\be
l(\bigvee_ip_i)_{(0,r)}:=\begin{cases}[0,r]= \uom_{(0,r)}^{Sh((0,1)_L)}\text{  if  }\bigvee_ip_i\geq r\\
[0,\bigvee_ip_i]\text{  if  }0<\bigvee_ip_i<r\\
\emptyset \text{  if  } \bigvee_ip_i=0
\end{cases}
\ee
On the other hand if we computed $\bigvee_il(p_i)_{(0,r)}$ for all contexts $(0,r)\in\mathcal{O}((0,1)_L)$ we would obtain
\be
\bigvee_il(p_i)_{(0,r)}=\bigvee_i\begin{cases}
\uom_{(0,r)}^{Sh((0,1)_L)}\text{  if  }p_i\geq r\\
\{(0,r^{'})\in\mathcal{O}((0,1)_L)|r^{'}\leq p_i\}\text{  if  }0<p_i<r\\
\emptyset \text{  if  } p_i=0
\end{cases}
\ee
Or equivalently
\be
\bigvee_il(p_i)_{(0,r)}=\bigvee_i\begin{cases}
[0,r]=\uom_{(0,r)}^{Sh((0,1)_L)}\text{  if  }p_i\geq r\\
[0,p_i]\text{  if  }0<p_i<r\\
\emptyset \text{  if  } p_i=0
\end{cases}
\ee
Thus it follows that $\bigvee_il(p_i)_{(0,r)}=l(\bigvee_ip_i)_{(0,r)}$ for all contexts $(0,r)\in\mathcal{O}((0,1)_L)$ thus
\be
(l\circ\mu)(\bigvee_i(\tilde{S}_i))=\bigvee_i(l\circ\mu)(\tilde{S}_i)
\ee
It follows that the logical description of the $\sigma$-additivity is
\be
(\epsilon^{\mu}\circ\Delta)(\bigvee_i\tilde{S}_i)=\bigvee_i(l\circ\mu)(\tilde{S}_i)
\ee
So far we have only described the classical aspects of the topos interpretation of probability. However, it is possible to extend such ideas to the quantum case. 
\section{Quantum Case}
We now would like to show how the interpretation of probabilities as truth values in an appropriate topos can be applied to the quantum case. So far, in the literature it was shown how, within the topos $Sh(\mv(\mh))$\footnote{Although in the previous lectures we have used the topos $Set^{\mv(\mh)}$, we know that given the Alexandrov topology on $\mv(\mh)$ then $Sh(\mv(\mh))\simeq Sets^{\mv(\mh)}$. So, in the following we will alternate freely between sheaves and presheaves.} (the topos formed by the collection of sheaves on the poset $\mv(\mh)$ of all abelian von Neumann algebras of a given Hilbert space $\mh$) the truth value of a proposition $A\in \Delta$, given a state $\psi$, is defined as\footnote{Note that here we have added a suffix \emph{org} to indicate original since we will now change the formulation of truth objects. }
\be\label{equ:tru}
v(A\in \Delta,\psi)(V)=[[\delta(\hat{E}[A\in \Delta])\in ^{org}\underline{\mathbb{T}}^{|\psi\rangle}]]_V:=\{V^{'}\subseteq V|\langle\psi|\delta(\hat{E}[A\in \Delta])_{V^{'}}|\psi\rangle=1\}\in \uom^{Sh(\mv(\mh))}_{V}
\ee
where, for the quantum case, the truth object is the sheaf $^{org}\underline{\mathbb{T}}^{|\psi\rangle}$ which is defined, for each context $V\in \mv(\mh)$ as the set
\be\label{equ:tobj}
^{org}\underline{\mathbb{T}}^{|\psi\rangle}_V:=\{\hat{\alpha}\in P(V)|\langle\psi|\hat{\alpha}|\psi\rangle=1\}
\ee
This definition of truth value works perfectly well when we consider a pure state $|\psi\rangle$. However, if we consider mixed states with associate density matrix say $\rho=\sum_{i=1}^Nr_i|\psi\rangle\langle\psi|$, then the obvious analogue of \ref{equ:tru}, namely
\be\label{equ:tru1} 
v(A\in \Delta,\rho)(V)=[[\delta(\hat{E}[A\in \Delta])\in ^{org}\underline{\mathbb{T}}^{\rho}]]_V:=\{V^{'}\subseteq V|tr(\rho\hat{E}[A\in \Delta])_{V^{'}}=1\}\in \uom^{Sh(\mv(\mh))}_{V}
\ee
doesn't work since it does not separate the states $\rho$. \\
In this case the truth object is 
\be\label{equ:tobj1}
^{org}\underline{\mathbb{T}}^{\rho}_V:=\{\hat{\alpha}\in P(V)|tr(\rho\hat{\alpha})=1\}
\ee
To see why this is the case consider our usual example for $\mh=\Cl^4$ and consider two density matrices

\[\rho_1=\begin{pmatrix} 3/4& 0& 0&0\\	
0&1/4&0 &0\\
0&0&0&0\\
0&0&0&0
  \end{pmatrix}\;\;\;\rho=\begin{pmatrix}3/5&0&0&0\\
  0&2/5&0&0\\
  0&0&0&0\\
  0&0&0&0
  \end{pmatrix}
  \]
Then the condition for $tr(\rho\hat{P})=1$ is that $\hat{P}\geq \hat{P}_1+\hat{P}_2$ where $\hat{P}_1=diag(1,0,0,0) $ and $\hat{P}_2=diag(0,1,0,0)$. Similarly the condition for $tr(\rho_1\hat{P})=1$ is that $\hat{P}\geq \hat{P}_1+\hat{P}_2$. It follows that $\rho$ and $\rho_1$ have the same support , namely $\hat{P}_1+\hat{P}_2$. However, from the definition of truth values, the element $v(A\in \Delta,\rho)(V)$ depends only on the support of $\rho$, and $\rho_1$, thus it is not possible to distinguish the two. 

On the other hand, truth values that do separate the density matrices $\rho$ are the one parameter family of truth values
\be\label{equ:measure}
v(A\in \Delta,\rho)^r(V):=\{V^{'}\subseteq V|tr(\rho\hat{E}[A\in \Delta])_{V^{'}}\geq r\}\in \uom^{Sh(\mv(\mh))}_{V}
\ee
for $r\in (0,1]$. As we can see from the above formula, such truth values introduce truth probabilities different from one. In particular, 
what \ref{equ:measure} expresses is the truth value for the proposition $A\in \Delta$ to be true with probability at least $r$.\\
However, the problem with such one parameter family is that it represents a collection of objects, one for each $r$. So we need to group, somehow, these objects together and show that such a family can be considered as a single object. This can be done by enlarging the topos $Sh(\mv(\mh))_{V}$. \\\\\\
So the main steps needed, in order to correctly express probabilities as truth values of quantum propositions in an appropriate topos are: 
\begin{itemize}
\item [1)] Define the truth values which separate the density matrices so as to be an object in the topos. This, as was hinted to above, can be done by enlarging the topos we are working with.
\item [2)] Define a correct probability measure on the state space $\us$. 
\item [3)] Find a way to relate 1) and 2).
\end{itemize}
\subsection{Measure on the Topos State Space}
We are now interested in constructing a measure $\mu$ on the state space $\us$. As such it should some how define a size or weight for each sub-object of $\us$. However we will not consider all sub-objects of $\us$ but only a collection of them which we define as ``measurable". These collection of ``measurable"sub-objects of $\us$ will be the collection of all clopen sub-objects of $\us$. The reason being that we are yet again taking ideas form classical physics. There we know that a proposition is defined as a measurable subset of the state space. Hence in the topos formulation of quantum theory we ascribe the status of ``measurable" to all clopen sub-objects of $\us$, i.e., all propositions. It could be the case that one can consider a larger collection of sub-objects, but surely $Sub_{cl}(\us)$ is the minimal such collection. 
In any case we will define a probability measure on $Sub_{cl}(\us)$ such that $\us$ will have measure 1 and $\ps{0}$ will have measure zero. 

As we will see, the measure defined on $\us$ will be such that there exists a bijective correspondence between such a measure and states of the quantum system, i.e. $\mu\mapsto\rho_{\mu}$ is a bijection\footnote{It should be notes at this point that the correspondence between measures and state is present also in the context of classical physics. In fact in that case, as we have seen in previous lectures, a pure state (i.e. a point) is identified with the Dirac measure, while a general state is simply a probability measure on the state space. Such a probability measure assigns a number in the interval $[0,1]$ called the weight and that in a sense tells you the `size' of the measurable set.} (for $\rho$ a density matrix).

So the question is : how is this measure defined?  A suitable measure on $\us$ is defined as follows: \\\\
for each density matrix $\rho$ we have
\ba\label{ali:measure}
\mu_{\rho}: Sub_{cl}(\us)&\rightarrow& \Gamma\underline{[0,1]}^{\succeq}\\
\underline{S}=(S_V)_{V\in\mv(\mh)}&\mapsto&\mu_{\rho}(\underline{S}):=(tr(\rho\hat{P}_{\ps{S}_V}))_{V\in\mv(\mh)}
\ea
where $\underline{[0,1]}^{\succeq}$ is the sheaf of order-reversing functions from $\mv(\mh)$ to $[0,1]$. Recall that $\hat{P}_{S_V}=\mathfrak{S}^{-1}\ps{S}_V$ where 
\begin{equation}\label{equ:smap}
\mathfrak{S}:P(V)\rightarrow \Sub_{cl}(\Sig)_V\hspace{.2in}
\end{equation}

The detailed definition of the sheaf $\underline{[0,1]}^{\succeq}$ is as follows
\begin{Definition}
The presheaf $\underline{[0,1]}^{\succeq}:\mv(\mh)\rightarrow Sets$ is defined on:
\begin{enumerate}
\item Objects: For each context $V$ we obtain $\underline{[0,1]}^{\succeq}_V:=\{f:\downarrow V\rightarrow [0,1]|f \text{ is order reversing} \}$ which is a set of order-reversing functions.

\item Morphisms: given $i:V^{'}\subseteq V$ then the corresponding presheaf map is 
\ba
\underline{[0,1]}^{\succeq}(i_{V^{'}V}):\underline{[0,1]}^{\succeq}_V&\rightarrow& \underline{[0,1]}^{\succeq}_{V^{'}}\\
f&\mapsto&f_{|\downarrow V^{'}}
\ea

\end{enumerate}

\end{Definition}

Thus the measure $\mu_{\rho}$ defined in \ref{ali:measure} takes a clopen sub-object of $\us$ and defines an order reversing function 
\be
\mu_{\rho}(\underline{S}):\mv(\mh)\rightarrow [0,1]
\ee
such that for each $V\in \mv(\mh)$, $\mu_{\rho}(\underline{S})(V)$ defined the expectation value, with respect to $\rho$, of the projection operators $\hat{P}_{\ps{S}_V}$ to which the sub-object $\ps{S}$ corresponds to in $V$. Therefore, given two contexts $V^{'}\subseteq V$, since $\hat{P}_{S_{V^{'}}}\geq\hat{P}_{S_V}$ then $tr(\rho\hat{P}_{S_{V^{'}}})\geq tr(\rho\hat{P}_{S_V})$. In detail we have, for each $V\in \mv(\mh)$
\ba
\mu_{\rho}(\underline{S})(V):\ps{1}_V&\rightarrow &\ps{[0,1]}^{\succeq}\\
\{*\}&\mapsto&\mu_{\rho}(\underline{S})(V)(\{*\}):=\mu_{\rho}(\underline{S})(V)
\ea
where $\mu_{\rho}(\underline{S})(V):\downarrow V\rightarrow [0,1]$.

It is worth mentioning some properties of $\mu_{\rho}$
\begin{enumerate}
\item For each $V\in \mv(\mh)$ then $\mu_{\rho}(\ps{0})(V)=tr(\rho\hat{0})=0$ therefore globally
\be
\mu_{\rho}(\ps{0})=(0)_{V\in \mv(\mh)}
\ee
\item  For each $V\in \mv(\mh)$ then $\mu_{\rho}\us(V)=tr(\rho\hat{1})=1$ therefore globally
\be
\mu_{\rho}(\us)=(1)_{V\in \mv(\mh)}
\ee
\item Given two disjoint clopen sub-objects $\ps{S}$ and $\ps{T}$ of $\us$ we then have for each context $V\in \mv(\mh)$ that
\ba\label{ali:additive}
\mu_{\rho}(\ps{S}\vee\ps{T})(V)&=&\mu_{\rho}((\ps{S}\vee\ps{T})_V)\\
&=&\mu_{\rho}(\ps{S}_V\cup\ps{T}_V)\\
&=&tr(\rho(\hat{P}_{\ps{S}_V}\vee\hat{P}_{\ps{T}_V}))\\
&=&tr(\rho(\hat{P}_{\ps{S}_V}+\hat{P}_{\ps{T}_V}))\\
&=&tr(\rho(\hat{P}_{\ps{S}_V})+tr(\rho\hat{P}_{\ps{T}_V}))\\
&=&\mu_{\rho}(\ps{S})(V)+\mu_{\rho}(\ps{T})(V)
\ea
It follows globally that
\be
\mu_{\rho}(\ps{S}\vee\ps{T})=\mu_{\rho}(\ps{S})(V)+\mu_{\rho}(\ps{T})
\ee
This is the property of finite additivity.
\item A generalisation of the above property, i.e. a property of which \ref{ali:additive} is a special case is the following: given two arbitrary clopen sub-objects $\ps{S}$ and $\ps{T}$ of $\us$, for all $V\in \mv(\mh)$ we obtain
\ba
\mu_{\rho}(\ps{S}\vee\ps{T})(V)+\mu_{\rho}(\ps{S}\wedge\ps{T})(V)&=&tr(\rho(\hat{P}_{\ps{S}_V}\vee\hat{P}_{\ps{T}_V}))+tr(\rho(\hat{P}_{\ps{S}_V}\wedge\hat{P}_{\ps{T}_V}))\\
&=&tr(\rho(\hat{P}_{\ps{S}_V}\vee\hat{P}_{\ps{T}_V}+\hat{P}_{\ps{S}_V}\wedge\hat{P}_{\ps{T}_V}))\\
&=&tr(\rho(\hat{P}_{\ps{S}_V}+\hat{P}_{\ps{T}_V}))\\
&=&tr(\rho(\hat{P}_{\ps{S}_V}))+tr(\rho(\hat{P}_{\ps{T}_V}))\\
&=&\mu_{\rho}(\ps{S})(V)+\mu_{\rho}(\ps{T})(V)
\ea
which globally gives
\be
\mu_{\rho}(\ps{S}\vee\ps{T})+\mu_{\rho}(\ps{S}\wedge\ps{T})=\mu_{\rho}(\ps{S})+\mu_{\rho}(\ps{T})
\ee
\item Because the collection of clopen sub-objects of $\us$ forms a Heyting algebra and not a Boolean algebra it follows that
since 
\be
\ps{S}\vee\neg\ps{S}<\us
\ee
then
\be
\mu_{\rho}(\ps{S}\vee\neg\ps{S})<1_{V\in\mv(\mh)}
\ee
\item The closest one can get to $\sigma$-additivity is the following: given a countable infinite family $(\ps{S}_i)_{i\in I}$ for open sub-objects such that for each context $V\in \mv(\mh)$ the clopen subsets $(\ps{S}_i)_V\subseteq \us_V$ (for all $i\in I$) are pairwise disjoint then we have
\be
(\mu_{\rho})\bigvee_{i\in I}\ps{S}_i))(V)=tr(\rho (\bigvee_{i\in I}\hat{P}_{\ps{S}_i}))=tr(\rho(\sum_{i\in I}\hat{P}_{\ps{S}_i}))=\sum_{i\in I}tr(\rho\hat{P}_{\ps{S}_i})=\sum_{i\in I}\mu_{\rho}(\ps{S}_i)(V)
\ee
\end{enumerate}
\subsection{Deriving a State from a Measure}
So far we have defined a particular measure given a quantum state $\rho$. We are now interested in doing the opposite, since in the end we want to show that there is a bijection between the two. Thus we would like to first give an abstract characterisation of a measure with no reference to a state and then show how such a measure can uniquely determine a state $\rho$.

Thus the abstract characterisation of a measure is as follows
\begin{Definition}\label{Def:measure}
A measure $\mu$ on the state space $\us$ is a map
\ba
\mu:Sub_{cl}(\us)\us&\rightarrow&\Gamma\ps{[0,1]^{\geq}}\\
\ps{S}=(\ps{S}_V)V\in\mv(\mh)&\mapsto&(\mu(\ps{S}_V))_{V\in\mv(\mh)}
\ea
such that the following condition hold
\begin{enumerate}
\item $\mu(\us)=1_{V(\mh)}$
\item for all $\ps{S}$ and $\ps{T}$ in $Sub_{cl}(\us)$ then $\mu(\ps{S}\vee\ps{T})+\mu(\ps{S}\wedge\ps{T})=\mu(\ps{S})+\mu(\ps{T})$

\end{enumerate}
\end{Definition}

Given such an abstract definition of measure all properties defined in the previous section follow. We now want to show that each such measure $\mu$ uniquely determines a state $\rho_{\mu}$. Since above we show that each state $\rho$ determined a measure, then combining the two result we end up with a bijective correspondence between the space of states $\rho$ and the space of measures $\mu$ on $\us$.
\begin{Theorem}
Given a measure $\mu$ as defined above then there exist a unique state $\rho$ ``associated" to that measure.
\end{Theorem}

\begin{proof}
Let us define 
\ba
m:P(\mh)&\rightarrow& [0,1]\\
P&\mapsto&m(p)
\ea
such that i) $m(\hat{1})=1$; ii) if $\hat{P}\hat{Q}=0$ (are orthogonal) then $m(\hat{P}\vee\hat{Q})=m(\hat{P}+\hat{Q})=m(\hat{P})+m(\hat{Q})$ such a map is called in the literature a\emph{finitely additive probability measure on the projections of $V$}. We now want to define a unique such finite additive measure given the probability measure $\mu$. To this end we define
\be\label{equ:defmu}
m(\hat{P}):=\mu(\ps{S}_{\hat{P}})(V)=\mu(\ps{S}_{\hat{P}})_V
\ee
such that $m(\hat{1}):=1$. In the above formula $\ps{S}$ is some clopen sub-object $\ps{S}\subseteq \us$ such that for $V$ we have $\mathfrak{S}^{-1}(\ps{S}_V)=\hat{P}$. 

In order for this definition to be well defined it has to be independent on which sub-object of $\us$ is chosen to represent the projection operator. In fact it could be the case that for a contexts $V$ and $V^{'}$ we have that $\ps{S}_V$ and $\ps{T}_{V^{'}}$ both correspond to the same projection operator $\hat{P}=\mathfrak{S}^{-1}(\ps{S}_V)=\mathfrak{S}^{-1}(\ps{T}_{V^{'}})$. If this is the case then we need to show that $\mu(\ps{S}_{\hat{P}})_V=\mu(\ps{T}_{\hat{P}})_{V^{'}}$. To show this we first take the case for which here exist to sub-objects of $\us$ such that they correspond to the same projector at the same context\footnote{In our usual sample of $\Cl^4$ this is the case for $\delta^o(\hat{P}_1)_{V_{\hat{P_2}}}$ and $\delta^o(\hat{P}_3)_{V_{\hat{P_2}}}$} $V$, i.e. $\ps{T}_V=\ps{S}_V$. We then obtain that 
\ba
\mu(\ps{S})(V)+\mu(\ps{T})(V)&\stackrel{\ref{Def:measure}}{=}&\mu(\ps{S}\vee\ps{T})(V)+\mu(\ps{S}\wedge\ps{T})(V)\\
&=&\mu(\ps{S}\vee\ps{T})_V+\mu(\ps{S}\wedge\ps{T})_V\\
&=&\mu(\ps{S}_V\cup\ps{T}_V)+\mu(\ps{S}_V\cap\ps{T}_V)\\
&=&\mu(\ps{S}_V)+\mu(\ps{S}_V)\\
&=&\mu(\ps{S})(V)+\mu(\ps{S})(V)
\ea
Therefore 
\be\label{equ:samecontext}
\mu(\ps{S})(V)=\mu(\ps{T})(V)
\ee
So for the same context we manage to prove the result. What about different contexts? Let us consider the case in which $\ps{S}_V=\ps{T}_{V^{'}}$ and they both correspond to the projection operator $\hat{P}$. This means that $\hat{P}\in V$ and $\hat{P}\in V^{'}$ therefore $\hat{P}\in V\cap V^{'}$. So for example $\delta(\hat{P})_V=\ps{S}_V=\ps{T}_V$. We should not that although $\delta^o(\hat{P})_{V^{'}\cap V}=\hat{P}$ it is not the case that $\mathfrak{S}^{-1}(\ps{S}_{V^{'}\cap V})=\hat{P}$ or $\mathfrak{S}^{-1}(\ps{T}_{V^{'}\cap V})$. Given such situation we obtain
\ba
\mu(\ps{S})(V)&\stackrel{\ref{equ:samecontext}}{=}&\mu(\delta(\hat{P}))(V)\\
&=&\mu(\delta(\hat{P}))(V^{'}\cap V)\\
&=&\mu(\delta(\hat{P}))(V^{'})\\
&\stackrel{\ref{equ:samecontext}}{=}&\mu(\ps{T})(V^{'})
\ea
Hence, the definition in \ref{equ:defmu} is well defined. We next have to show that $m$ is actually a finitely additive probability measure on the projections in $\mh$. To this end consider two orthogonal projection operators $\hat{P}$ and $\hat{Q}$, such that $\mathfrak{S}^{-1}(\ps{T}_V)=\hat{Q}$ and $\mathfrak{S}^{-1}(\ps{S}_V)=\hat{P}$. Then $(\ps{S}\vee\ps{T})_V$ corresponds to the projection operator $\hat{P}\vee\hat{Q}$. We then obtain
\ba
m(\hat{P}\vee\hat{Q})&:=&\mu(\ps{S}\vee\ps{T})(V)\\
&=&\mu(\ps{S})(V)+\mu(\ps{T})(V)+\mu(\ps{S}\wedge\ps{T})(V)\\
&=&\mu(\ps{S})(V)+\mu(\ps{T})(V)\\
&=:&m(\hat{P})+m(\hat{Q})
\ea
Finally it is easy to show that $m(\hat{0})=0$. These results together prove that, given a measure $\mu$ on $\us$ we can uniquely defined a  finitely additive probability measure on the projections in $\mh$. Through Gleason's theorem\footnote{Sucha  theorem shows that a quantum state is uniquely determined by the values it takes on projections. Since clopen sub-objects have components which correspond to projections, Gleason's theorem applies.} it is possible to show that such a probability measure corresponds to a state $\rho_{\mu}$. Thus we have the following chain: $\mu\mapsto m\mapsto\rho_{\mu}$.
\end{proof}

We have so managed to show that there exists a one to one correspondence between states and measures as defined in \ref{Def:measure}. We are now ready to tackle issue number 3):  how to relate such a measure with a generalised version of \ref{equ:measure}.\\
One thing to notice is that the measure in \ref{ali:measure} is defined on all clopen sub-objects of the state space $\us$, not only those deriving from the process of daseinisation, while the truth objects \ref{equ:tobj}, \ref{equ:tobj1} and thus the truth values \ref{equ:tru}, \ref{equ:tru1}, are only defined on those particular sub-objects which are derived from daseinisation. Thus another issues to solve is the following:
\begin{itemize}
\item [4)] Enlarge the truth objects in \ref{equ:tobj}, \ref{equ:tobj1} so as to include all clopen sub-objects of the state space $\us$.\\
We will tackle this issue in the following section.
\end{itemize}
\subsection{New Truth Object}
We now would like to define a truth object which takes into account all clopen sub-objects of the state space $\us$, not only those coming from daseinisation. This is because the measure $\mu$ described in previous sections is defined in all clopen sub-objects of $\us$ which in this context correspond to measurable sub-objects. 

Both the truth objects $^{org}\underline{\mathbb{T}}^{|\psi\rangle}$ and $^{org}\underline{\mathbb{T}}^{\rho}$ are defined in such a way (see previous lectures) such that the global sections $\Gamma(^{org}\underline{\mathbb{T}}^{|\psi\rangle})$ give all the clopen sub-objects of $\us$ coming from daseinisation.\\
We now would like to enlarge $\mathbb{T}^{|\psi\rangle}$ such that the global sections of the enlarged truth object give us all the sub-objects of the state space, i.e. what we are looking for is an object $\underline{\mathbb{T}}$, such that\begin{enumerate}
\item $\Gamma \underline{\mathbb{T}}=Sub_{cl}\us$
\item $^{org}\underline{\mathbb{T}}^{|\psi\rangle}\in \Gamma \underline{\mathbb{T}}$\end{enumerate}
Without going into the details of the proof that can be found in \cite{probabilities}, we will simply state the results of their findings.\\
\subsubsection{Pure State Truth Object}
As mentioned above, we would like to be able to define truth values for all clopen sub-objects of $\us$, not just those coming from daseininsation. Thus the truth object $\mathbb{T}^{|\psi\rangle}$ has to now be defined as a general sub-object of $P_{cl}(\us)$, such that $\Gamma \mathbb{T}^{|\psi\rangle}=Sub_{cl}(\us)$. 

Thus the general truth object as referred to a state $|\psi\rangle$ is the sheaf $\underline{\mathbb{T}}^{|\psi\rangle}$, such that for a given context $V$ we have\footnote{Note that $\us_{|\downarrow V}$ indicates the sheaf $\us$ defined on the lower set $\downarrow V$.} 
\be
\underline{\mathbb{T}}^{|\psi\rangle}_V:=\{\underline{S}\in Sub_{cl}(\us_{|\downarrow V})|\forall V^{'}\subseteq V, |\psi\rangle\langle\psi|\leq \hat{P}_{S_{V^{'}}}\}
\ee
In words, what this truth object tells you is a condition in terms of which any sub-object $\underline{S}$ of the state space $\us$ can be said to be true. In particular, for a sub-object $\underline{S}\subseteq\us$ to be totally true in a given state $|\psi\rangle$ it has to be such that for all $V\in \mv(\mh)$ (i.e. locally) it corresponds to a projection operator $\hat{P}_{S_{V}}\in P(V)$ which corresponds to an available proposition in $V$. This proposition is then required to be true with respect to $|\psi\rangle$, i.e. $\langle\psi|\hat{P}_{S_{V}}|\psi\rangle=1$.\\
The corresponding global elements are
\be
\Gamma(\underline{\mathbb{T}}^{|\psi\rangle})=\{\underline{S}\in Sub_{cl}(\us)|\forall V^{'}\subseteq \mv(\mh), |\psi\rangle\langle\psi|\leq \hat{P}_{S_{V}}\}
\ee
From the above definitions it is easy to see that indeed conditions 1. and 2. above are satisfied, namely $$\Gamma(\underline{\mathbb{T}}^{|\psi\rangle})=Sub_{cl}\us$$ and 
$$^{org}\underline{\mathbb{T}}^{|\psi\rangle}\in \Gamma \underline{\mathbb{T}}^{|\psi\rangle}$$\\
So we have managed to solve problem 4) and define a state dependent truth object for all sub-objects of the state space, not only those derived from daseinisation. This implies that now the expression $[[\underline{S}\in \underline{\mathbb{T}}^{|\psi\rangle}]]$ makes sense and can be evaluated also for $\underline{S}\neq\underline{\delta{P}}$. \\
The truth value in \ref{equ:tru} now becomes
\be
v(A\in \Delta, |\psi\rangle)=[[\underline{\delta(\hat{E}[A\in \Delta])}\in \underline{\mathbb{T}}^{|\psi\rangle}]]
\ee
\subsubsection{Density Matrix Truth Object}
The enlarged truth object as referred to a density matrix $\rho$ is the sheaf $\underline{\mathbb{T}}^{\rho, r}$ such that, for each context $V\in \mv(\mh)$ we have
\be\label{equ:rho}
\underline{\mathbb{T}}^{\rho, r}_V:=\{\underline{S}\in Sub_{cl}(\us_{|\downarrow V})|\forall V^{'}\subseteq V, tr(\rho\hat{P}_{S_{V^{'}}})\geq r\}
\ee
The corresponding global elements are
\be\label{equ:globrho}
\Gamma(\underline{\mathbb{T}}^{\rho, r}):=\{\underline{S}\in Sub_{cl}(\us)|\forall V\subseteq \mv(\mh), tr(\rho\hat{P}_{S_{V}})\geq r\}
\ee
What these global elements tell you is which (general) propositions, represented by sub-objects of the state space are true with probability at least $r$. \\
Moreover, if we have two distinct numbers $r_1\le r_2\leq 1$, then 
$$\underline{\mathbb{T}}^{\rho, r_2}\subseteq\underline{\mathbb{T}}^{\rho, r_1}$$
\begin{proof}
From the definition in equation \ref{equ:rho} $\underline{\mathbb{T}}^{\rho, r_2}$ is a family of clopen sub-objects $\ps{S}\subseteq\us_{\downarrow V}$ 
such that for all $V^{'}\in \downarrow V$ we have $tr(\rho\hat{P}_{S_{V^{'}}})\geq r_2$. However $r_1\le r_2$ therefore $tr(\rho\hat{P}_{S_{V^{'}}})\geq r_1$, for all $V^{'}\in \downarrow V$. It follows that $\underline{\mathbb{T}}^{\rho, r_2}\subseteq\underline{\mathbb{T}}^{\rho, r_1}$
\end{proof}
This means that the collection of sub-objects (propositions) of the state space, which are true with probability at least $r_1$, are more than the collection of sub-objects (propositions) of the state space which are true with a bigger probability (at least $r_2$). \\
Similarly, as above, conditions 1. and 2. are satisfied by the newly defined object \ref{equ:rho}.\\
Therefore, also for a density matrix truth object we have managed to solve problem 4), i.e. we have managed to define a density matrix dependent truth object for all sub-objects of the state space, not only those derived from daseinisation. \\
Thus the expression $[[\underline{S}\in \underline{\mathbb{T}}^{\rho, r}]]$ makes sense and can be evaluated for all $\underline{S}\in Sub(\us)$.\\ 
Equation \ref{equ:measure} now becomes
\be
v(A\in \Delta,\rho)^r(V)=[[\delta(\hat{E}[A\in \Delta])\in \underline{\mathbb{T}}^{\rho,r}]]_V:=\{V^{'}\subseteq V|tr(\rho\hat{E}[A\in \Delta])_{V^{'}}\geq r\}\in \uom^{Sh(\mv(\mh))}_{V}
\ee
\subsection{Generalised Truth Values}
We now try tackling issue number 1), i.e. how to put together the one parameter family of truth values in \ref{equ:measure}, so as to be itself an object in a topos. To this end one needs to extend the topos from $Sh(\mv(\mh)$ to $Sh(\mv(\mh)\times (0,1)_L)$, so that now the stages/contexts are pairs $\langle V, r\rangle$. Such a category $\mv(\mh)\times (0,1)_L$ can be given the structure of a poset as follow:
\be
\langle V^{'}, r^{'}\rangle\leq\langle V, r\rangle\text{ iff } V^{'}\leq V\;,\text{ and }\; r^{'}\leq r
\ee
Being a poset $\mv(\mh)\times (0,1)_L$ is equipped with the (lower) Alexander topology, where the basic open are $(\downarrow V\times (0,r))$
\\
This enlargement makes sense intuitively, since we are trying to combine the topos definition of probabilities outlined in section \ref{sec:prob}, which makes use of the topos $Sh((0,1)_L))$ and the concept of truth values which makes use of the topos $Sh(\mv(\mh))$. In particular, what we are trying to combine is the following:
\begin{itemize}
\item The one parameter family of truth values
\be\label{equ:sieve1}
v(A\in \Delta,\rho)^r(V)=\{V^{'}\subseteq V|tr(\rho\delta(\hat{E}[A\in \Delta])_{V^{'}}\geq r\}\in \uom^{Sh(\mv(\mh))}_{V}
\ee
which gives us sieves for each stage $V$ in $\uom^{Sh(\mv(\mh))}$. 
\item The topos definition of probabilities 
\be\label{equ:sieve2}
[[\underline{S}\in \underline{\mathbb{T}}^{\mu}]](0,r):=\{(0,r^{'})\leq (0,r)|\mu(S)\geq (0,r^{'})\}\in\underline{\Omega}^{Sh(\mathcal{O}((0,1)_L)}
\ee
which gives us, for each context $(0,r)$, a sieve in $\uom^{Sh(\mathcal{O}(0,1)_L)}$.
\end{itemize}
In order to meaningfully combine the two objects above we first need to go back to the the way in which topos probabilities where implemented for classical physics. In this context we ended up with the commuting diagram
\[\xymatrix{
Sub(X)\ar[rr]^{\mu}\ar[dd]^{\Delta}&&[0,1]\ar[dd]^l\\
&&\\
Sub_{Sh(\mathcal{O}((0,1)_L))}(\underline{X})\ar[rr]^{\epsilon^{\mu}}&&\Gamma\uom^{Sh(\mv(\mh)\times \mathcal{O}((0,1)_L))}\\
}\]
We are now looking for a quantum analogue of this. We already have the definition of a measure on the state space. The next step is to define the map of Heyting algebras 
\be
l:\Gamma\ps{[0,1]}^{\geq}\rightarrow\Gamma\uom^{Sh(\mv(\mh)\times \mathcal{O}((0,1)_L))}
\ee
Which is defined as:

\ba\label{ali:l}
l:\Gamma\ps{[0,1]}^{\geq}&\rightarrow&\Gamma\uom^{Sh(\mv(\mh)\times \mathcal{O}((0,1)_L))}\\
\gamma&\mapsto& l(\gamma)
\ea
where $l(\gamma)(\langle V, r\rangle):=\{\langle V^{'}, r^{'}\rangle\leq \langle V, r\rangle|\gamma(V^{'})\geq r^{'}\}$.

$\{\langle V^{'}, r^{'}\rangle\leq \langle V, r\rangle|\gamma(V^{'})\geq r^{'}\}$ is a sieve on $\langle V, r\rangle\in \mathcal{O}(\mv(\mh)\times (0,1)_L)$ since $\gamma$ is a nowhere increasing function. In fact, to each density matrix $\rho$ and to each projection operator $\hat{P}$ there corresponds a global element $\gamma_{\hat{P}, \rho}\in\Gamma\ps{[0,1]}^{\geq}$ defined by
\be
\gamma_{\hat{P}, \rho}(V):=tr(\rho\delta^o(\hat{P})_V)
\ee
When applying the map $l$ to such section we obtain, for each context $\langle V, \rho\rangle$:
\ba
l(\gamma_{\hat{P}, r})(\langle V, \rho\rangle)&=&\{\langle V^{'}, r^{'}\rangle\leq \langle V, r\rangle|\gamma_{\hat{P}, \rho}(V^{'})\geq r^{'}\}\\
&=&\{\langle V^{'}, r^{'}\rangle\leq \langle V, r\rangle|tr(\rho\delta^o(\hat{P})_{V^{'}})\geq r^{'}\}
\ea
If $\hat{P}=\delta([\hat{E}[A\in \Delta])$ then we get
\ba
l(\gamma_{\delta([\hat{E}[A\in \Delta]), \rho})(\langle V, \rho\rangle)&=&\{\langle V^{'}, r^{'}\rangle\leq \langle V, r\rangle|\gamma_{\delta([\hat{E}[A\in \Delta]), \rho}(V^{'})\geq r^{'}\}\\
&=&\{\langle V^{'}, r^{'}\rangle\leq \langle V, r\rangle|tr(\rho\delta([\hat{E}[A\in \Delta])_{V^{'}})\geq r^{'}\}
\ea
An important result is the following
\begin{Theorem}
The map $l$ as defined in \ref{ali:l} separated the elements in $\Gamma\ps{[0,1]}^{\geq}$, i.e. it is injective.
\end{Theorem}

\begin{proof}
Assume that $\gamma_1\neq\gamma_2$. Then there will exist a context $V_i$ such that $\gamma_1(V_i)\neq\gamma_2(V_i)$. Let us assume that $\gamma_1(V_i)>\gamma_2(V_i)$. If we then apply the $l$ map we obtain respectively
\be
l(\gamma_2)(\langle V_i, \gamma_2(V_i)\rangle):=\{\langle V^{'}, r^{'}\rangle\leq \langle V_i, \gamma_2(V_i)\rangle|\gamma_2(V^{'})\geq r^{'}\}=\downarrow \langle V_i, \gamma_2(V_i)\rangle\}
\ee
and
\be
l(\gamma_1)(\langle V_i, \gamma_2(V_i)\rangle):=\{\langle V^{'}, r^{'}\langle\leq \langle V_i, \gamma_2(V_i)\rangle|\gamma_1(V^{'})\geq r^{'}\}\subset \downarrow \langle V_i, \gamma_2(V_i)\rangle\}
\ee
Therefore $l(\gamma_1)\neq l(\gamma_2)$

\end{proof}
Since the map $l$ represents the quantum analogue of the classical map $l:[0,1]\rightarrow \Gamma(\uom^{(0,1)})$ used to define probabilities we would like to check if the quantum $l$ has the same properties, in particular if it preserves joins. To this end consider a family $(\gamma_i)_{i\in I}$ of global section of $\ps{[0,1]}^{\geq}$. We then obtain for each context $( \langle V, r\rangle)$
\ba
l(\bigvee_i\gamma_i)( \langle V, r\rangle)&=&\{ \langle V^{'}, r^{'}\rangle\leq  \langle V, r\rangle|(\bigvee_i\gamma_i)(V^{'})\geq r^{'}\}\\
&=&\{ \langle V^{'}, r^{'}\rangle\leq  \langle V, r\rangle|\sup_i\gamma_i(V^{'})\geq r^{'}\}\\
&=&\bigcup_i\{ \langle V^{'}, r^{'}\rangle\leq  \langle V, r\rangle|\gamma_i(V^{'})\geq r^{'}\}\\
&=&\bigvee_il(\gamma_i)( \langle V, r\rangle)
\ea 
It follows that 
\be
l(\bigvee_i\gamma_i)=\bigvee_il(\gamma_i)
\ee

Now that we have both maps $\mu$ and $l$ we can combine the two sieves defined in \ref{equ:sieve1} and \ref{equ:sieve2} in a unique sieve in $\uom^{Sh(\mv(\mh)\times (0,1)_L)}$ by composing the maps $\mu^{\rho}$ and $l$ as follows
\be
l\circ \mu^{\rho}(\ps{\delta(\hat{P})})\in \Gamma(\uom^{Sh(\mv(\mh)\times (0,1)_L)})
\ee
Therefore, for each context $\langle V, r\rangle$ we obtain
\ba
l\circ \mu^{\rho}(\ps{\delta(\hat{P})})(\langle V, r\rangle)&=&l(\gamma_{\delta(\hat{P}), \rho})(\langle V, r\rangle)\\
&=&\{\langle V^{'}, r^{'}\rangle\leq \langle V, r\rangle|\gamma_{\delta(\hat{P}), \rho}(V^{'})\geq r^{'}\}\\
&=&\{\langle V^{'}, r^{'}\rangle\leq \langle V, r\rangle|tr(\rho\delta(\hat{P})_{V^{'}})\geq r^{'}\}
\ea

%
%
%
%
%
%
In other words, for each context $\langle V, r\rangle$ we obtain the following truth value
\ba\label{ali:newtruth}
v(A\in \Delta,\rho)(\langle V, r\rangle)&:=&\{\langle V^{'}, r^{'}\rangle\leq \langle V, r\rangle|tr(\rho\delta(\hat{E}[A\in \Delta])_{V^{'}}\geq r\}\\
&:=&\{\langle V^{'}, r^{'}\rangle\leq \langle V, r\rangle|\mu^{\rho}(\delta(\hat{E}[A\in \Delta])_{V^{'}})\geq r\}
\ea
Mathematically, this is a well defined element (sieve) of $\uom^{Sh(\mv(\mh)\times (0,1)_L)}$. Moreover, it was shown that such a truth value separates the density matrices and thus separates the measures.\\
However, in order to give physical meaning to the above expression, as representing the truth value of a quantum proposition regarding a physical system, we need to translate all the other objects which were defined in the topos $Sh(\mv(\mh))$, to objects in the topos $Sh(\mv(\mh)\times (0,1)_L)$.
This will be done by defining the quantum analogues of the maps $\Delta$ and $\epsilon$ which we are still missing. In this way we can effectively obtain the quantum analogue of the commuting classical diagram above. In particular, the quantum map $\Delta$ will be identified through the geometric morphisms $p_1^*:Sh(\mh(\mv))\rightarrow Sh(\mv(\mh)\times (0,1)_L)$ induced by the projection map $p_1:\mv(\mh)\times (0,1)_L\rightarrow \mv(\mh)$. Therefore we obtain
\be
p_1^*:Sub_{cl}\us^{Sh(\mv(\mh))}\rightarrow Sub_{cl}(\us^{Sh(\mv(\mh)\times (0,1)_L)})
\ee
Where $\us^{Sh(\mv(\mh))}\in Sh(\mv(\mh))$ while $\us^{Sh(\mv(\mh)\times (0,1)_L)}\in Sh(\mv(\mh)\times (0,1)_L)$

In particular, we will obtain the following:
\begin{itemize}
\item \emph{State space} $$\us^{Sh(\mv(\mh)\times (0,1)_L)}_{\langle V, r\rangle}:=(p_1^*\us^{Sh(\mh(\mv))})_{\langle V, r\rangle}=\us^{Sh(\mh(\mv))}_{p_1(\langle V, r\rangle)}=\us^{Sh(\mh(\mv))}_V$$
\item \emph{Propositions} $$\underline{\delta(\hat{P})}^{Sh(\mv(\mh)\times (0,1)_L)}_{\langle V, r\rangle}:=(p_1^*\underline{\delta(\hat{P})}^{Sh(\mh(\mv))})_{\langle V, r\rangle}=\underline{\delta(\hat{P})}^{Sh(\mh(\mv))}_{p_1(\langle V, r\rangle)}=\underline{\delta(\hat{P})}^{Sh(\mh(\mv))}_V$$
\item \emph{Truth Object} $$ \underline{\mathbb{T}}^{\rho}_{Sh(\mv(\mh)\times (0,1)_L)}(\langle V, r\rangle):= \underline{\mathbb{T}}^{\rho, r}V=\{\underline{S}\in Sub_{cl}(\us_{|\downarrow V})|\forall V^{'}\subseteq V\;,\;tr(\rho\hat{P}_{\underline{S}_{V^{'}}}\geq r\}$$
Equation \ref{ali:newtruth} now becomes
\ba
v(A\in \Delta,\rho)(\langle V, r\rangle)&=&[[p_1^*\underline{\delta(\hat{E}[A\in \Delta])}\in\underline{\mathbb{T}}^{\rho}]]({\langle V, r\rangle})\\
& :=&\{\langle V^{'}, r^{'}\rangle\leq \langle V, r\rangle|tr(\rho\delta(\hat{E}[A\in \Delta])_{V^{'}}\geq r\}\\
&:=&\{\langle V^{'}, r^{'}\rangle\leq\langle V, r\rangle|\mu^{\rho}(\delta(\hat{E}[A\in \Delta])_{V^{'}})\geq r\}
\ea
\end{itemize}

Having done that we can complete our commuting diagram by defining
\ba
\epsilon^{\rho}:Sub_{cl}(\us^{Sh(\mv(\mh)\times (0,1)_L)}&\rightarrow &\Gamma\uom^{Sh(\mv(\mh)\times (0,1)_L)}\\
\ps{S}&\mapsto&[[\ps{S}\in \ps{\mathbb{T}}^{\rho}]]
\ea
Thus the quantum version of the above diagram is
\[\xymatrix{
Sub_{cl}(\us^{Sh(\mv(\mh)})\ar[rr]^{\mu^{\rho}}\ar[dd]^{p_1^*}&&\Gamma\ps{[0,1]}^{\geq}\ar[dd]^l\\
&&\\
Sub_{cl}(\us^{Sh(\mv(\mh)\times (0,1)_L)})\ar[rr]^{\epsilon^{\rho}}&&\Gamma\uom^{Sh(\mv(\mh)\times \mathcal{O}((0,1)_L))}\\
}\]

We then have that 
\be
\epsilon^{\rho}(p_1^*\ps{\delta(\hat{P})})=[[p_1^*\ps{\delta(\hat{P})}\in\ps{\mathbb{T}}^{\rho}]]=l\circ \mu^{\rho}(\ps{\delta(\hat{P})})
\ee
This equation gives the precise statement of the relation between truth values and probability measures in a topos.
Moreover, we know that, given a state $\rho$ then the measure $\mu^{\rho}$ has the $\sigma$-additivity property such that for a cauntable increasing family of clopen sub-objects $(\ps{S}_i)_{i\in N}$ then
\be
\mu^{\rho}(\bigvee_i\ps{S}_i)=\bigvee_i\mu^{\rho}(\ps{S}_i)
\ee
If we then combine this result with the fact that $l$ preserves joins we then get that
\be
(l\circ\mu^{\rho})(\bigvee_i\ps{S}_i)=\bigvee_i(l\circ \mu^{\rho})(\ps{S}_i)
\ee

However from the commutativity of the above diagram it follows that 
\be
(\epsilon^{\rho}\circ p_1^*)(\bigvee_i\ps{S}_i)=(l\circ\mu^{\rho})(\bigvee_i\ps{S}_i)=\bigvee_i(l\circ \mu^{\rho})(\ps{S}_i)
\ee
This represents the logical reformulation of $\sigma$-additivity.

So we have seen that even in the quantum case there is a clear relation between truth values and probability measures. \\Thus in both classical and quantum theory, probabilities can be faithfully expressed in terms of truth values in sheaf topoi with an intuitionistic logic.

\end{document}